\documentclass[10pt]{article}
\pdfoutput=1

\usepackage{amssymb}
\usepackage{amsmath}
\usepackage{amsthm}
\usepackage{mathtools}
\usepackage[usenames,dvipsnames]{xcolor}
\usepackage{epsfig}
\usepackage{dcolumn}
\usepackage{tikz}
\usepackage{tikz-cd}
\usetikzlibrary{shapes.geometric,arrows,svg.path}
\usepackage{upgreek}
\usepackage{setspace}
\usepackage{enumitem}
\usepackage{array,multirow,bigdelim,arydshln}
\usepackage{appendix}
\usepackage[export]{adjustbox}
\usepackage{xparse}
\usepackage[utf8]{inputenc}
\usepackage{microtype}
\usepackage{bm}
\usepackage{braket}
\usepackage{dsfont}
\usepackage{nccmath}
\usepackage{scalerel}

\usepackage{jheppub-mod}

\usepackage{hyperref}
\hypersetup{
	colorlinks=true,
	urlcolor=Maroon,
	linkcolor=Maroon,
	citecolor=Maroon,
	bookmarks=true,
	pdfauthor={Sebastian Mizera},
	pdftitle={Aspects of Scattering Amplitudes and Moduli Space Localization},
	pdfdisplaydoctitle=true,
	pdfstartview=FitH,
	backref=false,
	pagebackref=false
}

\definecolor{orcidlogocol}{HTML}{A6CE39}
\tikzset{
	orcidlogo/.pic={
		\fill[orcidlogocol] svg{M256,128c0,70.7-57.3,128-128,128C57.3,256,0,198.7,0,128C0,57.3,57.3,0,128,0C198.7,0,256,57.3,256,128z};
		\fill[white] svg{M86.3,186.2H70.9V79.1h15.4v48.4V186.2z}
		svg{M108.9,79.1h41.6c39.6,0,57,28.3,57,53.6c0,27.5-21.5,53.6-56.8,53.6h-41.8V79.1z M124.3,172.4h24.5c34.9,0,42.9-26.5,42.9-39.7c0-21.5-13.7-39.7-43.7-39.7h-23.7V172.4z}
		svg{M88.7,56.8c0,5.5-4.5,10.1-10.1,10.1c-5.6,0-10.1-4.6-10.1-10.1c0-5.6,4.5-10.1,10.1-10.1C84.2,46.7,88.7,51.3,88.7,56.8z};
	}
}

\newcommand\orcidicon[1]{\raisebox{.05em}{\href{https://orcid.org/#1}{\mbox{\scalerel*{
				\begin{tikzpicture}[yscale=-1,transform shape]
				\pic{orcidlogo};
				\end{tikzpicture}
			}{|}}}}}

\def \be {\begin{equation}}
\def \ee {\end{equation}}
\def \nn {\nonumber}
\def \la {\langle}
\def \ra {\rangle}
\def \C {\mathbb{C}}
\def \om {\bm{\omega}}
\def \Om {\bm{\Omega}}
\def \O {\mathbf{\Omega}}
\def \F {\mathbf{F}}
\def \I {\mathbb{I}}

\def \A {\mathbf{A}}
\def \B {\mathbf{B}}
\def \C {\mathbb{C}}

\def \SL {\mathrm{SL}}
\def \PSL {\mathrm{SL}}
\def \GL {\mathrm{GL}}

\def \CP {\mathbb{CP}}
\def \Z {\mathbb{Z}}

\def \J {\mathbf{J}}
\def \M {\mathcal{M}}
\def \MM {\mathbf{M}}
\def \L {\mathcal{L}}
\def \a {\mathrm{A}}
\def \b {\mathrm{B}}
\def \c {\mathrm{C}}
\def \d {\mathrm{D}}
\def \KN {\mathrm{KN}}
\def \PT {\mathrm{PT}}
\def \FB {\mathrm{FB}}

\DeclareMathOperator{\im}{im}
\DeclareMathOperator{\Res}{Res}

\DeclareMathOperator{\Conf}{Conf}
\DeclareMathOperator{\ord}{ord}
\DeclareMathOperator{\Pf}{Pf}

\newcommand{\overbar}[1]{\mkern 1.5mu\overline{\mkern-1.5mu#1\mkern-1.5mu}\mkern 1.5mu}

\newcommand{\WidestEntry}{$-1$}%
\newcommand{\SetToWidest}[1]{\makebox[\widthof{\WidestEntry}]{$#1$}}%

\newtheorem{theorem}{Theorem}[section]

\newtheorem{proposition}{Proposition}[section]
\newtheorem{lemma}{Lemma}[section]

\newtheorem{definition}{Definition}[section]

\theoremstyle{definition}
\newtheorem{example}{Example}[section]
\newtheorem*{acknowledgements}{Acknowledgements}
\newtheorem*{outline}{Outline}

\allowdisplaybreaks
\graphicspath{{figures/}}
\renewcommand{\thefigure}{\thesection.\arabic{figure}}

\title{Aspects {\itshape\LARGE of}\hspace{.1em} Scattering Amplitudes\\ {\itshape\LARGE and} Moduli Space Localization}
\author{Sebastian Mizera \orcidicon{0000-0002-8066-5891}}
\affiliation{Perimeter Institute for Theoretical Physics, Waterloo, ON N2L 2Y5, Canada}
\affiliation{Department of Physics \& Astronomy, University of Waterloo, Waterloo, ON N2L 3G1, Canada}
\emailAdd{smizera@pitp.ca}

\abstract{%
We elaborate on the recent proposal that intersection numbers of certain cohomology classes on the moduli space of genus-zero Riemann surfaces with $n$ punctures, $\mathcal{M}_{0,n}$, compute tree-level scattering amplitudes in quantum field theories with a finite spectrum of particles.
The relevant cohomology groups are twisted by representations of the fundamental group $\pi_1(\mathcal{M}_{0,n})$ that describes how punctures braid around each other on the Riemann surface.
Such a structure can be used to link the space of Riemann surfaces with the space of kinematic invariants.
Intersection numbers of said cohomology classes---whose representatives we call \emph{twisted forms}---can be shown to fully localize on the boundaries of $\M_{0,n}$, which are in one-to-one map with Feynman diagrams.
In this work we develop systematic approaches towards accessing such boundary information.
We prove that when twisted forms are logarithmic, their intersection numbers have a simple expansion in terms of trivalent Feynman diagrams weighted by residues, allowing only for massless propagators on the internal and external lines.
It is also known that in the massless limit intersection numbers have a different localization formula on the support of so-called scattering equations.
Nevertheless, for physical applications one also needs to study non-logarithmic forms as they are responsible for propagation of massive states.
We utilize the natural fibre bundle structure of $\M_{0,n}$---which allows for a direct access to the boundaries---to introduce recursion relations for intersection numbers that ``integrate out'' puncture-by-puncture.
The resulting recursion involves only linear algebra of certain matrices describing braiding properties of $\M_{0,n}$ and evaluating one-dimensional residues, thus paving a way for explicit analytic computations of scattering amplitudes.
Together with the previous reformulation of the tree-level S-matrix of string theory in terms of twisted forms, the results of this work complete a unified geometric framework for studying scattering amplitudes from genus-zero Riemann surfaces.
We show that a web of dualities between different homology and cohomology groups allows for deriving a host of identities among various types of amplitudes computed from the moduli space, which in this setup become a consequence of linear algebra.
Throughout this work we emphasize that algebraic computations can be supplemented---or indeed replaced---by combinatorial, geometric, and topological ones.
}

\begin{document}

\maketitle
\setcounter{page}{2}

\newpage
\section{Introduction}

\textsc{It has long been known} that moduli spaces of Riemann surfaces with marked points play an important role in theoretical physics. For instance, interpreting a Riemann surface physically as a real manifold, observables in two-dimensional quantum gravity are computed in terms of the \emph{intersection theory} on such moduli spaces \cite{Witten:1989ig,Witten:1990hr,Kontsevich:1992ti}, see also \cite{Mirzakhani2007,mirzakhani2007weil,Dijkgraaf:2018vnm}. More specifically, \emph{intersection numbers}---given by integrals of certain cohomology classes over compactified moduli spaces---have a physical interpretation as computing correlation functions in this theory. It is natural to ask whether observables in other theories can admit a similar description. Clearly, if this were to be the case, one needs to endow moduli spaces with an additional structure that enhances the purely topological character of intersection theory.

The natural candidate for such a structure is the so-called \emph{local system}, which allows for a connection between the space of Riemann surfaces and the space of kinematic invariants. This turns out to be rigid enough to allow for  essentially only two ways of constructing intersection theory of cohomology classes ``twisted'' by a local system, at least at genus zero. The first of them is known to compute tree-level scattering amplitudes in string theory \cite{Mizera:2017cqs}, while the second---previously missed---way appears to give rise to amplitudes for point-particle scattering \cite{Mizera:2017rqa}. In this work we elaborate upon this proposal.

A local system is a representation of the fundamental group of a given space, which packages together the information about its monodromy properties. In order to define integration on such spaces (as opposed to their covering) one needs to introduce the notion of homology and cohomology groups with coefficients in a local system. The study of such objects is an old and rich topic of research dating back to the 1930s \cite{reidemeister1935uberdeckungen,10.2307/1968884,10.2307/1969099,leray1946anneau,leray1946structure}, though most results of our interest are considerably more recent \cite{aomoto1975vanishing,aomoto1977structure,aomoto1987gauss,10.1093/qmath/38.4.385,1983113,zbMATH03996010,MANA:MANA19941660122,cho1995,matsumoto1998,aomoto2011theory}. Since our goal is to find practical applications of this theory, we will often simplify its exposition, in the process of which we hope not to lose sight of the inherent beauty of the subject.

Throughout this work we focus on the moduli space $\M_{0,n}$ of genus-zero Riemann surfaces with $n$ punctures. The fundamental group $\pi_1(\M_{0,n})$ is its only non-trivial homotopy group. It is generated by the loops $\circlearrowleft_{ij}$ in which a puncture $z_j$ encircles another $z_i$ and goes back to its original position. Equivalently $\pi_1(\M_{0,n})$ can be described as a braid group of $n$ strands out of which three are fixed due to the $\SL(2,\C)$ redundancy on the Riemann sphere $\Sigma$. In order to construct a local system on $\M_{0,n}$ we introduce the following one-form:
\be
\omega = \frac{1}{\Lambda^2} \sum_{1 \leq i<j \leq n} \!\!\!\! 2 p_i {\cdot} p_j \, d\log(z_i - z_j),
\ee
where $z_i$'s denote inhomogeneous coordinates on each $\Sigma \cong \CP^1$. To each $z_i$ we associate an ingoing momentum $p_i^\mu$ and impose momentum conservation $\sum_i p_i^\mu =0$ in an arbitrary-dimensional Minkowski space. In order to keep $\omega$ dimensionless it is normalized by a mass-scale parameter $\Lambda$. One can check that the requirement of the absence of spurious poles at infinity implies a quantization condition on the masses of each particle, $m_i^2 = -p_i^2$, given by $m_i^2 \in \Lambda^2 \Z /2$. The above one-form gives rise to one of the simplest local systems, which associates a non-zero complex number $\exp \int_\gamma \omega$ to a given path $\gamma$. For example, each loop $\circlearrowleft_{ij}$ is represented by a coefficient $\exp(4\pi i  p_i {\cdot} p_j / \Lambda^2)$.

Let us focus on the question how the above topological properties of $\M_{0,n}$ affect the algebraic structure of differential forms. One can talk about differential forms with coefficients in a local system. For our purposes, however, it is sufficient to work with their infinitesimal version, which is governed by the Knizhnik--Zamolodchikov-like connections:
\be
\nabla_{\pm\omega} = d \pm \omega\wedge,
\ee
labelled by a $\pm$ sign. Because $\omega$ is a closed one-form, both $\nabla_{\pm\omega}$ are integrable connections, meaning that they square to zero. This naturally gives rise to the notion of \emph{twisted forms} $\varphi_\pm$, which are differential forms that are $\nabla_{\pm\omega}$-closed modulo the $\nabla_{\pm\omega}$-exact ones. They belong to the cohomology classes:
\be\label{intro-cohomology}
\varphi_{\pm} \;\in\; \frac{ \nabla_{\pm\omega}\text{-closed forms}}{\nabla_{\pm\omega}\text{-exact forms}}.
\ee
This is nothing but a version of the standard de Rham cohomology twisted by the one-form $\omega$. It is known that $\varphi_\pm$ belonging to the above cohomology classes can only be $(n{-}3)$-forms \cite{aomoto1975vanishing}, which coincides with the complex dimension of $\M_{0,n}$. From now on we will take both $\varphi_\pm$ to be holomorphic.

The requirement of $\SL(2,\C)$-invariance imposes certain transformation properties on the twisted forms. To be concrete, one finds that $\varphi_\pm$ need to transform with M\"obius weights $\pm2m_i^2 / \Lambda^2$ in each puncture $z_i$. Thus their choice is constrained by the masses $m_i$ of the external particles. For instance, when all masses are equal to $m_i = \Lambda$ examples of allowed twisted forms are:
\be\label{intro-example-forms}
\varphi_- = d\mu_n, \qquad \varphi_+ = \frac{d\mu_n}{ (z_{12} z_{23} \cdots z_{n1})^2 },
\ee
where $z_{ij} = z_i {-} z_j$ and the measure equals to $d\mu_n = z_{jk} z_{j\ell} z_{k\ell} \bigwedge_{i\neq j,k,\ell} dz_i$, where $(z_j, z_k, z_\ell)$ are three arbitrary punctures fixed using the action of the $\SL(2,\C)$ redundancy.

The most natural proposal for extracting observables from twisted forms is to integrate them over the moduli space as follows:
\be
\int_{\M_{0,n}} \!\! \varphi_- \wedge \varphi_+.
\ee
Since the wedge product $\varphi_- \wedge \varphi_+$ vanishes for two holomorphic top-forms, it is clear that the integral receives no contributions from the bulk of $\M_{0,n}$. However, as written above the integral is actually not well-defined near the boundaries $\partial \M_{0,n}$, where $\varphi_\pm$ could diverge (it is a consequence of the fact that $\M_{0,n}$ is non-compact). This issue gives rise to a ``$0/0$ problem'' near $\partial \M_{0,n}$. It can be remedied by making the integrand compactly-supported, for example by constructing a form $\varphi_+^c$ which vanishes in the neighbourhood of $\partial \M_{0,n}$. In order for the integral to remain a bilinear between the two cohomology classes, $\varphi_+^c$ needs to be cohomologous to $\varphi_+$, i.e., $\varphi_+^c = \varphi_+ + \nabla_\omega \xi$ for some $\xi$ (such a choice always exists). At the same time as regulating the integration, $\xi$ contains non-holomorphic contributions, which make the integral non-zero, while $\nabla_\omega$ gives it dependence on kinematics. This defines the \emph{intersection number} $\braket{\varphi_- | \varphi_+}_\omega$ of twisted forms $\varphi_\pm$:
\be\label{intro-intersection-number}
\braket{\varphi_- | \varphi_+}_\omega = \frac{1}{(-2\pi i \Lambda^2)^{n-3}} \int_{\M_{0,n}} \!\!\varphi_- \wedge \varphi_+^c,
\ee
which we normalized for later convenience. By the above arguments, intersection numbers necessarily localize near the boundary $\partial \M_{0,n}$ of the moduli space.

Recall that codimension-one boundaries of the moduli space $\M_{0,n}$ correspond to configurations in which two or more punctures coalesce on the Riemann sphere. Equivalently, one can think of the surface stretching into an infinitely-long tube with a subset of punctures on one side of the throat and the complementary set on the other. The target-space interpretation of the throat is a propagator stretching between the two sets of particles associated to each puncture. Higher-codimension boundaries are obtained by developing more throats, and in particular the maximal-codimension ones (vertices of $\partial \M_{0,n}$) correspond to trivalent trees with the total of $n{-}3$ propagators, e.g.,
\be
\begin{aligned}
	\includegraphics[scale=1]{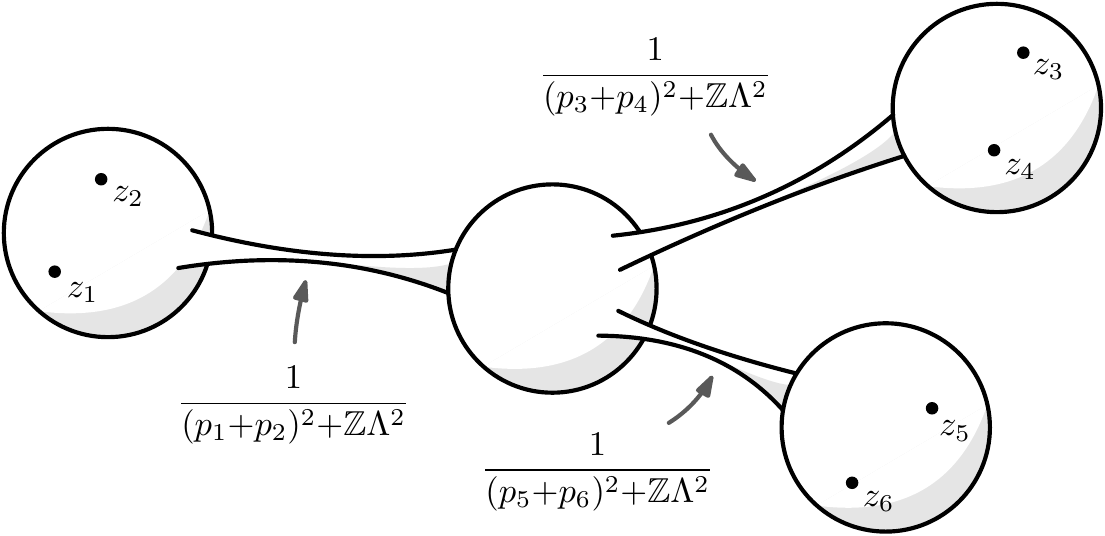}
\end{aligned}
\ee
In the limit when the throat becomes infinitely-long, one can also think of the Riemann surface as bubbling into many spheres touching each other through emergent punctures \cite{PMIHES_1969__36__75_0}, which have the interpretation of particles propagating between the individual spheres.

One can show that the intersection numbers \eqref{intro-intersection-number} are always rational functions of kinematic invariants with simple poles of the form $(p_{i} {+} p_{j} {+}\ldots)^2 + \Z \Lambda^2$, which leads us to propose that they compute tree-level scattering amplitudes in quantum field theories. For example, using the pair of twisted forms from \eqref{intro-example-forms} we find:
\be
\Braket{ d\mu_n | \frac{d\mu_n}{ (z_{12} z_{23} \cdots z_{n1})^2 }}_{\!\omega} = \sum_{\substack{\text{trivalent}\\ \text{planar trees }{\!\cal T}}} \prod_{\substack{\text{internal}\\ \text{edges }e\in{\cal T}}} \,\frac{1}{ p_e^2 + \Lambda^2},
\ee
where $p_e^\mu$ is the momentum flowing through the edge $e$. These are planar amplitudes in the massive $\text{Tr}(\phi^3)$ theory. Of course, amplitudes in the standard $\phi^3$ theory can be obtained by symmetrizing the twisted form $\varphi_+$ from \eqref{intro-example-forms}.

More generally, kinematic singularities of intersection numbers are governed by the behaviour of the twisted forms near the boundaries of the moduli space. For example, a pole in the $(p_i {+} p_j {+} \ldots)^2$-channel can only appear if $\varphi_\pm$ combined have a double pole on the corresponding boundary of $\M_{0,n}$ (when punctures $z_i, z_j, \ldots$ coalesce). Higher-degree poles in $\varphi_+$ shift the kinematic poles towards higher-mass exchanges, and likewise increasing the poles of $\varphi_-$ gives rise to more-tachyonic states. We give examples of intersection numbers that illustrate these rules throughout the paper.

Given the special role of logarithmic forms in the theory of complex manifolds \cite{deligne1970equations}, it is natural to look for simplifications of intersection numbers in those cases. Indeed, we show that logarithmic twisted forms $\varphi_\pm$ can only give rise to scattering amplitudes with massless internal and external states. We make this statement precise by proving that in the logarithmic cases $\la \varphi_- | \varphi_+ \ra_{\omega}$ evaluates to (Theorem~\ref{theorem-21}):
\be
\la \varphi_- | \varphi_+ \ra_{\omega} = \sum_{v \in \partial \M_{0,n}} \!\!\frac{\Res_{v}(\varphi_-) \Res_{v}(\varphi_+)}{\prod_{e \in {\cal T}} p^2_{e}},
\ee
where the sum goes over all $(2n{-}5)!!$ vertices of the boundary $\partial \M_{0,n}$, which are in one-to-one map with trivalent trees $\cal T$ with inverse propagators $p_{e}^2$. Each numerator in the sum factors into two residues around the corresponding vertex, which can be evaluated by a change of variables from $z_i$'s into $\SL(2,\C)$-invariant cross-ratios \cite{Koba:1969kh,Brown:2009qja}. What is more, the above formula shows that intersection numbers become $\Lambda$-independent. This fact can be used to prove that in the logarithmic cases they compute the strict low-energy limit of open- and closed-string amplitudes (Theorem~\ref{theorem-22}). In the non-logarithmic cases one finds a weaker result that the low-energy limits of string amplitudes and intersection numbers coincide.

As a matter of fact, much of the physical intuition about intersection numbers an be obtained by studying their relations to string theory. Moreover, string theory correlation functions can be treated as a factory of twisted forms, which after plugging into intersection numbers compute scattering amplitudes in field theory. For example, we can use the textbook form of the correlation function of massless gauge bosons (modulo the plane-wave contractions) \cite{green1988superstring} and treat it as a twisted form,
\be\label{intro-varphi-gauge}
\varphi^{\text{gauge}}_{\pm,n} = d\mu_n \int \prod_{i=1}^{n} d\theta_i d\tilde{\theta}_i\, \frac{\theta_k \theta_\ell}{z_k {-} z_\ell} \exp\! \left( - \sum_{i \neq j}\frac{\theta_i \theta_j p_i {\cdot} p_j + \tilde{\theta}_i \tilde{\theta}_j \varepsilon_i {\cdot} \varepsilon_j + 2(\theta_i {-} \theta_j) \tilde{\theta}_i \varepsilon_i {\cdot} p_j }{z_i {-} z_j \mp \Lambda^2 \theta_i \theta_j}
\right).
\ee
Here $\varepsilon_i^\mu$ is the polarization vector of the $i$-th boson. The expression is written as an integral over Grassmann variables $\theta_i, \tilde\theta_i$ and is independent of the specific choice of the special labels $k,\ell$. The dependence on the mass scale $\Lambda$ is auxiliary and drops out from the final expression. One can show that after plugging \eqref{intro-varphi-gauge} into the intersection pairing
\be
\Braket{\,\frac{\mathrm{Tr}(T^{c_{1}} T^{c_{2}} \cdots T^{c_{n}} )\, d\mu_n}{z_{12}\, z_{23}\cdots z_{n1}} + \cdots\, | \varphi^{\text{gauge}}_{+,n}}_{\!\!\omega} = {\cal A}^{\text{YM}}_n,
\ee
we obtain tree-level scattering amplitudes of gluons in the Yang--Mills theory. As the other twisted form we used a cyclic combination of the poles $1/(z_i {-} z_{i+1})$ weighted by a $\text{U}(N)$ colour trace, where $c_i$ denotes a colour of the $i$-th particle. The ellipsis stands for a sum over all inequivalent permutations of the labels. Similarly, we find that using another copy of \eqref{intro-varphi-gauge} one obtains amplitudes in Einstein gravity:
\be
\Big< \phantom{\frac{1}{1}}\!\!\! \widetilde{\varphi}^{\text{gauge}}_{-,n}\, \Big|\, \varphi^{\text{gauge}}_{+,n}\, \Big>_{\!\omega} = {\cal A}^{\text{GR}}_n.
\ee
We stress that even though the computations are performed on a worldsheet, they result in field-theory scattering amplitudes without any stringy corrections!

Let us also comment on the massless limit, $\Lambda \to 0$, of intersection numbers of more general twisted forms. Since $\omega$ behaves as ${\sim}\Lambda^{-1}$, we have $\nabla_{\pm\omega} \to \pm\omega\wedge$ and hence the two cohomology classes \eqref{intro-cohomology} for holomorphic forms $\varphi_\pm$ degenerate into a single one,
\be\label{intro-cohomology-massless}
\varphi_\pm \;\in\; \frac{\{\text{holomorphic }(n{-}3)\text{-forms on }\M_{0,n}\}}{\omega \wedge \{\text{holomorphic }(n{-}4)\text{-forms on }\M_{0,n}\}}.
\ee
The numerator simplifies since all holomorphic $(n{-}3)$-forms are $\omega \wedge $-closed. The space \eqref{intro-cohomology-massless} is therefore the space of $(n{-}3)$-forms on the support of the constraint $\omega=0$, often referred to as the \emph{scattering equations} \cite{Cachazo:2013gna}. This suggests that there should exist an alternative formula for intersection numbers localizing on $\omega=0$ (a finite number of points in the bulk of $\M_{0,n}$) instead of the boundary $\partial\M_{0,n}$. In fact, using Morse theory one can prove that $\la \varphi_- | \varphi_+ \ra_\omega$ in the massless limit $\Lambda\to 0$ takes a simple form \cite{Mizera:2017rqa}:
\be\label{intro-scattering-equations-localization}
\lim_{\Lambda\to 0}\la \varphi_- | \varphi_+ \ra_\omega = \Res_{\omega=0} \left( \frac{\lim_{\Lambda \to 0} \varphi_-\, \widehat{\varphi}_+}{\prod_{i\neq j,k,\ell} \omega_i} \right),
\ee
where $\omega = \Lambda^{-2} \sum_{i\neq j,k,\ell} \omega_i dz_i$ and $\widehat{\varphi}_+$ is $\varphi_+$ stripped from the differential $\bigwedge_{i\neq j,k,\ell}dz_i$. The result coincides with a localization formula discovered by Cachazo, He, and Yuan \cite{Cachazo:2013hca} and is known to compute scattering amplitudes in various massless quantum field theories. For instance, sending $\Lambda\to0$ in \eqref{intro-varphi-gauge} gives a scattering equations representation of Yang--Mills and gravity amplitudes in terms of a Pfaffian of a certain matrix. The right-hand side of the formula \eqref{intro-scattering-equations-localization} has a worldsheet interpretation as a correlation function of string theory in the ambitwistor space \cite{Mason:2013sva}. We show that the localization on $\omega=0$ happens only in the massless case. Nevertheless, it does not facilitate explicit computations, since by the Abel's impossibility theorem algebraic solutions of the constraints $\omega=0$ do not exist for $n>5$.

The advantage of working in the cohomological formulation is that it allows us to understand many aspects of scattering amplitudes geometrically and topologically. For instance, once can ask about the dimensions of the twisted cohomology groups \eqref{intro-cohomology}, or equivalently how many linearly-independent twisted forms exist for a given $n$. Since only $(n{-}3)$-forms can be twisted forms (meaning that all other twisted cohomology groups vanish \cite{aomoto1975vanishing}), this question is easily answered by $(-1)^{n-3} \chi(\M_{0,n})$, where $\chi(\M_{0,n})$ is the Euler characteristic of the moduli space.

Let us illustrate how to perform a computation of $\chi(\M_{0,n})$ in the simplest way. We consider the following fibration of the moduli space $\M_{0,n}$, which splits into a ``product'' of $n{-}3$ one-dimensional spaces:
\be
\begin{aligned}
	\includegraphics[scale=1]{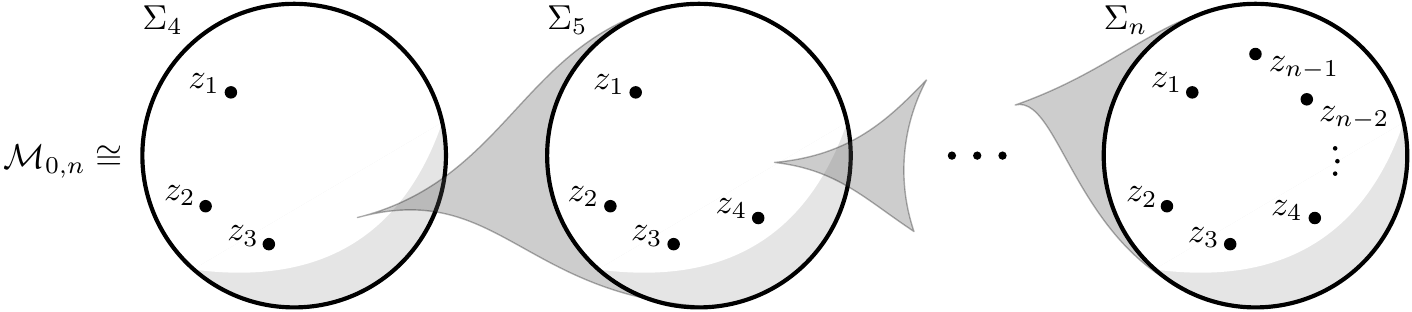}\label{fig-intro}\!
\end{aligned}
\ee
Here we used the $\SL(2,\C)$ redundancy to fix $(z_1, z_2, z_3)$ in some arbitrary positions. The first integrated puncture $z_4$ lives on $\Sigma_4 \cong \CP^1 {\setminus} \{z_1, z_2, z_3\}$. For each point on $\Sigma_4$ we have another punctured Riemann sphere $\Sigma_5 \cong \CP^1 {\setminus} \{z_1,z_2,z_3,z_4\}$ on which $z_5$ is defined. The procedure is repeated until we reach the space of $z_n$, $\Sigma_n \cong \CP^1 {\setminus} \{z_1, z_2, \ldots, z_{n-1}\}$. The order in which particles are being ``peeled out'' is arbitrary and will only affect intermediate steps of our computations, but not the final result.

The Euler characteristic of each surface $\Sigma_p$ is equal to that of a sphere (two) minus one for each removed puncture, i.e., $\chi(\Sigma_p) = 3{-}p$. Since the Euler characteristic of the whole space is a product of those of each fibre, we find
\be
\chi({\cal M}_{0,n}) = \chi(\Sigma_4) \chi(\Sigma_5) \cdots \chi(\Sigma_n) = (-1)^{n-3}(n{-}3)!,
\ee
and hence there are $(n{-}3)!$ linearly-independent twisted forms \cite{aomoto1987gauss}, as first shown by Aomoto.\footnote{The results of Aomoto regarding the basis of twisted homology and cohomology groups \cite{aomoto1975vanishing,aomoto1977structure,aomoto1987gauss,10.1093/qmath/38.4.385} were not widely known in the physics literature at the time and many of their consequences were later found independently by various authors, see, e.g., \cite{Bern:2008qj,BjerrumBohr:2009rd,Stieberger:2009hq,Mafra:2011nv,Mafra:2011nw,Huang:2016tag}.}

Notice that the above fibration also gives a direct access to the boundaries of the moduli space, which are now simply neighbourhoods of the removed points. Let us make use of this simplification. First notice that by forgetting about the Riemann spheres one-by-one from the right in \eqref{fig-intro} we obtain the sequence of lower-puncture moduli spaces:
\be
\M_{0,n} \,\to\, \M_{0,n-1} \,\to\, \ldots \,\to\, \M_{0,4} \,\to\, \M_{0,3},
\ee
which stops at $\M_{0,3}$ equal to a single point. Following this idea we can ``integrate out'' puncture-by-puncture and obtain a local system on each $\M_{0,p}$ starting from that on $\M_{0,n}$. A new feature is the appearance of higher-rank local systems, i.e., those defined by \emph{matrix-valued} one-forms on $\M_{0,p}$,
\be
\om_p = \frac{1}{\Lambda^2} \sum_{1 \leq i<j \leq p} \!\!\!\! \O^{ij}_p\, d\log(z_i - z_j).
\ee
Here each $\om_p$ is an $(n{-}3)!/(p{-}3)! \times (n{-}3)!/(p{-}3)!$ matrix. They give rise to representations of $\pi_1(\M_{0,p})$ given by path-ordered exponentials ${\cal P}\exp \int_\gamma \om_p$ for any $\gamma$. The matrices $\Om_p^{ij}$ contain all the information about how different fibres of $\M_{0,p}$ braid around each other, and hence we call them \emph{braid matrices}. We show that all the cases can be mapped to those with massless external particles, which allows us to focus on the problem with $p_i^2 =0$.

Integrability of the connections constructed out of each $\om_p$ imposes some conditions on the braid matrices, such as the infinitesimal form of the Yang--Baxter relations \cite{AIF_1987__37_4_139_0}. They are however not enough fix the form of $\Om_p^{ij}$ completely. We exploit this freedom to make some simplifying choices of bases on each fibre and show that all $\Om_{p-1}^{ij}$ can be determined recursively from $\Om_{p}^{k\ell}$ starting with the initial condition $\Om_{n}^{ij} = (p_i {+} p_j)^2$ (Lemma~\ref{lemma-3-3}). We find that each braid matrix is a polynomial in the kinematic invariants and their eigenvalues coincide with allowed factorization channels. Therefore they can be thought of as coarse-grained versions of $(p_i {+} p_j)^2$.

Our idea for computing intersection numbers is to take care of one puncture at a time, such that starting with twisted forms $\bm\varphi^\pm_n = \varphi_\pm$ on $\M_{0,n}$ we follow the sequence:
\be
\bm\varphi^{\pm}_n \;\mapsto\; \bm\varphi^{\pm}_{n-1} \;\mapsto\; \ldots \;\mapsto\; \bm\varphi^{\pm}_4 \;\mapsto\; \bm\varphi^{\pm}_3.
\ee
Here each $\bm\varphi^{\pm}_p$ is a vector-valued $(p{-}3)$-form on $\M_{0,p}$ of length $(n{-}3)!/(p{-}3)!$. In particular, $\bm\varphi_3^\pm$ are vectors of functions of length $(n{-}3)!$ that by $\SL(2,\C)$-invariance have to be rational functions of only physical data with all punctures integrated out. We make the above mapping concrete by giving its explicit form:
\be\label{intro-twisted-form-map}
\bm\varphi^\pm_{p-1} = \sum_{q=1}^{p-1} \Res_{z_p = z_q}\!
\begin{pmatrix*}[c]
	\MM_{pq3}^\pm\, \bm\varphi_p^{\pm}\\
	\MM_{pq4}^\pm\, \bm\varphi_p^{\pm} \\
	\vdots\\
	\MM_{p q,p-1}^\pm\,\bm\varphi_p^{\pm} 
\end{pmatrix*},
\ee
where the matrices $\MM_{pqr}^\pm$ have closed-form expression in terms of the braid matrices $\Om_{p}^{ij}$ (Lemmata~\ref{lemma-3-5}--\ref{lemma-3-6}). Each $\MM_{pqr}^\pm$ describes the effect of the puncture $z_p$ approaching $z_q$ on the $r$-th entry (indexed by $r=3,4,\ldots,p{-}1$) of the vector $\bm\varphi_{p-1}^\pm$. Notice that only a few leading orders in the expansion around $z_p {=} z_r$ are needed, depending on the degree of the pole of $\bm\varphi_p^\pm$, as the result enters a residue computation.

\pagebreak
Finally, we show that intersection numbers can be computed simply by a scalar product of the vectors $\bm\varphi_3^\pm$ (Theorem~\ref{theorem-31}):
\be\label{intro-intersection-final}
\braket{ \varphi_- | \varphi_+ }_\omega = \bm\varphi^-_3 \cdot \bm\varphi^+_3.
\ee
Together with the map \eqref{intro-twisted-form-map} this gives recursion relations for intersection numbers on the fibred moduli space. Thus the evaluation of intersection numbers requires only matrix algebra and computing one-dimensional residues. Moreover, we design the above recursion relations such that for planar amplitudes only $1$ out of the $(n{-}3)!$ entries of $\bm\varphi_{3}^\pm$ contributes to the final result \eqref{intro-intersection-final}. They become our main computational tool.

The interpretation of physical observables as pairings between vectors spaces extends beyond the intersection numbers $\braket{\varphi_- | \varphi_+}_\omega$. It was recently understood that also scattering amplitudes in string theory can be written as bilinears of different homology and cohomology groups on the moduli space \cite{Mizera:2017cqs}. This allows us to trivialize the problem of finding relations between scattering amplitudes and other objects defined on $\M_{0,n}$, which can now be understood as a consequence of dualities (existence of pairings) between various isomorphic vector spaces (Proposition~\ref{proposition}). For example, the Kawai--Lewellen--Tye relations \cite{Kawai:1985xq} or basis expansion identities become linear algebra statements \cite{Mizera:2017cqs,Mizera:2017rqa}. Aspects of homology with coefficients in the local system are reviewed in Appendix~\ref{app:aspects}.

\pagebreak

\begin{outline}
	In Section~\ref{sec:intersection-numbers} we introduce the main objects of this work and study their properties. We start by reviewing basic aspects of the moduli space $\M_{0,n}$ in Section~\ref{sec:basics}. In Section~\ref{sec:local-systems} we introduce the notions of local systems and twisted differential forms, followed by the definition of intersection numbers in Section~\ref{sec:definition}. After giving a few simple examples we proceed by studying the relations to other objects, such as string theory integrals in Section~\ref{sec:relations}. In Section~\ref{sec:logarithmic} we turn our focus to logarithmic forms and prove a simple formula for their intersection numbers. We close by commenting on certain shift relations between local systems, which allow us to map computations between massive and massless external states in Section~\ref{sec:shift-relations}.
	
	In Section~\ref{sec:recursion-relations} we discuss recursion relations, starting with a review of the fibre bundle structure of $\M_{0,n}$ in Section~\ref{sec:fibration}. We give the first version of recursion relations in Section~\ref{sec:recursion-relations-first}, before entirely solving for twisted cohomology groups on all fibres in Section~\ref{sec:solving} by introducing a natural set of orthonormal bases on each fibre. This allows us to present the final form of recursion relations of twisted forms in Section~\ref{sec:recursion-relations-reprise}.
	
	In Section~\ref{sec:examples} we give further examples of intersection numbers. We begin with a discussion of gauge and gravity scattering amplitudes in Section~\ref{sec:Yang-Mills-and-gravity}. In Section~\ref{sec:Kac-Moody-currents} we exemplify the use of recursion relations by computing intersection numbers of Kac--Moody correlators, which can serve as building blocks for more complicated twisted forms.
	
	We conclude the main body of the paper in Section~\ref{sec:conclusion}, where we summarize the results and comment on outstanding issues.
	
	To this paper we attach an extensive Appendix~\ref{app:aspects}, in which we give a pedagogical review of various aspects of homology with local coefficients and their applications in string theory. We start with the more familiar case of compact Riemann surfaces in Appendix~\ref{app:compact-Riemann-surfaces}, before discussing non-compact surfaces obtained by adding punctures and in particular the case $\M_{0,4}$ in Appendix~\ref{app:punctures}. We finish with a brief summary of results for higher-$n$ moduli spaces in Appendix~\ref{app:generalization}.
\end{outline}

\vfill
\begin{acknowledgements}\normalfont
	We thank N.~Arkani-Hamed, F.~Cachazo, S.~Caron-Huot, E.~Casali, E.~Delabaere, N.~Early, H.~Frellesvig, D.~Fuchs, F.~Gasparotto, A.~Guevara, S.~He, C.J.~Howls, S.~Laporta, A.~Maloney, M.K.~Mandal, P.~Mastrolia, L.~Mattiazzi, I.~Pesando, A.~Pokraka, O.~Schlotterer, A.~Schwarz, P.~Tourkine, and E.~Witten for useful discussions at various stages of this work. This research was supported in part by Perimeter Institute for Theoretical Physics. Research at Perimeter Institute is supported by the Government of Canada through the Department of Innovation, Science and Economic Development Canada and by the Province of Ontario through the Ministry of Research, Innovation and Science.
\end{acknowledgements}

\pagebreak
\section[Intersection Numbers of Twisted Differential Forms]{\label{sec:intersection-numbers}Intersection Numbers of Twisted Differential Forms}

\textsc{In this section} we give the definition of intersection numbers of twisted differential forms and study their general properties. In order to emphasize how much about these objects can be learned by general arguments, we postpone the bulk of explicit computations until later sections, and here focus mostly on the relations to geometry and string theory amplitudes. Before doing so, we start by reviewing topological aspects of moduli spaces on which these twisted differential forms live.\footnote{Note that in several instances we will depart from the notational conventions used in \cite{Mizera:2017cqs,Mizera:2017rqa}.}

\subsection{\label{sec:basics}Basics of Moduli Spaces of Punctured Spheres}

\textsc{Let us consider} the configuration space of $n \geq 3$ punctures (marked points) on a genus-zero Riemann surface $\Sigma \cong \mathbb{CP}^1$,
\be
{\cal M}_{0,n} := \Conf_n(\mathbb{CP}^1) / \PSL(2,\C),
\ee
where the action of the automorphism group $\text{SL}(2,\C)$ allows us to fix positions of three punctures, by convention $(z_1,z_{n-1},z_n)$, where $z_i$ denotes the position of the $i$-th puncture. Hence the complex dimension of $\M_{0,n}$ is $n{-}3$. More explicitly we can write $\M_{0,n}$ as
\be\label{M0n}
{\cal M}_{0,n} = \{ (z_2, z_3, \ldots, z_{n-2}) \in (\mathbb{CP}^1)^{n-3} \;|\; z_i \neq z_j \text{ for all } i \neq j \}.
\ee
We will refer to it as the \emph{moduli space} of $n$-punctured Riemann spheres. From now on we identify $z_i$'s with inhomogeneous coordinates on each $\CP^1$. One can show that $\M_{0,n}$ is path-connected and aspherical (also called the Eilenberg--MacLane space of type $K(\pi, 1)$), i.e., the fundamental group $\pi_1({\cal M}_{0,n})$ is the only non-trivial homotopy group of ${\cal M}_{0,n}$.

The fundamental group is generated by the loops $\circlearrowleft_{ij}$ given by the $j$-th puncture going around the $i$-th puncture in the counter-clockwise direction and returning to its original position without encircling any other puncture. Clearly they are symmetric, $\circlearrowleft_{ij} = \circlearrowleft_{ji}$, and the inverse of each element $\circlearrowleft_{ij}^{-1}$ corresponds to a loop in the clockwise direction. We define $\circlearrowleft_{ii}$ to be the identity element. In addition, since any $\circlearrowleft_{ij}$ can be deformed into a composition of $n{-}1$ other loops $\circlearrowleft_{kj}^{-1}$ for $k {\neq} i$ and $j$ fixed, the dimension of $\pi_1(\M_{0,n})$ is $n(n{-}3)/2$. Equivalently $\pi_1(\M_{0,n})$ can be described as the braid group of $n$ distinguishable strands out of which three are held fixed. For example
\be
\begin{aligned}
	\includegraphics[scale=1]{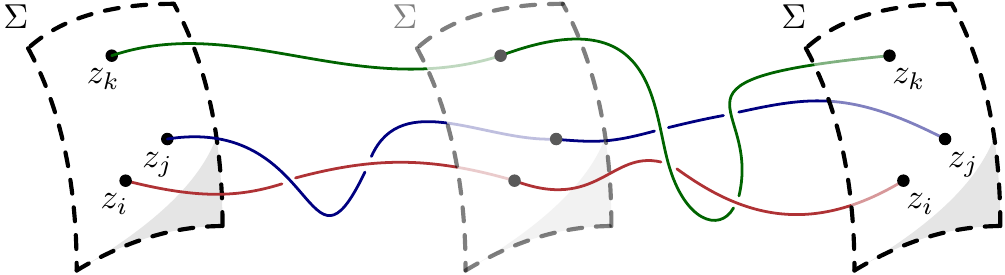}
\end{aligned}
\ee
represents a braid corresponding to $\circlearrowleft_{ij} {\circ} \circlearrowleft_{ik}$, obtained by composing $\circlearrowleft_{ij}$ in the first step with $\circlearrowleft_{ik}$ in the second.\footnote{The fundamental group of $\Conf_{n-1}(\C)$ is called the \emph{pure braid group} \cite{10.2307/1969218} (sometimes also \emph{colored} or \emph{dyed} braid group depending on a translation from Russian \cite{arnold1969cohomology}). Using the $\PSL(2,\C)$ quotient to fix three punctures to $(0,1,\infty)$ we can treat $\M_{0,n}$ as the configuration space $\Conf_{n-3}(\C{\setminus}\{0,1\})$ and use the result on braid groups---such as \cite{Fadell_Neuwirth_1962,arnold1969cohomology}---with minimal modifications, cf. \cite{10.1007/978-1-4612-4264-2_8}.}

One of the main tools in studying topological properties of spaces are homology and cohomology groups. For example, the abelianization of the fundamental group $\pi_1(\M_{0,n})$ is the 1-st singular homology group $H_1(\M_{0,n}, \Z)$. The cohomology ring of ${\cal M}_{0,n}$ with integer coefficients, $H^{\ast}({\cal M}_{0,n}, \Z)$, was essentially studied in the seminal work of Arnold \cite{arnold1969cohomology}, who showed that it is generated by the symmetric elements $\lambda_{ij}$, which are the analogues of $\circlearrowleft_{ij}$, modulo the relations
\be\label{Arnold-relations}
\lambda_{ij} \lambda_{jk} + \lambda_{jk} \lambda_{ki} + \lambda_{ki} \lambda_{ij} = 0.
\ee
An explicit realization as differential forms is given by:
\be\label{lambda}
\lambda_{ij} := d\log(z_i - z_j),
\ee
where the multiplication is given by the exterior product $\wedge$. Here it is understood that for the fixed punctures we have $dz_1 {=} dz_{n-1} {=} dz_n {=} 0$. The Poincar\'e polynomial of ${\cal M}_{0,n}$ is given by
\be\label{Poincare-polynomial}
P(t) := \sum_{k=0}^{2n-6} t^k \dim H^k({\cal M}_{0,n}, \Z) = \prod_{m=1}^{n-3} \big(1 + (m{+}1)t \big).
\ee
In particular, we can read off the dimension of the $1$-st cohomology group,
\be
\dim H^1({\cal M}_{0,n}, \Z) = \frac{1}{2}n(n-3),
\ee
which is in agreement with the previous counting for the fundamental group, and by extension also $H_1(\M_{0,n},\Z)$.

In the sequel we will be mostly interested in the $(n{-}3)$-forms that have correct covariance properties under the M\"obius group transformations. Recall that the action of $\PSL(2,\C)$ on each $z_i$ is given by
\be
\quad z_i \mapsto \frac{\a z_i + \b}{\c z_i + \d} \qquad \text{with} \qquad \a\d - \b\c = 1.
\ee
Hence \eqref{lambda} transform as $\lambda_{ij} \mapsto \lambda_{ij} - d\log(\c z_i + \d) - d\log(\c z_j + \d)$. Let us introduce the following $\PSL(2,\C)$-covariant measure:
\begin{align}
d\mu_n := \frac{\bigwedge_{i=1}^{n} dz_i}{\text{vol }\PSL(2,\C)} = (z_1 {-} z_{n-1})(z_{n-1} {-} z_n)(z_1 {-} z_n) \bigwedge_{i=2}^{n-2} dz_i,
\end{align}
where $\text{vol }\PSL(2,\C)$ is the conventional notation used to denote a quotient by the action of $\PSL(2,\C)$, which we stated explicitly in the gauge where $(z_1,z_{n-1},z_n)$ are fixed. It transforms as
\be
d\mu_n \;\mapsto\; d\mu_n\, \prod_{i=1}^{n} (\c z_i + \d)^{-2}.
\ee
The simplest $\PSL(2,\C)$-invariant differential $(n{-}3)$-form one can write down is the so-called Parke--Taylor form \cite{Parke:1986gb}:
\be\label{Parke-Taylor-form}
\text{PT}(\alpha) := \frac{d\mu_n}{\prod_{i=1}^{n} \left( z_{\alpha(i)} - z_{\alpha(i+1)} \right)},
\ee
where $\alpha$ is a permutation of $n$ labels and we use the cyclic identification $\alpha(n{+}1) = \alpha(1)$. After taking into account reflections, there are $(n{-}1)!/2$ distinct Parke--Taylor forms that can be written down. However, according to \eqref{Poincare-polynomial} the dimension of the $(n{-}3)$-rd cohomology group with integer coefficients is
\be
\dim H^{n-3}({\cal M}_{0,n}, \Z) = (n-2)!,
\ee
which means that over $\Z$ there are exactly $(n{-}2)!$ linearly-independent $(n{-}3)$-forms. A basis of them can be formed by using $(n{-}2)!$ Parke--Taylor forms \eqref{Parke-Taylor-form}. The dimension of the cohomology ring $H^\ast(\M_{0,n},\Z)$ is $P(1)=(n{-}1)!/2$. Note that the same results are valid over $\C$.

The moduli space has a natural fibration obtained by forgetting one puncture, $\M_{0,n} \to \M_{0,n{-}1}$, with the fibre being an $(n{-}1)$-punctured Riemann sphere, as already hinted at by \eqref{Poincare-polynomial}. This recursive nature will be exploited later on for explicit computations and hence we postpone further discussion of fibre bundles until Section~\ref{sec:fibration}.

The boundary divisor $\partial \M_{0,n} := \overbar{\M}_{0,n}{\setminus}\M_{0,n}$ corresponds to nested configurations in which two or more particles collide with each other on the Riemann surface. However, these degenerations are obscured by the coordinate representation \eqref{M0n}, whose divisor has non-normally-crossing components (for example when $z_1 {=} z_2 {=} z_3$). This problem is alleviated by the use of the Deligne--Mumford compactification $\overbar{\M}_{0,n}$ \cite{PMIHES_1969__36__75_0}. We will discuss its specific realization in Section~\ref{sec:logarithmic}.

\subsection{\label{sec:local-systems}From Local Systems to Twisted Forms}

\textsc{Let us introduce} an additional structure on ${\cal M}_{0,n}$ given by a representation---for now Abelian---of its fundamental group,
\be
\L_\omega: \quad \pi_1({\cal M}_{0,n}) \;\to\; \C^\times.
\ee
It is called a \emph{local system} (or a \emph{locally-constant sheaf}) \cite{10.2307/1969099}. For every path $\gamma \in \pi_1({\cal M}_{0,n})$ the local system ${\cal L}_\omega$ assigns a non-zero complex number by integrating a closed one-form $\omega$ over $\gamma$ and exponentiating the result:
\be\label{local-system-map}
\quad \gamma \;\mapsto\; \exp \int_{\gamma} \omega.
\ee
Since $\omega$ is a closed form, its value depends only on the homotopy class of $\gamma$. Let us confirm that the representation \eqref{local-system-map} has the correct group properties. For every pair of paths $\gamma_1, \gamma_2$ we have
\be
\exp \int_{\gamma_1 \circ \gamma_2} \!\!\! \omega = \left(\exp \int_{\gamma_2} \omega \right) \left(\exp \int_{\gamma_1} \omega\right),
\ee
and hence concatenation of paths $\gamma_1 \circ \gamma_2$ corresponds to multiplication of the $\C^\times$ coefficients. In particular, any path contractible to a point has a coefficient $1$. Similarly, the coefficient of a path $\gamma^{-1}$ is the inverse of that associated to $\gamma$.

In this work we consider a particular class of local systems, associated to representations of the fundamental group $\pi_1(\M_{0,n})$, given by \cite{Mizera:2017cqs,Mizera:2017rqa}
\be\label{omega}
\omega := \frac{1}{\Lambda^2}\sum_{1 \leq i<j \leq n} 2p_i {\cdot} p_j\, d \log (z_i - z_j).
\ee
Here to each puncture we associate a Lorentz vector $p_i^\mu$ in an arbitrary space-time dimension, contracted with Minkowski metric in the mostly-plus signature, $p_i {\cdot} p_j := p_i^\mu p_j^\nu \eta_{\mu\nu}$ (we use Einstein notation for the contraction of indices). They satisfy momentum conservation $\sum_{i=1}^{n} p_i^\mu = 0$. In order to make the one-form $\omega$ dimensionless, we normalized it with a mass scale $\Lambda^2$ (it might be possible to introduce multiple mass scales into the problem, which give rise to different local systems, but for simplicity we do not consider them here). Unless specified otherwise, we use \emph{generic} (or non-singular) kinematics, i.e., one for which the kinematic invariants $p_i {\cdot} p_j$ do not satisfy any extra linear relations beyond those implied by momentum conservation. In particular we allow the momenta to be embedded in four space-time dimensions, as it only requires that every $5{\times}5$ minor of the Gram matrix $[p_i {\cdot} p_j]$ vanishes.

Since in a generic gauge-fixing infinity is not a special point on ${\cal M}_{0,n}$, the loop $\circlearrowleft_{\infty i}$ where the $i$-th puncture goes around infinity is contractible to a point, which implies the condition:
\be
1 = \exp \int_{\circlearrowleft_{\infty i}} \!\!\! \omega = \exp \left( -\frac{4\pi i m_i^2}{\Lambda^2} \right),
\ee
where we used the fact that $\sum_{j\neq i } p_i {\cdot} p_j = - p_i^2 =: m_i^2$ is the mass squared of the $i$-th particle. This requirement imposes the quantization condition on the masses, $m_i^2 \in \Lambda^2\Z / 2$ for every $i$. Under these conditions, there are exactly $n(n{-}3)/2$ linearly-independent kinematic invariants $2p_i {\cdot} p_j$ modulo $\Lambda^2 \Z$ and hence the Abelian representation of $\pi_1(\M_{0,n})$ given by \eqref{local-system-map} with \eqref{omega} is faithful. In our applications we will mostly focus on the cases where all the external particles are massless, $m_i =0$, or have the same mass, $m_i = \Lambda$.

Let us consider a generalization of differential forms to forms with coefficients in the local system ${\cal L}_\omega$ (also called \emph{local coefficients}). For a single-valued differential $k$-form $\varphi \in \Omega^{k}(\M_{0,n})$ and a coefficient $\exp \int_\gamma \omega \in \L_\omega$ we write it as
\be\label{cocycle-with-local-coefficient}
\varphi \otimes \exp \textstyle\int_\gamma \omega \;\in\; \Omega^{k}(\M_{0,n}) \otimes \L_\omega,
\ee
where $\Omega^{k}(\M_{0,n})$ denotes the space of smooth $k$-forms on $\M_{0,n}$, and $\gamma \in \pi_1(\M_{0,n})$ is an arbitrary path connecting a fixed point $p$ to the coordinate of $\varphi$.

In fact, correlation functions of vertex operators in string theory have the form \eqref{cocycle-with-local-coefficient} before integrating over the puncture coordinates, see, e.g., \cite{green1988superstring}. The meaning of \eqref{cocycle-with-local-coefficient} is that such correlators are not globally defined, but rather depend on the history of punctures travelling around each other on the Riemann surface, called the \emph{monodromy}. For example, the effect of two punctures $z_i$ and $z_j$ participating in the loop $\circlearrowleft_{ij}$ is multiplying the correlator by the phase $\exp \left( 4\pi i p_i{\cdot} p_j / \Lambda^2 \right)$. The local system is a way of encapsulating monodromy properties of such correlation functions. Of course, one can also define them on the covering space of $\M_{0,n}$, but the local-system description will prove to be more advantageous for explicit computations.

For the form with local coefficients \eqref{cocycle-with-local-coefficient} to be well-defined we impose that any $\varphi \otimes \exp\int_\gamma\! \omega$ is $\PSL(2,\C)$-invariant. Since, using momentum conservation, $\omega$ has the transformation property
\be
\omega \;\mapsto\; \omega - \frac{2}{\Lambda^2} \sum_{i=1}^{n} m_i^2\, d\log (\c z_i + \d),
\ee
it means that $\varphi$ needs to transform in a covariant way:
\be\label{SL2C-covariance}
\varphi \;\mapsto\; \varphi\; \prod_{i=1}^{n} (\c z_i + \d)^{2m_i^2 / \Lambda^2}.
\ee
A form $\varphi$ with the above property is said to carry a \emph{M\"obius weight} $2m_i^2 / \Lambda^2 \in \Z$ in the coordinate $z_i$.

One can define $k$-th cohomology with coefficients in the local system, $H^k(\M_{0,n}, \L_\omega)$ \cite{10.2307/1969099}, which is the space of closed forms with coefficients in ${\cal L}_\omega$ modulo the exact ones (strictly speaking one should also impose $\SL(2,\C)$-invariance, but we leave it implicit from now on),
\be\label{cohomology-with-local-coefficients}
H^k(\M_{0,n}, {\cal L}_{\omega}) := \frac{\{ \varphi \otimes \exp \textstyle\int_\gamma \omega \,\in\, \Omega^{k}(\M_{0,n}) \otimes {\cal L}_\omega \,|\, d( \varphi \otimes \exp\textstyle\int_\gamma \omega ) = 0 \}}{d \left( \Omega^{k-1}(\M_{0,n}) \otimes {\cal L}_\omega \right)}.
\ee
However, for our practical purposes we would like to construct a de Rham analogue of this cohomology group, so that all computations can be done purely on differential forms. Since the definition \eqref{cohomology-with-local-coefficients} only depends on the way the differential $d$ acts on forms and ${\cal L}_\omega$, let us analyze this action more closely. We start by noticing that $d$ acts on a form $\xi \in \Omega^{k-1}(\M_{0,n})$ with a local coefficient in a $\gamma$-independent way:
\begin{align}
d \left( \xi \otimes \exp \textstyle\int_\gamma \omega \right) &= d\xi \otimes \exp \textstyle\int_\gamma \omega + (-1)^{k-1} \xi \otimes \omega \wedge \exp \textstyle\int_\gamma \omega \nn\\
&= \Big( d\xi + \omega \wedge \xi \Big) \otimes \exp \textstyle\int_\gamma \omega.\label{image-map}
\end{align}
Thus it is natural to define the Gauss--Manin connection $\nabla_{\omega} := d + \omega \wedge$, which acts just on the forms. Since $\omega$ is closed, it is straightforward to check that the connection is \emph{integrable} (flat), i.e., $\nabla_{\omega}^2 = 0$. Hence the quotient in \eqref{cohomology-with-local-coefficients} translates to a quotient of $\varphi$ by terms of the form $\nabla_{\omega} \xi$. Repeating the same exercise for the numerator of \eqref{cohomology-with-local-coefficients} we find that $\varphi$ should satisfy $\nabla_\omega \varphi = 0$.

Therefore matching the kernels and images of $\nabla_{\omega}$ we can identify cohomology groups $H^k(\M_{0,n},\nabla_\omega)$ isomorphic to $H^k({\cal M}_{0,n}, \L_{\omega})$, explicitly given by
\be\label{twisted-cohomology}
H^k(\M_{0,n}, \nabla_{\omega}) := \frac{\{ \varphi \in \Omega^{k}(\M_{0,n}) \,|\, \nabla_{\omega} \varphi = 0 \}}{\nabla_\omega \Omega^{k-1}(\M_{0,n})}.
\ee
In other words, it is the space of $\nabla_{\omega}$-closed $k$-forms on $\M_{0,n}$ modulo those which are $\nabla_{\omega}$-exact. This space is called the $k$-th \emph{twisted cohomology}. Its isomorphism to cohomology with coefficients in the local system $\L_\omega$ follows essentially from the twisted version of Grothendieck's comparison theorem \cite{PMIHES_1966__29__95_0} proven by Deligne \cite{deligne1970equations}.

Under some non-resonance assumptions---which in our cases translate to generic kinematics---Aomoto showed that twisted cohomology is concentrated in the middle dimension, i.e., only $H^{k}(\M_{0,n}, \nabla_\omega)$ for $k= \dim_{\C} \M_{0,n}$ is non-trivial \cite{aomoto1975vanishing}. In other words, $k$-forms $\varphi$ which are $\nabla_\omega$-closed are always $\nabla_\omega$-exact unless $k = n{-}3$. This allows us to compute the dimension of the twisted cohomology purely topologically, since the Euler characteristic of $\M_{0,n}$ is related to twisted Betti numbers through
\be\label{vanishing}
\chi(\M_{0,n}) = \sum_{k=0}^{2n-6} (-1)^k \dim H^k(\M_{0,n}, \nabla_{\omega}) = (-1)^{n-3} \dim H^{n-3}(\M_{0,n}, \nabla_\omega).
\ee
Here we used the fact that the twisted and topological Euler characteristics coincide, $\chi(\M_{0,n},\L_\omega)=\chi(\M_{0,n})$, since $\L_\omega$ is a flat line bundle. The Euler characteristic can be easily obtained from the Poincar\'e polynomial in \eqref{Poincare-polynomial} evaluated at $t=-1$, which gives
\be
\chi(\M_{0,n}) = (-1)^{n-3}(n-3)!
\ee
and therefore we conclude
\be\label{dimension}
\dim H^{n-3}(\M_{0,n}, \nabla_\omega) = (n-3)!.
\ee
To our knowledge this statement was first stated explicitly in \cite{aomoto1987gauss}. We will give another derivation in Section~\ref{sec:fibration}.

The elements $[\varphi] \in H^{n-3}(\M_{0,n}, \nabla_\omega)$ are called \emph{twisted cocycles}. Throughout this work we will abuse the notation by writing $\varphi \in H^{n-3}(\M_{0,n}, \nabla_{\omega})$ for any representative of the cohomology class $[\varphi]$, and refer to $\varphi$ as a \emph{twisted differential form}, or simply \emph{twisted form}. Two twisted forms in the same cohomology class are said to be \emph{cohomologous}. 

In order to constructs invariants between twisted forms one needs a dual space. A natural dual is given by the cohomology associated to the local system ${\cal L}_{-\omega}$ obtained by sending $\omega \to -\omega$, since ${\cal L}_{-\omega} \otimes {\cal L}_{\omega}$ is trivial (monodromy-invariant).\footnote{The alternative choice is a pairing between ${\cal L}_{\overbar{\omega}}$ and ${\cal L}_{\omega}$, which gives rise to integrals in closed string theory, see Section~\ref{sec:closed-string-amplitudes}.} The duality induces a pairing between twisted forms of the two types, which will be discussed in the next subsection.

For the sake of clarity let us introduce short-hand notation for the two cohomologies:
\be
H^{n-3}_\omega := H^{n-3}(\M_{0,n}, \nabla_{\omega}),\qquad H^{n-3}_{-\omega} := H^{n-3}(\M_{0,n}, \nabla_{-\omega}).
\ee
For twisted forms to be well-defined we also need that they transform with M\"obius weights $\pm 2m_i^2 / \Lambda^2$ for each $z_i$. Since the twisted cohomologies are concentrated in the middle dimension $n{-}3$, it is natural to take both of them to be holomorphic top forms, i.e., proportional to $d\mu_n$. Note that in those cases the $\nabla_{\pm\omega}$-closedness condition is always satisfied. For physical applications we will consider twisted forms with poles only on the boundary divisor $\partial\M_{0,n}$, i.e., those of the form $1/(z_i {-} z_j)$, as well as poles at the infinities ${\cal I}$. The latter type of poles, happening when $z_i \to \infty$, are \emph{spurious} meaning that they will not contribute to the final result of our computations. The space of rational holomorphic forms with such properties is called $\Omega^{n-3,0}(\ast \partial \M_{0,n} {\cup} {\cal I})$. We will denote generic elements of $H^{n-3}_{\pm \omega}$ which are also in $\Omega^{n-3,0}(\ast \partial \M_{0,n} {\cup} {\cal I})$ by $\varphi_+$ and $\varphi_-$. Finally, it is reasonable to impose that $\varphi_\pm$ themselves do not have kinematic poles involving $p_i {\cdot} p_j$, so that all singularities arise geometrically from the local system.

At this stage it is instructive to give examples of twisted forms. For instance, when all external masses are equal to an integer (as opposed to half-integer) multiple of the mass-scale, $m_i^2 = \mathrm{M}\Lambda^2$ for $\mathrm{M} \in \Z$, we have the following generalization of the Parke--Taylor forms from \eqref{Parke-Taylor-form},
\be\label{massive-PT}
\frac{d\mu_n}{\prod_{m=1}^{1\pm \mathrm{M}} \prod_{i=1}^{n} \left(z_{\alpha_m\! (i)} - z_{\alpha_m\! (i+1)}\right)} \;\in\; H^{n-3}_{\pm\omega},
\ee
where each $\alpha_m$ is a permutation of $n$ elements for $m=1,2\ldots,1{\pm}\mathrm{M}$, where the number of permutations is chosen to assure the correct $\SL(2,\C)$-transformation properties \eqref{SL2C-covariance}. Except for the degenerate case $\pm \mathrm{M} = -1$, $(n{-}3)!$ of such twisted forms give a basis of their respective cohomology groups. For $\pm \mathrm{M} = -1$ a basis has to be constructed out of twisted forms obtained by a product of $\PSL(2,\C)$-invariant cross-ratios times $d\mu_n$. It is straightforward to verify that the above massive Parke--Taylor forms have the correct M\"obius weights. Notice that for $\mathrm{M}=0$ the same twisted forms \emph{can} belong to the two cohomologies $H^{n-3}_{\pm\omega}$, even though they represent distinct cohomology classes.

At this stage let us mention that much of the structure of twisted cohomologies on $\M_{0,n}$ can be understood purely combinatorially through the theory of hyperplane arrangements, see, e.g.,~\cite{SB_1971-1972__14__21_0,Orlik1980,Schechtman1991,Esnault1992,orlik2013arrangements}. For instance, by sending one puncture to infinity we can apply the result that the number of bounded chambers (regions with finite volume) in $\M_{0,n}(\mathbb{R})$ equals to the dimension of $H^{n-3}_\omega$ \cite{aomoto1975vanishing}. Indeed, explicit counting gives $(n{-}3)!$ chambers, cf.~\cite{Cachazo:2016ror,Cachazo:2019ngv}.\footnote{To be more specific, the characteristic polynomial $\chi_{\cal A}(t)$ of the Selberg arrangement ${\cal A}$ defined by the pole locus of $\omega$ (with $z_n$ set to infinity) evaluated at $t {=-} 1$ gives the total number of chambers in $\M_{0,n}(\mathbb{R})$ up to a sign, $|\chi_{\cal A}(-1)| = (n{-}1)!/2$, while $|\chi_{\cal A}(1)| = (n{-}3)!$ is the number of bounded chambers, see, e.g., \cite{orlik2013arrangements} for examples. The characteristic polynomial is related to the Poincar\'e polynomial \eqref{Poincare-polynomial} by $P(t) = (-t)^{n-3}\chi_{\cal A}(-1/t)$ \cite{orlik2001arrangements}.} Based on the combinatorics of the arrangement of hyperplanes $\{ z_i {-} z_j = 0\}$ one can also consistently construct bases of twisted cohomologies known as the $\beta\textbf{nbc}$ bases \cite{falk1997betanbc}.

\subsection{\label{sec:definition}Definition and Properties of Intersection Numbers}

\textsc{With the above definitions} in place we can finally turn to the main object of our study. Given two middle-dimensional forms, their most natural pairing is given by the integral
\be\label{intersection-naive}
\int_{{\cal M}_{0,n}} \!\!\! \varphi_- \wedge \varphi_+
\ee
over the whole moduli space $\M_{0,n}$. However, this integral is not well-defined since ${\cal M}_{0,n}$ is not compact (in particular it does not contain all its limit points and hence is not closed). Moreover, away from the boundaries $\partial \M_{0,n}$ the integrand vanishes, as it is given by a wedge product of two holomorphic top-forms. It means that---one way or another---the appropriate regularization of \eqref{intersection-naive} will necessarily localize on $\partial \M_{0,n}$. Following \cite{cho1995} we give the natural definition of the intersection pairing between $\varphi_-$ and $\varphi_+$ \cite{Mizera:2017rqa}.
\begin{definition}\label{definition-intersection-number}
	The intersection number of two twisted forms $\varphi_\pm \in H^{n-3}(\M_{0,n}, \nabla_{\pm\omega})$ is given by
	\be\label{intersection-definition}
	\braket{ \varphi_- | \varphi_+ }_\omega := \frac{1}{(-2\pi i \Lambda^2)^{n-3}} \int_{\M_{0,n}} \!\!\! \varphi_- \wedge \iota_\omega(\varphi_+),
	\ee
	where $\iota_\omega(\varphi_+) \in H^{n-3}_c(\M_{0,n},\nabla_\omega)$ denotes a form cohomologous to $\varphi_+$ but with compact support.
\end{definition}
Here $H^{n-3}_c(\M_{0,n},\nabla_\omega)$ has the same definition as $H^{n-3}(\M_{0,n},\nabla_\omega)$, but with $\Omega^k_c(\M_{0,n})$, the space of $k$-forms on $\M_{0,n}$ with compact support (that is, those that vanish in the small neighbourhood of $\partial \M_{0,n}$). We will show explicit ways of constructing the map 
\be
\iota_\omega:\; H^{n-3}(\M_{0,n}, \nabla_{\omega}) \,\to\, H^{n-3}_c(\M_{0,n}, \nabla_{\omega})
\ee
throughout this work. Notice that compactly-supported forms cannot be holomorphic, and hence the map $\iota_\omega$ will have to introduce non-holomorphic contributions near $\partial \M_{0,n}$. These are precisely the terms that will make the intersection numbers non-vanishing. Of course, \eqref{intersection-definition} can be also defined by imposing compact support on the form $\varphi_-$ instead (or at the same time), which yields the same answer. It is not implausible that a similar result can be obtained by compactifying the moduli space $\M_{0,n}$, but we will not attempt it here, as the above implementation appears easier to execute technically.

Due to the opposite M\"obius weights of $\varphi_+$ and $\varphi_-$, the intersection number is independent of the $\PSL(2,\C)$ fixing. By definition it is linear in both arguments and satisfies the cohomology relations
\be
\la \varphi_- {+} \nabla_{-\omega} \xi\, | \varphi_+ \ra_{\omega} = \la \varphi_- |\, \varphi_+ {+} \nabla_\omega \xi \ra_{\omega} = \la \varphi_- | \varphi_+ \ra_{\omega}
\ee
for any $(n{-}4)$-form $\xi$. The overall normalization $1/(-2\pi i \Lambda^2)^{n-3}$ is introduced to make the result non-transcendental (in particular not contain factors of $\pi$) and gives it non-trivial mass dimension. From the definition \eqref{intersection-definition}, we have the following symmetry:
\be\label{intersection-number-symmetry}
\braket{ \varphi_- | \varphi_+ }_\omega = \braket{ \varphi_+ | \varphi_- }_{-\omega},
\ee
where on the right-hand side the intersection number is evaluated using $-\omega$ instead of $\omega$. Note that this includes a change of the overall sign by $(-1)^{n-3}$ due to $\Lambda$-dependent normalization in \eqref{intersection-definition}.

Let us understand the reason why \eqref{intersection-definition} provides an invariant associated to the two twisted cohomologies. Modding out by $\nabla_{-\omega}$-exact terms in $\varphi_-$ gives us the equality:
\begin{align}
0 &= \int_{\M_{0,n}} \!\!\! \nabla_{-\omega} \xi \wedge \iota_\omega (\varphi_+)\nn\\
&= (-1)^{n-3}\int_{\M_{0,n}} \xi \wedge \nabla_\omega \iota_\omega (\varphi_+),
\end{align}
where the second line is obtained by expanding $\nabla_{-\omega}\xi = d\xi - \omega \wedge \xi$ followed by integration by parts and commuting $\omega$ through $\xi$. Boundary terms do not contribute since the integrand has compact support. The final result implies that $\nabla_{\omega} \iota_\omega(\varphi_+)$ and hence also $\nabla_{\omega} \varphi_+$ vanish on the level of the twisted cohomology. This is nothing but the $\nabla_\omega$-closedness condition on $\varphi_+$. Similar computation shows that quotient of $\varphi_+$ by $\nabla_\omega$-exactness implies $\nabla_{-\omega}$-closedness of $\varphi_-$. This is precisely the sense in which the two cohomology groups $H^{n-3}_{\pm\omega}$ are dual to each other.

One can show that result of \eqref{intersection-definition} is a rational function of kinematic invariants $p_i {\cdot} p_j$, external masses $m_i^2$, the mass scale $\Lambda^2$, and other physical quantities (such as polarization vectors and colours) that may enter through $\varphi_\pm$. It therefore has the correct ingredients to compute tree-level scattering amplitudes in quantum field theories.

In order to gain intuition about the type of manipulations involved in the evaluation of intersection numbers let us consider the simplest case for $n=4$. In the standard $\SL(2,\C)$ fixing the moduli space $\M_{0,4}$ can be written as $\M_{0,4} = \{z_2 \in \CP^1 \,|\, z_2 \neq z_1, z_3, z_4\}$ and has the boundary divisor consisting of three points, $\partial \M_{0,4} = \{z_1, z_3, z_4\}$. For a given $\varphi_+ \in \Omega^{1,0}(\ast\{z_1, z_3, z_4, \infty\})$, its compactly supported cousin $\iota_\omega(\varphi_+)$ needs to vanish in a small neighbourhood of each of the three removed points. We construct this form explicitly as follows:
\be\label{compactly-supported-form}
\iota_\omega(\varphi_+) = \varphi_+ - \nabla_\omega \Bigg( \sum_{i=1,3,4} \Theta(\varepsilon^2 - |z_2 {-} z_i|^2)\, \nabla_\omega^{-1} \varphi_+ \Bigg),
\ee
where $\Theta(x)$ is the Heaviside step function equal to $1$ for $x \geq 0$ and $0$ otherwise. Therefore each term inside the sum has support only on the infinitesimal disk around $z_2 = z_i$ with radius $\varepsilon$. Notice that the step function is non-holomorphic. The function $\psi := \nabla_\omega^{-1} \varphi_+$ involves a formal inverse of $\nabla_\omega$, which is understood as a solution of the equation $\nabla_\omega \psi = \varphi_+$. Since the result is multiplied by a step function, it is enough to know $\psi$ as a holomorphic expansion around near each $z_2 = z_i$, which is given uniquely by $\psi_i = \sum_{k} \psi_{i,k} (z_2 {-} z_i)^k$. By construction, the compactly-supported form $\iota_\omega(\varphi_+)$ in \eqref{compactly-supported-form} is cohomologous to $\varphi_+$. Evaluating the action of $\nabla_\omega$ on the terms in the brackets we find
\be\label{compactly-supported-form-2}
\iota_\omega(\varphi_+) = \Bigg( 1 - \!\!\! \sum_{i=1,3,4} \!\! \Theta(\varepsilon^2 - |z_2 {-} z_i|^2) \Bigg) \varphi_+ - \sum_{i=1,3,4} \delta(\varepsilon^2 - |z_2 {-} z_i|^2)\, \nabla^{-1}_\omega \varphi_+.
\ee
The first term is obtained by $\nabla_\omega$ acting on $\nabla_\omega^{-1} \varphi_+$ for each term in the sum in \eqref{compactly-supported-form}. It is simply a regularization of $\varphi_+$ which vanishes in an infinitesimal neighbourhood of each component of $\partial \M_{0,4}$ and equal to $\varphi_+$ otherwise. The second term involves a sum over $\nabla_\omega^{-1} \varphi_+$ weighted by Dirac delta functions localizing the result on the circle around each $z_2 = z_i$ with radius $\varepsilon$. Therefore the resulting one-form $\iota_\omega(\varphi_+)$ has compact support, as required.

After plugging into the Definition~\ref{definition-intersection-number}, the terms proportional to $\varphi_+$ in \eqref{compactly-supported-form-2} vanish, since they are wedged with $\varphi_-$. Only the terms involving delta functions survive and hence we find:
\begin{align}
\braket{ \varphi_- | \varphi_+ }_\omega &= \frac{1}{2\pi i \Lambda^2} \sum_{i=1,3,4} \int_{\M_{0,4}} \!\!\!\delta(\varepsilon^2 - |z_2 {-} z_i|^2)\, \varphi_- \nabla_{\omega}^{-1} \varphi_+ \nn\\
&= \frac{1}{\Lambda^2} \sum_{i=1,3,4} \Res_{z_2 = z_i} \!\left( \varphi_- \nabla_{\omega}^{-1} \varphi_+ \right).\label{n-4-result}
\end{align}
As expected, the intersection number localizes as a sum over residues around the boundary components of $\partial \M_{0,4}$.

The above result means that in order to compute $\braket{\varphi_- | \varphi_+}_\omega$ one needs to first find $\nabla^{-1}_\omega \varphi_+$ locally near each $z_2 = z_i$ and then evaluate residues around these points. It is precisely due to the $\nabla_{\omega}^{-1} \varphi_+$ factor that intersection numbers feature combinations of the kinematic invariants $(p_i {+} p_j)^2 {+} \Z \Lambda^2$ in denominators. The evaluation of higher-$n$ intersection numbers is conceptually similar, though the boundary structure of $\M_{0,n}$ becomes more involved and likewise the form of $\iota_\omega(\varphi_+)$ increases in complexity. In Section~\ref{sec:recursion-relations} we will exploit the simplicity of the result \eqref{n-4-result} to construct recursion relations for intersection numbers on a fibration of $\M_{0,n}$ into $n{-}3$ one-dimensional spaces.

In order to illustrate general properties of intersection numbers let us give the result of computing \eqref{n-4-result} for a few examples. We start with the case of massless external kinematics and the Parke--Taylor form $\varphi_+ = \mathrm{PT}(1234)$. Evaluating $\nabla_{\omega}^{-1}\varphi_+$ around $z_2 = z_1$ we find the expansion:
\begin{align}
&\nabla_{\omega}^{-1}\mathrm{PT}(1234) = \frac{\Lambda^2}{(p_1{+}p_2)^2} + \frac{\Lambda ^2 (p_1 {+} p_3)^2}{(p_1 {+} p_2)^2\left((p_1 {+} p_2)^2 + \Lambda ^2\right)}\frac{(z_2{-}z_1) (z_3{-}z_4)}{(z_1{-}z_3) (z_1{-}z_4)} \\
&\qquad\qquad\quad{+} \frac{\Lambda^2 (p_1{+}p_3)^2 (z_2{-}z_1)^2 (z_3{-}z_4)  \left((z_1{-}z_4) \left((p_1{+}p_3)^2{-}\Lambda^2\right) {+} (z_1{-}z_3) \left((p_2{+}p_3)^2{-}\Lambda^2\right)\right)}{(p_1{+}p_2)^2 \left((p_1 {+}p_2)^2 + \Lambda^2\right) \left((p_1 {+}p_2)^2 + 2 \Lambda ^2\right) (z_1{-}z_3)^2 (z_1{-}z_4)^2} {+} \ldots.\nn
\end{align}
Here the leading order has a simple pole in $(p_1 {+} p_2)^2$ corresponding to a propagation of a massless state, while the first subleading term has propagators with mass $0$ and $\Lambda$, the second subleading has $0$, $\Lambda$, $\sqrt{2}\Lambda$, etc. By symmetry, around $z_2=z_3$ the expansion takes a similar form, but with poles in the $(p_2{+}p_3)^2$-channel,
\be
\nabla_{\omega}^{-1} \mathrm{PT}(1234) = -\frac{\Lambda^2}{(p_2{+}p_3)^2}
+\frac{\Lambda^2 (p_1{+}p_3)^2}{(p_2{+}p_3)^2 ((p_2{+}p_3)^2 + \Lambda^2)}\frac{(z_2{-}z_3) (z_1{-}z_4)}{(z_1{-}z_3) (z_3{-}z_4)} + \ldots,
\ee
while around $z_2=z_4$ we have:
\be
\nabla_{\omega}^{-1} \mathrm{PT}(1234) = -\frac{\Lambda^2}{(p_1{+}p_3)^2+\Lambda^2}\frac{(z_2{-}z_{4})(z_1{-}z_3) }{(z_1{-}z_4) (z_4{-}z_3) } + \ldots.\label{expansion-nabla-inverse-3}
\ee
This time the expansion starts at the subleading order as a consequence of $\mathrm{PT}(1234)$ not having a pole at $z_2 = z_4$. Correspondingly, the massless pole in the kinematic variable $(p_2 {+}p_4)^2 = (p_1{+}p_3)^2$ is absent.

Plugging the above expansions into \eqref{n-4-result}, say with $\varphi_-$ equal to the same Parke--Taylor form we find:
\be
\Braket{ \frac{d\mu_4}{z_{12} z_{23} z_{34} z_{41}} | \frac{d\mu_4}{z_{12} z_{23} z_{34} z_{41}} }_{\!\omega} \;=\;  \frac{1}{(p_1 {+} p_2)^2} + \frac{1}{(p_2 {+} p_3)^2}.
\ee
Here we used the notation $z_{ij} := z_i {-} z_j$. This result exemplifies a general feature of intersection numbers: they become rational functions with simple poles in the kinematic variables. Consider another example obtained by changing $\varphi_-$:
\be\label{PT-1243-1234}
\Braket{ \frac{d\mu_4}{z_{12} z_{24} z_{43} z_{31}} | \frac{d\mu_4}{z_{12} z_{23} z_{34} z_{41}} }_{\!\omega} \;=\; -\frac{1}{(p_1 {+} p_2)^2}, 
\ee
which suggests that kinematic poles can appear in the intersection number when the two twisted forms share poles in the moduli space. This is already clear from \eqref{n-4-result} and the above solutions of $\nabla_{\omega}^{-1}\varphi_+$: since $\nabla_\omega^{-1}$ decreases the degree of the pole by one, each residue can be non-zero only if the degrees of the poles of $\varphi_\pm$ add up to at least two. Notice that in the above example there are two poles which are shared: in $z_1 {-} z_2$ and $z_3 {-} z_4$. These are however the same singularities from the point of view of $\M_{0,n}$ and correspond to the Riemann surface becoming an infinitely long tube with punctures $z_1$, $z_2$ on one side of the throat and $z_3$, $z_4$ on the other. In our $\SL(2,\C)$ fixing, we manifested the fact that there are only three types of possible singularities coming from degeneration when $z_2$ approaches one of $\{z_1, z_3, z_4\}$. Thus we find only one Feynman diagram contributing to the intersection number.

Let us consider an example with double poles in the twisted form $\varphi_-$:
\be
\Braket{ \frac{d\mu_4}{z_{13}^2 z_{24}^2} | \frac{d\mu_4}{z_{12} z_{23} z_{34} z_{41}} }_{\!\omega} = \frac{1}{(p_1 {+} p_3)^2 + \Lambda^2}.
\ee
Here the presence of a double pole induced a propagation of a massive state in the $(p_1 {+} p_3)^2$-channel with mass $\Lambda$ coming from \eqref{expansion-nabla-inverse-3}. Using the relation \eqref{intersection-number-symmetry} that exchanges $\varphi_-$ and $\varphi_+$, would produce tachyonic poles. Let us see another example with double poles in $\varphi_+$:
\be\label{example-235}
\Braket{ \frac{d\mu_4}{z_{12} z_{23} z_{34} z_{41}} | \frac{d\mu_4}{z_{12}^2 z_{34}^2} }_{\!\omega} \;=\;  \frac{(p_1 {+} p_3)^2}{(p_1 {+} p_2)^2 \left((p_1 {+} p_2)^2 - \Lambda^2\right)}.
\ee
This time we find not only a massless state, but also a tachyon with mass $i\Lambda$, being exchanged in the $(p_1 {+} p_3)^2$-channel. A general rule of thumb is that introducing higher-degree poles to $\varphi_-$ allows for a propagation of more-massive states, while higher-degree poles in $\varphi_+$ produces more-tachyonic propagators.

When all four external states have mass $\Lambda$ the simplest intersection number reads
\be
\Braket{ d\mu_4 | \frac{d\mu_4}{(z_{12} z_{23} z_{34} z_{41})^2} }_{\!\omega} \;=\; \frac{1}{(p_1 {+} p_2)^2 + \Lambda^2} + \frac{1}{(p_2 {+} p_3)^2 + \Lambda^2}.
\ee
In this case $\varphi_-$ does not have any poles other than the spurious ones at infinity, while $\varphi_+$ has all consecutive double poles of the form $z_i {-} z_{i+1}$. The resulting intersection number has massive states propagating in both planar channels. One the level of the local system the singularities are actually of the same type as in \eqref{example-235}: $2p_1 {\cdot} p_2 - \Lambda^2$, however now $p_1^2 = p_2^2 = -\Lambda^2$, and as a consequence the physical interpretation of this pole changes.

The above examples are by no means exhaustive and were merely supposed to give an intuition about the objects under study. Further examples of intersection numbers of more general twisted forms will be given in Section~\ref{sec:examples}.

\subsection{\label{sec:relations}Relations to Other Objects}

\textsc{One can learn a great deal} about properties of intersection numbers from studying their relations to other objects, most notably scattering amplitudes in string theory and the Cachazo--He--Yuan (CHY) formulation, which is what we consider in this subsection. As our main focus rests on intersection numbers, the exposition will be brief and more details are left until the Appendix~\ref{app:aspects}.

We will make repeated use of the following result, which is elementary but has many important consequences.

\begin{proposition}\label{proposition}
	Let $U\cong V \cong W \cong X$ be four isomorphic complex vector spaces with non-degenerate bilinear pairings denoted by $\la u | v \ra$, $\la u | x \ra$, $\la w | x \ra$, $\la w | v \ra$ for $u\in U, v\in V, w\in W, x\in X$, which are normalized such that for every $w \in W$ there exists $\eta(w) \in U$ for which $\la \eta(w) | v \ra = \la w | v \ra$ and $\la \eta(w) | x \ra = \la w | x \ra$ for all $v \in V$, $x \in X$. Then the bilinears between basis vectors $\{u_a\}_{a=1}^{\dim U} {\in} U$, $\{v_b\}_{b=1}^{\dim V} {\in} V$, $\{w_c\}_{c=1}^{\dim W} {\in} W$, $\{x_d\}_{d=1}^{\dim X} {\in} X$ are related by
	\be\label{proposition-first}
	\la u_a | x_d \ra = \sum_{b=1}^{\dim U} \la u_a | v_b \ra \la w_b^\vee | x_d \ra,
	\ee
	for orthonormal bases $\{ v_b^\vee \}_{b=1}^{\dim V} {\in }V$, $\{ w_c^\vee \}_{c=1}^{\dim W} {\in }W$ such that $\la w_{c}^\vee | v_b \ra = \la w_c | v_b^\vee \ra = \delta_{cb}$. Alternatively, this result can be stated as
	\be\label{proposition-second}
	\la u_a | x_d \ra = \sum_{b,c=1}^{\dim U} \la u_a | v_b \ra \, \mathbf{S}_{bc} \, \la w_c | x_d \ra,
	\ee
	where $\mathbf{S}$ is a $(\dim U){\times}(\dim U)$ matrix with entries $\mathbf{S}_{bc} := \la w_b^\vee | v_c^\vee \ra$, while its inverse $\mathbf{S}^{-1}$ has entries $\mathbf{S}_{cb}^{-1} := \la w_c | v_b \ra$.
\end{proposition}
\begin{proof}
	Since $U \cong W$ we can express any $u_a \in U$ in terms of a basis $\{\eta(w_c)\}_{c=1}^{\dim U} \in U$ as
	\be\label{proposition-u}
	u_a = \sum_{c=1}^{\dim U} \alpha_{ac}\, \eta(w_c)
	\ee
	for some coefficients $\alpha_{ac} \in \mathbb{C}$. Projecting with $\la \,\bullet\, | v_b \ra$ and solving for $\alpha_{ac}$ we find
	\be
	\alpha_{ac} = \sum_{b=1}^{\dim U} \la u_a | v_b \ra\, \mathbf{S}_{bc}.
	\ee
	Plugging back into \eqref{proposition-u} and projecting with $\la \,\bullet\, | x_d \ra$ gives the required result \eqref{proposition-second}. Applying this result to $U=W$, $X=V$ with orthonormal sets of bases we obtain
	\be
	\la w_b^\vee | v_c^\vee \ra = \sum_{a,e=1}^{\dim U} \la w_b^\vee | v_a \ra\, \mathbf{S}_{ae} \la w_e | v_c^\vee \ra = \mathbf{S}_{bc}.
	\ee
	Finally, using \eqref{proposition-second} with $w_c \leftrightarrow w_c^\vee$ gives $\mathbf{S}_{bc} = \delta_{bc}$ and hence \eqref{proposition-first}.
\end{proof}

In our applications the vector spaces will be various twisted homology and cohomology groups, whose dualities will imply a web of relations between different physical objects, such as intersection numbers, string theory amplitudes, etc. One can think of the Proposition~\ref{proposition} as describing different ways of inserting identity operators, $\la u_a | x_d \ra = \la u_a | \I | x_d \ra$ \cite{Mizera:2017rqa}, for instance with
\be
\I = \sum_{b,c=1}^{\dim U} | v_b \ra \, \mathbf{S}_{bc} \, \la w_c |.
\ee
The assumption about consistent normalization is only important for the overall constant multiplying the right-hand side to be one. It is satisfied for all the pairings we will consider, except for the intersection number \eqref{intersection-definition} itself, which we normalized by an additional factor of $1/(-2\pi i \Lambda^2)^{n-3}$ for other reasons. It can be easily accounted for by inserting overall factors of $(-2\pi i \Lambda^2)^{n-3}$ wherever necessary.

\subsubsection{\label{sec:open-string-amplitudes}Open-String Amplitudes}

\textsc{Genus-zero scattering amplitudes} in open string theory can be understood in the language of local systems as follows \cite{Mizera:2017cqs}. We introduce the $(n{-}3)$-rd homology with coefficients in the local system (or \emph{twisted homology}), $H_{n-3}(\M_{0,n}, \L_\omega)$. Its elements are given by the topological middle-dimensional cycles, e.g., those corresponding to insertions of punctures on the circle $\mathbb{RP}^1 \subset \CP^1$ with an ordering $\alpha$. For simplicity of notation we can fix $z_n$ to infinity and $\alpha(n){=}n$ to write such integration cycles $\Delta(\alpha)$ as
\be
\Delta(\alpha) := \{ (z_{2}, z_{3}, \ldots, z_{n-2}) \in \mathbb{R}^{n-3} \,|\,  z_{\alpha(1)} < z_{\alpha(2)} < \cdots < z_{\alpha(n-1)}  \}.
\ee
In order to construct elements of the twisted homology one also needs to specify local coefficients $\exp \int_\gamma \omega$. For example, we can pick a continuous family of $\gamma$'s that connect an arbitrary point $p$ with to points in $\Delta(\alpha)$. The choice of $p$ only affects the overall phase of the local coefficient and we fix it be the so-called Koba--Nielsen factor \cite{Koba:1969rw,Koba:1969kh}:
\be\label{twisted-cycle}
\Delta(\alpha) \otimes \KN := \Delta(\alpha) \otimes e^{i\pi \phi(\alpha)} \!\!\!\!\! \prod_{1 \leq i < j \leq n} (z_i - z_j)^{2\alpha'\! p_i {\cdot} p_j}.
\ee
Here the phase $\phi(\alpha)$ is chosen such that the local coefficient is real-valued everywhere on $\Delta(\alpha)$, or in other words, equal to $\prod_{i<j} |z_i {-} z_j|^{2\alpha'\! p_i{\cdot}p_j}$, where the vertical lines denote the absolute value. Notice that we identified the mass scale $1/\Lambda^2$ with inverse string tension parameter $\alpha'$, which is conventional in string theory. The \emph{twisted cycle} given in \eqref{twisted-cycle} is technically an element of the \emph{locally-finite} twisted homology $H_{n-3}^\omega := H_{n-3}^{\text{lf}}(\M_{0,n}, \L_\omega)$ since it is non-compact. More details on homology with coefficients in the local system are given in the Appendix~\ref{app:aspects}.

The above twisted homology is Poincar\'e dual to the twisted cohomology $H^{n-3}_{\omega}$ by the following pairing:
\be\label{open-string-amplitude}
\la \Delta(\alpha) \otimes \KN \,|\, \varphi_+ \ra := \int_{\Delta(\alpha)} \!\!\! \KN \; \varphi_+.
\ee
This is the open string amplitude. Here the Koba--Nielsen factor arises as the plane-wave part of the correlation function of vertex operators, while the remainder of the correlator (together with the integration measure) is called $\varphi_+$. The above integral is only defined formally, as it does not generically converge for physical scattering processes. This is a major challenge for explicit computations of open-string integrals \eqref{open-string-amplitude}, in particular because kinematic limits cannot be commuted with the integration. Aspects of analytic continuation of \eqref{open-string-amplitude} are discussed in the Appendix~\ref{app:aspects}.

The number of linearly-independent integrals of the type \eqref{open-string-amplitude} is exactly $\left((n{-}3)!\right)^2$, given that dimensions of both $H_{n-3}^\omega$ and $H^{n-3}_\omega$ are $(n{-}3)!$ \cite{aomoto1987gauss}. In order to put this fact to use we apply Proposition~\ref{proposition} with $U{=}H_{n-3}^{\omega}$, $V{=}X{=}H^{n-3}_\omega$, and $W{=}H^{n-3}_{-\omega}$, giving a decomposition formula expressing the integral \eqref{open-string-amplitude} in terms of an arbitrary basis of integrals over twisted forms $\varphi_{+,a}$:
\be\label{decomposition}
\int_{\Delta(\alpha)} \!\!\! \KN \; \varphi_+ = \sum_{a=1}^{(n-3)!} \la \Phi_{-,a}^\vee \,|\, \varphi_+ \ra_\omega \int_{\Delta(\alpha)} \!\!\! \KN \; \Phi_{+,a}.
\ee
One can make use of it in efficient computations of open-string amplitudes. In particular, specializing to the massless case and choosing the basis $\{\Phi_{+,a}\}_{a=1}^{(n-3)!} \in H^{n-3}_{\omega}$ on the right-hand side to consists of Parke--Taylor forms, there exist algorithmic methods of expanding the corresponding integral around the low-energy limit \cite{Broedel:2013aza,Mafra:2016mcc}. In this case the dual basis $\{\Phi_{-,a}^\vee\}_{a=1}^{(n-3)!} \in H^{n-3}_{-\omega}$ is given by \cite{Mafra:2011nv,Mafra:2011nw}:
\be
\text{PT}(\alpha)^\vee := \frac{d\mu_n}{(z_{\alpha(1)} {-} z_{\alpha(n-1)})(z_{\alpha(n-1)} {-} z_{\alpha(n)})(z_{\alpha(1)} {-} z_{\alpha(n)})} \prod_{i=2}^{n-2} \sum_{j=1}^{i-1} \frac{( p_{\alpha(i)}{+}p_{\alpha(j)})^2}{z_{\alpha(i)}{-}z_{\alpha(j)}} \frac{z_{\alpha(j)}{-}z_{\alpha(n)}}{z_{\alpha(i)}{-}z_{\alpha(n)}}.
\ee
The two bases are orthonormal in the sense that for $(n{-}3)!$ permutations $\widehat{\alpha}$, $\widehat{\beta}$ of the labels $(2,3,\ldots,n{-}2)$ they satisfy
\be\label{intersection-PT-PTvee}
\braket{ \PT(1,\widehat{\alpha},n{-}1,n)^\vee \,|\, \PT(1,\widehat{\beta},n{-}1,n) }_\omega = \delta_{\widehat{\alpha} \widehat{\beta}},
\ee
which will be confirmed later on by direct computation. Note that the orthonormality might not hold for other choices of permutations. We will construct alternative pairs of orthonormal bases in Section~\ref{sec:solving}. There exists a natural homological counterpart of the decomposition \eqref{decomposition}, which allows for expanding open-string integrals in a basis of twisted cycles \cite{Mizera:2017cqs}.\footnote{In fact, the above consideration can be applied to a wide variety of integrals, such as Feynman multi-loop integrals. Recent explorations into this topic were given in \cite{Mastrolia:2018uzb,Frellesvig:2019kgj}.}

For completeness and later reference, let us mention an alternative way of decomposing the integral \eqref{open-string-amplitude} which follows from Proposition~\ref{proposition},
\be\label{open-string-decomposition}
\int_{\Delta(\alpha)} \!\!\! \KN \, \varphi_+ = \sum_{a,b=1}^{(n-3)!} \left( \int_{\Delta(\alpha)} \!\!\! \KN \, \Phi_{+,a} \right) \mathbf{C}_{ab}\, \la \Phi_{-,b} \,|\, \varphi_+ \ra_\omega,
\ee
where $\mathbf{C}_{ab} := \la \Phi_{-,a}^\vee \,|\, \Phi_{+,b}^\vee \ra_\omega$ or alternatively $\mathbf{C}_{ab}^{-1} = \la \Phi_{-,a} \,|\, \Phi_{+,b} \ra_\omega$. In the special case when the kinematics is massless and all $\Phi_{-,a}$, $\Phi_{+,b}$ are Parke--Taylor forms, the intersection matrix $\mathbf{C}$ is often referred to as the \emph{field-theory} KLT matrix \cite{Bern:1998sv,BjerrumBohr:2010hn}. Various closed-form expressions and recursions were obtained in \cite{Bern:1998sv,BjerrumBohr:2010hn,BjerrumBohr:2010ta,BjerrumBohr:2010yc,Carrasco:2016ldy} by other means. In such cases, Mafra, Schlotterer, and Stieberger found that when the left-hand side of \eqref{open-string-decomposition} is a superstring amplitude, $\la \Phi_{-,b} \,|\, \varphi_+ \ra_\omega$ on the right-hand side coincides with super Yang--Mills amplitudes without any $\alpha'$ corrections \cite{Mafra:2011nv}, see also \cite{Mafra:2011nw,Huang:2016tag,Azevedo:2018dgo}.

Using the definitions given in \eqref{intersection-definition} and \eqref{open-string-amplitude} one can prove directly that the low-energy limits, $\alpha' \to 0$ (or equivalently $\Lambda \to \infty$, in which massive modes decouple), of intersection numbers and open string amplitudes coincide. Note that whenever we talk about low-energy limit we specialize to massless external states. As a shortcut, let use the result of \cite{Mafra:2011nw}, who showed that
\be
\lim_{\alpha' \to 0} \braket{\Delta(1,\widehat{\alpha},n{-}1,n) \otimes \KN \,|\, \PT(1,\widehat{\beta},n{-}1,n)^\vee} = \frac{1}{(\alpha')^{n-3}}\,\delta_{\widehat{\alpha}\widehat{\beta}},
\ee
up to a sign. After plugging it into the right-hand side of \eqref{decomposition} and taking the low-energy limit we obtain the anticipated relation
\be\label{open-string-limit}
\lim_{\alpha' \to 0} (\alpha')^{n-3}\!\! \int_{\Delta(\alpha)} \!\!\! \KN \, \varphi_+ = \lim_{\alpha' \to 0} \la \PT(\alpha) \,|\, \varphi_+ \ra.
\ee
We will improve upon this result in Section~\ref{sec:logarithmic}.

\subsubsection{\label{sec:closed-string-amplitudes}Closed-String Amplitudes}

\textsc{In a similar fashion}, scattering amplitudes of closed strings admit a natural interpretation in terms of the twisted theory \cite{Mizera:2017cqs}. One can introduce the antiholomorphic counterpart of twisted cohomology, denoted by $H^{n-3}_{\overbar\omega} := H^{n-3}(\M_{0,n}, \nabla_{\overbar{\omega}})$, whose generic element we will call $\overbar{\vartheta_+}$. It is isomorphic to $H^{n-3}_\omega$ by the pairing
\be\label{closed-string-amplitude}
\la \overbar{\vartheta_+} | \varphi_+ \ra := \int_{\M_{0,n}} \!\!\!\!\! |\KN|^2\; \varphi_+ \wedge \overbar{\vartheta_+} ,
\ee
where $|\KN|^2$ the modulus squared of the Koba--Nielsen factor, $|\KN|^2 := \prod_{i<j} (|z_i {-} z_j|^2)^{2\alpha' p_i {\cdot} p_j}$, which is single-valued on the whole moduli space. Here we made use of the fact that closed string correlator factorizes holomorphically into left- and right-moving contributions, $\varphi_+$ and $\overbar{\vartheta_+}$ respectively. As with the case of open string amplitudes, the integral \eqref{closed-string-amplitude} is defined only formally and needs to be regularized before its evaluation. Analytic continuation of \eqref{closed-string-amplitude} and more details behind the definition of $H^{n-3}_{\overbar{\omega}}$ are given in the Appendix~\ref{app:aspects}.

Naturally, there are only $((n{-}3)!)^2$ linearly-independent pairings of the form \eqref{closed-string-amplitude}. We can put this statement to use by applying the Proposition~\ref{proposition} twice: first with $U {=} H^{n-3}_{\overbar\omega}$, $V{=}X{=}H^{n-3}_{\omega}$, $W {=} H^{n-3}_{-\omega}$, and then with $U{=}W{=}H^{n-3}_{\overbar\omega}$, $V{=}H^{n-3}_{-\overbar{\omega}}$ (with the obvious definition), $X{=}H^{n-3}_{\omega}$, to obtain the decomposition formula:
\be\label{closed-string-decomposition}
\int_{\M_{0,n}} \!\!\!\!\! |\KN|^2\;  \varphi_+ \wedge \overbar{\vartheta_+} = \sum_{a,b=1}^{(n-3)!} \la \overbar{\vartheta_{+}} | \overbar{\Theta_{-,a}^{\,\vee}} \ra_{\overbar\omega}\, \la \Phi_{-,b}^\vee | \varphi_+ \ra_\omega
\int_{\M_{0,n}} \!\!\!\!\! |\KN|^2\; \Phi_{+,b} \wedge \overbar{\Theta_{+,a}}
\ee
in terms of arbitrary bases $\{\overbar{\Theta_{+,a}}\}_{a=1}^{(n-3)!} \in H^{n-3}_{\overbar{\omega}}$ and $\{\Phi_{+,b}\}_{b=1}^{(n-3)!} \in H^{n-3}_{\omega}$ of twisted forms. Recall that the orthonormality condition means $\la \overbar{\Theta_{+,a}} | \overbar{\Theta_{-,b}^\vee} \ra_{\overbar\omega} = \la \Phi_{-,a}^\vee | \Phi_{+,b} \ra_{\omega} = \delta_{ab}$. Here the antiholomorphic intersection numbers have a definition analogous to that in \eqref{intersection-definition}. Since for real kinematics we have ${\cal L}_{\pm\omega} \cong {\cal L}_{\mp\overbar{\omega}}$, they can be expressed in terms of the holomorphic ones as follows:
\be\label{anti-holo-relation}
\la \overbar{\vartheta_+} | \overbar{\vartheta_-} \ra_{\overbar\omega} = \la \vartheta_- | \vartheta_+ \ra_{\omega}.
\ee
Note the exchange of $\vartheta_+$ and $\vartheta_-$ on the right-hand side. Let us mention that \eqref{closed-string-decomposition} can be alternatively expressed in terms of the intersection matrices (and/or their inverses) in a way analogous to \eqref{open-string-decomposition}.

Closed-string amplitudes can be also \emph{homologically split} using the Proposition~\ref{proposition} with $U{=}H^{n-3}_{\overbar{\omega}}$, $V{=}H_{n-3}^{\overbar{\omega}}$, $W{=}H^{\omega}_{n-3}$, $X{=}H^{n-3}_{\omega}$, in terms of two copies of open-string integrals and intersection numbers of twisted cycles.\footnote{Let us remark that homological splitting is a rather generic property of two-dimensional conformal field theories, see, e.g., \cite{Dotsenko:1984nm,Dotsenko:1984ad,10.1093/qmath/38.4.385,Mimachi2003,Mimachi2004,Vanhove:2018elu}. Another example is the partition function of $\text{SL}(2,\C)$ Chern--Simons theory which can be written as a quadratic combination of the $\text{SU}(2)$ Chern--Simons contributions weighted by intersection numbers \cite{Witten:2010cx}.} One of the most important applications of homological splitting are the KLT relations \cite{Kawai:1985xq}, which are obtained as follows. Specializing to the case of integration cycles $\Delta(\alpha)$ we obtain
\be\label{homological-factorization}
\int_{\M_{0,n}} \!\!\!\!\! |\KN|^2\; \varphi_+ \wedge \overbar{\vartheta_+} = \sum_{\alpha, \beta} \left(\int_{\Delta(\alpha)} \!\!\! \overbar{\KN}\; \overbar{\vartheta_+} \right) \mathbf{H}_{\alpha\beta} \left(\int_{\Delta(\beta)} \!\!\! \KN\; \varphi_+ \right),
\ee
where the sum goes over two sets of $(n{-}3)!$ permutations of $n$ labels. The (homological) intersection matrix, often called the \emph{string-theory} KLT matrix, and its inverse are given by
\be
\mathbf{H}_{\alpha\beta} = \la \Delta(\alpha)^\vee \otimes \KN \,|\, \Delta(\beta)^\vee \otimes \overbar{\KN} \ra, \qquad \mathbf{H}_{\alpha\beta}^{-1} = \la \Delta(\alpha) \otimes \KN \,|\, \Delta(\beta) \otimes \overbar{\KN} \ra.
\ee
Further discussion and explicit examples can be found in \cite{Mizera:2017cqs} and the Appendix~\ref{app:aspects}.

As in the open-string case, the field-theory limit of closed string amplitudes is governed by intersection numbers of twisted forms. One way to make it apparent is to use the results of \cite{Mafra:2011nw}, which imply
\be
\lim_{\alpha' \to 0} \braket{ \overbar{\PT(1,\widehat{\alpha},n{-}1,n)} \,|\, \PT(1,\widehat{\beta},n{-}1,n)^\vee} =  \left(\frac{-2\pi i}{\alpha'}\right)^{\!n-3}\! \delta_{\widehat{\alpha}\widehat{\beta}},
\ee
up to a sign. Plugging it on the right-hand side of \eqref{closed-string-decomposition} and taking the low-energy limit we find
\be\label{closed-string-limit}
\lim_{\alpha' \to 0} \!\left(\frac{-\alpha'}{2\pi i}\right)^{\!\!n-3}\!\!\!\! \int_{\M_{0,n}} \!\!\!\!\!\!\!\! |\KN|^2\;  \varphi_+ \wedge \overbar{\vartheta_+} {=} \lim_{\alpha' \to 0} \sum_{\widehat{\alpha}} \la \PT(1,\widehat{\alpha},n{-}1,n) \,|\, \varphi_+ \ra_\omega\, \la \PT(1,\widehat{\alpha},n{-}1,n)^\vee \,|\, \vartheta_+ \ra_\omega,
\ee
where we also applied \eqref{anti-holo-relation}. Note that the sum does not generically simplify under Proposition~\ref{proposition} since the Parke--Taylor forms and their duals are both on the left-hand side of the intersection pairings.\footnote{An alternative way of computing the low-energy limit of closed string integrals is by using%
	\be
	\lim_{\alpha' \to 0} \left(\frac{-\alpha'}{2\pi i}\right)^{\!n-3}\!\! \int_{\M_{0,n}}\!\!\! \iota_\omega(\varphi_+) \wedge \overbar{\vartheta_+},
	\ee
	where $\iota_\omega(\varphi_+)$ has compact support, which allows us to send $|\KN|^2 \to 1$ inside the integrand. Notice that in contrast with the intersection numbers, the integrand involves both holomorphic and antiholomorphic forms, and hence its singularity structure is different before taking the $\alpha' \to 0$ limit.
}
This result will be reconsidered in Section~\ref{sec:logarithmic}.

\subsubsection{Scattering Equations}

\textsc{Building up} on the work on twistor string theory by Witten and others  \cite{Witten:2003nn,Roiban:2004yf,Berkovits:2004jj,Cachazo:2012da,Cachazo:2012kg}, Cachazo, He, and Yuan proposed formulae for scattering amplitudes of massless quantum field theories in arbitrary space-time dimension \cite{Cachazo:2013hca,Cachazo:2013iea}. They take the form of localization integrals on $\M_{0,n}$ supported on a finite set of solution of the so-called \emph{scattering equations} \cite{Cachazo:2013gna}:
\be\label{scattering-equations}
\sum_{\substack{j=1\\ j\neq i}}^{n} \frac{2p_i {\cdot} p_j}{z_i - z_j} = 0
\ee
for $i=1,2,\ldots,n$ (by $\PSL(2,\C)$-invariance only $n{-}3$ of them are linearly independent). Here we consider only massless kinematics, i.e., $m_i^2 =0$ for all $i$. Scattering equations appeared before in the context of string theory in \cite{Fairlie:1972zz,RobertsThesis,Fairlie:2008dg,Gross:1987ar,Gross:1987kza}. They have the following interpretation. 

\begin{figure}[t]
	\centering
	\includegraphics[width=.8\columnwidth]{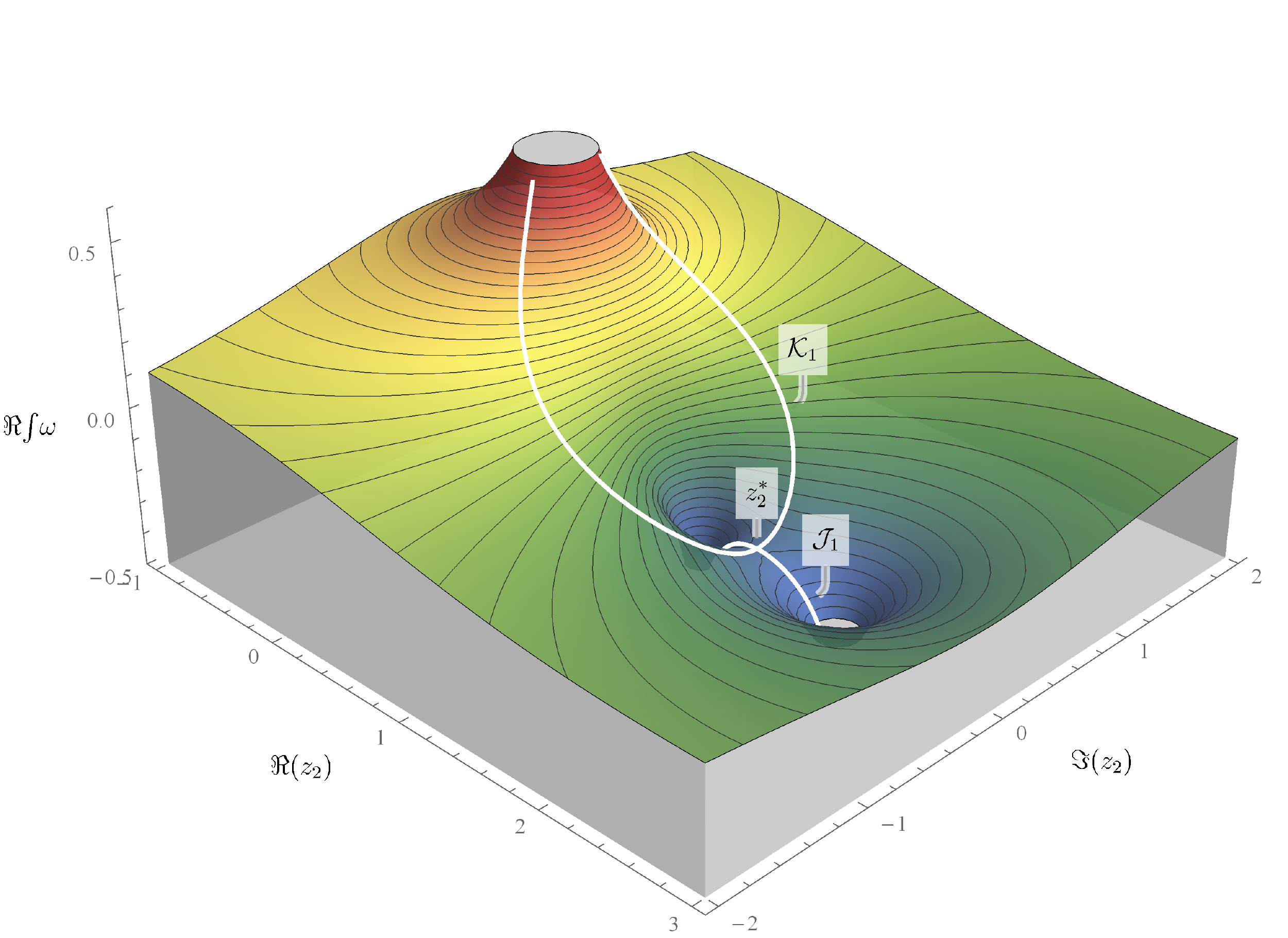}
	\caption{\label{fig:Morse-function}Example Morse function for $n{=}4$. Here we fixed $(z_1,z_3,z_4)=(0,1,2)$ and used massless kinematics with $p_2 {\cdot} p_1 <0$ and $p_2 {\cdot} p_3, p_2 {\cdot} p_4 > 0$. The paths of steepest descent ${\cal J}_1$ and ascent ${\cal K}_1$ (indicated in white) extend from the only critical point located at $z_2^\ast = 2 p_2 {\cdot} p_1 / p_2 {\cdot} (2p_1 {+} p_3)$.}
\end{figure}

Let us consider the moduli space $\M_{0,n}$ furnished with a ``height function'' given by $\Re {\textstyle\int} \omega$, or explicitly
\be\label{Morse-function}
\Re {\textstyle\int} \omega = \Re\, \bigg( \sum_{1 \leq i<j \leq n} \!\!\!\! 2p_i {\cdot} p_j \log(z_i {-} z_j) \bigg),
\ee
which associates a real number to each point in $\M_{0,n}$. It is called a \emph{Morse function}, see Figure~\ref{fig:Morse-function} for an example. Notice that it is single-valued, and hence globally defined, when the kinematics is real. Denoting coordinates on the space $\M_{0,n} {\times} \mathbb{R}$ by $(z_2,z_3,\ldots,z_{n-2},\tau)$, we can consider the following pair of symplectic forms labelled by a sign:
\be
\varpi_\pm = i\, g_{j\overbar{k}}\, dz^j \wedge d\overbar{z}^{k} \,\pm\,
d\left( \Im {\textstyle\int} \omega \right)
\wedge d\tau,
\ee
where the first term is the K\"ahler form associated to the complete positive-definite metric $ds^2 = g_{j\overbar{k}} (dz^j \otimes d\overbar{z}^k +  d\overbar{z}^k \otimes dz^j )$ on $\M_{0,n}$. By varying $\varpi_\pm$ with respect to $dz^j$ and $d\overbar{z}^k$ we obtain the Hamiltonian flow equations:
\be\label{vector-field}
\frac{d\overbar{z}^k}{d\tau} = \pm  g^{j\overbar{k}} \frac{\partial\, \Re\! \int\! \omega}{\partial z^j},
\qquad
\frac{dz^j}{d\tau} = \pm  g^{j\overbar{k}} \frac{\partial\, \Re\! \int\! \omega}{\partial \overbar{z}^k}
\ee
with the conserved Hamiltonian $\Im \!\int\! \omega$ obtained by varying with respect to $d\tau$. Here we used $\partial \Im\! \int\!\omega / \partial z^j = -i\, \partial \Re\! \int\!\omega /\partial z^j$ and $\partial \Im\! \int\!\omega / \partial \overbar{z}^k = i\, \partial \Re\! \int\!\omega /\partial \overbar{z}^k$. In addition, along any trajectory of the vector field defined by \eqref{vector-field} we have:
\be
\frac{d\, \Re \!\int\! \omega}{d \tau} = \sum_{j=2}^{n-2} \frac{\partial \,\Re \!\int\! \omega}{\partial z^j} \frac{d z^j}{d\tau} + \sum_{k=2}^{n-2} \frac{\partial \,\Re \!\int\! \omega}{\partial \overbar{z}^k} \frac{d \overbar{z}^k}{d\tau}  = \pm 2 g^{j\overbar{k}} \frac{\partial \,\Re \!\int\! \omega}{\partial z^j} \frac{\partial \,\Re \!\int\! \omega}{\partial \overbar{z}^k}.
\ee
The right-hand side has a definite sign. Therefore the value of the Morse function $\Re \!\int\! \omega$ is always increasing (decreasing) along any non-constant flow in the $+$ ($-$) case. In particular, the flow reaches its minimum (maximum) when the right-hand side vanishes, i.e., when $\omega=0$, or equivalently \eqref{scattering-equations}. Note that by Cauchy--Riemann equations the critical points of $\int\! \omega$ and its real part \eqref{Morse-function} are the same. Let us assume that each such critical point is isolated and non-degenerate. Since $\int \!\omega$ is holomorphic, a complex version of the Morse lemma implies that there exist local coordinates $(w_2, w_3, \ldots, w_{n-2})$ around each critical point in which it has the expansion
\be
{\textstyle\int} \omega \,=\, {\textstyle\int} \omega \,\big|_{w_i =0} \,+\, \sum_{i=2}^{n-2} w_i^2 \,+\, \ldots.
\ee
Taking its real part and treating $\M_{0,n}$ as a real manifold with coordinates $w_j = x_j + i y_j$ we find
\be
\Re {\textstyle\int} \omega \,=\, \Re {\textstyle\int} \omega\,\big|_{w_i=0} \,+\, \sum_{i=2}^{n-2} x_i^2 \,-\, \sum_{i=2}^{n-2} y_i^2 + \ldots,
\ee
which means that the Morse function has a shape of a saddle near the critical point, i.e., it has exactly $n{-}3$ independent directions in which it decreases and $n{-}3$ in which it increases. One says that it has a \emph{Morse index} $n{-}3$. Denoting the number of critical points with a Morse index $\lambda$ by $\mathrm{M}_\lambda$ we have the following equality \cite{milnor2016morse}:
\be
\chi(\M_{0,n}) = \sum_{\lambda=0}^{2n-6} (-1)^\lambda\, \mathrm{M}_\lambda = (-1)^{n-3}\, \mathrm{M}_{n-3},
\ee
where we used the fact that all $\mathrm{M}_\lambda$ for $\lambda \neq n{-}3$ vanish. Using the value of the Euler characteristic computed before we find that there are exactly $(n{-}3)!$ critical points. In fact, this argument was first given by Aomoto to prove the dimension of twisted cohomology groups and construct their bases in \cite{aomoto1975vanishing,aomoto1987gauss}. The counting was also obtained in \cite{Cachazo:2013iaa,Cachazo:2013gna} by direct computation.

One may wonder what happens upon the inclusion of external masses, $m_i^2 \neq 0$. In this case the Morse function \eqref{Morse-function} is not generic near the infinites of $\M_{0,n}$ and thus the above arguments do not apply (more specifically one should require that the pole locus of $\omega$ coincides with the divisor of $\M_{0,n}$). In the massive case \eqref{Morse-function} is in fact not even well-defined on $\M_{0,n}$ since it is not $\PSL(2,\C)$-invariant. These issues can be alleviated by fixing one of the punctures to infinity. At any rate, we will see shortly that localization of intersection numbers on the critical points will require taking the massless limit.\footnote{In certain restricted cases that exploit planarity one can also treat massive momenta as massless ones in higher space-time dimension, which allows for the use of scattering equations \cite{Dolan:2013isa,Naculich:2015zha,Naculich:2015coa,Cachazo:2015aol}, cf. the discussion in Section~\ref{sec:shift-relations}.}

The middle-dimensional submanifold of $\M_{0,n}$ generated by the downward (upward) flow and passing through the $a$-th critical point is called the path of steepest descent ${\cal J}_a$ (steepest ascent ${\cal K}_a$). Examples of these paths and further discussion of Morse theory is given in Appendix~\ref{app:aspects}. For more general background on Morse (or Picard--Lefschetz) theory see, e.g.,~\cite{milnor2016morse,Witten:2010cx,aomoto2011theory}.

Using the above construction one can prove that CHY formulae are special cases of intersection numbers in the limit where the internal and external kinematics becomes massless, $\Lambda^2 \to 0$ (recall that it also implies $m_i^2 \to 0$ for all $i$) \cite{Mizera:2017rqa}. We start by applying Proposition~\ref{proposition} with $U {=} H^{n-3}_{-\omega}$, $V{=} H^{-\omega}_{n-3}$, $W{=} H^{\omega}_{n-3}$, $X{=}H^{n-3}_{\omega}$ to write:
\be\label{intersection-steepest}
\la \varphi_- | \varphi_+ \ra_\omega = \frac{1}{(2\pi i \Lambda^2)^{n-3}}\sum_{a=1}^{(n-3)!} \left(\int_{{\cal K}_a} \!\!\! \KN^{-1} \varphi_- \right) \left( \int_{{\cal J}_a} \!\!\! \KN\, \varphi_+ \right).
\ee
Here we used the fact that $\{{\cal J}_a\}_{a=1}^{(n-3)!}$ and $\{{\cal K}_a\}_{a=1}^{(n-3)!}$ provide natural orthonormal bases of twisted homologies, namely
\be
\la {\cal J}_a \otimes \KN \,|\,  {\cal K}_b \otimes \KN^{-1} \ra = \delta_{ab},
\ee
since for $a=b$ they intersect at the $a$-th critical point and for $a \neq b$ they do not intersect at all. The orientations and branches of $\KN = \exp \int_\gamma \omega$ and $\KN^{-1} = \exp {-}\!\int_\gamma \omega$ are arbitrary as long as they cancel at the critical points. Note that the precise shape of ${\cal J}_a$ and ${\cal K}_a$ will not be needed in our derivation.

Crucially, the choice of steepest descent and ascent paths means that the right-hand side of \eqref{intersection-steepest} always converges, and hence we can study the asymptotic behaviour as $\Lambda \to 0$ (or equivalently $\alpha' \to \infty$). Note that the role of ${\cal J}_a$ and ${\cal K}_a$ switches between $H^{\omega}_{n-3}$ and $H^{-\omega}_{n-3}$. The leading order comes from the saddle point approximation, see, e.g., \cite{bleistein1986asymptotic}, which up to an overall sign is given by
\begin{align}
\lim_{\Lambda \to 0}\, \la \varphi_- | \varphi_+ \ra_\omega = \frac{1}{(2\pi i \Lambda^2)^{n-3}} \sum_{a=1}^{(n-3)!} & \left( \sqrt{\frac{ (2\pi)^{n-3} }{\det \left[\partial^2 \!\int\! \omega / \partial z_i \partial z_j \right]}} \KN^{-1} \lim_{\Lambda \to 0}\widehat{\varphi}_- \right)_{\!\!z_i = z_i^{(a)}} \nn\\
&\times \left( \sqrt{\frac{ (2\pi)^{n-3} }{\det \left[-\partial^2 \!\int\! \omega / \partial z_i \partial z_j \right]}} \KN \lim_{\Lambda \to 0}\widehat{\varphi}_+ \right)_{\!\!z_i = z_i^{(a)}}.
\end{align}
Here each integral localizes on a saddle point, whose position is denoted by $(z_2^{(a)}, z_3^{(a)}, \ldots, z_{n-2}^{(a)})$. Recall the notation $\widehat{\varphi}_{\pm} d^{n-3}z := \varphi_\pm$. Since on the saddle point we have $\KN {\times} \KN^{-1} = 1$ and the Hessians combine, we can rewrite this limit more concisely as
\begin{align}
\lim_{\Lambda \to 0}\, \la \varphi_- | \varphi_+ \ra_\omega &= \frac{1}{(-\Lambda^2)^{n-3}} \sum_{a=1}^{(n-3)!} \frac{1}{\det \left[\partial^2 \!\int\! \omega / \partial z_i \partial z_j \right]} \lim_{\Lambda \to 0} \widehat{\varphi}_-\, \widehat{\varphi}_+ \nn\\
&= \frac{1}{(-\Lambda^2)^{n-3}} \Res_{\omega = 0} \left(\frac{\lim_{\Lambda \to 0} \varphi_-\, \widehat{\varphi}_+}{\prod_{i=2}^{n-2} \omega_i}\right),\label{CHY-formula}
\end{align}
up to a sign. In the second line we rewrote it as the CHY formula localizing on the solutions of the scattering equations \cite{Cachazo:2013hca}, where $\omega =: \sum_{i=2}^{n-2} \omega_i dz_i$. To emphasize the massless limit, $\Lambda \to 0$, we explicitly kept it acting on the twisted forms, though in practice they usually have homogeneous scaling in $\Lambda$. We will additionally see in Section~\ref{sec:logarithmic} that when the twisted forms are logarithmic, the relation between their intersection number and CHY formula becomes exact.

We warn the reader that for twisted forms $\varphi_\pm$ with non-simple poles the CHY formula generally produces unphysical (non-unitary) kinematic poles, which are not present in intersection numbers, as is already clear from the $\Lambda\to 0$ limit of \eqref{example-235}.

One can restate the above formulae by defining the massless limit, $\Lambda \to 0$, of the twisted cohomologies $H^{n-3}_{\pm\omega} \cap \Omega^{n-3,0}(\ast\partial\M_{0,n})$, which both degenerate to
\be\label{cohomology-limit}
H^{n-3,0}(\M_{0,n},\, \omega\wedge) := \frac{\Omega^{n-3,0}(\ast \partial\M_{0,n})}{\omega \wedge \Omega^{n-4,0}(\ast \partial\M_{0,n})}.
\ee
Self-duality of this cohomology group is given by the CHY integral \eqref{CHY-formula}, understood as a limit of the intersection number from \eqref{intersection-definition}. Cohomology groups of this type are important in the study of combinatorial and algebraic properties of hyperplane arrangements and braid groups \cite{Orlik1980,Esnault1992,Schechtman:1994bb,orlik2001arrangements}.

One can also construct other proofs of the localization on the support of scattering equations. For instance, at $n=4$ we can use the residue formula for intersection numbers \eqref{n-4-result} together with the fact that $\nabla_\omega^{-1} \varphi_+ \to \widehat{\varphi}_{+}/\omega_2$ in the massless limit to find:
\begin{align}
\lim_{\Lambda\to 0} \la \varphi_- | \varphi_+ \ra_\omega &= \frac{1}{\Lambda^2} \sum_{i=1,3,4} \Res_{z_2 = z_i}\! \left( \frac{\lim_{\Lambda\to 0} \varphi_- \, \widehat{\varphi}_+}{\omega_2}  \right)\nn\\
&= - \frac{1}{\Lambda^2} \Res_{\omega =0} \! \left( \frac{\lim_{\Lambda\to 0} \varphi_- \, \widehat{\varphi}_+}{\omega_2}  \right),
\end{align}
where in the second line we used the fact that $\varphi_\pm$ have only poles on $\partial\M_{0,n}$, which allowed us to enclose the remaining pole at $\omega_2 =0$ by the residue theorem. Similar derivations can be repeated at higher-$n$, though the details become considerably more complicated, cf. \cite{Dolan:2013isa}, which is why the proof using Morse theory is the preferred one.

Tempting as it might be, computation of intersection numbers by expanding the right-hand side of \eqref{intersection-steepest} around the saddle points does not seem practical, especially given limited understanding of the asymptotic methods for higher-dimensional integrals at subleading orders \cite{doi:10.1098/rspa.1997.0122,delabaere2002,delabaere2010singular,soton377148}. In fact, already the localization formula \eqref{CHY-formula} is not practical for direct computations beyond $n \leq 5$, since by the Abel--Ruffini theorem positions of the saddle points cannot be determined analytically. There has been considerable progress in solving the scattering equations numerically \cite{Cachazo:2013gna,Weinzierl:2014vwa,Dolan:2014ega,He:2014wua,Huang:2015yka,Sogaard:2015dba,Bosma:2016ttj,Zlotnikov:2016wtk,Du:2016fwe,Cachazo:2016ror,Farrow:2018cqi,Liu:2018brz}, as well as evaluating the CHY formula without solving them explicitly \cite{Dolan:2013isa,Kalousios:2013eca,Dolan:2014ega,Naculich:2014rta,Naculich:2014naa,Kalousios:2015fya,Cachazo:2015nwa,Baadsgaard:2015voa,Baadsgaard:2015ifa,Baadsgaard:2015hia,Cardona:2015eba,Cardona:2015ouc,Dolan:2015iln,Lam:2016tlk,Gomez:2016bmv,Huang:2016zzb,Bjerrum-Bohr:2016juj,Cardona:2016gon,Du:2016wkt,Chen:2016fgi,Chen:2017edo,Huang:2017ydz,Zhou:2017mfj,Gao:2017dek,He:2018pue}. In Section~\ref{sec:recursion-relations} we will present recursion relations for intersection numbers which---in the massless limit---allow to evaluate CHY formulae analytically.

\subsection{\label{sec:logarithmic}Logarithmic Forms, Boundaries, and Dihedral Coordinates}

\textsc{Ever since the work of Deligne} \cite{deligne1970equations} it has been appreciated that differential forms with logarithmic singularities play a special role in the theory of complex manifolds (see also \cite{saito1980theory}), and those relevant for scattering amplitudes are certainly no exception.\footnote{This statement extends to scattering amplitudes computed on spaces other than $\M_{g,n}$, see, e.g., \cite{ArkaniHamed:2012nw,Arkani-Hamed:2013jha}.}

Recall that a form on ${\cal M}_{0,n}$ is called \emph{logarithmic} if it has at most simple poles along the boundary divisor $\partial \M_{0,n}$. The standard way of constructing such logarithmic forms is to consider a wedge of $n{-}3$ building blocks $\lambda_{ij} = d\log(z_i - z_j)$ such that the set of edges connecting the labels $i,j$ forms a tree, see, e.g.,~\cite{aomoto1987gauss,terasoma_2002}. Note that having simple poles in $z_i{-}z_j$ is a necessary but not a sufficient condition for logarithmicity. However, neither the degree nor the position of poles are preserved under cohomology relations. In fact, any (massless) twisted form can be brought into a logarithmic form by the basis decomposition formula \eqref{decomposition}, at a cost of introducing possible explicit dependence on $2p_i {\cdot} p_j$ and $\Lambda^2$. For instance, twisted forms in bosonic string theory are non-logarithmic, but are cohomologous to logarithmic forms with kinematic poles of the type $(p_i {+} p_j)^2 {-} \Lambda^2$, thus signalling propagation of a tachyon.

As discussed in Section~\ref{sec:local-systems}, $\PSL(2,\C)$-invariance implies that logarithmic twisted forms are only allowed when the \emph{external} states are massless, which is the case we will focus on in this subsection. We will see explicitly that this special class of forms also corresponds to all \emph{internal} states being massless.

Since logarithmic forms are defined by their behaviour near $\partial \M_{0,n}$, we first need to study the boundary structure of the moduli space. The standard way of understanding it is through the Deligne--Mumford compactification $\overbar{\M}_{0,n}$ \cite{PMIHES_1969__36__75_0}, which includes nodal surfaces as the limiting configurations in which two or more punctures coalesce on the Riemann surface. Beyond $n{=}4$ this cannot be achieved in a global set of coordinates---such as $(z_2, z_3, \ldots, z_{n-2})$ used so far---and one needs to find local coordinates suitable for studying each degeneration. This procedure is known as an (iterated) \emph{blow-up}. Here we follow the construction of Brown \cite{Brown:2009qja}, who introduced an open covering of the compactified moduli space with certain subspaces $\M_{0,n}^\delta \subset \overbar{\M}_{0,n}$ we will describe now.

For the sake of simplicity let us first consider the case when $\delta$ is the canonical permutation of $n$ labels, $\delta = (12\cdots n)$, and introduce the $\PSL(2,\C)$-invariant cross-ratios:
\be
u_{ij} := \frac{(z_i - z_{j+1})(z_{i+1} - z_j)}{(z_i - z_j)(z_{i+1} - z_{j+1})}.
\ee
These are in fact the variables in which the original Koba--Nielsen integrals were defined \cite{Koba:1969kh}. Let us consider the set $\chi_\delta$ of all possible chords of an $n$-gon with vertices ordered according to $\delta$. There are $n(n{-}3)/2$ of them. To each chord $\{ i,j \} \in \chi_\delta$ extending from the $i$-th to $j$-th vertex we associate the cross-ratio $u_{ij}$. One can show that all identities between such variables are given by constraints of the form:
\be\label{crossing-identities}
\prod_{\substack{p \leq i < q\\ r \leq j < s}}  u_{ij} + \prod_{\substack{q \leq k < r\\ s \leq \ell < p}}  u_{k\ell} = 1,
\ee
for any choice of distinct labels $p,q,r,s$. Note that the inequalities above are defined in the cyclic sense. The identities \eqref{crossing-identities} reduce the number of independent cross-ratios to $n{-}3$ that give coordinates on the space $\M_{0,n}^\delta$ \cite{Brown:2009qja}. Note that there are multiple ways of choosing the independent variables. They can be classified by $C_{n-2} = \frac{1}{n{-}1}\binom{2(n{-}2)}{n{-}2}$ triangulations of an $n$-gon with $n{-}3$ chords, where $C_k$ is the $k$-th Catalan number. The set of $u_{ij}$'s corresponding to such a choice defines \emph{dihedral coordinates} on the moduli space, and the hypersurfaces $\{u_{ij} = 0\}$ correspond to its codimension-one boundaries. Higher-codimension boundaries are obtained by intersection of these hypersurfaces and in particular the maximal-codimension boundaries ($C_{n-2}$ vertices of $\partial \M_{0,n}^\delta$) are given by
\be
\bigcap_{\{i,j\} \in {\cal T}_n} \!\!\{ u_{ij} =0 \},
\ee
each labelled by a triangulation ${\cal T}_n$ of the $n$-gon. The whole boundary divisor $\partial \M_{0,n}$ is obtained by stitching together $\partial \M_{0,n}^\delta$ for $(n{-}1)!/2$ inequivalent permutations $\delta$ (modulo the dihedral group). Note that the definition of $u_{ij}$ and constraints \eqref{crossing-identities} need to be changed for each different choice of ordering $\delta$. There is certain redundancy in covering of $\overbar{\M}_{0,n}$ with copies of $\M_{0,n}^\delta$: for example each of the $(2n-5)!!$ vertices is covered $2^{n-3}$ times. We refer the reader to \cite{Brown:2009qja} for more details on the dihedral extension of the moduli space (see also \cite{Mafra:2011nw,delaCruz:2017zqr,Arkani-Hamed:2017mur,Brown:2018omk} for other relevant literature).

In order to understand physical significance of the dihedral coordinates, let us compute the form of the local system ${\cal L}_\omega$ in these variables. For a given choice of $\delta$ we have:
\be\label{omega-dihedral}
\omega = \frac{1}{\Lambda^2} \sum_{\{i,j\} \in \chi_\delta}\!\!\! (p_{i{+}1} {+} p_{i{+}2} {+} \ldots {+} p_j)^2\, d\log u_{ij},
\ee
where the momentum variable could also be written as $(p_{j+1}{+}p_{j+2}{+}\ldots{+}p_{i})^2$ by momentum conservation and hence it does not depend on the orientation of the chord, similarly to $u_{ij} = u_{ji}$ (all labels are modulo $n$). Here we have written the form on the $n(n{-}3)/2$-dimensional space spanned by all $u_{ij}$'s, which can be projected down to $\M_{0,n}^\delta$ after imposing \eqref{crossing-identities}. Note that since \eqref{crossing-identities} are in general non-linear, this might complicate the form of $\omega$, but preserves its singularity properties.

The local system in those variables has a simple meaning in terms of kinematic poles it can produce. To each chord $\{i,j\}$ of an $n$-gon it associates a kinematic pole in $(p_{i{+}1} {+} p_{i{+}2} {+} \ldots {+} p_j)^2$.
One can repeat this procedure by adding chords that do not cross the previous ones until one obtains a triangulation of the $n$-gon with $n{-}3$ non-crossing chords. The resulting kinematic pole is that of the corresponding trivalent graph:
\be
\begin{aligned}
	\includegraphics[scale=1]{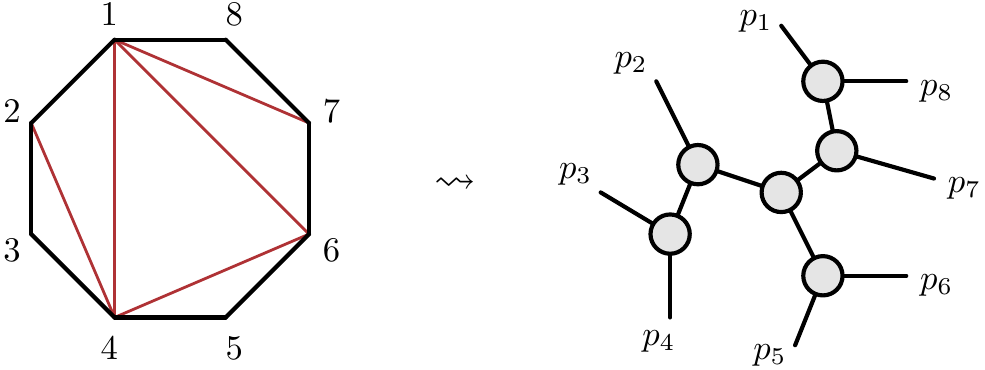}
\end{aligned}
\ee

Let us make the above statements precise by evaluating intersection numbers of logarithmic twisted forms, which will result in an expression summing over residues around each vertex belonging to $\partial \M_{0,n}$, or equivalently all trivalent trees.\footnote{Recall that the residue around $v = \bigcap_{\{i,j\}\in {\cal T}_n} \{ u_{ij} = 0\}$ is given by an integral over the tubular contour $\bigwedge_{\{i,j\}\in {\cal T}_n} \allowbreak \{ |u_{ij}| = \varepsilon\}/(2\pi i)$, which in our case can be obtained by $n{-}3$ consecutive one-dimensional residues, see, e.g.,~\cite{griffiths2014principles}.} Our derivation follows the method of Cho and Matsumoto \cite{cho1995,matsumoto1998}.

\begin{theorem}\label{theorem-21}
	When both twisted forms $\varphi_\pm$ are logarithmic, their intersection number evaluates to
	\be\label{intersection-logarithmic}
	\la \varphi_- | \varphi_+ \ra_\omega = \sum_{v \in \partial \M_{0,n}} \frac{\Res_{v}(\varphi_-) \Res_{v}(\varphi_+)}{\prod_{a=1}^{n-3} p^2_{I_{a}^v}},
	\ee
	where the sum goes over all vertices $v$ in the boundary divisor $\partial \M_{0,n}$. For each $v$, $\{ p_{I_{a}^v}^2 \}_{a=1}^{n-3}$ denotes the set of factorization channels associated to this vertex.
\end{theorem}
\begin{proof}
	The proof proceeds by explicitly constructing the form $\iota_\omega(\varphi_+) \in H^{n-3}_c(\M_{0,n},\nabla_{\omega})$ cohomologous to $\varphi_+$ and with compact support along all components of the boundary divisor $\partial\M_{0,n}$. As a consequence, the integral of $\varphi_- \wedge \iota_\omega(\varphi_+)$ over $\overbar{\M}_{0,n}$ will localize on the maximal-codimension components (vertices) of $\partial\M_{0,n}$. We start with the cases $n=4,5$, from which the general pattern will be clear.
	
	Since $\M_{0,4}$ is one-dimensional, it only has zero-dimensional boundaries, $v \in \partial \M_{0,4}$. For the $v$-th boundary component we construct the regulating function $h_v$, which equals to $1$ in the infinitesimal neighbourhood $V_v$ of this boundary, equals to $0$ on the outside of another infinitesimal neighbourhood $U_v \supsetneq V_v$, and smoothly interpolates between the two values in $U_v {\setminus} V_v$ with a topology of an annulus. We then write:
	\begin{align}
	\iota_\omega(\varphi_+) = \varphi_+ - \nabla_{\omega} \Big( \sum_{v} h_{v} \psi_v \Big),
	\end{align}
	which by construction is cohomologous to $\varphi_+$. Here $\psi_v$ denotes the unique holomorphic solution of $\nabla_\omega \psi_v = \varphi_+$ in the small neighbourhood $U_v$ of the $v$-th boundary. After acting with $\nabla_\omega$ on the terms in the bracket we find:
	\be\label{compact-form}
	\iota_\omega(\varphi_+) = \Big(1-\sum_v h_v \Big) \varphi_+ - \sum_v dh_v \psi_v.
	\ee
	Clearly both terms vanish in neighbourhoods $V_v$ for every $v$ and hence the form $\iota_\omega(\varphi_+)$ has compact support. Let us evaluate $\psi_v$ for a logarithmic form $\varphi_+$. Denoting the local coordinate around the $v$-th boundary by $u_I$, the differential equation $\nabla_\omega \psi_v = \varphi_+$ becomes to leading order:
	\be
	d\psi_v + \frac{1}{\Lambda^2} \left( p_I^2\, \frac{du_{I}}{u_{I}} + \ldots \right) \wedge \psi_{v} = \Res_{v} \left(\varphi_+\right) \frac{du_{I}}{u_{I}} + \ldots,
	\ee
	where the ellipsis denotes terms subleading in $u_{I}$, and the residue is around $u_{I}=0$. Hence we find:
	\be
	\psi_{v} = \Lambda^2\, \frac{\Res_{v} \left(\varphi_+\right)}{p_I^2} + \ldots
	\ee
	with subleading terms of order ${\cal O}(u_{I})$. Let us plug in \eqref{compact-form} into the Definition~\ref{definition-intersection-number}:
	\be
	\braket{ \varphi_- | \varphi_+ }_\omega = -\frac{1}{2\pi i \Lambda^2} \int_{\M_{0,4}} \varphi_- \wedge \left( \Big(1-\sum_v h_v \Big) \varphi_+ - \sum_v dh_v \psi_v\right)
	\ee
	The first term vanishes since it is a wedge of two holomorphic top-forms, while the second term has only support on the annuli $U_v {\setminus} V_v$, which yields:
	\be
	\braket{ \varphi_- | \varphi_+ }_\omega = -\frac{1}{2\pi i \Lambda^2} \sum_{v} \int_{U_v {\setminus} V_v} d \Big( \varphi_- \wedge h_v \psi_v \Big),
	\ee
	where we wrote it as a total differential since $d\varphi_- = d\psi_v = 0$. Hence we can use Stokes' theorem to localize on $\partial(U_v{\setminus} V_v) = \partial U_v - \partial V_v$, out of which only the component $\partial V_v$ contributes since $h_v(\partial V_v) =1$, while $h_v(\partial U_v) =0$. Given that $\varphi_-$ is logarithmic we find:
	\be
	\braket{ \varphi_- | \varphi_+ }_\omega = \, \sum_{v} \left( \frac{1}{2\pi i}\oint_{\partial V_v} \!\!\varphi_- \right) \frac{\Res_{v} \left(\varphi_+\right)}{p_I^2},
	\ee
	which, after writing the term in the brackets as $\Res_v \left( \varphi_- \right)$, is the required result for $n=4$.
	
	In the case of $\M_{0,5}$, we have one-dimensional boundaries $s \in \partial\M_{0,5}$, as well as zero-dimensional ones $v \in \partial\M_{0,5}$, which can be written as $v = s \cap \tilde{s}$ subject to the compatibility conditions described above. For each $s$ we introduce the regulating functions $h_s$ with a similar definition to those used in the $n{=}4$ case, where now $U_s$, $V_s$ are two-complex-dimensional neighbourhoods of $s$. We construct the following two-form cohomologous to $\varphi_+$: 
	\be\label{two-form-compact}
	\iota_\omega(\varphi_+) = \varphi_+ - \nabla_\omega \left( \sum_{\tilde{s}} h_{\tilde{s}} \Big( \sum_{s \prec \tilde{s}} dh_{s}\, \psi_{s\tilde{s}} + \prod_{s \prec \tilde{s}} (1-h_{s}) \, \psi_{\tilde{s}} \Big) \right),
	\ee
	where $\psi_s$ denotes the holomorphic one-form obtained by solving $\nabla_\omega \psi_s = \varphi_+$ in the small neighbourhood $U_s$ (and similarly for $\psi_{\tilde{s}}$), while $\psi_{s\tilde{s}}$ denotes the holomorphic function given by the solution of $\nabla_\omega \psi_{s\tilde{s}} = \psi_{s} - \psi_{\tilde{s}}$ on $U_s \cap U_{\tilde{s}}$. Here we made use of some auxiliary ordering of the hypersurfaces, written as $s \prec \tilde{s}$, in order to make the above expression more concise. Let us confirm that \eqref{two-form-compact} has compact support by first evaluating it on every $U_{\tilde{s}} {\setminus} \cup_{s \neq \tilde{s}} U_{s}$, where $h_{s} = dh_{s} = 0$ and hence:
	\be\label{compact-two-form-UmU}
	\iota_\omega(\varphi_+) = (1-h_{\tilde{s}}) \varphi_+ - dh_{\tilde{s}} \wedge \psi_{\tilde{s}},
	\ee
	which indeed vanishes on $V_{\tilde{s}} \setminus \cup_{s \neq \tilde{s}} U_{s}$. Moreover, \eqref{compact-two-form-UmU} does not contribute to the intersection number due to the presence of the holomorphic forms $\varphi_+$ and $\psi_{\tilde{s}}$ which vanish after wedging with $\varphi_-$ in the Definition~\ref{definition-intersection-number}. Since a similar thing happens for $\M_{0,5} {\setminus} \cup_{s} U_s$, where $\iota_{\omega}(\varphi_+) = \varphi_+$, it remains to check the regions $U_s \cap U_{\tilde{s}}$ for $s \neq \tilde{s}$. There \eqref{two-form-compact} reduces to
	\be\label{compact-two-form-UU}
	\iota_\omega(\varphi_+) = (1-h_s)(1-h_{\tilde{s}}) \varphi_+ - (1-h_{\tilde{s}}) dh_s \wedge \psi_s - (1-h_s) dh_{\tilde{s}} \wedge \psi_{\tilde{s}} + dh_{s}\wedge dh_{\tilde{s}}\, \psi_{s\tilde{s}}.
	\ee
	On $V_s \cap V_{\tilde{s}}$ all terms vanish and therefore the form \eqref{two-form-compact} has compact support. The final term in \eqref{compact-two-form-UU} is the only one not containing holomorphic contributions, and hence the intersection number localizes on the product of annuli $(U_s {\setminus} V_s) \wedge (U_{\tilde{s}} {\setminus} V_{\tilde{s}})$ around each vertex $v = s \cap \tilde{s} \in \partial \M_{0,5}$. Therefore we need to evaluate $\psi_s$, $\psi_{\tilde{s}}$, and $\psi_{s\tilde{s}}$. As before, using local coordinates $u_I$ and $u_{\tilde{I}}$ around $s$ and $\tilde{s}$ respectively, we obtain:
	\be
	\psi_{s} = \Lambda^2\, \frac{\Res_{s} \left(\varphi_+\right)}{p_I^2}  + \ldots, \qquad
	\psi_{\tilde{s}} = \Lambda^2\, \frac{\Res_{\tilde{s}} \left(\varphi_+\right)}{p_{\tilde{I}}^2} + \ldots.
	\ee
	Note that they are both one-forms, say in $du_{\tilde{I}}$ and $du_{I}$ respectively. In order to solve $\nabla_{\omega} \psi_{s\tilde{s}} = \psi_s - \psi_{\tilde{s}}$ we expand both sides around $(u_I, u_{\tilde{I}})=(0,0)$, giving
	\be
	d\psi_{s\tilde{s}} + \frac{1}{\Lambda^2} \left( p_I^2 \frac{du_I}{u_I} + p_{\tilde{I}}^2 \frac{du_{\tilde{I}}}{u_{\tilde{I}}} + \ldots \right) \wedge \psi_{s\tilde{s}} =  \Lambda^2\, \frac{\Res_{v} \left(\varphi_+\right)}{p_{\tilde{I}}^2} \frac{du_I}{u_I} + \Lambda^2\, \frac{\Res_{v} \left(\varphi_+\right)}{p_I^2} \frac{du_{\tilde{I}}}{u_{\tilde{I}}} + \ldots,
	\ee
	where the residue $\Res_{v}(\varphi_+)$ denotes $\Res_{\tilde{s}} ( \Res_{s} (\varphi_+)) = -\Res_{s} ( \Res_{\tilde{s}} (\varphi_+))$ in the Leray sense. Thus, we obtain the solution
	\be
	\psi_{s\tilde{s}} = \Lambda^4 \frac{\Res_v (\varphi_+)}{p_I^2\, p_{\tilde{I}}^2}.
	\ee
	Recall that the compatibility condition for $s \cap \tilde{s}$ to be a vertex is that the labels in $I$ and $\tilde{I}$ are either disjoint or one is contained within the other. Finally, we can plug all the ingredients into the Definition~\ref{definition-intersection-number},
	\begin{align}
	\braket{ \varphi_- | \varphi_+ }_\omega &= \frac{1}{(2\pi i \Lambda^2)^2} \sum_{v = s \cap \tilde{s}} \int_{(U_s {\setminus} V_s) \wedge (U_{\tilde{s}} {\setminus} V_{\tilde{s}})} \!\!\!\!\! \varphi_- \wedge dh_{s}\wedge dh_{\tilde{s}}\, \psi_{s\tilde{s}}\nn\\
	&=  \sum_{v = s \cap \tilde{s}} \left( \frac{1}{(2\pi i)^2} \oint_{\partial V_{\tilde{s}}} \oint_{\partial V_{s}} \varphi_- \right) \frac{\Res_v (\varphi_+)}{p_I^2\, p_{\tilde{I}}^2},
	\end{align}
	where we used the same arguments as in the $n{=}4$ case to rewrite the first line as a total derivative and localize on $\partial V_{\tilde{s}} \wedge \partial V_s$ after using Stokes' theorem twice. The term in the brackets equals to $\Res_v (\varphi_-)$ and hence we obtain the required result for $n{=}5$.
	
	The proof strategy in the general-$n$ case is entirely analogous. One starts by constructing the form $\iota_\omega(\varphi_+)$ cohomologous to $\varphi_+$, but has compact support near all components of the boundary divisor $\partial \M_{0,n}$. The explicit expression for $\iota_\omega(\varphi_+)$ can be found in \cite{cho1995} and we will not repeat it here. It suffices to extract its non-holomorphic component, which, as before, has support only near the vertices $v = H_1 \cap H_2 \cap \cdots \cap H_{n-3} \in \partial \M_{0,n}$, given by the $(n{-}3)$-annuli $(U_{H_1} {\setminus} V_{H_1}) \wedge (U_{H_2} {\setminus} V_{H_2}) \wedge \cdots \wedge (U_{H_{n-3}} {\setminus} V_{H_{n-3}})$, where $H_a$ denote the hypersurfaces of $\partial \M_{0,n}$. We find
	\be
	\iota_\omega(\varphi_+) = dh_{H_1} \wedge dh_{H_2} \wedge \cdots \wedge dh_{H_{n-3}}\, \psi_{H_1 H_2 \cdots H_{n-3}} + \ldots,
	\ee
	where the ellipsis denotes holomorphic contributions. In dihedral coordinates $(u_{I_1^v}, u_{I_2^v}, \ldots,\allowbreak u_{I_{n-3}^v})$ around $v$ the function $\psi_{H_1 H_2 \cdots H_{n-3}}$ becomes to leading order
	\be
	\psi_{H_1 H_2 \cdots H_{n-3}} = \Lambda^{2(n-3)} \frac{\Res_v (\varphi_+)}{\prod_{a=1}^{n-3} p_{I_{a}^v}^2} + \ldots.
	\ee
	Recall that a compatibility condition requires that all pairs of labels in $\{I_a^v\}_{a=1}^{n-3}$ are either disjoint or one is contained within the other. After using the Definition~\ref{definition-intersection-number}, the result localizes on the $(n{-}3)$-tori around each vertex $v \in \partial \M_{0,n}$:
	\be
	\braket{ \varphi_- | \varphi_+ }_{\omega} = \sum_{v = H_1 \cap H_2 \cap \cdots \cap H_{n-3}} \left(\frac{1}{(2\pi i)^{n-3}} \oint_{\partial V_{H_{n-3}}} \!\!\!\!\!\!\!\!\!\! \cdots \; \oint_{\partial V_{H_{2}}} \oint_{\partial V_{H_{1}}} \varphi_- \right) \frac{\Res_v (\varphi_+)}{\prod_{a=1}^{n-3} p_{I_{a}^v}^2}.
	\ee
	This is the formula \eqref{intersection-logarithmic}.
\end{proof}

It is expected that a similar derivation can be made for intersection numbers of non-logarithmic forms, though more care is required in ensuring that the Laurent expansion of $\nabla^{-1}_\omega \varphi_+$ near each $H_a$ is valid everywhere along this hypersurface, see, e.g., \cite{Frellesvig:2019kgj} for some examples. 

Notice that for $n=4$ as a consequence of $\varphi_\pm$ not having poles in places other than $\partial\M_{0,4} = \{z_1, z_3, z_4\}$ in the standard $\SL(2,\C)$ fixing, the residues in the numerators of \eqref{intersection-logarithmic} corresponding to the $(p_1{+}p_2)^2$-, $(p_2{+}p_3)^2$-, and $(p_2{+}p_4)^2$-channels add up to zero:
\be
\Res_{z_2 = z_1} (\varphi_\pm) + \Res_{z_2 = z_3} (\varphi_\pm) + \Res_{z_2 = z_4} (\varphi_\pm) = 0.
\ee
This is the moduli space interpretation of the relations found by Bern, Carrasco, and Johansson (BCJ) called the \emph{kinematic Jacobi relations} \cite{Bern:2008qj}. It would be interesting to investigate the interplay between the formula \eqref{intersection-logarithmic} and BCJ relations further.

There are two important features of the intersection numbers \eqref{intersection-logarithmic}: they are symmetric under the exchange $\varphi_- \leftrightarrow \varphi_+$, and are independent of $\Lambda$ (assuming $\varphi_\pm$ are). This straight away gives us the following result.

\begin{theorem}\label{theorem-22}
	For logarithmic twisted forms $\varphi_+$ and $\vartheta_+$ the low-energy theory limit of open-string amplitudes is given by
	\be\label{proof-open}
	\lim_{\alpha' \to 0} (\alpha')^{n-3} \int_{\Delta(\alpha)} \!\!\! \KN\, \varphi_+ \;=\; \la \PT(\alpha) | \varphi_+ \ra_\omega,
	\ee
	while the low-energy limit of closed-string amplitudes is given by
	\be\label{proof-closed}
	\lim_{\alpha' \to 0} \left(\frac{-\alpha'}{2\pi i}\right)^{\!n-3}\!\! \int_{\M_{0,n}} \!\!\!\!\! |\KN|^2\; \varphi_+ \wedge \overbar{\vartheta_+} \;=\; \la \vartheta_+ | \varphi_+ \ra_\omega,
	\ee
	up to overall signs. Intersection numbers of such twisted forms coincide with the Cachazo--He--Yuan formula,
	\be\label{proof-CHY}
	\left(-\alpha'\right)^{n-3} \Res_{\omega = 0} \left( \frac{\vartheta_+\, \widehat{\varphi}_+}{\prod_{i=2}^{n-2} \omega_i} \right) \;=\; \la \vartheta_+ | \varphi_+ \ra_\omega
	\ee
	assuming that $\varphi_+$ and $\vartheta_+$ have homogeneous scaling in $\alpha'$.
\end{theorem}
\begin{proof}
	We identify $\alpha' = 1/\Lambda^2$. Around each codimension-one boundary of $\M_{0,n}$ the integrals on the left-hand sides of \eqref{proof-open} and \eqref{proof-closed} behave as $\sim 1/\alpha'$, and hence to leading order in $\alpha' \to 0$ the integrals diverge as $\sim 1/(\alpha')^{n-3}$, which after normalizing by $(\alpha')^{n-3}$ is the same scaling as that of the intersection number of twisted forms. Therefore \eqref{open-string-limit} implies \eqref{proof-open}, since $\PT(\alpha)$ is logarithmic. Using \eqref{closed-string-limit}, the symmetry of the intersection pairing, and the fact that $\text{PT}(\alpha)^\vee$ is logarithmic we have:
	\be
	\lim_{\alpha' \to 0} \left(\frac{-\alpha'}{2\pi i}\right)^{\!n-3} \!\!\int_{\M_{0,n}} \!\!\!\!\! |\KN|^2\; \varphi_+ \wedge \overbar{\vartheta_+} = \sum_{\widehat{\alpha}} \la \vartheta_+ | \PT(1,\widehat{\alpha},n{-}1,n)^\vee \ra_\omega\, \la \PT(1,\widehat{\alpha},n{-}1,n) \,|\, \varphi_+ \ra_\omega,
	\ee
	which after using the Proposition~\ref{proposition} with $U{=}W{=}H^{n-3}_{-\omega}$ and $V{=}X{=}H^{n-3}_{\omega}$ implies \eqref{proof-closed}. Finally, \eqref{proof-CHY} follows from \eqref{CHY-formula} and the $\alpha'$-homogeneity property.
\end{proof}

One can also give proofs of the above statements by direct computation of string theory integrals in the $\alpha' \to 0$ limit along the lines of those given in the Appendix~\ref{app:punctures}.\footnote{An alternative---but equivalent---proof of the equality between the low-energy limit of closed-string amplitudes, the CHY formalism, and the intersection numbers in the logarithmic case was given in \cite{Mizera:2017rqa} and can be summarized as follows.
	Consider the homologically-split form of the closed-string amplitudes from \eqref{homological-factorization} with logarithmic forms $\varphi_+$ and $\vartheta_+$. Notice that exchanging $\overbar{\KN} \overbar{\vartheta_+}$ for $\KN^{-1} \vartheta_+$ in the open-string integrals only affects the global sign of the leading order ${\cal O}((\alpha')^{3-n})$ by $(-1)^{n-3}$, but not its value (it can, and does affect the subleading orders). Since monodromies of ${\cal L}_{\overbar{\omega}}$ are also that of ${\cal L}_{-\omega}$ for real-valued kinematics, the intersection matrix $\mathbf{H}_{\alpha\beta}$ remains unchanged. The result is simply the homologically-split form of the intersection number $\la \vartheta_{+} | \varphi_+ \ra_\omega$, which by the $\alpha'$-independence property proves \eqref{proof-closed} and by extension also \eqref{proof-CHY}.} 

The results of Theorem~\ref{theorem-22} imply that in logarithmic cases the low-energy behaviour of string amplitudes is governed by the critical points of the Morse function. It should not be confused with the high-energy limit, which also localizes on the critical points but is exponentially suppressed and depends on the direction in which such limit is taken \cite{Gross:1987ar,Gross:1987kza}.

The above results by no means imply that the low-energy and massless limits \emph{only} agree in the logarithmic cases. In fact, in Section~\ref{sec:Yang-Mills-and-gravity} we will discuss examples of twisted forms that despite being non-logarithmic give rise to intersection numbers homogeneous in $\alpha'$, thus extending the result of the above theorem. As is already clear from dimensional analysis, when intersection numbers become non-homogeneous in $\alpha'$ the relationship between low-energy limit of strings amplitudes and the CHY formalism breaks down. However, following the discussion at the beginning of this section one can always express twisted forms in terms of logarithmic ones first and then apply the above theorem after stripping away non-homogeneous terms.

According to the explicit formula \eqref{intersection-logarithmic}, intersection numbers of logarithmic forms have a simple geometric interpretation as a sum over trivalent graphs represented as vertices in the moduli space. In the special case when $\varphi_\pm$ have unit residues (up to a sign), the intersection number can be determined purely combinatorially. 
For instance, as is evident from their definition, the massless Parke--Taylor forms $\PT(\alpha)$ have $\pm 1$ residues only on the vertices of $\partial\M_{0,n}$ corresponding to planar diagrams with respect to the same permutation $\alpha$ (see, e.g., \cite{Skinner:2010cz,Mizera:2017cqs} for their explicit logarithmic representation). This fact immediately implies that
\be\label{PT-PT}
\braket{ \PT(\alpha) | \PT(\beta) }_\omega \;=\; (-1)^{w(\alpha|\beta)+1}\!\!\!\!\! \sum_{{\cal T} \in {\cal G}(\alpha) \cap {\cal G}(\beta)} \frac{1}{\prod_{e \in {\cal T}}  p_e^2},
\ee
where ${\cal G}(\alpha)$ is the set of all trivalent trees planar with respect to $\alpha$ and the product in the denominator goes over kinematic variables $\{p_e^2\}_{e\in {\cal T}}$ associated to each edge $e$ of the tree ${\cal T}$. The overall sign can be fixed by a more careful computations and it turns out to depend on the relative winding number $w(\alpha|\beta)$ between the two permutations \cite{Cachazo:2013iea,Mizera:2016jhj}. The intersection numbers \eqref{PT-PT} compute colour-ordered amplitudes in the so-called \emph{bi-adjoint} scalar theory \cite{Cachazo:2013iea}.

In order to illustrate how to use the formula \eqref{intersection-logarithmic} in practice, let us finish this subsection by proving that $\PT(1,\alpha,n{-}1,n)$ and $\PT(1,\beta,n{-}1,n)^\vee$ are orthonormal with respect to the intersection pairing at low multiplicity.

\begin{example}
	For $n=4$ we consider the twisted forms:
	\be
	\PT(1234) = \frac{d\mu_4}{z_{12} z_{23} z_{34} z_{41}}, \qquad \PT(1234)^\vee = -\frac{(p_1{+}p_2)^2\, d\mu_4}{z_{13} z_{34} z_{42} z_{21} }.
	\ee
	Their orthonormality already follows from \eqref{PT-1243-1234}. However, let us perform this computation using the prescription from \eqref{intersection-logarithmic} instead. Given that $\PT(1234)$ is planar, it is sufficient to consider $\delta=(1234)$ with $\chi_{(1234)} = \{\{1,3\}, \{2,4\}\}$. The corresponding cross-ratios are:
	\be
	u_{13} = \frac{z_{14}z_{23}}{z_{13}z_{24}}, \qquad u_{24} = \frac{z_{21}z_{34}}{z_{24}z_{31}},
	\ee
	which are responsible for the two planar channels: $(p_2{+}p_3)^2$ and $(p_1{+}p_2)^2$ respectively since
	\be
	\omega = \frac{1}{\Lambda^{2}} \Big( (p_2{+}p_3)^2 d\log u_{13} + (p_1{+}p_2)^2 d\log u_{24} \Big).
	\ee
	Using \eqref{crossing-identities} with $(p,q,r,s)=(1,2,3,4)$ they satisfy the constraint $u_{13} + u_{24} = 1$. We can express $\PT(1234)$ in two ways in terms of these dihedral coordinates:
	\be
	\PT(1234) = -d\log \frac{u_{13}}{u_{13}{-}1} = d\log \frac{u_{24}}{u_{24}{-}1},
	\ee
	where we also used the fact that $d\mu_4 = z_{13}z_{34}z_{14}\, dz_2$. Similarly for $\PT(1234)^\vee$ we have:
	\be
	\PT(1234)^\vee = (p_1{+}p_2)^2\, d\log (u_{13}{-}1) = (p_1{+}p_2)^2 d\log u_{24}.
	\ee
	Therefore the only place where the two forms have non-vanishing residue at the same time is $u_{24}=0$ (recall that we ignored the non-planar channel with respect to the ordering $\delta=(1234)$ since $\PT(1234)$ has a zero residue there):
	\be
	\Res_{u_{24}=0} \Big(\PT(1234)\Big) = 1, \qquad \Res_{u_{24}=0} \Big(\PT(1234)^\vee \Big) = (p_1{+}p_2)^2.
	\ee
	Hence out of the three terms in the sum \eqref{intersection-logarithmic} only one contributes and we find:
	\be
	\braket{\PT(1234)^\vee | \PT(1234)}_\omega = \frac{(p_1 {+} p_2)^2}{(p_1 {+} p_2)^2} = 1.
	\ee
	Therefore the two twisted forms create orthonormal bases of $H^{1}_{\pm\omega}$.
\end{example}

\begin{example}
	For $n=5$ the twisted cohomology groups are two-dimensional and we consider the following basis of $H^{2}_\omega$:
	\be\label{PT-5}
	\PT(12345) = \frac{d\mu_5}{z_{12}z_{23}z_{34}z_{45}z_{51}}, \qquad \PT(13245) = \frac{d\mu_5}{z_{13}z_{32}z_{24}z_{45}z_{51}},
	\ee
	as well as the dual basis of $H^{2}_{-\omega}$, which using the notation $s_{ij} := (p_i {+} p_j)^2$ is given by:
	\begin{gather}
	\PT(12345)^\vee = \frac{ s_{12}\, d\mu_5}{z_{14} z_{45} z_{52} z_{21}} \left( s_{13} \frac{z_{15}}{z_{13}z_{35}} + s_{23} \frac{z_{25}}{z_{23}z_{35}}\right),\\
	\PT(13245)^\vee = \frac{ s_{13}\, d\mu_5}{z_{14} z_{45} z_{53} z_{31}} \left( s_{12} \frac{z_{15}}{z_{12}z_{25}} + s_{23} \frac{z_{35}}{z_{32}z_{25}} \right).
	\end{gather}
	Let us exploit the fact that \eqref{PT-5} are planar with respect to two different orderings, as well as their symmetry under the exchange $2 \leftrightarrow 3$. We start by computing the two intersection numbers $\la \PT(1\widehat{\alpha}45)^\vee | \PT(12345) \ra_\omega$ with $\widehat{\alpha} \in \{(23), (32)\}$ using $\delta = (12345)$, which has the set of chords $\chi_{(12345)} = \{\{1,3\}, \{1,4\}, \{2,4\}, \{2,5\}, \{3,5\}\}$. The cross-ratios $u_{ij}$ are subject to relations $u_{13} + u_{24} u_{25} = 1$ and its cyclic relabellings. The local system is defined by the one-form:
	\be
	\omega = \frac{1}{\Lambda^2} \Big( s_{23}\, d\log u_{13} + s_{51}\, d\log u_{14} + s_{34}\, d\log u_{24} + s_{12} \, d\log u_{25} + s_{45}\, d\log u_{35} \Big),
	\ee
	where we used momentum conservation. There are five triangulations of a pentagon with non-crossing pairs of chords from $\chi_{(12345)}$, which are in one-to-one map with pairs of compatible factorization channels. Let us start by considering the two sets of dihedral coordinates $(u_{13}, u_{35})$ and $(u_{25}, u_{35})$, in terms of which we have:
	\be
	\PT(12345) = - d\log \frac{u_{13}}{u_{13}{-}1} \wedge d\log \frac{u_{35}}{u_{35}{-}1} = d\log \frac{u_{25}}{u_{25}{-}1} \wedge d\log \frac{u_{35}}{u_{35}{-}1},
	\ee
	while for the dual basis:
	\begin{align}
	\PT(12345)^\vee &= -s_{12} \left( s_{23}\, d\log \frac{u_{13}}{u_{13}{-}1} + s_{13}\, d\log \frac{u_{13}u_{35}{-}1}{u_{13}{-}1} \right) \wedge d\log u_{35} \nn\\
	&= s_{12} \left( s_{23}\, d\log \frac{u_{25}}{u_{25}{-}1} + s_{13}\, d\log u_{25} \right) \wedge d\log u_{35}
	\end{align}
	and
	\begin{align}
	\PT(13245)^\vee &= -s_{13} \left( s_{12}\, d\log \frac{u_{13}u_{35}{-}1}{u_{13}{-}1} + s_{23}\, d\log \frac{u_{13}u_{35}{-}1}{u_{13}} \right) \wedge d\log u_{35}\nn\\
	&= s_{13} \Big( s_{12}\, d\log u_{25} + s_{23}\, d\log (u_{25}{-}1)\Big) \wedge d\log u_{35}.
	\end{align}
	Evaluating residues around the origin of $(u_{13},u_{35})$ we obtain:
	\begin{gather}
	\Res_{u_{13}=0} \Res_{u_{35}=0} \Big(
	\PT(12345)
	\Big) = -1,\\
	\Res_{u_{13}=0} \Res_{u_{35}=0} \Big(
	\PT(12345)^\vee
	\Big) = -s_{12}s_{23},\qquad \Res_{u_{13}=0} \Res_{u_{35}=0} \Big(
	\PT(13245)^\vee
	\Big) = s_{13}s_{23}.\nn
	\end{gather}
	Similarly, around the other vertex parametrized by $(u_{25}, u_{35})$ we find
	\begin{gather}
	\Res_{u_{25}=0} \Res_{u_{35}=0} \Big(
	\PT(12345)
	\Big)
	= 1,\\
	\Res_{u_{25}=0} \Res_{u_{35}=0} \!\Big(
	\PT(12345)^\vee
	\Big)
	= s_{12}(s_{13}{+}s_{23}),\;\; \Res_{u_{25}=0} \Res_{u_{35}=0} \!\Big(
	\PT(13245)^\vee
	\Big)
	= s_{12}s_{13}.\nn
	\end{gather}
	Repeating similar exercises for the remaining vertices reveals that for each of them the product of two residues is always zero. Thus there are only two non-zero contributions to the sum in \eqref{intersection-logarithmic} giving
	\begin{gather}
	\la \PT(12345)^\vee | \PT(12345) \ra_\omega = \frac{s_{12}s_{23}}{s_{23}s_{45}} + \frac{s_{12} ( s_{13} {+} s_{23} )}{s_{12}s_{45}} = 1,\\
	\la \PT(13245)^\vee | \PT(12345) \ra_\omega = -\frac{s_{13}s_{23} }{s_{23}s_{45}} + \frac{s_{12}s_{13}}{s_{12}s_{45}} = 0,
	\end{gather}
	where in the first line we used that $s_{12}{+}s_{13}{+}s_{23} = s_{45}$ by momentum conservation. Exchanging $2\leftrightarrow 3$ gives us the remaining two intersection numbers. This leads to the final result:
	\be
	\Braket{ \left(\begin{array}{c}
			\PT(12345)^\vee \\
			\PT(13245)^\vee \\
		\end{array}\right) | \left(\begin{array}{c}
			\PT(12345) \\
			\PT(13245) \\
		\end{array}\right)^{\!\!\intercal} }_{\!\omega} = \left(\begin{array}{cc}
		1 & 0 \\
		0 & 1 \\
	\end{array}\right),
	\ee
	which shows that the two bases of $H^2_{\pm\omega}$ are orthonormal.
\end{example}

\subsection{\label{sec:shift-relations}Comment on Shift Relations for Local Systems}

\textsc{Up until now}, intersection numbers corresponding to different external masses were considered separately, as they are computed with distinct local systems.
We close this section by commenting on shift relations between different local systems that allow us to relate them to each other. Let us consider a local system ${\cal L}_{\omega'}$ defined by the following one-form:
\be
\omega' \,:=\, \omega \,+ \!\!\! \sum_{1 \leq i<j \leq n} \!\!\! n_{ij}\, d\log(z_i - z_j)
\ee
for any $n_{ij} = n_{ji} \in \Z$. Setting $n_{ii}=0$ we have $n(n{-}1)/2$ such integers. We can introduce shifted differential forms:
\be\label{form-shift}
\varphi_\pm' \,:=\, \varphi_\pm \!\!\! \prod_{1\leq i<j \leq n} (z_i - z_j)^{\mp n_{ij}},
\ee
which remain single-valued on $\M_{0,n}$ and hence are elements of $H^{n-3}_{\pm\omega'}$. With these assignments the intersection number of the primed and unprimed twisted forms coincide:
\be\label{shift-relation}
\la \varphi_- | \varphi_+ \ra_{\omega} = \la \varphi_-' | \varphi_+' \ra_{\omega'},
\ee
as is straightforward to show, for example from the representation \eqref{intersection-steepest}.

We can interpret the local system ${\cal L}_{\omega'}$ as that corresponding to shifted kinematics, i.e., $2 p'_i {\cdot} p'_j = 2p_i {\cdot} p_j + \Lambda^2 n_{ij}$.\footnote{We thank F.~Cachazo for this observation.} From here it follows that masses $m_i'{}^{2} := - p_i'{}^2$ of the primed momentum vectors are given by 
\be
m'_i{}^2 = m_i^2 + \frac{\Lambda^2}{2} \sum_{j \neq i} n_{ij}.
\ee
Thus by choosing $n_{ij}$'s appropriately we can map computations with different external masses to each other. A particularly useful application is to map all computations of intersection numbers to the massless case, as it allows to exploit simplifications entailed by logarithmicity of differential forms.

For instance, let us illustrate how to evaluate $\braket{ d\mu_n | \PT(\alpha|\beta) }_\omega$ when all external masses are equal to $m_i = \Lambda$ using the result of the previous subsection. We can map this problem into the massless one, $m_i'{}^2=0$ by choosing $n_{ij} = -1$ when $i$ and $j$ are adjacent, say in the cyclic ordering $\alpha$, and $n_{ij} = 0$ otherwise. With this assignments we find that every propagator appearing in the planar ordering $\alpha$ gets shifted from massive in $\omega$ to massless in $\omega'$:
\be
\big(p_{\alpha(i)} + p_{\alpha(i+1)} + \cdots + p_{\alpha(j)}\big)^2 + \Lambda^2 = \big(p'_{\alpha(i)} + p'_{\alpha(i+1)} + \cdots + p'_{\alpha(j)}\big)^2.
\ee
Notice that according to \eqref{form-shift}, $\varphi_- = d\mu_n$ becomes $\varphi'_- = \PT(\alpha)$, while $\varphi_{+} = \PT(\alpha|\beta)$ becomes $\varphi'_{+} = \PT(\beta)$ in the massless case. Therefore we can use the shift relation \eqref{shift-relation} and the result of \eqref{PT-PT} to obtain
\be\label{massive-biadjoint}
\braket{ d\mu_n | \PT(\alpha|\beta) }_{\omega} \;=\; \braket{ \PT(\alpha) | \PT(\beta) }_{\omega'} \;=\; (-1)^{w(\alpha|\beta)+1}\!\!\!\!\!  \sum_{{\cal T} \in {\cal G}(\alpha) \cap {\cal G}(\beta)} \frac{1}{\prod_{e \in {\cal T}}  \left( p_e^2 + \Lambda^2 \right)},
\ee
which are colour-ordered amplitudes in the massive bi-adjoint scalar theory. We stress that the choice of $n_{ij}$ was auxiliary and the same result can be obtained by other means.

As another example, let us consider setting $\alpha=\beta$ and symmetrizing over all $(n{-}1)!/2$ cyclically-inequivalent permutations $\alpha$:
\be
\bigg< d\mu_n \,\bigg|\, \sum_{\alpha} \PT(\alpha|\alpha) \bigg>_{\!\!\omega} \;=\; 2^{n-3}\, \sum_{\cal T} \frac{1}{\prod_{e \in {\cal T}}  \left( p_e^2 + \Lambda^2 \right)},
\ee
where the right-hand side, obtained from \eqref{massive-biadjoint} and linearity of the intersection pairing, involves a sum over all trivalent trees $\cal T$. The result is the scattering amplitude in the massive cubic scalar theory.

\pagebreak
\section[Recursion Relations from Braid Matrices]{\label{sec:recursion-relations}Recursion Relations from Braid Matrices}

\textsc{In this section} we introduce recursion relations for intersection numbers. We start by explaining the fibre bundle structure of moduli spaces and how twisted cohomologies decompose on their fibres in a recursive way. This leads to a simple way of computing intersection numbers by ``integrating out'' puncture by puncture, which will become our main computational tool.

Since using the shift relations discussed in Section~\ref{sec:shift-relations} one can map any calculation to the one with all external masses vanishing, $m_i {=} 0$, in this section we only consider this case.

\subsection{\label{sec:fibration}Fibred Moduli Spaces and Twisted Cohomologies}

\textsc{Moduli spaces} of Riemann surfaces with punctures have natural fibration structure. For $\M_{0,n}$ it is obtained by the forgetful map
\be\label{M0n-fibration}
\pi: \, \M_{0,n} \to \M_{0,n-1},
\ee
which acts by removing one of the punctures, say $z_n$. For clarity of notation in this section we use the $\PSL(2,\C)$ quotient to fix the positions of $(z_1,z_2,z_3)$. For a generic point $b \in \M_{0,n-1}$ the fibre $\pi^{-1}(b)$ becomes
\be
\pi^{-1}(b) = \{ z_n \in \CP^1 \,|\, z_n \neq z_1, z_2, \ldots, z_{n-1} \},
\ee
that is a Riemann sphere with $n{-}1$ punctures removed. Its Euler characteristic equals to $\chi(\CP^1) - (n{-}1) = 3-n$. Since Euler characteristics on a fibre bundle multiply, we obtain the recursion relation
\be\label{Euler-characteristic-recursion}
\chi(\M_{0,n}) = (3-n)\, \chi(\M_{0,n-1}),
\ee
which, using the fact that $\M_{0,3}$ is a single point with $\chi(\M_{0,3}) = 1$, gives the well-known result $\chi(\M_{0,n}) = (-1)^{n-3}(n-3)!$, cf. \cite{Harer1986,penner1988}.

It is of course not the only fibration of $\M_{0,n}$ available. For example, one can consider maps $\M_{0,n} \to \M_{0,n-k}$ that forget $k \leq n{-}4$ points at the same time. In fact, as will become clear shortly, for our purposes it is the most convenient to consider the case $k = n{-}4$, in which the base space is one-dimensional,
\be
\pi_4:\, \M_{0,n} \,\to\, \Sigma_{4} := \{z_4 \in \CP^1 \,|\, z_4 \neq z_1, z_2, z_3\}
\ee
with the fibre $\pi^{-1}_4(z_4)$, which is the configuration space of of $n{-}4$ points $(z_5,z_6,\ldots,z_n)$ on a Riemann sphere with four points removed, $\CP^{1} {\setminus} \{z_1,z_2,z_3,z_4\}$. We will denote it by $\pi^{-1}_4(z_4) =: {\cal F}_{n,4}$.

Let us turn our attention to the structure of cohomologies with local coefficients on this fibre bundle, which can be studied using spectral sequences,  see, e.g., \cite{bott2013differential,fomenko2016homotopical}.\footnote{We thank D.~Fuchs, A.~Schwarz, and E.~Witten for suggesting this to us.} Specifically, we can understand $H^{n-3}(\M_{0,n}, {\cal L}_\omega)$ on the total space as a cohomology group on the base space $\Sigma_{4}$ with a new local system ${\cal H}_\omega$,
\be\label{cohomology-base}
H^{n-3}(\M_{0,n}, {\cal L}_\omega) \,\cong\, H^{1}(\Sigma_{4}, {\cal H}_{\omega}).
\ee
Here ${\cal H}_\omega$ is of higher rank, meaning that it is a higher-dimensional representation of the fundamental group of the base $\pi_1(\Sigma_4)$ acting on a fibre cohomology. We can write
\be
{\cal H}_\omega \,\cong \bigcup_{z_4 \in \Sigma_4}  H^{n-4}({\cal F}_{n,4},\, {\cal L}_{\imath_4^\ast \omega}),
\ee
where $\imath_4^\ast \omega$ is the pullback of $\omega$ by the inclusion map $\imath_4 \!: {\cal F}_{n,4} \hookrightarrow \M_{0,n}$. In other words, the right-hand side of \eqref{cohomology-base} is a cohomology on $\Sigma_4$ with coefficients in cohomologies on each fibre ${\cal F}_{n,4}$, which in turn have coefficients in the local system ${\cal L}_{\imath^\ast_4 \omega}$ inherited from the original ${\cal L}_\omega$. One can show that both cohomologies are concentrated in the dimension of $\Sigma_4$ and ${\cal F}_{n,4}$, i.e., $1$ and $n{-}4$ respectively. Repeating the arguments that led to \eqref{Euler-characteristic-recursion} we find $\chi({\cal F}_{n,4}) = (-1)^{n-4} (n-3)!$ and hence the dimension of $H^{n-4}({\cal F}_{n,4},\, {\cal L}_{\imath_4^\ast \omega})$ is $(n-3)!$.

In order to understand practical implications of the above result, let us turn to twisted differential forms. The above decomposition means that  $\varphi \otimes \exp\int_\gamma \omega \in H^{n-3}(\M_{0,n}, {\cal L}_\omega)$ can be written as the expansion
\be
\varphi \otimes \exp{\textstyle\int_{\gamma}} \,\omega = \sum_{a=1}^{(n-3)!} \sigma_a \wedge \left( f_a \otimes \exp {\textstyle\int_{\gamma}} \,\omega \right),
\ee
where $\sigma_a \in \Omega^1(\Sigma_4)$ are one-forms on $\Sigma_4$ that do not depend on variables $(z_5,z_6,\ldots,z_n)$, while $f_a \in \Omega^{n-4}({\cal F}_{n,4})$ with local coefficients $\exp \int_\gamma \omega$ form a basis of the fibre cohomology and can depend on $z_4$. The sum goes over the full basis and the path $\gamma \in \pi_1(\M_{0,n})$ is arbitrary.

Following the steps from \eqref{image-map} let us see how differentials acts on $\xi \otimes \exp{\textstyle\int}_{\!\gamma} \omega$ with $\xi \in \Omega^{n-4}(\M_{0,n})$. We first consider the derivative in the direction of the fibre, $d_{{\cal F}} := \sum_{i=5}^{n} dz_i \,\partial/\partial z_i$:
\begin{align}
d_{{\cal F}}\! \left( \xi \otimes \exp{\textstyle\int}_{\!\gamma} \omega \right) &= d_{{\cal F}} \bigg( \sum_{a=1}^{(n-3)!} \sigma_a \wedge \left( g_a \otimes \exp {\textstyle\int_{\gamma}} \,\omega \right) \bigg)\nn\\
&= \sum_{a=1}^{(n-3)!} \sigma_a \wedge \big( d g_a + \imath_4^\ast \omega \wedge g_a \big) \otimes \exp {\textstyle\int_{\gamma}} \,\omega,\label{differential-first}
\end{align}
where $g_a \in \Omega^{n-5}({\cal F}_{n,4})$. We used $d_{{\cal F}} \sigma_a {=} 0$, $d_{{\cal F}} g_a {=} dg_a$, and the fact that $d_{{\cal F}}\! \int_{\gamma} \omega = \imath^\ast_4 \omega$ selects only the components projected on ${\cal F}_{n,4}$. Thus, recognizing the term in the final bracket as a twisted differential acting on $g_a$, we find according to our expectations that the cohomology with local coefficients on the fibre can be understood as a twisted cohomology $H^{n-4}({\cal F}_{n,4}, \nabla_{\imath^\ast_4 \omega})$ with the connection $\nabla_{\imath^\ast_4 \omega} := d + \imath^\ast_4 \omega \wedge$. It is of exactly the same type as the one on the total space $\M_{0,n}$.

Let us turn to the derivative acting in the direction on the base space, $d_{\Sigma_4} := dz_4\, \partial / \partial z_4$:
\begin{align}
d_{\Sigma_4} \left( \xi \otimes \exp{\textstyle\int}_{\!\gamma} \omega \right) &= d_{\Sigma_4} \bigg( \sum_{a=1}^{(n-3)!} \tau_a \wedge \left( f_a \otimes \exp {\textstyle\int_{\gamma}} \,\omega \right) \bigg)\nn\\
&= \sum_{a=1}^{(n-3)!} \!\bigg( d\tau_a \wedge f_a \otimes \exp {\textstyle\int_{\gamma}} \,\omega + \tau_a \wedge d_{\Sigma_4}\! \left( f_a \otimes \exp {\textstyle\int_{\gamma}} \,\omega \right) \bigg),\label{base-space-connection}
\end{align}
where $\tau_a \in {\cal O}(\Sigma_4)$ is a function on $\Sigma_4$ and we used $d_{\Sigma_4} \tau_a = d\tau_a$. In contrast with \eqref{differential-first}, the second term in \eqref{base-space-connection} fails to factorize cleanly into base space and fibre contributions for two reasons: the term $d_{\Sigma_4} f_a$ does not necessarily vanish, and the one-form $d_{\Sigma_4} \int_\gamma \omega$ can still depend on the coordinates on ${\cal F}_{n,4}$. Both of these issues encode the fact that the fibre bundle is non-trivial and describe how $\Sigma_4$ and ${\cal F}_{n,4}$ are braided together. Nevertheless, since $\{f_a\}_{a=1}^{(n-3)!}$ forms a basis, there always exists an expansion of the form:
\be\label{omega-sigma}
d_{\Sigma_4} \!\left( f_a \otimes \exp {\textstyle\int_{\gamma}} \,\omega \right) = \sum_{b=1}^{(n-3)!}\, (\om^{+}_4)_{ab} \wedge f_b \otimes \exp {\textstyle\int_{\gamma}} \,\omega
\ee
for a $(n{-}3)! {\times} (n{-}3)!$ matrix-valued one-form $\om^{+}_{4}$ on $\Sigma_4$. We will construct it explicitly in the following subsections, and for the time being it suffices to know that it exists. Plugging \eqref{omega-sigma} into \eqref{base-space-connection} we find the factorized expression:
\be
d_{\Sigma_4} \left( \xi \otimes \exp{\textstyle\int}_{\!\gamma} \omega \right) = \sum_{a=1}^{(n-3)!} \Big(  d\tau_a + \sum_{b=1}^{(n-3)!} \!\! \tau_b \wedge (\om_4^+)_{ba} \Big) \wedge f_a \otimes \exp {\textstyle\int_{\gamma}} \,\omega.
\ee
Thus we can define ${\bm\nabla}^+_{4} \tau_a := d\tau_a + \sum_b \tau_b \wedge (\om^+_4)_{ba}$. It is a Gauss--Manin connection of rank $(n{-}3)!$. Note that $\om_4^+$ acts from the right and its form depends on the choice of the basis $\{f_a\}_{a=1}^{(n-3)!}$ on the fibre. Hence we should think of the twisted cohomology on $\Sigma_4$ as the space of vector-valued one-forms which are closed but not exact with respect to the connection ${\bm\nabla}^+_4$. This gives us the base space twisted cohomology $H^1(\Sigma_4, {\bm\nabla}^+_4)$.

We can also consider the dual cohomology group $H^{n-3}(\M_{0,n}, {\cal L}_{-\omega})$ and decompose it in an analogous manner. This gives $H^{n-4}({\cal F}_{n,4}, \nabla_{-\imath^\ast_4 \omega})$ on the fibre, as well as $H^{1}(\Sigma_4, {\bm\nabla}_4^-)$ on the base. However, the connection ${\bm\nabla}_4^-$ is in general not related to ${\bm\nabla}_4^+$ by a sign flip. By choosing orthonormal bases on the fibre we will show in the Lemma~\ref{lemma-3-1} that the dual connection ${\bm\nabla}_4^- \tau_a := d\tau_a - \sum_b \tau_b \wedge (\om_4^-)_{ba}$ is obtained by matrix transpose $\om_{4}^- = (\om_{4}^+)^\intercal$, a feature invisible at the level of rank-one twisted cohomologies.

We can carry on by considering a sequence of fibrations in which one strips away puncture-by-puncture until the whole moduli space in decomposed into a ``product'' of $n{-}3$ one-dimensional spaces, cf.~\eqref{fig-intro}. We can summarize it on the diagram:
\be\label{fibration-diagram}
\begin{tikzcd}
{{\cal M}_{0,n}} \arrow[d, "\pi_4"'] & {\cal F}_{n,4} \arrow[l, "\imath_4"', hook'] \arrow[d, "\pi_5"'] & {\cal F}_{n,5} \arrow[d, "\pi_6"'] \arrow[l, "\imath_5"', hook'] & \cdots \arrow[l, "\imath_6"', hook'] & {\cal F}_{n,n-1} \arrow[d, "\pi_n"'] \arrow[l, "\imath_{n-1}"', hook'] & {\cal F}_{n,n} \arrow[l, "\imath_n"', hook'] \\
\Sigma_4  & \Sigma_5 & \Sigma_6 & \cdots & \Sigma_n &
\end{tikzcd}
\ee
Notice a dependence on an auxiliary ordering in which particles are being peeled off, which affects intermediate steps of calculations but not the final result (in fact, considering different orderings allows for strong consistency checks). Here each $\Sigma_p$ is the Riemann sphere on which $z_p$ lives,
\be
\Sigma_p := \{ z_p \in \CP^1 \,|\, z_p \neq z_1, z_2, \ldots, z_{p-1} \},
\ee
while ${\cal F}_{n,p}$ denotes the configuration space of the remaining $n{-}p$ points on $\CP^1 {\setminus} \{z_1, z_2, \ldots, z_p\}$:
\be
{\cal F}_{n,p} := \{ (z_{p+1}, z_{p+2}, \ldots, z_n) \in (\CP^1)^{n-p} \,|\, z_i \neq z_1, z_2, \ldots, z_p, z_j \text{ for all } i \neq j\}.
\ee
In the edge cases ${\cal F}_{n,3} = \M_{0,n}$ and ${\cal F}_{n,n}$ is a single point. Since the Euler characteristic of $\Sigma_p$ is given by $\chi(\Sigma_p) = 3-p$, by previous arguments we find:
\be
\chi({\cal F}_{n,p}) = \prod_{q=p+1}^{n} \!\chi(\Sigma_q) = (-1)^{n-p} \frac{(n-3)!}{(p-3)!}.
\ee
Given that the twisted cohomology on ${\cal F}_{n,p}$ is concentrated in the dimension $n{-}p$ and its dimension is equal to the rank of the local system on $\Sigma_p$, we find that ${\bm\nabla}_{p}^+$ is of rank $(n-3)!/(p-3)!$. For conciseness let us introduce the notation $\kappa_{n,p} := (n-3)!/(p-3)!$.  Therefore the twisted cohomology on the each $\Sigma_p$ is given by
\be\label{base-space-cohomology}
H^1(\Sigma_p, {\bm\nabla}_p^\pm) := \frac{\{ \bm\varphi^\pm_p \in \Omega^1(\Sigma_p) \otimes \C^{\kappa_{n,p}} \,|\, {\bm\nabla}_p^\pm \bm\varphi^\pm_p = 0 \}}{{\bm\nabla}_p^\pm\, {\cal O}(\Sigma_p) \otimes \C^{\kappa_{n,p}}}
\ee
for $p = 4,5,\ldots,n$. That is, it is the space of $\bm\nabla_p^\pm$-closed vector-valued one-forms $\bm\varphi_p^\pm$ modulo the ${\bm\nabla}_p^\pm$-exact forms. As before, we only consider twisted forms that have a M\"obius weight $0$ in the variable $z_p$. We will discuss how to construct the connections ${\bm\nabla}_p^{\pm}$ using relations analogous to \eqref{omega-sigma} in the following subsections. On each fibre we have rank-one twisted cohomologies $H^{n-p}({\cal F}_{n,p}, {\nabla}_{\pm\omega_p})$, where $\omega_{p}$ is simply the projection of the original form $\omega$:
\be
\omega_{p} := \imath^\ast_{p} \cdots \imath_{5}^\ast \imath_{4}^\ast \omega = \omega \big|_{dz_{4} = dz_5 = \cdots = dz_{p} = 0 }.
\ee
Hence they behave in a similar way to $H^{n-3}_{\pm\omega}$ except for the fact that a subset of punctures are locked into some fixed positions.

The reason for considering the sequence of fibrations \eqref{fibration-diagram} is the following. We will follow a strategy of ``integrating out'' punctures one-by-one, starting with $z_n$ and ending with $z_4$. Thus, twisted forms on $\M_{0,n}$ will be projected onto forms on $\M_{0,n-1}$, then on $\M_{0,n-2}$, etc., until the trivial space $\M_{0,3}$ is reached. The result of this procedure will be functions that are independent of the puncture coordinates $z_i$, which will in turn enter the computation of intersection numbers. In the sequel we construct recursion relations that allow for a concrete implementation of this idea.

For previous discussion of fibrations of twisted cohomologies in the context of their intersection theory see~\cite{Ohara98intersectionnumbers,OST2003,MOY2004}.\footnote{It is known that one can construct a rank-$(n{-}2)!$ Gauss--Manin connection in an \emph{auxiliary} variable $z_0 \in \CP^1 {\setminus} \{0,1,\infty\}$, whose horizontal sections compute certain \emph{finite} open-string integrals for $n$ punctures in the asymptotic limit $z_0 {\to} 0$, and for $n{+}1$ punctures at $z_0 {\to} 1$ \cite{terasoma_2002,Broedel:2013aza}. The two limits are connected by an associator \cite{zbMATH04202514,LE1995193}, thus providing a recursion for these integrals in $n$. Using this setup Terasoma proved that coefficients of the Taylor expansion of such open-string integrals around $\alpha' {=} 0$ are $\mathbb{Q}$-linear combinations of multiple zeta values (MZVs) \cite{terasoma_2002}. The conjecture \cite{Stieberger:2013wea} that coefficients of closed-string integrals fall into a special class of \emph{single-valued} MZVs \cite{Brown:2013gia} was recently proven by Brown and Dupont \cite{Brown:2018omk}, see also \cite{Schlotterer:2018zce,Vanhove:2018elu}.}

\subsection{\label{sec:recursion-relations-first}Recursion Relations}

\textsc{Even before knowing} the exact description of twisted cohomologies on all fibres, it is instructive to derive the simplest form of recursion relations for intersection numbers. It will not only teach us about their general structure, but also give an idea of what further aspects of fibred cohomologies will be needed.

Let us consider the fibre ${\cal F}_{n,p-1}$ and twisted forms $\varphi_\pm \in H^{n-p+1}({\cal F}_{n,p-1}, {\nabla}_{\pm \omega_{p-1}})$. We wish to express intersection numbers $\braket{ \varphi_- |\varphi_+}_{\omega_{p-1}}$, given as in Definition~\ref{definition-intersection-number}, recursively in terms of intersection numbers on ${\cal F}_{n,p}$. Let us start by writing the general fibre bundle decomposition of $\varphi_\pm$:
\be\label{form-expansion}
\varphi_- = \sum_{a=1}^{\kappa_{n,p}}  \sigma_{-,a} \wedge f_{-,a}^\vee,\qquad \varphi_+ =  \sum_{a=1}^{\kappa_{n,p}} \sigma_{+,a} \wedge f_{+,a},
\ee
in terms of bases $f_{-,a}^\vee$ and $f_{+,a}$ of $H^{n-p}({\cal F}_{n,p}, \nabla_{\pm \omega_p})$. There always exists a choice of such bases which is orthonormal in the sense $\braket{ f_{-,a}^\vee | f_{+,b} }_{\omega_p} = \delta_{ab}$. Coefficients of the expansion are twisted forms $\sigma_{\pm,a} \in H^{1}(\Sigma_p, {\bm\nabla}^{\pm}_p)$, which are vector-valued one-forms with $a=1,2,\ldots,\kappa_{n,p}$. We can find them simply by projecting both sides of \eqref{form-expansion} with $\braket{ \,\bullet\, | f_{+,b}}_{\omega_p}$ and  $\braket{ f^\vee_{-,b} |  \,\bullet\,}_{\omega_p}$ respectively:
\be\label{sigma-forms}
\braket{ \varphi_- | f_{+,b} }_{\omega_{p}} = \sigma_{-,b},\qquad
\braket{ f^\vee_{-,b} | \varphi_+ }_{\omega_{p}} = \sigma_{+,b},
\ee
where we used the orthonormality condition and the fact that $\sigma_{\pm,a}$ are constants from the perspective of ${\cal F}_{n,p}$. Although the left-hand sides of \eqref{sigma-forms} are one-forms on $\Sigma_p$, we will abuse the terminology and still refer to them as intersection numbers.

We can now use the Definition~\ref{definition-intersection-number} applied to ${\cal F}_{n,p-1}$ in order write out the intersection number $\braket{\varphi_- | \varphi_+ }_{\omega_{p-1}}$ decomposed into the base space and fibre contributions:
\be
\braket{\varphi_- | \varphi_+ }_{\omega_{p-1}} = \frac{1}{(-2\pi i \Lambda^2)^{n-p+1}}  \sum_{a,b=1}^{\kappa_{n,p}} \int_{\Sigma_p} \left( \sigma_{-,a} \wedge \iota_{\om_p^+}(\sigma_{+,b}) \int_{{\cal F}_{n,p}} f_{-,a}^\vee \wedge \iota_{\omega_p}(f_{+,b}) \right),
\ee
where $\iota_{\om_p^+}$ and $\iota_{\omega_p}$ are maps that turn twisted forms into compactly-supported forms in their respective cohomology classes. Recognizing that the second integral is nothing other than $\braket{ f_{-,a}^\vee | f_{+,b} }_{\omega_p}$ weighted by powers of $-2\pi i \Lambda^2$ and plugging in the expression of $\sigma_{\pm,a}$ from \eqref{sigma-forms} we find:
\be
\braket{\varphi_- | \varphi_+ }_{\omega_{p-1}} = -\frac{1}{2\pi i \Lambda^2} \sum_{a=1}^{\kappa_{n,p}} \int_{\Sigma_p} \braket{ \varphi_- | f_{+,a} }_{\omega_{p}} \wedge \iota_{\om_p^+}\!\left(\braket{ f^\vee_{-,a} | \varphi_+ }_{\omega_{p}}\right).
\ee
This expression simplifies further as a localization formula on each of the $p{-}1$ points in the boundary divisor $\partial \Sigma_p$ of $\Sigma_p$.

\begin{lemma}\label{lemma-3-2}
	The intersection number $\braket{\varphi_- | \varphi_+ }_{\omega_{p-1}}\!$ evaluates to
	\be\label{lemma-22}
	\braket{\varphi_- | \varphi_+ }_{\omega_{p-1}} = \frac{1}{\Lambda^2} \sum_{r=1}^{p-1} \sum_{a=1}^{\kappa_{n,p}} \Res_{z_p = z_r} \!\left(  \braket{ \varphi_- | f_{+,a} }_{\omega_{p}} \psi_{+,a}^r \right),
	\ee
	where $\psi_{+,a}^r$ is the unique $\kappa_{n,p}$-vector of holomorphic functions solving $\bm\nabla_{p}^+ \psi_{+,a}^r = \braket{ f^\vee_{-,a} | \varphi_+ }_{\omega_{p}}$ around $z_p = z_r$. 
\end{lemma}
\begin{proof}
	We use a generalization of the proof of Theorem~\ref{theorem-21} in the $n=4$ case. The starting point is an explicit construction of the form $\iota_{\om_p^+} (\sigma_{+,a})$ with compact support around $\partial \Sigma_p = \{z_1, z_2, \ldots, z_{p-1}\}$, which reads
	\be\label{iota-sigma}
	\iota_{\om_p^+} (\sigma_{+,a}) = \sigma_{+,a} - \bm\nabla_p^+ \left( \sum_{r=1}^{p-1} h_{r}\, \psi_{+,a}^r \right).
	\ee
	By construction it is cohomologous to $\sigma_{+,a}$. The regulator function $h_{r}$ is equal to one in the infinitesimal disk around $z_r$, say $V_r = \{ |z_p {-} z_r| < \varepsilon \}$, it is zero outside of the disk $U_r = \{ |z_p {-} z_r | < \tilde\varepsilon \}$ for $0 < \varepsilon < \tilde\varepsilon \ll 1$, and interpolates smoothly between the two values on $U_r {\setminus} V_r$. The vector of functions $\psi_{+,a}^r$ is the unique holomorphic solution of
	\be
	\bm\nabla_{p}^+ \psi_{+,a}^r = d\psi_{+,a}^r + \sum_{b=1}^{\kappa_{n,p}} \psi_{+,b}^r \wedge (\om_p^+)_{ba} = \sigma_{+,a}
	\ee
	around $z_p = z_r$ (in this step one assumes that the kinematics is generic). Therefore \eqref{iota-sigma} becomes
	\be\label{iota-sigma-2}
	\iota_{\om_p^+} (\sigma_{+,a}) = \left( 1 - \sum_{r=1}^{p-1} h_r \right) \sigma_{+,a} - \sum_{r=1}^{p-1} dh_{r}\, \psi_{+,a}^r,
	\ee
	which vanishes on each $V_q$ since there we have $h_q {=} 1$, $h_{r\neq q} {=} dh_r {=} 0$. Hence $\iota_{\om_p^+} (\sigma_{+,a})$ has compact support. Plugging it into the expression for intersection numbers we find
	\be
	\braket{\varphi_- | \varphi_+ }_{\omega_{p-1}} = \frac{1}{2\pi i \Lambda^2} \sum_{r=1}^{p-1} \sum_{a=1}^{\kappa_{n,p}} \int_{U_{r} {\setminus} V_r} \sigma_{-,a} \wedge dh_r\, \psi_{+,a}^r,
	\ee
	given that the first term in \eqref{iota-sigma-2} vanishes once contracted with the holomorphic forms $\sigma_{-,a}$. We also used the fact that the integrand has support only on the annuli $U_{r} {\setminus} V_r$ where $dh_r$ are non-zero. Rewriting the integrand as $-d(\sigma_{-,a} h_r\, \psi_{+,a}^r)$ we can use Stokes' theorem to localize the integral on $\partial (U_{r} {\setminus} V_r) = \partial U_r - \partial V_r$. Due to the fact that $h_r{=}0$ on $\partial U_r$, while $h_r{=}1$ on $\partial V_r$, the result can be written as a sum of contour integrals around all components of $\partial \Sigma_p$:
	\be
	\braket{\varphi_- | \varphi_+ }_{\omega_{p-1}} = \frac{1}{\Lambda^2} \sum_{r=1}^{p-1} \sum_{a=1}^{\kappa_{n,p}} \frac{1}{2\pi i}\oint_{ |z_p {-} z_r| = \varepsilon} \!\!\! \sigma_{-,a}\, \psi^r_{+,a}.
	\ee
	After substituting the expressions for $\sigma_{\pm,a}$ it gives the required result.
\end{proof}

Since the right-hand side of \eqref{lemma-22} is expressed entirely in terms of residues of quantities computed on ${\cal F}_{n,p}$, this result provides the earliest example of a recursion relation for intersection numbers. It is clear that in order to make explicit computations we need the connections $\bm\nabla_p^+$, which will be constructed in Section~\ref{sec:braid-matrices}. Although not necessary, the choice of orthonormal bases of twisted forms greatly simplifies the recursion and thus we build them in Section~\ref{sec:fibration-bases}. The final simplification will be given in Section~\ref{sec:recursion-relations-reprise}, which will result in a streamlined version of the recursion relations.

Due to the fact that $\psi^{r}_{+,a}$ are evaluated in the neighbourhood of each $z_r$ by taking residues, only a few leading terms of their Laurent expansion around $z_p = z_r$ are actually needed. Thus the differential equation $\bm\nabla_{p}^+ \psi_{+,a}^r = \sigma_{+,a}$ can be solved by a holomorphic expansion around $z_p = z_r$. Equating the left- and right-hand sides order-by-order in $(z_p{-}z_r)$ fixes the holomorphic expansion of $\psi_{+,a}^r$ uniquely. This relies on an important assumption that $\Res_{z_p = z_r}(\om_p^+)$ is generic enough (for example, it does not have non-negative eigenvalues), which physically corresponds to the assumption of generic kinematics.\footnote{The manipulations in Lemma~\ref{lemma-3-2} require \emph{any} solution $\psi_{+,a}^r$ of the Pfaffian system of differential equations $\bm\nabla_{p}^+ \psi_{+,a}^r = \sigma_{+,a}$. However, the homogeneous solution of this system, given simply by $\psi_{+}^r = \mathbf{c}\, {\cal P}\exp ( - \int_{z_r}^{z_p} \om_p^+ )$ for a constant vector $\mathbf{c}$, is multi-valued in the neighbourhood of each $z_p = z_r$ and hence does not have a holomorphic expansion. This is why we only consider inhomogeneous solutions.}

By counting powers of $(z_p {-} z_r)$ it is straightforward to check that the residue can be non-zero only if
\be
\ord_{z_p = z_r} (\sigma_{-,a}) + \ord_{z_p = z_r} (\sigma_{+,a}) \,\leq\, -2,
\ee
where $\ord_{z_p = z_r}$ denotes the order of a zero at $z_p = z_r$, taken as the minimum of orders for each element in a vector. Therefore $\psi_{+,a}^r$ needs to be expanded only up to order $(z_p {-} z_r)^k$ for $k = -\ord_{z_p = z_r} (\sigma_{-,a}) - 1$. Moreover, since $\om_p^+$ are logarithmic, the expansion starts at the order $(z_p {-} z_r)^m$ for $m = \ord_{z_p = z_r} (\sigma_{+,a}) + 1$.

For completeness and later convenience let us also give an alternative formula for the recursion, which can be obtained by imposing compact support on the forms $\sigma_{-,a}$ instead of $\sigma_{+,a}$. In addition, it tells us the relation between the two connections $\bm\nabla^+_p$ and $\bm\nabla^-_p$.

\begin{lemma}\label{lemma-3-1}
	The intersection number $\braket{\varphi_- | \varphi_+ }_{\omega_{p-1}}\!$ evaluates to
	\be\label{intersection-recursion-alternative}
	\braket{\varphi_- | \varphi_+ }_{\omega_{p-1}} = -\frac{1}{\Lambda^2} \sum_{r=1}^{p-1} \sum_{a=1}^{\kappa_{n,p}} \Res_{z_p = z_r} \!\left(  \psi_{-,a}^r \braket{ f^\vee_{-,a} | \varphi_+ }_{\omega_{p}} \right),
	\ee
	where $\psi_{-,a}^r$  is the unique $\kappa_{n,p}$-vector of holomorphic functions solving
	\be
	\bm\nabla_{p}^- \psi_{-,a}^r := d\psi_{-,a}^r - \sum_{b=1}^{\kappa_{n,p}} \psi_{-,b}^r \wedge (\om_{p}^-)_{ba} = \braket{ \varphi_- | f_{+,a} }_{\omega_{p}}
	\ee
	around $z_p = z_r$. Provided that $\la f_{-,a}^\vee | f_{+,b} \ra_{\omega_p} = \delta_{ab}$ the Gauss--Manin connection $\bm\nabla_{p}^-$ is related to $\bm\nabla_{p}^+$ by $(\om_p^{-})_{ba} = (\om_p^{+})_{ab}$.
\end{lemma}
\begin{proof}
	The formula \eqref{intersection-recursion-alternative} is proven in the same way as that in Lemma~\ref{lemma-3-2}. To prove the second part of the lemma we use the fact that $\la f_{-,a}^\vee | f_{+,b} \ra_{\omega_p}$ is independent of the variables on the base space $\Sigma_p$ to write
	\be
	0 = d_{\Sigma_p} \la f_{-,a}^\vee | f_{+,b} \ra_{\omega_p},
	\ee
	where $d_{\Sigma_p}$ denotes the differential on $\Sigma_p$. By bilinearity of the intersection pairing we obtain
	\be
	0 = \sum_{c=1}^{\kappa_{n,p}} (\om_p^-)_{ac}\, \la f_{-,c}^\vee | f_{+,b} \ra_{\omega_p} + \sum_{c=1}^{\kappa_{n,p}} (\om_p^+)_{bc}\, \la f_{-,a}^\vee | f_{+,c} \ra_{\omega_p}.
	\ee
	Using the orthonormality condition $\la f_{-,a}^\vee | f_{+,b} \ra_{\omega_p} = \delta_{ab}$ we find $(\om_p^-)_{ab} = (\om_p^+)_{ba}$.
\end{proof}

Finally, let us comment on the fact that even though we are considering the case of moduli spaces $\M_{0,n}$, analogous recursion relations can be written down for intersection numbers defined on any other space admitting a fibration.\footnote{In the cases when orthonormal bases are not known, i.e., $\la f_{-,a}^\vee | f_{+,b}\ra_{\omega_p} =: \mathbf{C}_{ab} \neq \delta_{ab}$, following the above steps gives $\sigma_{-,a} = \sum_b \la \varphi_- | f_{+,b} \ra_{\omega_p} \mathbf{C}^{-1}_{ba}$, $\sigma_{+,a} = \sum_b \mathbf{C}^{-1}_{ab} \la f_{-,b}^\vee | \varphi_+ \ra_{\omega_p}$, and $(\bm\omega_{p}^+)_{ab} = \sum_c \mathbf{C}^{-1}_{bc} \la f_{-,c}^\vee | (d_{\Sigma_{p}} \!{+} \omega_{p-1} {\wedge}) f_{+,a} \ra_{\omega_p}$, $(\bm\omega_{p}^-)_{ab} = -\sum_c \la (d_{\Sigma_{p}} \!{-} \omega_{p-1} {\wedge}) f_{-,a}^\vee | f_{+,c} \ra_{\omega_p} \mathbf{C}_{cb}^{-1}$. The recursion relations keep their form; the functions $\psi_{\pm,a}^r$ are solutions of $\bm\nabla_p^\pm \psi_{\pm,a}^r = \sigma_{\pm,a}$ with $\sigma_{\pm,a}$ and $\bm\omega_p^{\pm}$ given above.}

\subsection{\label{sec:solving}Solving for Twisted Cohomologies on All Fibres}

\textsc{Before being able to perform} explicit computations with the above recursion relations we need to construct the connections $\bm\nabla^\pm_p$ as well as orthonormal sets of bases on each fibre ${\cal F}_{n,p}$. We address these two problems in turn.

\subsubsection{\label{sec:braid-matrices}Braid Matrices}

\textsc{Since by Lemma}~\ref{lemma-3-1} the one-forms $\om_p^-$ are related to $\om_p^+$ by a transpose for appropriate choices of bases, in order to specify $\bm\nabla_{p}^{\pm}$ it is enough to find $\om_p^+$, which is what we focus on now.

In fact, given the natural embedding $\Sigma_p \subset \M_{0,p}$ in the moduli space of $p \geq 3$ punctures on a Riemann sphere, it is more intuitive to do computations on the full space $\M_{0,p}$ first and then project down to $\Sigma_p$ at the end the computation by simply setting $dz_4 = \cdots = dz_{p-1} = 0$. The one-form $\widetilde\om_p^+$ on $\M_{0,p}$ takes the general form
\be\label{braid-matrix-ansatz}
\widetilde\om^+_p := \frac{1}{\Lambda^2} \sum_{1 \leq i<j \leq p} \!\!\!\! \O^{ij}_{p}\, d\log(z_i - z_j),
\ee
which will serve as an ansatz for the following computations. Here $\Om^{ij}_p$ are $\kappa_{n,p} {\times} \kappa_{n,p}$ matrices for each $i,j,p$. We set $\O^{ij}_p = \O^{ji}_p$ and $\O^{ii}_p = 0$. Additionally, imposing that $\widetilde\om_p^+$ does not have poles at infinity gives:
\be
\sum_{j=1}^{p} \O^{ij}_p = 0
\ee
for each $i$. In the edge case $p{=}n$ we have $\Om^{ij}_n = (p_i {+} p_j)^2$ by matching with \eqref{omega} for massless kinematics. In a certain sense $\Om_p^{ij}$ for $p{<}n$ provide a generalization of these kinematic invariants.

The matrix $\widetilde\om_p^+$ can be understood as providing a higher-dimension representation of the fundamental group on $\M_{0,p}$,
\be\label{higher-rank-local-system}
\pi_1(\M_{0,p}) \,\to\, \GL( \kappa_{n,p}, \C) 
\ee
given by the path-ordered exponential ${\cal P}\exp \int_\gamma \widetilde\om^+_p$ for each $\gamma \in \pi_1(\M_{0,p})$. Any path contractible to a point is represented as the identity matrix $\mathbb{I}$, and for every pair of paths $\gamma_1$ and $\gamma_2$ we have the relation:
\be\label{flatness-condition}
{\cal P}\exp \int_{\gamma_1 \circ \gamma_2} \!\!\! \widetilde\om_p^+ = \left({\cal P}\exp \int_{\gamma_2} \widetilde\om_p^+ \right) \left({\cal P}\exp \int_{\gamma_1} \widetilde\om_p^+ \right),
\ee
where the product on the right-hand side is given by matrix multiplication. It follows that the matrices $\Om_p^{ij}$ describe how the punctures $z_i$ and $z_j$ braid around each other on $\M_{0,p}$ and hence we will call them \emph{braid matrices}.

Flatness condition on the local system \eqref{higher-rank-local-system} can be stated more conveniently as integrability of the connection, $({\bm\nabla}^+_p)^2 = 0$. Since $\widetilde\om_p^+$ is a closed one-form, the only non-trivial constraint it imposes on braid matrices is $\widetilde\om_p^+ \wedge \widetilde\om_p^+ = 0$, which in its expanded forms reads:
\be
\sum_{\substack{1 \leq i<j \leq p\\ 1 \leq k<\ell \leq p}} \!\!\!\! \O^{ij}_p\, \O^{k\ell}_p\, d\log(z_i - z_j) \wedge d\log(z_k - z_\ell) = 0.
\ee
Identities between quadratic combinations of $\lambda_{ij} = d\log(z_i - z_j)$ are generated by the Arnold relations \eqref{Arnold-relations} and yield the following conditions on the commutators of braid matrices:
\be\label{infinitesimal-pure-braid-relations}
\left[ \O^{ij}_p,\, \O^{k\ell}_p \right] =0, \qquad
\left[ \O^{ij}_p + \O^{jk}_p,\, \O^{ik}_p \right] = 0
\ee
for distinct $i,j,k,\ell$. The first set of constraints means that two loops $\circlearrowleft_{ij}$ and $\circlearrowleft_{k\ell}$ involving four distinct punctures do not affect each other, while the second is an infinitesimal form of the Yang--Baxter relations.\footnote{Defining $\mathbf{r}_{ij} := \int \Om^{ij}_p d\log(z_i {-} z_j)$ it is straightforward to check that \eqref{infinitesimal-pure-braid-relations} and \eqref{Arnold-relations} imply
	\be
	[ \mathbf{r}_{ij}, \mathbf{r}_{ik} ] +  [ \mathbf{r}_{ij}, \mathbf{r}_{jk} ] + [ \mathbf{r}_{ik}, \mathbf{r}_{jk} ] = 0,
	\ee
	which are the classical Yang--Baxter relations, see, e.g.,~\cite{jimbo1990yang}. They can be understood as the linearized version of the quantum Yang--Baxter relations
	\be
	\mathbf{R}_{ij} \mathbf{R}_{ik} \mathbf{R}_{jk} = \mathbf{R}_{jk} \mathbf{R}_{ik} \mathbf{R}_{ij}
	\ee
	for $\mathbf{R}_{ij} := \mathbb{I} + \alpha' \mathbf{r}_{ij} + \ldots$ in the limit $\alpha' \to 0$.
}
Together, the constraints \eqref{infinitesimal-pure-braid-relations} are known as the \emph{infinitesimal pure braid relations} \cite{aomoto1973theoreme,kohno1985serie,AIF_1987__37_4_139_0}.

The above general properties of braid matrices are not enough to fix their explicit expressions needed for computations of intersection numbers. As remarked before, $\widetilde\om_p^+$ depends on the choice of bases of twisted forms through the relation \eqref{omega-sigma}. Let us make a particular choice of such a basis. We start by introducing the following basis of $(n{-}3)!$ twisted forms on $H^{n-3}_\omega$ called the \emph{fibration basis}:
\be\label{fibration-basis}
\FB(\rho) := \bigwedge_{i=4}^{n} d\log \left(\frac{z_i - z_{\rho(i)}}{z_i - z_1}\right)
\ee
for a word $\rho$ of length $n{-}3$ (indexing starts with $i{=}4$) and $\rho(i) = 3,4,\ldots,i{-}1$. Recall that $dz_1 {=} dz_2 {=} dz_3 {=} 0$. By construction the forms $\FB(\rho)$ are $\PSL(2,\C)$ invariant.

In fact, this gives a natural indexing scheme we will use in the remainder of this section: for a collection of $\kappa_{n,p}$ length-$(n{-}p)$ words $Q$ indexing objects on ${\cal F}_{n,p}$, we construct $\kappa_{n,p-1} = (p-3) \kappa_{n,p}$ length-$(n{-}p{+}1)$ words $qQ$ with $q = 3,4,\ldots,p{-}1$ that index objects on ${\cal F}_{n,p-1}$. To make this more clear, let us introduce vectors of twisted forms called $\F_{p}^+$ for $3 \leq p \leq n$ defined by the recursion:
\be\label{fibration-basis-coefficients}
(\F_{p-1}^+)_q := \left( \frac{dz_p}{z_p - z_q} - \frac{dz_p}{z_p - z_1} \right) \wedge \F_{p}^+,
\ee
where the $q$-th entry of the vector $\F_{p-1}^+$ of length $\kappa_{n,p-1}$ is defined in terms of the length-$\kappa_{n,p}$ vector $\F_p^+$. This gives us $\F_3^+$ that consists of $(n{-}3)!$ twisted forms given in \eqref{fibration-basis} ordered lexicographically. The boundary condition for the recursion is $\F_n^+ = 1$.

The fibration basis owes its name to the fact that by the recursion \eqref{fibration-basis-coefficients} it provides bases for \emph{all} the fibre cohomologies. Note that in \eqref{fibration-basis-coefficients} we wrote out the $d\log$ form explicitly in order to avoid ambiguities in the way its is extended from $\Sigma_p$ to $\M_{0,p}$.

We can translate from twisted forms to forms with local coefficients by simply modifying the initial condition to $\F_n^+ = 1 \otimes \exp \int_\gamma \omega$, which does not affect the form of the recursion \eqref{fibration-basis-coefficients}. We use this setup until the end of this subsection, as it matches the form of the constraint \eqref{omega-sigma}, which now becomes:
\be\label{fibration-basis-DE}
d \mathbf{F}_{p}^+ = \widetilde\om_{p}^+ \wedge  \mathbf{F}_{p}^+.
\ee
The rows and columns of the matrices $\widetilde\om_p^+$ are indexed using our new convention. For $n{=}p$ we find simply $\widetilde\om_n^+ = \omega$ according to our expectations. The idea is to use both \eqref{fibration-basis-coefficients} and \eqref{fibration-basis-DE} in order to find $\Om_p^{ij}$ and hence $\widetilde\om_{p}^+$ recursively in $p$. We have the following result (recall that we are not interested in the cases when both $i,j \leq 3$).

\begin{lemma}\label{lemma-3-3}
	The entries $(\Om_{p-1}^{ij})_{qr}$ of the braid matrix $\Om^{ij}_{p-1}$ for $q,r=3,4,\ldots,p{-}1$ are determined recursively in terms of higher-$p$ matrices $\Om_{p}^{k\ell}$ as follows:

	\be
	( \O^{ij}_{p-1} )_{qr} \,=\, \begin{dcases}
		\O^{ij}_p \qquad &\text{if}\qquad q=r,\quad q\neq i,j,\\
		\O^{pj}_p + \O^{qj}_p \qquad &\text{if}\qquad q=r=i,\quad j\neq 1,2,\\
		\O^{pj}_p + \O^{qj}_p + \O^{pq}_p \qquad &\text{if}\qquad q=r=i,\quad j=1,2,\\
		\O^{ip}_p + \O^{iq}_p \qquad &\text{if}\qquad q=r=j,\\
		-\O^{pj}_p \qquad &\text{if}\qquad q=i,\quad r=j,\\
		-\O^{ip}_p \qquad &\text{if}\qquad q=j,\quad r=i,\\
		\O^{ip}_p \qquad &\text{if}\qquad j=1,\quad r=i, \quad q\neq i,\\
		\O^{pr}_p \qquad &\text{if}\qquad j=2,\quad q=i, \quad r\neq i,\\
		0 \qquad &\text{otherwise},
	\end{dcases}
	\ee
	for $i=4,5,\ldots,p{-}1$, $j=1,2,\ldots,p{-}1$, and $i\neq j$.
\end{lemma}

\begin{proof}
	The proof proceeds by solving the differential equation $d\F_{p-1}^+ = \widetilde\om_{p-1}^+ \wedge \F_{p-1}^+$ assuming that $\widetilde\om_p^+$ is already determined through $d\F_p^+ = \widetilde\om_p^+ \wedge \F_p^+$. We use the general form of $\widetilde\om_p^+$ from \eqref{braid-matrix-ansatz} and the recursion for the fibration basis \eqref{fibration-basis-coefficients}. We start with the case $i=4,5,\ldots,p{-}1$, $i \neq q$, for which the differential equation for the $q$-th element of $\F_{p-1}^+$ reads:
	\begin{align}
	\frac{\partial (\F_{p-1}^{+})_q}{\partial z_i}  &= \left( \frac{1}{z_p {-} z_q} - \frac{1}{z_p {-} z_1}\right) \frac{\partial \F^+_{p}}{\partial z_i}\nn\\
	&= \frac{1}{\Lambda^2}\left( \frac{1}{z_p {-} z_q} - \frac{1}{z_p {-} z_1}\right) \Bigg( \sum_{\substack{j=1\\ j \neq i,q}}^{p-1} \frac{\O^{ij}_p\, \F^+_{p}}{z_i - z_j} + \frac{\O^{iq}_p\, \F^+_{p}}{z_i - z_q} + \frac{\O_{p}^{ip}\, \F^+_{p}}{z_i - z_p} \Bigg),\label{lemma-30}
	\end{align}
	where we used the differential equation for $\F_p^+$. The first two terms in the sum already have the required form. On the last one we use the following partial fraction identity:
	\be
	\left( \frac{1}{z_p {-} z_q} {-} \frac{1}{z_p {-} z_1}\right) \frac{1}{z_i {-} z_p} = \left( \frac{1}{z_p {-} z_q} {-} \frac{1}{z_p {-} z_1}\right) \frac{1}{z_i {-} z_q} {-} \left( \frac{1}{z_p {-} z_i} {-} \frac{1}{z_p {-} z_1}\right) \left(\frac{1}{z_i {-} z_q} {+} \frac{1}{z_i {-} z_1} \right)\!.
	\ee
	After substituting it back into \eqref{lemma-30} and using the definition \eqref{fibration-basis-coefficients} to express the right-hand side in terms of $\F_{p-1}^+$ we obtain:
	\be\label{lemma-32}
	\frac{\partial (\F_{p-1}^+)_q}{\partial z_i} = \frac{1}{\Lambda^2} \Bigg( \sum_{\substack{j=1\\ j \neq q}}^{p-1} \frac{\O^{ij}_{p}\, (\F_{p-1}^+)_q}{z_i - z_j} + \frac{\left(\O^{ip}_p + \O^{iq}_p\right) (\F_{p-1}^+)_q}{z_i - z_q} - \frac{\O^{ip}_p\, (\F^{+}_{p-1})_i}{z_i - z_q} + \frac{\O^{ip}_p\, (\F^{+}_{p-1})_i}{z_i - z_1}\Bigg).
	\ee
	The right-hand side is expressed in terms of elements of $\F_{p-1}^+$. In the remaining case when $i=q$ we have the differential equation:
	\begin{align}
	\frac{\partial (\F_{p-1}^+)_q}{\partial z_q} &= \frac{(\F_{p-1}^+)_q}{(z_p - z_q)^2} +  \left( \frac{1}{z_p {-} z_q} - \frac{1}{z_p {-} z_1}\right) \frac{\partial \F_p^+}{\partial z_q}\nn\\
	&= \frac{(\F_{p-1}^+)_q}{(z_p - z_q)^2} + \frac{1}{\Lambda^2}\left( \frac{1}{z_p {-} z_q} - \frac{1}{z_p {-} z_1}\right) \Bigg( \sum_{\substack{j=1\\ j\neq q}}^{p-1} \frac{\O^{qj}_p\, \F^+_p}{z_q - z_j} + \frac{\O^{qp}_p\, \F_{p}^+}{z_q - z_p}\Bigg).
	\end{align}
	On the last term we use the identity:
	\be
	\left( \frac{1}{z_p {-} z_q} - \frac{1}{z_p {-} z_1} \right) \frac{1}{z_q {-} z_p} = \left( \frac{1}{z_p {-} z_q} - \frac{1}{z_p {-} z_1}\right)\frac{1}{z_q {-} z_1} - \frac{1}{(z_p {-} z_q)^2},
	\ee
	which gives
	\be\label{lemma-35}
	\frac{\partial (\F^{+}_{p-1})_q}{\partial z_q} = \frac{1}{\Lambda^2} \Bigg( \sum_{\substack{j=1\\ j\neq q}}^{p-1} \frac{\O^{qj}_p\, (\F_{p-1}^+)_q}{z_q - z_j} + \frac{\O^{qp}_p\, (\F^{+}_{p-1})_q}{z_q - z_1} + \frac{\left( \Lambda^{2}\, \I - \O_{p}^{qp} \right) \F^{+}_{p}}{(z_p - z_q)^2} \Bigg),
	\ee
	where all the terms except for the last one are in the required form. We have the following cohomology relation on $\Sigma_p$:
	\be\label{lemma-36}
	0 = \Bigg[\Bigg(\I\frac{\partial}{\partial z_p} + \frac{1}{\Lambda^2} \sum_{j=1}^{p-1} \frac{\O^{pj}_p}{z_p - z_j}\Bigg) \frac{\Lambda^2}{z_p - z_q}\Bigg] \F_{p}^+ = - \frac{\left(\Lambda^2\, \I - \O^{qp}_p\right) \F_{p}^+}{(z_p - z_q)^2} + \frac{1}{z_p - z_q} \sum_{\substack{j=1\\ j\neq q}}^{p-1} \frac{\O^{pj}_p\, \F^+_p }{z_p - z_j}.
	\ee
	Let us use the partial fraction identity on the final term:
	\be
	\frac{1}{(z_p {-} z_q)(z_p {-} z_j)} =  \left[\left( \frac{1}{z_p {-} z_q} - \frac{1}{z_p {-} z_1}\right) - \left(\frac{1}{z_p {-} z_j} - \frac{1}{z_p {-} z_1}\right)\right] \frac{1}{z_q {-} z_j}.
	\ee
	Therefore \eqref{lemma-36} gives the identity:
	\be
	\frac{\left(\Lambda^2\,\I - \O^{qp}_p\right) \F^+_p}{(z_p - z_q)^2} = \sum_{\substack{j=1\\ j\neq q}}^{p-1} \frac{\O^{pj}_p (\F^+_{p-1})_q}{z_q - z_j} - \sum_{\substack{j=3\\ j\neq q}}^{p-1} \frac{\O^{pj}_p (\F^+_{p-1})_j}{z_q - z_j} - \left(\frac{1}{z_p {-} z_2} - \frac{1}{z_p {-} z_1} \right) \frac{\O^{p2}_p\, \F^+_p}{z_q - z_2}.
	\ee
	In order to resolve the final term we use the following cohomology relation on $\Sigma_p$:
	\begin{align}
	0 &= \left[\left(\I\frac{\partial}{\partial z_p} + \frac{1}{\Lambda^2}\sum_{r=1}^{p-1} \frac{\O^{pr}_p}{z_p - z_r}\right) \frac{\Lambda^2}{z_q - z_2}\right] \F^+_p = \left( \frac{\O^{p1}_p}{z_p - z_1} + \frac{\O^{p2}_p}{z_p - z_2} + \sum_{r=3}^{p-1} \frac{\O^{pr}_p}{z_p - z_r}  \right) \frac{\F^+_p}{z_q - z_2} \nn\\
	&= \left(\frac{1}{z_p {-} z_2} - \frac{1}{z_p {-} z_1} \right) \frac{\O^{p2}_p\, \F^+_p}{z_q - z_2} + \sum_{r=3}^{p-1} \frac{\O^{pr}_{p}\, (\F^+_{p-1})_{r}}{z_q - z_2},
	\end{align}
	where in the final transition we used the fact that
	\be
	\O^{p1}_p = -\sum_{r=2}^{p-1} \O^{pr}_p.
	\ee
	Plugging everything back into \eqref{lemma-35} we obtain the final result:
	\begin{align}
	\frac{\partial (\F^+_{p-1})_q}{\partial z_q} = \frac{1}{\Lambda^2} \Bigg(& \sum_{\substack{j=3\\ j\neq q}}^{p-1} \frac{\left(\O^{pj}_p + \O^{qj}_p \right) (\F_{p-1}^+)_{q}}{z_q - z_j} - \sum_{\substack{j=3\\ j\neq q}}^{p-1} \frac{\O^{pj}_p\, (\F_{p-1}^+)_j}{z_q - z_j} + \sum_{\substack{r=3\\ r\neq q}}^{p-1} \frac{\O^{pr}_p\, (\F_{p-1}^+)_r}{z_q - z_2} \nn\\
	&\quad + \frac{\left(\O^{p1}_p + \O^{q1}_p + \O^{pq}_p \right)(\F^+_{p-1})_{q}}{z_q - z_1} + \frac{\left(\O^{p2}_p + \O^{q2}_p + \O^{pq}_p \right) (\F_{p-1}^+)_q }{z_q - z_2} \Bigg).\label{lemma-41}
	\end{align}
	Reading off entries of the matrices $\O^{ij}_{p-1}$ from \eqref{lemma-32} and \eqref{lemma-41} gives the required recursion relations.
\end{proof}

Note that even though the label $n$ is not displayed explicitly, the braid matrices depend on the total number of particles through the boundary condition $\Om_n^{ij} = (p_i {+} p_j)^2$. Introducing the notation $s_{i_1 i_2 \cdots i_k} = (p_{i_1} {+} p_{i_2} {+} \ldots {+} p_{i_k})^2$ let us give several examples of braid matrices at low multiplicity.

\begin{example}
	For $n{=}5$ we have $\Om_{5}^{ij} = s_{ij}$, while the non-vanishing braid matrices at $p{=}4$ are:
	\be
	\O_{4}^{41} = \left(
	\begin{array}{cc}
		s_{14} & s_{45} \\
		0 & s_{145} \\
	\end{array}
	\right),\qquad
	\O_{4}^{42} = \left(
	\begin{array}{cc}
		s_{24} & 0 \\
		s_{35} & s_{245} \\
	\end{array}
	\right),\qquad
	\O_{4}^{43} = \left(
	\begin{array}{cc}
		s_{34} {+} s_{45} & -s_{45} \\
		-s_{35} & s_{34} {+} s_{35} \\
	\end{array}
	\right),
	\ee
	where we used momentum conservation to simplify the entries. Their eigenvalues are $\text{eig}(\Om_4^{4j}) = \{s_{j4}, s_{j45}\}$.
\end{example}

\begin{example}
	For $n{=}6$ apart from $\Om_{6}^{ij} = s_{ij}$, at $p{=}5$ we have the following non-zero matrices:
	\begin{gather}
	\O_{5}^{41} = \left(
	\begin{array}{ccc}
	s_{14} & s_{46} & 0 \\
	0 & s_{146} & 0 \\
	0 & s_{46} & s_{14} \\
	\end{array}
	\right)\!,\;
	\O_{5}^{42} = \left(
	\begin{array}{ccc}
	s_{24} & 0 & 0 \\
	s_{36} & s_{246} & s_{56} \\
	0 & 0 & s_{24} \\
	\end{array}
	\right)\!,\;
	\O_{5}^{43} = \left(
	\begin{array}{ccc}
	s_{34} {+} s_{46} & - s_{46} & 0 \\
	- s_{36} & s_{34} {+} s_{36} & 0 \\
	0 & 0 & s_{34} \\
	\end{array}
	\right)\!,\nn\\
	\O_{5}^{45} = \O_{5}^{54} = \left(
	\begin{array}{ccc}
	s_{45} & 0 & 0 \\
	0 &  s_{45} {+} s_{56} & - s_{56} \\
	0 & - s_{46} & s_{45} {+} s_{46} \\
	\end{array}
	\right),\qquad
	\O_{5}^{51} = \left(
	\begin{array}{ccc}
	s_{15} & 0 & s_{56} \\
	0 & s_{15} & s_{56} \\
	0 & 0 & s_{156} \\
	\end{array}
	\right),\\
	\O_{5}^{52} = \left(
	\begin{array}{ccc}
	s_{25} & 0 & 0 \\
	0 & s_{25} & 0 \\
	s_{36} & s_{46} & s_{256} \\
	\end{array}
	\right),\qquad
	\O_{5}^{53} = \left(
	\begin{array}{ccc}
	s_{35} {+} s_{56} & 0 & -s_{56} \\
	0 & s_{35} & 0 \\
	-s_{36} & 0 & s_{35} {+} s_{36} \\
	\end{array}
	\right).\nn
	\end{gather}
	Their eigenvalues are $\text{eig}(\Om_5^{ij}) = \{s_{ij}, s_{ij}, s_{ij6}\}$. For $p{=}4$ we have only three non-zero braid matrices:
	\begin{gather}
	\O_{4}^{41} = \left(
	\begin{array}{cccccc}
	s_{14} & s_{46} & 0 & s_{45} & 0 & 0 \\
	0 & s_{146} & 0 & 0 & s_{45} {+} s_{56} & -s_{56} \\
	0 & s_{46} & s_{14} & 0 & -s_{46} & s_{45} {+} s_{46} \\
	0 & 0 & 0 & s_{145} & s_{46} & s_{56} \\
	0 & 0 & 0 & 0 & s_{1456} & 0 \\
	0 & 0 & 0 & 0 & 0 & s_{1456} \\
	\end{array}
	\right),\qquad\!\!
	\O_{4}^{42} = \left(
	\begin{array}{cccccc}
	s_{24} & 0 & 0 & 0 & 0 & 0 \\
	s_{36} & s_{246} & s_{56} & 0 & 0 & 0 \\
	0 & 0 & s_{24} & 0 & 0 & 0 \\
	s_{35} {+} s_{56} & 0 & -s_{56} & s_{245} & 0 & 0 \\
	0 & s_{35} & 0 & s_{36} & s_{2456} & 0 \\
	-s_{36} & 0 & s_{35} {+} s_{36} & s_{36} & 0 & s_{2456} \\
	\end{array}
	\right),\nn\\
	\O_{4}^{43} \!=\! \left(
	\begin{array}{cccccc}
	s_{34} {+} s_{45} {+} s_{46} & - s_{46} & 0 & - s_{45} & 0 & 0 \\
	-s_{36} & s_{34} {+} s_{36} {+} s_{45} {+} s_{56} & -s_{56} & 0 & -s_{45} {-} s_{56} & s_{56} \\
	0 & -s_{46} & s_{34} {+} s_{45} {+} s_{46} & 0 & s_{46} & -s_{45} {-} s_{46} \\
	-s_{35} {-} s_{56} & 0 & s_{56} & s_{34} {+} s_{35} {+} s_{46} {+} s_{56}
	& -s_{46} & -s_{56} \\
	0 & -s_{35} & 0 & -s_{36} & s_{34} {+} s_{35} {+} s_{36} & 0 \\
	s_{36} & 0 & -s_{35} {-} s_{36} & -s_{36} & 0 & s_{34} {+} s_{35} {+} s_{36} \\
	\end{array}
	\right)\!.
	\end{gather}
	Eigenvalues of these matrices are $\text{eig}(\Om_{4}^{4j}) = \{ s_{j4}, s_{j4}, s_{j45}, s_{j46}, s_{j456}, s_{j456} \}$.
\end{example}
Notice the factors $s_{ij}$ on the diagonal of each $\O_{p}^{ij}$. As a matter fact, departure from the purely diagonal matrix $s_{ij}\I$ is what measures non-triviality of the fibre bundle. Based on explicit computations for $n \leq 10$ we conjecture that eigenvalues of braid matrices always correspond to physical factorization channels of the form $s_{i_1 i_2 \cdots i_k}$.

Projected onto the subspace $\Sigma_p$ the one-form \eqref{braid-matrix-ansatz} becomes:
\be
\om_p^+ = \frac{1}{\Lambda^2} \sum_{i = 1}^{p-1} \Om_{p}^{ip} \frac{dz_p}{z_p - z_i},
\ee
which is the one entering the definitions of $\bm\nabla_p^\pm$. Note that only $\Om_p^{ip} = \Om_p^{pi}$ appear in this expression. However, the other braid matrices are still needed in order to compute $\Om_{p-1}^{k\ell}$ recursively. It was in fact one of the main reasons to perform computations on the full moduli space $\M_{0,p}$ in the first place.

\subsubsection{\label{sec:fibration-bases}Fibration Bases}

\textsc{Let us complete the description} of bases of twisted forms on each ${\cal F}_{n,p}$ by giving a definition of dual bases. Following the recursive structure of $\F_p^+$ from \eqref{fibration-basis-coefficients} we introduce:
\be\label{dual-fibration-basis}
(\F^-_{p-1})_q := \left(\frac{dz_p}{z_p - z_q} - \frac{dz_p}{z_p - z_2}\right) (\Om_p^{pq})^\intercal \wedge \F^-_p,
\ee
which we call the dual fibration bases. Note the transpose on the braid matrix.  Once again, boundary conditions are $\F_n^- = 1$ for twisted forms and $\F_n^{-} = 1 \otimes \exp \int_\gamma -\omega$ for forms with local coefficients. In order to complete the description of bases needed in our recursion relations we prove the following result.

\begin{lemma}
	The fibration bases $\F_p^\pm$ are orthonormal with respect to the intersection pairing, i.e., $\la \F_p^- | (\F_p^+)^\intercal \ra_{\omega_p} = \I$.
\end{lemma}
\begin{proof}
	We use recursion in $p$ starting with the trivial case $p=n$, for which $\la 1 | 1\ra_{\omega_n} = 1$. Assuming that the lemma is true for $p$, we evaluate $\la (\F_{p-1}^-)_q | (\F_{p-1}^+)_s^\intercal \ra_{\omega_{p-1}}$ for each $q,s$ using the result of Lemma~\ref{lemma-3-2},
	\be\label{FF-definition}
	\la (\F_{p-1}^-)_q | (\F_{p-1}^+)_s^\intercal \ra_{\omega_{p-1}} = \frac{1}{\Lambda^2} \sum_{r=1}^{p-1} \Res_{z_p = z_r} \!\left(  \braket{ (\F_{p-1}^-)_q | (\F_p^+)^\intercal }_{\omega_{p}} \!(\bm\Psi^r)_s \right).
	\ee
	For each $r$ and $s$, $(\bm\Psi^r)_s$ is a $\kappa_{n,p} {\times} \kappa_{n,p}$ matrix-valued holomorphic function satisfying $\bm\nabla_{p}^+ (\bm\Psi^r)^\intercal_s = \braket{ \F_p^- | (\F^+_{p-1})^\intercal_s }_{\omega_{p}}$ around $z_p = z_r$. Notice that $(\bm\Psi^r)_s$ is transposed here as $\om_p^+$ acts on the same index as the one contracted in \eqref{FF-definition}. Using the definition of $\F_{p-1}^\pm$ and the inductive assumption we have:
	\be\label{lemma-63}
	\braket{ (\F_{p-1}^-)_q | (\F_p^+)^\intercal }_{\omega_{p}} = \left(\frac{dz_p}{z_p - z_q} - \frac{dz_p}{z_p - z_2}\right) (\Om_p^{pq})^\intercal,
	\ee
	as well as
	\be
	\braket{ \F_p^- | (\F^+_{p-1})^\intercal_s }_{\omega_{p}} = \left( \frac{dz_p}{z_p - z_s} - \frac{dz_p}{z_p - z_1} \right) \I.
	\ee
	The equation \eqref{lemma-63} tells us that it is enough to solve for $(\bm\Psi^r)^\intercal_s$ up to the order $(z_p {-} z_r)^{0}$ for any $r$. Expanding the differential equation for $(\bm\Psi^r)^\intercal_s$ we find:
	\be
	\frac{\partial (\bm\Psi^r)_s^\intercal}{\partial z_p} + (\bm\Psi^r)_s^\intercal \frac{1}{\Lambda^2} \sum_{j=1}^{p-1} \frac{\Om_p^{pj}}{z_p - z_j} = \left( \frac{1}{z_p - z_s} - \frac{1}{z_p - z_1} \right) \I.
	\ee
	Therefore to leading order we have:
	\be\label{Psi-expansion}
	(\bm\Psi^r)^\intercal_s = \Lambda^2\, (\Om_p^{pr})^{-1} (\delta_{rs} - \delta_{r1}) + \ldots,
	\ee
	where the ellipsis denotes terms of order ${\cal O}(z_p - z_r)$. Note that braid matrices themselves are independent of the coordinates $z_i$. Plugging \eqref{lemma-63} and \eqref{Psi-expansion} back into \eqref{FF-definition} gives:
	\be
	\la (\F_{p-1}^-)_q | (\F_{p-1}^+)_s^\intercal \ra_{\omega_{p-1}} {=} \sum_{r=1}^{p-1} \Res_{z_p = z_r} \!\left( \! \left(\frac{dz_p}{z_p - z_q} {-} \frac{dz_p}{z_p - z_2}\right) (\Om_p^{pq})^\intercal ((\Om_p^{pr})^{\intercal})^{-1} (\delta_{rs} {-} \delta_{r1}) {+} \ldots\! \right)\!.
	\ee
	Recalling that $q,s \geq 3$, the only non-zero residue comes from $q{=}r{=}s$, giving
	\be
	\la (\F_{p-1}^-)_q | (\F_{p-1}^+)_s^\intercal \ra_{\omega_{p-1}} = \Res_{z_p = z_q} \!\left(  \frac{dz_p}{z_p - z_q} (\Om_p^{pq})^\intercal ((\Om_p^{pq})^{\intercal})^{-1} \delta_{qs} + \ldots \right) = \delta_{qs} \I,
	\ee
	which proves the required result.
\end{proof}

The above lemma shows that \eqref{dual-fibration-basis} defines an orthonormal dual basis $\FB(\rho)^\vee$ of $H^{n-3}_{-\omega}$, for words $\rho$ of length $n{-}3$ and $\rho(i) = 3,4,\ldots,i{-}1$, analogous to that in \eqref{fibration-basis}. The orthonormality can be stated as $\la \FB(\rho)^\vee | \FB(\sigma) \ra_\omega = \delta_{\rho\sigma}$,

Note that while orthonormality on $\M_{0,n}$ can be always achieved for any set of $(n{-}3)!$ twisted forms by an appropriate rotation of bases, the above result shows that fibration bases have a much stronger property: orthonormality on \emph{all} the fibres at once, which greatly simplifies the form of the recursion relations. By contrast, bases of $\PT(\alpha)$ and $\PT(\alpha)^\vee$ we used before do not have this property. The choice of fibration bases is however not unique: for example we made a choice of special punctures $z_1$ and $z_2$ in the definition of $\F^+_p$ and $\F^-_p$ respectively. Changing this choice together with the range of $q$ yields other sets of orthonormal bases on each ${\cal F}_{n,p}$.

We close this subsection by giving a translation between the fibration bases and Parke--Taylor bases. The specific choice of the definition \eqref{fibration-basis} was made such that the final element of the fibration basis $\FB(\rho)$ coincides with the Parke--Taylor form with the canonical ordering, i.e.,
\be\label{FB-PT}
\FB(34\cdots n{-}1) = \PT(12\cdots n).
\ee
In fact, all the elements of the fibration basis are already written in a form that makes them linear combinations of $\PT(\alpha)$ algebraically (in other words, on the level of $H^{n-3}(\M_{0,n}, \Z)$). We can write this relation using a rotation matrix $\mathbf{T}$,
\be\label{bases-rotation}
\FB(\rho) = \sum_{\widehat{\alpha}} \mathbf{T}_{\rho\widehat{\alpha}}\, \PT(3\widehat{\alpha}12), \qquad \FB(\rho)^\vee = \sum_{\widehat{\alpha}} (\mathbf{T}^{\intercal})^{-1}_{\rho\widehat{\alpha}}\, \PT(3\widehat{\alpha}12)^\vee,
\ee
where we used orthonormality of the fibration and Parke--Taylor bases to arrive at the second equation, which relates the dual bases to each other. The sums go over $(n{-}3)!$ permutations $\widehat{\alpha}$ of labels $\{4,5,\ldots,n\}$. Here the choice of fixed labels $\{1,2,3\}$ reflects the $\SL(2,\C)$ fixing employed in this section. It yields the specific orderings of $\PT(\alpha)^\vee$, which (unlike $\PT(3\widehat{\alpha}12) = \PT(123\widehat{\alpha})$) are not cyclic-invariant. Let us give explicit expressions for the rotation matrix $\mathbf{T}$ and its inverse at low multiplicity, which are found to have entries consisting of only $0,\pm 1$.
\begin{example}
	For $n{=}5$ we find that the two sets of bases are related by
	\be
	\left(
	\begin{array}{c}
		\FB(33) \\
		\FB(34) \\
	\end{array}
	\right)=
	\left(
	\begin{array}{cc}
		\SetToWidest{1} & \SetToWidest{1} \\
		1 & 0 \\
	\end{array}
	\right)\left(
	\begin{array}{c}
		\PT(12345) \\
		\PT(12354) \\
	\end{array}
	\right),
	\ee
	while their duals obey
	\be
	\left(
	\begin{array}{c}
		\FB(33)^\vee \\
		\FB(34)^\vee \\
	\end{array}
	\right)=
	\left(
	\begin{array}{cc}
		\SetToWidest{0} & \SetToWidest{1} \\
		1 & -1 \\
	\end{array}
	\right) \left(
	\begin{array}{c}
		\PT(34512)^\vee \\
		\PT(35412)^\vee \\
	\end{array}
	\right).
	\ee
\end{example}

\begin{example}
	For $n{=}6$ the rotation matrix translates between the two bases as follows,
	\be
	\left(
	\begin{array}{c}
		\FB(333) \\
		\FB(334) \\
		\FB(335) \\
		\FB(343) \\
		\FB(344) \\
		\FB(345) \\
	\end{array}
	\right)=
	\left(
	\begin{array}{cccccc}
		\SetToWidest{1} & \SetToWidest{1} & \SetToWidest{1} & \SetToWidest{1} & \SetToWidest{1} & \SetToWidest{1} \\
		1 & 1 & 1 & 0 & 0 & 0 \\
		1 & 0 & 1 & 1 & 0 & 0 \\
		1 & 1 & 0 & 0 & 1 & 0 \\
		1 & 1 & 0 & 0 & 0 & 0 \\
		1 & 0 & 0 & 0 & 0 & 0 \\
	\end{array}
	\right)\left(
	\begin{array}{c}
		\PT(123456) \\
		\PT(123465) \\
		\PT(123546) \\
		\PT(123564) \\
		\PT(123645) \\
		\PT(123654) \\
	\end{array}
	\right),
	\ee
	and similarly using the second equation in \eqref{bases-rotation} we find
	\be
	\left(
	\begin{array}{c}
		\FB(333)^\vee \\
		\FB(334)^\vee \\
		\FB(335)^\vee \\
		\FB(343)^\vee \\
		\FB(344)^\vee \\
		\FB(345)^\vee \\
	\end{array}
	\right)=
	\left(
	\begin{array}{cccccc}
		\SetToWidest{0} & \SetToWidest{0} & \SetToWidest{0} & \SetToWidest{0} & \SetToWidest{0} & \SetToWidest{1} \\
		0 & 0 & 1 & -1 & 0 & 0 \\
		0 & 0 & 0 & 1 & 0 & -1 \\
		0 & 0 & 0 & 0 & 1 & -1 \\
		0 & 1 & -1 & 1 & -1 & 0 \\
		1 & -1 & 0 & -1 & 0 & 1 \\
	\end{array}
	\right)\left(
	\begin{array}{c}
		\PT(345612)^\vee \\
		\PT(346512)^\vee \\
		\PT(354612)^\vee \\
		\PT(356412)^\vee \\
		\PT(364512)^\vee \\
		\PT(365412)^\vee \\
	\end{array}
	\right).
	\ee
\end{example}

We notice that even though in going from the left to the right equation in \eqref{bases-rotation} we used orthonormality with respect to the intersection pairing, in the above examples the equalities hold algebraically, that is, without using any cohomology relations.

\subsection{\label{sec:recursion-relations-reprise}Recursion Relations (Reprise)}

\textsc{We finish this section} by giving a streamlined version of the recursion relations for intersection numbers. The first step is using Proposition~\ref{proposition} with $U{=}W{=}H^{n-3}_{-\omega}$ and $V{=}X{=}H^{n-3}_{\omega}$ and fibration bases to write out $\braket{\varphi_- | \varphi_+}_\omega$ for a pair of twisted forms $\varphi_\pm \in H^{n-3}_{\pm\omega}$ as:
\be\label{recursion-first-line}
\braket{\varphi_- | \varphi_+}_\omega = \braket{\varphi_- | \F_3^{+}}_\omega^\intercal \braket{ \F_3^{-} | \varphi_+ }_\omega.
\ee
Recall that both $\F_{3}^\pm$ are length-$(n{-}3)!$ vectors of twisted forms. We define the vectors:
\be\label{bold-varphi}
\bm\varphi_3^- := \braket{\varphi_- | \F_3^{+}}_\omega, \qquad \bm\varphi_3^+ := \braket{ \F_3^{-} | \varphi_+ }_\omega,
\ee
such that the intersection number is given by their contraction $(\bm\varphi_3^-)^\intercal \bm\varphi_3^+$. At first sight this might not seem like an efficient representation, but since the recursion relations already depend on $(n{-}3)!$-vectors of twisted forms in the intermediate steps, no efficiency is actually sacrificed. Moreover, it lets us treat both forms entering the intersection number separately.

Given that $\varphi_\pm$ provides the input of our computation and $\bm\varphi_3^\pm$ is the output, our idea is to construct intermediate vectors of twisted forms $\bm\varphi_p^\pm$ which interpolate between the two ends:
\be\label{recursion-map}
\bm\varphi_n^\pm \,\to\, \bm\varphi_{n-1}^\pm \,\to\, \cdots \,\to\, \bm\varphi_{4}^\pm \,\to\, \bm\varphi_{3}^\pm,  
\ee
where $\bm\varphi_n^\pm := \varphi_\pm$. Each $\bm\varphi_p^\pm$ is a length-$\kappa_{n,p}$ vector of twisted forms on $\M_{0,p}$, given by a straightforward generalization of \eqref{bold-varphi}: $\bm\varphi_p^- := \la \varphi_- | \F_p^+ \ra_{\omega_p}$ and $\bm\varphi_p^+ := \la \F_p^- | \varphi_+ \ra_{\omega_p}$. The explicit recursion relation realizing \eqref{recursion-map} takes the form
\be\label{recursion}
(\bm\varphi_{p-1}^\pm)_r = \sum_{q=1}^{p-1} \Res_{z_p = z_q} \!\Big( \MM^\pm_{pqr}\, \bm\varphi_{p}^{\pm} \Big).
\ee
The matrices $\MM_{pqr}^\pm$ encode how the behaviour of the puncture $z_p \in \Sigma_p$ around each boundary component $z_q \in \partial\Sigma_p$ affects the $r$-th entry of $\bm\varphi_{p-1}^\pm$. We follow by giving their definition in terms of braid matrices $\Om^{pj}_{p}$.

Let us start by comparing the recursion for $\bm\varphi_{p-1}^-$ from \eqref{recursion} with Lemma~\ref{lemma-3-2} by sending $\varphi_+ \to (\F^+_{p-1})_r$. Equating the two formulae gives us 
\be\label{P-recursion-minus}
\bm\nabla_{p}^+ \MM_{pqr}^- = \frac{1}{\Lambda^2} \braket{ \F_p^- | (\F_{p-1}^+)_r }_{\omega_p} = \frac{1}{\Lambda^2}\left( \frac{dz_p}{z_p - z_r} - \frac{dz_p}{z_p - z_1} \right) \I.
\ee
Similarly, the $\bm\varphi_{p-1}^+$ case can be compared with Lemma~\ref{lemma-3-1} upon the change $\varphi_- \to (\F_{p-1}^-)_r$, which gives the condition
\be\label{P-recursion-plus}
\bm\nabla_p^- \MM_{pqr}^+ = -\frac{1}{\Lambda^2} \braket{ (\F_{p-1}^-)_r | \F_{p}^+ }_{\omega_p} = -\frac{1}{\Lambda^2} \left(\frac{dz_p}{z_p - z_r} - \frac{dz_p}{z_p - z_2}\right) (\Om_p^{pr})^\intercal.
\ee
Therefore finding $\MM_{pqr}^\pm$ amounts to solving these differential equations. Note that they only involve quantities known ahead of time (in particular they are independent of $\varphi_\pm$) and hence the most computationally involved part of the calculation needs to be performed only once.

Since the matrices $\MM_{pqr}^\pm$ enter a residue computation, we only need their expansion around each $z_p = z_q$. Hence we proceed by writing the holomorphic expansion:
\be
\MM_{pqr}^\pm = \sum_{k=0}^{\infty} {}^{k}\MM_{pqr}^\pm\, (z_p - z_q)^k ,
\ee
which starts at the constant term since the right-hand sides of \eqref{P-recursion-minus} and \eqref{P-recursion-plus} are logarithmic. Recall that the ranges of indices are $p=4,5,\ldots,n$, $q=1,2,\ldots,p{-}1$, $r=3,4,\ldots,p{-}1$. We obtain the following solutions.

\begin{lemma}\label{lemma-3-5}
	The leading matrix $^{0} \MM^-_{pqr}$ is given by:
	\be\label{P0-minus}
	{}^0 \MM_{pqr}^- = \begin{dcases}
		-\left(\O^{p1}_p \right)^{-1} & \text{if}\quad q = 1,\\
		\left(\O^{pr}_p\right)^{-1} & \text{if}\quad q = r,\\
		0 & \text{otherwise}.
	\end{dcases}
	\ee
	The subleading matrices ${}^{k} \MM_{pqr}^-$ for $k>0$ are given by the recursion:
	\be\label{Pk-minus}
	{}^{k} \MM_{pqr}^- = \Bigg( (-1)^{k-1}\! \left( \frac{\delta_{q\neq r}}{(z_q {-} z_r)^k} - \frac{\delta_{q\neq 1}}{(z_q {-} z_1)^k}\right) \I - \sum_{\substack{j=1\\ j\neq q}}^{p-1} \sum_{i=0}^{k-1} \frac{ {}^{i}\MM_{pqr}^- (-1)^{k-i-1} }{(z_q{-}z_j)^{k-i}} \,\O^{pj}_p \Bigg)\! \left( \O^{pq}_p + k \Lambda^2 \I \right)^{-1} \!\!\!\!\!\!\!,\;\;
	\ee
	where the matrices $\Om_{p}^{pj}$ are given in Lemma~\ref{lemma-3-3}.
\end{lemma}
\begin{proof}
	Expanding the differential equation \eqref{P-recursion-minus} around $z_p = z_q$ we obtain:
	\begin{align}
	&\sum_{k=1}^{\infty} {}^{k}\MM_{pqr}^-\, k\, (z_p{-}z_q)^{k-1} + \frac{1}{\Lambda^2} \sum_{k=0}^{\infty} {}^{k}\MM_{pqr}^-\, (z_p {-} z_q)^{k-1} \, \O^{pq}_p  \nn\\
	&\qquad+ \frac{1}{\Lambda^2} \sum_{\substack{j = 1\\ j\neq q}}^{p-1} \sum_{i=0}^{\infty} \sum_{\ell=1}^{\infty} \frac{{}^{i}\MM_{pqr}^- (-1)^{\ell-1} (z_p {-} z_q)^{i+\ell-1} }{(z_q {-} z_r)^{\ell}}\, \O^{pj}_p \\
	&= \frac{1}{\Lambda^2} \left( \frac{\delta_{qr}}{z_p{-}z_r} - \frac{\delta_{q1}}{z_p{-}z_1} \right)\I + \frac{1}{\Lambda^2} \sum_{k=1}^{\infty} (z_q {-} z_p)^{k-1} \left( \frac{\delta_{q\neq r}}{(z_q {-} z_r)^k} - \frac{\delta_{q\neq 1}}{(z_q {-} z_1)^k} \right) \I,\nn
	\end{align}
	where we used the fact that $(z_p {-} z_j)^{-1} = \sum_{i=1}^{\infty} (z_q {-} z_p)^{i-1} (z_q {-} z_j)^{-i}$. At order $(z_p {-} z_q)^{-1}$ we find
	\be
	{}^{0}\MM_{pqr}^- \, \O^{pq}_p = \left( \delta_{qr} - \delta_{q1} \right)\I,
	\ee
	which gives \eqref{P0-minus}. At order $(z_p {-} z_q)^{k-1}$ for $k>0$ after rearranging terms we obtain
	\be
	{}^{k} \MM_{pqr}^- \left( k \Lambda^2 \I +  \O^{pq}_p \right) + \sum_{\substack{j=1\\ j\neq q}}^{p-1} \sum_{i=0}^{k-1} \frac{{}^{i}\MM_{pqr}^- (-1)^{k-i-1}}{(z_q {-} z_j)^{k-i}}\, \O^{pj}_p = (-1)^{k-1}\! \left( \frac{\delta_{q\neq r}}{(z_q {-} z_r)^k} - \frac{\delta_{q\neq 1}}{(z_q {-} z_1)^k}\right) \I,
	\ee
	from which \eqref{Pk-minus} follows.
\end{proof}

\begin{lemma}\label{lemma-3-6}
	The leading matrix $^{0} \MM^+_{pqr}$ is given by:
	\be\label{P0-plus}
	{}^{0}\MM_{pqr}^+ = \begin{dcases}
		- (\O^{pr}_p)^\intercal \left((\O^{p2}_p)^\intercal \right)^{-1} & \text{if}\quad q = 2,\\
		\I & \text{if}\quad q = r,\\
		0 & \text{otherwise}.
	\end{dcases}
	\ee
	The subleading matrices ${}^{k} \MM_{pqr}^+$ for $k>0$ are given by the recursion:
	\begin{align}\label{Pk-plus}
	{}^{k}\MM_{pqr}^+  \!=\! \Bigg( &\!(-1)^{k-1}\!\left( \frac{\delta_{q\neq r}}{(z_q {-} z_r)^k} {-} \frac{\delta_{q\neq 2}}{(z_q {-} z_2)^k}\right) (\O^{pr}_p)^\intercal {-}\! \sum_{\substack{j=1\\ j\neq q}}^{p-1} \sum_{i=0}^{k-1} \frac{{}^{i}\MM_{pqr}^+ (-1)^{k-i-1}}{(z_q {-} z_j)^{k-i}} (\O^{pj}_p)^\intercal \!\Bigg)\!\! \left( (\O^{pq}_p)^\intercal {-} k \Lambda^2 \I \right)^{-1} \!\!\!\!\!\!\!,\;\;
	\end{align}
	where the matrices $\Om_{p}^{pj}$ are given in Lemma~\ref{lemma-3-3}.
\end{lemma}
\begin{proof}
	Expanding the differential equation \eqref{P-recursion-plus} around $z_p = z_q$ we obtain:
	\begin{align}
	&\sum_{k=1}^{\infty} {}^{k}\MM_{pqr}^+\, k\, (z_p{-}z_q)^{k-1} - \frac{1}{\Lambda^2}\sum_{k=0}^{\infty} {}^{k}\MM_{pqr}^+\, (z_{p} {-} z_{q})^{k-1}  (\O^{pq}_p)^\intercal \nn\\
	&\qquad - \frac{1}{\Lambda^2} \sum_{\substack{j = 1\\ j\neq q}}^{p-1} \sum_{i=0}^{\infty} \sum_{\ell=1}^{\infty} \frac{{}^{i}\MM_{pqr}^+\, (-1)^{\ell-1} (z_p {-} z_q)^{i+\ell-1} }{(z_q {-} z_r)^{\ell}}\, (\O^{pj}_p)^\intercal \\
	&= - \frac{1}{\Lambda^2} \left( \frac{\delta_{qr}}{z_p{-}z_r} - \frac{\delta_{q2}}{z_p{-}z_2} \right) (\O^{pr}_p)^\intercal - \frac{1}{\Lambda^2} \sum_{k=1}^{\infty} (z_q {-} z_p)^{k-1} \left( \frac{\delta_{q\neq r}}{(z_q {-} z_r)^k} - \frac{\delta_{q\neq 2}}{(z_q {-} z_2)^k} \right) (\O^{pr}_p)^\intercal,\nn
	\end{align}
	where we used the fact that $(\om_p^-)^\intercal = \om_p^+$. At the leading order $(z_p {-} z_q)^{-1}$ we have
	\be
	{}^{0}\MM_{pqr}^+ (\O^{pq}_p)^\intercal = \left( \delta_{qr} - \delta_{q2} \right) (\O^{pr}_p)^\intercal,
	\ee
	which gives \eqref{P0-plus}. At order $(z_p {-} z_q)^{k-1}$ for $k>0$ after rearranging terms we obtain
	\begin{align}
	{}^{k}\MM_{pqr}^+ \left( k \Lambda^2 \I -  (\O^{pq}_p)^\intercal \right) - \sum_{\substack{j=1\\ j\neq q}}^{p-1} \sum_{i=0}^{k-1} \frac{{}^{i}\MM_{pqr}^+\, (-1)^{k-i-1}}{(z_q {-} z_j)^{k-i}} (\O^{pj}_p)^\intercal = (-1)^{k}\! \left( \frac{\delta_{q\neq r}}{(z_q {-} z_r)^k} - \frac{\delta_{q\neq 2}}{(z_q {-} z_2)^k}\right) (\O^{pr}_p)^\intercal,&
	\end{align}
	from which \eqref{Pk-plus} follows.
\end{proof}

One of the consequences of the above expressions is that kinematic singularities of the intersection numbers can only occur when $(\Om^{pq}_{p} \mp k\Lambda^2 \I)^{-1}$ diverges, or in other words when one of the eigenvalues of $\Om^{pr}_p$ approaches $\mp k\Lambda^2$. Assuming that all of the eigenvalues take the form of kinematic invariants $s_{i_1 i_2 \cdots i_m}$, we find that $\la \varphi_- | \varphi_+ \ra_\omega$ involves a propagation of a particle with mass $\sqrt{k} \Lambda$ whenever $\bm\varphi_{p}^-$ has pole of order $k{+}1$ or higher, and likewise a tachyon with mass $i\sqrt{k} \Lambda$ propagates when $\bm\varphi_{p}^+$ has a pole of order $k{+}1$ or higher (one should of course allow for a possibility of cancellations between poles). In particular, when all the poles are logarithmic there are only massless internal states in agreement with Theorem~\ref{theorem-21}.

The above derivations prove the main result of the section.
\begin{theorem}\label{theorem-31}
	Recursion relations for intersection numbers $\braket{\varphi_- | \varphi_+}_\omega = (\bm\varphi_3^-)^\intercal \bm\varphi_3^+$ of twisted forms $\varphi_\pm \in H^{n-3}_{\pm\omega}$ are given by
	\be\label{theorem-3-1}
	(\bm\varphi_{p-1}^\pm)_r = \sum_{q=1}^{p-1} \Res_{z_p = z_q} \!\Big( \MM^\pm_{pqr}\, \bm\varphi_{p}^{\pm} \Big),
	\ee
	where $r=3,4,\ldots,p{-}1$ and $\bm\varphi_p^\pm \in \Omega^{p-3,0}(\M_{0,p}) \otimes \C^{\kappa_{n,p}}$ for $3 \leq p \leq n$. The initial conditions are given by $\bm\varphi^\pm_n = \varphi_\pm$ and the matrices $\MM_{pqr}^\pm$ are specified in Lemmata \ref{lemma-3-5}--\ref{lemma-3-6}.
\end{theorem}

Let us remark on the fact that just like $\Om_p^{ij}$, the matrices $\MM_{pqr}^\pm$ depend on the choice of bases of twisted forms on the fibres. Hence it is not unlikely that the above recursion can be further simplified by other choices of bases. Since the main computational bottleneck in the above definitions of $\MM_{pqr}^\pm$ comes from the necessity of inverting braid matrices, it is also worthwhile to seek recursions for the inverse matrices $(\Om_{p}^{ij})^{-1}$ themselves. Some computational efficiency can be gained by sending one of the fixed punctures to infinity, which we did not do here to keep the formulae more symmetric. Note that the above recursions were designed no to be $\pm$-symmetric, which is clear from \eqref{recursion-first-line} and the explicit form of the matrices $\MM_{pqr}^\pm$. A symmetric form can be obtained with minor modifications, for example by starting with the representation
\be
\braket{\varphi_- | \varphi_+}_\omega = \braket{\varphi_- | \F_3^{-}}_\omega^\intercal\, \la \F_3^{+} | (\F_3^{+})^\intercal \ra_\omega \braket{ \F_3^{-} | \varphi_+ }_\omega
\ee
instead of \eqref{recursion-first-line}. The resulting recursion becomes symmetric, at a cost of introducing the intersection matrix $\la\F_3^{+} | (\F_3^{+})^\intercal \ra_\omega$ that can be evaluated explicitly in terms of trivalent diagrams. In applications to basis expansion, as in \eqref{decomposition} and \eqref{closed-string-decomposition}, onto the fibration basis $\FB(\rho)$ it is only needed to compute one side of the recursion corresponding to $\bm\varphi_p^+$.

We finish by giving examples of the matrices $\MM_{pqr}^\pm$ at low multiplicity.

\begin{example}
	For $n=4$, using the notation $s=s_{12}$, $t=s_{23}$, $u=s_{13}$ with $s+t+u=0$, expansions of the one-by-one matrices $\MM^-_{pqr}$ around $z_p = z_q$ are given by:
	\begin{gather}
	\MM_{413}^- = -\frac{1}{t} + \frac{u }{t\left(t {+} \Lambda ^2\right)} \frac{(z_4{-}z_1)(z_2{-}z_3)}{(z_1{-}z_2)(z_1{-}z_3)} + \ldots,\nn\\
	\MM_{423}^- = 0 - \frac{1}{u {+} \Lambda^2} 	\frac{(z_4{-}z_2)(z_3{-}z_1)}{(z_1{-}z_2)(z_2{-}z_3)} + \ldots,\\
	\MM_{433}^- = \frac{1}{s} + \frac{u}{s(s{+}\Lambda^2)} 	\frac{(z_4{-}z_3)(z_1{-}z_2)}{(z_1{-}z_3)(z_2{-}z_3)} + \ldots,\nn
	\end{gather}
	while the $\MM^+_{pqr}$ matrices are to leading orders:
	\begin{gather}
	\MM_{413}^+ = 0 - \frac{s}{t {-} \Lambda^2}\frac{(z_4{-}z_1)(z_2{-}z_3)}{(z_1{-}z_2)(z_1{-}z_3)} + \ldots,\nn\\
	\MM_{423}^+ = - \frac{s}{u} + \frac{s t}{u
		(u {-} \Lambda^2)}\frac{(z_4{-}z_2)(z_3{-}z_1)}{(z_1{-}z_2)(z_2{-}z_3)} + \ldots,\\
	\MM_{433}^+ = 1 - \frac{t }{s {-} \Lambda^2} \frac{(z_4{-}z_3)(z_1{-}z_2)}{(z_1{-}z_3)(z_2{-}z_3)} + \ldots.\nn
	\end{gather}
\end{example}

\begin{example}
	For $n=5$ the one-by-one matrices $\MM_{5qr}^-$ to leading order have the expansion:
	\begin{gather}
	\MM_{513}^- = -\frac{1}{s_{15}} + \frac{1}{s_{15} (s_{15} {+} \Lambda^2)} \left( \frac{s_{25}}{z_1{-}z_2} +\frac{s_{15} {+} s_{35}}{z_1{-}z_3} + \frac{s_{45}}{z_1{-}z_4} \right) (z_5 {-} z_1) + \ldots,\nn\\
	\MM_{523}^- = 0 - \frac{1}{s_{25} {+} \Lambda^2} \left( \frac{1}{z_2{-}z_1}-\frac{1}{z_2{-}z_3} \right) (z_5 {-} z_2) + \ldots,\nn\\
	\MM_{533}^- = \frac{1}{s_{35}} - \frac{1}{s_{35}(s_{35} {+} \Lambda^2)} \left( \frac{s_{15}{+}s_{35}}{z_3{-}z_1} +  \frac{s_{25}}{z_3{-}z_2}+\frac{s_{45}}{z_3{-}z_4} \right)(z_5 {-} z_3) + \ldots,\\
	\MM_{543}^- = 0 - \frac{1}{s_{45} {+} \Lambda^2} \left( \frac{1}{z_4 {-} z_1} - \frac{1}{z_4 {-} z_3} \right) (z_5 {-} z_4) + \ldots,\nn
	\end{gather}
	while the $\MM_{5qr}^+$ matrices are:
	\begin{gather}
	\MM_{513}^+ = 0 - \frac{s_{35}}{s_{15} {-} \Lambda^2} \left( \frac{1}{z_1 {-} z_2} - \frac{1}{z_1 {-} z_3} \right) (z_5 {-} z_1) + \ldots,\nn\\
	\MM_{523}^+ = -\frac{s_{35}}{s_{25}} + \frac{s_{35}}{s_{25}(s_{25}{-}\Lambda^2)} \left( \frac{s_{15}}{z_2 {-} z_1} + \frac{s_{25} {+} s_{35}}{z_2 {-} z_3} + \frac{s_{45}}{z_2 {-} z_4} \right)(z_5 {-} z_2) + \ldots,\nn\\
	\MM_{533}^+ = 1 - \frac{1}{s_{35}{-}\Lambda^2} \left( \frac{s_{15}}{z_3 {-} z_1} + \frac{s_{25} {+} s_{35}}{z_3 {-} z_2} + \frac{s_{45}}{z_3 {-} z_4} \right)(z_5 {-} z_3) + \ldots,\\
	\MM_{543}^+ = 0 - \frac{s_{35}}{s_{45} {-} \Lambda^2} \left( \frac{1}{z_4 {-} z_2} - \frac{1}{z_4 {-} z_3} \right)(z_5 {-} z_4) + \ldots.\nn
	\end{gather}
	The $p=5$ matrices $\MM_{4qr}^-$ are two-by-two (with columns and rows labelled by $3,4$). The entries of $\MM_{413}^-$ have the expansion:
	\begin{align}
	(\MM_{413}^-)_{33} &= -\frac{1}{s_{14}} + \frac{s_{23} s_{24}{-}s_{35} s_{45}}{s_{14} s_{23} \left(s_{14}{+}\Lambda^2 \right)} \frac{(z_4 {-} z_1)(z_2 {-} z_3)}{(z_1 {-} z_2)(z_1 {-} z_3)} + \ldots,\nn\\
	(\MM_{413}^-)_{34} &= \frac{s_{45}}{s_{14} s_{23}} - \frac{s_{45} \left(s_{23}s_{24}{-}s_{35}s_{45}{+}s_{13}(s_{14}{+}\Lambda^2) \right)}{s_{14} s_{23} \left(s_{14} {+} \Lambda^2\right) \left(s_{23}{+}\Lambda^2\right)}\frac{(z_4 {-} z_1)(z_2 {-} z_3)}{(z_1 {-} z_2)(z_1 {-} z_3)} + \ldots,\nn\\
	(\MM_{413}^-)_{43} &= 0 + \frac{s_{35}}{s_{23} \left(s_{14}{+}\Lambda^2\right)} \frac{(z_4 {-} z_1)(z_2 {-} z_3)}{(z_1 {-} z_2)(z_1 {-} z_3)} + \ldots,\\
	(\MM_{413}^-)_{44} &= - \frac{1}{s_{23}} + \frac{\left(s_{24}{+}s_{25}\right) \left(s_{14}{+}\Lambda^2\right) {+} s_{45} \left(s_{14}{-}s_{35}{+}\Lambda^2\right)}{s_{23}
		\left(s_{14}{+}\Lambda^2\right) \left(s_{23}{+}\Lambda^2\right)}\frac{(z_4 {-} z_1)(z_2 {-} z_3)}{(z_1 {-} z_2)(z_1 {-} z_3)} + \ldots.\nn
	\end{align}
	The entries of $\MM_{423}^-$ are given by:
	\begin{align}
	(\MM_{423}^-)_{33} &= 0 -\frac{1}{s_{24}{+}\Lambda^2} \frac{(z_4 {-} z_2)(z_3 {-} z_1)}{(z_1 {-} z_2)(z_2 {-} z_3)} + \ldots,\nn\\
	(\MM_{423}^-)_{34} &= 0 + \ldots,\nn\\
	(\MM_{423}^-)_{43} &= 0 + \frac{s_{35}}{\left(s_{13}{+}\Lambda^2\right) \left(s_{24}{+}\Lambda^2\right)} \frac{(z_4 {-} z_2)(z_3 {-} z_1)}{(z_1 {-} z_2)(z_2 {-} z_3)} + \ldots,\\
	(\MM_{423}^-)_{44} &= 0 -\frac{1}{s_{13}{+}\Lambda^2} \frac{(z_4 {-} z_2)(z_3 {-} z_1)}{(z_1 {-} z_2)(z_2 {-} z_3)} + \ldots,\nn
	\end{align}
	and finally for $\MM_{433}^-$ we find:
	\begin{align}
	(\MM_{433}^-)_{33} &= \frac{s_{34}{+}s_{35}}{s_{12} s_{34}} {+} \frac{s_{24}(s_{34}{+}s_{35})(s_{34}{+}\Lambda^2){+}s_{35} \left(s_{12} s_{24}{+}s_{45}(s_{12}{+}s_{25}{+}\Lambda^2)\right)}{s_{12} s_{34} \left(s_{12}{+}\Lambda^2\right) \left(s_{34}{+}\Lambda^2\right)} \frac{(z_4 {-} z_3)(z_1 {-} z_2)}{(z_1 {-} z_3)(z_2 {-} z_3)} {+} \ldots,\nn\\
	(\MM_{433}^-)_{34} &= \frac{s_{45}}{s_{12} s_{34}} + \frac{s_{45} \left(s_{13}(s_{34}{+}\Lambda^2) {-} s_{45} s_{51}{+}s_{12} s_{24}\right)}{s_{12}
		s_{34} \left(s_{12}{+}\Lambda^2\right) \left(s_{34}{+}\Lambda^2\right)} \frac{(z_4 {-} z_3)(z_1 {-} z_2)}{(z_1 {-} z_3)(z_2 {-} z_3)} + \ldots,\nn\\
	(\MM_{433}^-)_{43} &= \frac{s_{35}}{s_{12} s_{34}} {-}\frac{s_{35} \left(s_{14} \left(s_{12}{+}s_{34}{+}\Lambda^2\right){+}\left(s_{12}{+}s_{15}\right)
		\left(s_{34}{+}s_{45}\right)\right)}{s_{12} s_{34} \left(s_{12}{+}\Lambda^2\right) \left(s_{34}{+}\Lambda^2\right)}  \frac{(z_4 {-} z_3)(z_1 {-} z_2)}{(z_1 {-} z_3)(z_2 {-} z_3)} + \ldots,\\
	(\MM_{433}^-)_{44} &= \frac{s_{34}{+}s_{45}}{s_{12} s_{34}} + \frac{s_{13} \left(s_{34}{+}s_{45}\right) \left(s_{34}{+}s_{45}{+}\Lambda^2\right){-}s_{14} s_{35} s_{45}}{s_{12} s_{34} \left(s_{12}{+}\Lambda^2\right) \left(s_{34}{+}\Lambda^2\right)}\frac{(z_4 {-} z_3)(z_1 {-} z_2)}{(z_1 {-} z_3)(z_2 {-} z_3)} + \ldots.\nn
	\end{align}
	Similarly we list the expansions of the entries for each $\MM_{4qr}^+$, starting with $\MM_{413}^+$,
	\begin{align}
	(\MM_{413}^+)_{33} &= 0 - \frac{s_{35}s_{45}{+}(s_{34}{+}s_{45})(s_{23}{-}\Lambda^2)}{(s_{14}{-}\Lambda ^2)(s_{23}{-}\Lambda^2)} \frac{(z_4 {-} z_1)(z_2 {-} z_3)}{(z_1 {-} z_2)(z_1 {-} z_3)} + \ldots,\nn\\
	(\MM_{413}^+)_{34} &= 0 + \frac{s_{35}}{s_{23}{-}\Lambda ^2} \frac{(z_4 {-} z_1)(z_2 {-} z_3)}{(z_1 {-} z_2)(z_1 {-} z_3)} + \ldots,\nn\\
	(\MM_{413}^+)_{43} &= 0 -\frac{s_{45} \left(s_{13}{+}\Lambda^2\right)}{\left(s_{14} {-} \Lambda ^2\right) \left(s_{23} {-} \Lambda^2\right)} \frac{(z_4 {-} z_1)(z_2 {-} z_3)}{(z_1 {-} z_2)(z_1 {-} z_3)} + \ldots,\\
	(\MM_{413}^+)_{44} &= 0 - \frac{s_{34}{+}s_{35}}{s_{23}{-}\Lambda^2}\frac{(z_4 {-} z_1)(z_2 {-} z_3)}{(z_1 {-} z_2)(z_1 {-} z_3)} + \ldots.\nn
	\end{align}
	For $\MM_{423}^+$ we find:
	\begin{align}
	(\MM_{423}^+)_{33} &= -\frac{s_{34}{+}s_{45}}{s_{24}} + \frac{s_{14} \left(s_{35} s_{45}{+}s_{13} \left(s_{34}{+}s_{45}\right)\right)}{s_{13} s_{24} \left(s_{24}{-}\Lambda^2\right)} \frac{(z_4 {-} z_2)(z_3 {-} z_1)}{(z_1 {-} z_2)(z_2 {-} z_3)} + \ldots,\nn\\
	(\MM_{423}^+)_{34} &= -\frac{s_{14} s_{35}}{s_{13} s_{24}} + \frac{s_{14} s_{35}  \left(s_{12} s_{34} {-} s_{23} \left(s_{14}{+}\Lambda^2\right)\right)}{s_{13} s_{24} \left(s_{13}{-}\Lambda ^2\right) \left(s_{24}{-}\Lambda ^2\right)} \frac{(z_4 {-} z_2)(z_3 {-} z_1)}{(z_1 {-} z_2)(z_2 {-} z_3)} + \ldots,\nn\\
	(\MM_{423}^+)_{43} &= \frac{s_{45}}{s_{24}} + \frac{\left(s_{14} s_{23}{-}s_{12} s_{34}\right) s_{45}}{s_{13} s_{24} \left(s_{24}{-}\Lambda^2\right)}\frac{(z_4 {-} z_2)(z_3 {-} z_1)}{(z_1 {-} z_2)(z_2 {-} z_3)} + \ldots,\\
	(\MM_{423}^+)_{44} &= -\frac{s_{24} \left(s_{34}{+}s_{35}\right) {+} s_{35} s_{45}}{s_{13} s_{24}} + \bigg(
	\frac{s_{23} \left(s_{24} \left(s_{34}{+}s_{35}\right){+}s_{35} s_{45}\right)}{s_{13} s_{24} \left(s_{13}{-}\Lambda ^2\right)}\nn\\
	&\qquad\qquad\qquad\quad -\frac{s_{35} s_{45} \left(s_{24} \left(s_{34}{+}s_{35}\right){+}s_{35} s_{45}{-}s_{13} s_{14}\right)}{s_{13} s_{24}
		\left(s_{13}{-}\Lambda ^2\right) \left(s_{24}{-}\Lambda ^2\right)}
	\bigg)\frac{(z_4 {-} z_2)(z_3 {-} z_1)}{(z_1 {-} z_2)(z_2 {-} z_3)} + \ldots.\nn
	\end{align}
	Finally, evaluating $\MM_{433}^+$ yields:
	\begin{align}
	(\MM_{433}^+)_{33} &= 1 -\frac{s_{14} \left(s_{12}{-}s_{45}{-}\Lambda ^2\right)}{\left(s_{12} {-} \Lambda^2\right) \left(s_{34} {-} \Lambda ^2\right)} \frac{(z_4 {-} z_3)(z_1 {-} z_2)}{(z_1 {-} z_3)(z_2 {-} z_3)} + \ldots,\nn\\
	(\MM_{433}^+)_{34} &= 0 -\frac{s_{14} s_{35}}{\left(s_{12}{-}\Lambda ^2\right) \left(s_{34}{-}\Lambda ^2\right)}\frac{(z_4 {-} z_3)(z_1 {-} z_2)}{(z_1 {-} z_3)(z_2 {-} z_3)} + \ldots,\nn\\
	(\MM_{433}^+)_{43} &= 0 + \frac{s_{45} \left(s_{13}{+}\Lambda^2\right)}{\left(s_{12}{-}\Lambda ^2\right) \left(s_{34}{-}\Lambda ^2\right)}\frac{(z_4 {-} z_3)(z_1 {-} z_2)}{(z_1 {-} z_3)(z_2 {-} z_3)} + \ldots,\\
	(\MM_{433}^+)_{44} &= 1 -\frac{s_{23} \left(s_{34}{-}\Lambda^2\right){+}\left(s_{23}{+}s_{35}\right) s_{45}}{\left(s_{12}{-}\Lambda ^2\right)
		\left(s_{34}{-}\Lambda ^2\right)}
	\frac{(z_4 {-} z_3)(z_1 {-} z_2)}{(z_1 {-} z_3)(z_2 {-} z_3)}
	+ \ldots.\nn
	\end{align}
\end{example}

\pagebreak
\section[Further Examples of Intersection Numbers]{\label{sec:examples}Further Examples of Intersection Numbers}

\textsc{In this section} we give several instructive examples of evaluating intersection numbers using the newly-introduced recursion relations, starting with a discussion of the Yang--Mills and gravity scattering amplitudes, whose twisted forms are inherited from string theory. We then move on to computing intersection numbers of Kac--Moody correlators that can serve as building blocks for more complicated twisted forms. Throughout this section we use the notation $s_{ij}:=(p_i{+}p_j)^2$.

It is not presently known what is the full space of scattering amplitudes admitting an intersection number representation. In the massless case there has been considerable interest in constructing S-matrices of various quantum field theories within the scattering equations formalism with much success, see, e.g., \cite{Cachazo:2013hca,Cachazo:2013iaa,Cachazo:2013iea,Mason:2013sva,Geyer:2014fka,Adamo:2014wea,Cachazo:2014nsa,Cachazo:2014xea,Weinzierl:2014ava,Naculich:2015zha,Naculich:2015coa,Casali:2015vta,Adamo:2015gia,delaCruz:2015raa,Lam:2015mgu,He:2016vfi,Cachazo:2016njl,He:2016dol,He:2016iqi,Zhang:2016rzb,Mizera:2017sen,Azevedo:2017lkz,Heydeman:2017yww,Cachazo:2018hqa,Mizera:2018jbh,He:2018pol,Geyer:2018xgb,Heydeman:2018dje,Geyer:2019ayz}. This suggests that the space of quantum field theories whose amplitudes have an interpretation in terms of intersection numbers should be even richer, especially given their direct link to string theory.

\subsection{\label{sec:Yang-Mills-and-gravity}Gauge and Gravity Scattering Amplitudes}

\textsc{Since the low-energy limits} of open-string integrals and intersection numbers coincide according to \eqref{open-string-limit}, it means that the leading $\alpha'$-order of the intersection numbers of the form $\la \PT(\alpha) |  \varphi^{\text{gauge}}_{+,n} \ra_\omega$ computes Yang--Mills amplitudes. Here $\varphi^{\text{gauge}}_{+,n}$ is a twisted form coming from a correlation function of massless gauge boson vertex operators,
either in the bosonic or superstring case, modulo the Koba--Nielsen factor, see \cite{green1988superstring,polchinski1998string} for details.

This result can be strengthened even further. The decomposition formula \eqref{open-string-decomposition} together with the results of \cite{Mafra:2011nv} imply that in the superstring case the intersection number is independent of $\alpha'$ and thus it computes \emph{exclusively} the low-energy limit of string amplitudes:
\be
\Braket{ \PT(\alpha) |  \varphi^{\text{gauge}}_{+,n} }_\omega = {\cal A}^{\text{YM}}(\alpha),
\ee
which are colour-ordered amplitudes in the pure $\text{U}(N)$ Yang--Mills theory with the ordering $\alpha$. Full tree-level amplitudes can be obtained by decorating $\PT(\alpha)$ with the colour trace $\text{Tr}(T^{c_{\alpha(1)}} T^{c_{\alpha(2)}} \cdots T^{c_{\alpha(n)}})$, where $T^{c_i} \in \text{U}(N)$ are the Chan--Paton factors associated to each puncture \cite{Paton:1969je}, and symmetrizing over all permutations $\alpha$. For simplicity here we only consider colour-ordered amplitudes and ignore the overall coupling constant.

There are many ways of representing correlation functions, as they are unique only up to cohomology relations. Let us consider the textbook form of the correlator in the Ramond--Neveu--Schwarz (RNS) formulation \cite{green1988superstring,polchinski1998string}, which after stripping away the Koba--Nielsen factor $\prod_{i<j} |z_i {-} z_j|^{2p_i {\cdot} p_j / \Lambda^2}$ in our normalization reads:
\be\label{varphi-gauge}
\varphi^{\text{gauge}}_{\pm,n} = d\mu_n\! \int \prod_{i=1}^{n} d\theta_i d\tilde{\theta}_i \frac{\theta_k \theta_\ell}{z_k {-} z_\ell} \exp\! \left( - \sum_{i \neq j}\frac{\theta_i \theta_j p_i {\cdot} p_j + \tilde{\theta}_i \tilde{\theta}_j \varepsilon_i {\cdot} \varepsilon_j + 2(\theta_i {-} \theta_j) \tilde{\theta}_i \varepsilon_i {\cdot} p_j }{z_i {-} z_j \mp \Lambda^2 \theta_i \theta_j}
\right).
\ee
It is written as an integral over Grassmann variables $\theta_i$, $\tilde{\theta}_i$ for each puncture $z_i$.\footnote{Recall that Grassmann variables are mutually anticommuting and the integral $\int \!d\theta_i (\alpha + \beta \theta_i) = \beta$ for bosonic $\alpha, \beta$ extracts the linear part in $\theta_i$. Since $\exp(\beta \theta_i) = 1 + \beta \theta_i$ we have $\int\! d\theta_i \exp(\beta \theta_i) = \beta$. See, e.g., \cite{deligne1999quantum} for a reference. Pfaffian $\Pf\mathbf{A}$ of an antisymmetric $2m {\times} 2m$ matrix $\mathbf{A}$ with entries $\mathbf{A}_{ij}$ is defined through $\Pf\mathbf{A} := \int \prod_{i=1}^{2m} d\theta_i \exp(- \frac{1}{2} \sum_{i \neq j}\theta_i \mathbf{A}_{ij} \theta_j)$. The Pfaffian is equal up to a sign to the square root of $\det \A$.} In full string theory $(z_i | \theta_i) \in \CP^{1|1}$ are coordinates on a super Riemann surface, but in our setup they lose this interpretation and together with $\tilde{\theta}_i$, $\theta_i$ serve as auxiliary variables. The factor $\theta_k \theta_\ell / (z_k {-} z_\ell)$ arises from a pair of vertex operators in the $-1$ picture. The choice of the labels $k,\ell$ is auxiliary, meaning that $\varphi^{\text{gauge}}_{\pm,n}$ with different choices of $k,\ell$ are cohomologous to each other. For simplicity we set $(k,\ell) = (1,2)$ in the following manipulations. Since each appearance of the polarization vector $\varepsilon_i^\mu$ is accompanied by $\tilde{\theta}_i$ and the Grassmann integration extracts only linear terms in $\tilde{\theta}_i$, the twisted form is also multi-linear in polarization vectors. Recall that for massless gauge bosons we have the conditions $p_i^2 = \varepsilon_i {\cdot} p_i = 0$. In addition, gauge invariance of the RNS correlator implies that upon the replacement $\varepsilon_i \to p_i$ for any $i$ the twisted form \eqref{varphi-gauge} becomes cohomologous to zero. For later purposes we introduced two types of twisted forms in \eqref{varphi-gauge}, labelled by a $\pm$ sign, depending on the factor of $\mp\Lambda^2$ in the denominators. Note that even though the twisted forms \eqref{varphi-gauge} come from a superstring computation, they give rise to pure Yang--Mills amplitudes without any supersymmetry (though inclusion of supersymmetry is certainly possible).

In order to get a handle on the combinatorics involved in integrating out Grassmann variables in \eqref{varphi-gauge} let us start by symmetrizing its exponent, after which its argument reads
\be\label{exponent-argument}
- \sum_{i \neq j} \left( \frac{(\theta_i p_i {+} \tilde{\theta}_i \varepsilon_i){\cdot} (\theta_j p_j {+} \tilde{\theta}_j \varepsilon_j) +\theta_i \tilde{\theta}_i \varepsilon_i {\cdot} p_j - \theta_j \tilde{\theta}_j \varepsilon_j {\cdot} p_i}{z_i {-} z_j} \pm \Lambda^2 \frac{\varepsilon_i {\cdot} \varepsilon_j \theta_i \theta_j \tilde{\theta}_i \tilde{\theta}_j}{(z_i{-}z_j)^2}\right).
\ee
Here we also used the fact that
\be
\frac{1}{z_i {-} z_j \mp \Lambda^2 \theta_i \theta_j} = \frac{1}{z_i {-} z_j} \pm \Lambda^2 \frac{\theta_i \theta_j}{(z_i{-}z_j)^2},
\ee
where the second term gives non-zero contributions only for summands proportional to $\varepsilon_i {\cdot} \varepsilon_j$ since all the others contain at least one factor of $\theta_i$. We find that in \eqref{exponent-argument} the first fraction is Gaussian in the Grassmann variables and the second gives a quartic perturbation. Hence it is convenient to rewrite the twisted form as
\be
\varphi^{\text{gauge}}_{\pm,n} = d\mu_n\! \int \prod_{i=1}^{n} d\theta_i d\tilde{\theta}_i \frac{\theta_1 \theta_2}{z_1 {-} z_2} \exp\! \left( -\frac{1}{2} \sum_{i \neq j} \left(\left(\begin{array}{c}
	\theta_i \\
	\tilde{\theta}_i \\
\end{array}\right)^{\!\!\intercal}\!\!
\left(\begin{array}{cc}
	\A_{ij} & -{\mathbf C}_{ji}\\
	{\mathbf C}_{ij} & \B_{ij} \\
\end{array}\right)\!\!
\left(\begin{array}{c}
	\theta_j \\
	\tilde{\theta}_j \\
\end{array}\right)
\pm \Lambda^2 \frac{2\varepsilon_i {\cdot} \varepsilon_j }{(z_i{-}z_j)^2}\, \theta_i \theta_j \tilde{\theta}_i \tilde{\theta}_j \right)
\right).
\ee
Here $\A_{ij}$, $\B_{ij}$, and ${\mathbf C}_{ij}$ are elements of $n{\times}n$ matrices given by
\begin{gather}
\A_{ij} := \begin{dcases}
	\phantom{\frac{1}{1}}\; 0 & \text{if}\quad i = j,\\
	\frac{2p_i {\cdot} p_j}{z_i {-} z_j} & \text{otherwise},
\end{dcases} \qquad
\B_{ij} := \begin{dcases}
	\phantom{\frac{1}{1}}\; 0 & \text{if}\quad i = j,\\
	\frac{2\varepsilon_i {\cdot} \varepsilon_j}{z_i {-} z_j} & \text{otherwise},
\end{dcases} \qquad
{\mathbf C}_{ij} := \begin{dcases}
	-\sum_{k\neq i} \frac{2\varepsilon_i {\cdot} p_k}{z_i {-} z_k} & \text{if}\quad i = j,\\
	\frac{2\varepsilon_i {\cdot} p_j}{z_i {-} z_j} & \text{otherwise}.
\end{dcases}
\end{gather}
Let us introduce a $2n {\times} 2n$ antisymmetric matrix $\mathbf\Psi$ constructed from the above blocks:
\be
{\mathbf\Psi} := \left(\begin{array}{cc}
	\A & -{\mathbf C}^\intercal\\
	{\mathbf C} & \B \\
\end{array}\right).
\ee
We can now consider performing the Grassmann integrals as an expansion in terms of the variable $\Lambda^2$, in which the leading term is the Pfaffian of the matrix $\mathbf\Psi$ with columns and rows $1,2$ removed, while the subleading ones involve Pfaffians of the same matrix with more of its entries removed. To be precise, we have up to an overall sign:
\be
\varphi^{\text{gauge}}_{\pm,n} = \frac{d\mu_n}{z_1 {-} z_2}\! \sum_{q=0}^{n/2-1}\,\! (\mp\Lambda^2)^{q}\! \sum_{\substack{\text{distinct}\\ \text{pairs}\\ \{i_k, j_k\}}} \left( \prod_{k=1}^{q} \frac{2 \varepsilon_{i_k}{\cdot}\varepsilon_{j_k}}{(z_{i_k} {-} z_{j_k})^2}\right) \Pf {\mathbf\Psi}_{12}^{i_1 j_1 i_2 j_2 \cdots i_q j_q}.
\ee
The second sum goes over all $q$ distinct unordered pairs $\{i_1, j_1\}, \{i_2, j_2\}, \ldots, \{i_q, j_q\}$ of labels from the set $\{ 3,4,\ldots,n\}$. We used the notation ${\mathbf\Psi}^{i_1j_1\cdots i_q j_q}_{12}$ to denote the matrix $\mathbf\Psi$ with columns and rows $\{1,2,i_1,j_1,\ldots, i_q,j_q,\allowbreak n{+}i_1, n{+}j_1,\ldots, n{+}i_q, n{+}j_q\}$ removed, such that it becomes a $2(n{-}2q{-}1) {\times}\allowbreak 2(n{-}2q{-}1)$ matrix.
Each term in the resulting expansion has the same mass dimension and is multi-linear in polarization vectors. In order to demystify the notation let us spell out a few leading terms:
\begin{align}
\varphi^{\text{gauge}}_{\pm,n} = &\frac{d\mu_n}{z_1 {-} z_2} \Bigg(\Pf {\mathbf\Psi}_{12} \mp \Lambda^2 \sum_{3\leq i_1 < j_1 \leq n} \frac{2\varepsilon_{i_1} {\cdot} \varepsilon_{j_1}}{(z_{i_1} {-} z_{j_1})^2} \Pf {\mathbf\Psi}^{i_1 j_1}_{12} \\
&\quad + \Lambda^4\, \Bigg(\sum_{3 \leq i_1 < j_1 < i_2 < j_2 \leq n} \!{+}\! \sum_{3 \leq i_1 < i_2 < j_1 < j_2 \leq n}\Bigg) \frac{2\varepsilon_{i_1} {\cdot} \varepsilon_{j_1}}{(z_{i_1} {-} z_{j_1})^2} \frac{2\varepsilon_{i_2} {\cdot} \varepsilon_{j_2}}{(z_{i_2} {-} z_{j_2})^2} \Pf {\mathbf\Psi}^{i_1 j_1 i_2 j_2}_{12} + \cdots \Bigg).\nn
\end{align}

In a similar spirit we can consider the low-energy limit of closed-string amplitudes, whose correlator splits into $\varphi^{\text{gauge}}_{+,n} \wedge \overbar{\widetilde{\varphi}^{\text{gauge}}_{+,n}}$, where the antiholomorphic part involves polarization vectors $\tilde{\varepsilon}^{\mu}_i$ in the place of $\varepsilon^{\mu}_i$. Noticing that in such cases the right-hand side of \eqref{closed-string-limit} involves only intersection numbers independent of $\alpha'$, we can perform the same manipulations as in the proof of Theorem~\ref{theorem-22} to arrive at\footnote{The only non-trivial step is realizing that $\la \PT(1,\widehat{\alpha},n{-}1,n)^\vee | \widetilde{\varphi}^{\text{gauge}}_{+,n} \ra_\omega = \la \widetilde{\varphi}^{\text{gauge}}_{-,n} | \PT(1,\widehat{\alpha},n{-}1,n)^\vee \ra_\omega$, where $\widetilde{\varphi}^{\text{gauge}}_{-,n}$ involves a sign change $\Lambda^2 \to - \Lambda^2$ compared to $\widetilde{\varphi}^{\text{gauge}}_{+,n}$ as in \eqref{varphi-gauge}.}
\be\label{intersection-gravity}
\Braket{ \widetilde{\varphi}^{\text{gauge}}_{-,n} | \varphi^{\text{gauge}}_{+,n} }_\omega = {\cal A}^{\text{GR}}_n,
\ee
which coincides with the low-energy limit associated to the above closed-string correlator: Einstein gravity amplitude. To be a bit more precise, it is the amplitude of ``${\cal N}=0$ supergravity'' involving the graviton $h_{\mu\nu}$, Kalb--Ramond field $B_{\mu\nu}$, and the dilaton $\phi$. Amplitudes of gravitons only are obtained by using the symmetrized version of he polarization vectors, $\varepsilon^{\mu\nu}_i := \varepsilon^{(\mu}_i \tilde{\varepsilon}^{\nu)}_i$.

At this stage it is instructive to compute $n=4$ amplitudes in the above theories. Fixing $\alpha = (1234)$ and using recursion relations from Theorem~\ref{theorem-31} it is enough to compute the two functions ${\bm\varphi}^{\pm}_3$ and multiply them together. For the Parke--Taylor form we find:
\be
{\bm\varphi}_3^- [\PT(1234)] = \frac{1}{s_{12}} + \frac{1}{s_{23}},
\ee
which can be obtained using the definition \eqref{bold-varphi}, the equality $\PT(1234) = \FB(3)$, and the formula in terms of trivalent graphs from \eqref{PT-PT}, or by direct computation. The above notation indicates that we are referring to ${\bm\varphi}_3^-$ computed with the boundary condition $\varphi_- = \PT(1234)$. In order to have a better understanding of how the final result arises, let us split the computation of ${\bm\varphi}_3^+[\varphi^{\text{gauge}}_{+,4}]$ into two pieces: these proportional to $\Lambda^0$ and $\Lambda^2$. For the first term we find:
\begin{align}
{\bm\varphi}_3^+[&\varphi^{\text{gauge}}_{+,4}|_{\Lambda^{0}}] = -4s_{12}s_{23} \left( \frac{\varepsilon_1 {\cdot} \varepsilon_2\,\varepsilon_3 {\cdot} \varepsilon_4}{s_{12} - \Lambda^2} + \frac{\varepsilon_2 {\cdot} \varepsilon_3\,\varepsilon_1 {\cdot} \varepsilon_4}{s_{23}} + \frac{\varepsilon_1 {\cdot} \varepsilon_3\,\varepsilon_2 {\cdot} \varepsilon_4}{s_{13}} \right)\\
&- \frac{8}{s_{13}} \Big( s_{12} (\varepsilon_1 {\cdot} \varepsilon_4\, \varepsilon_2 {\cdot} p_4\, \varepsilon_3{\cdot} p_1 + \varepsilon_2 {\cdot} \varepsilon_4\, \varepsilon_1{\cdot} p_4\, \varepsilon_3 {\cdot} p_2 + \varepsilon_1 {\cdot} \varepsilon_3\, \varepsilon_2 {\cdot} p_3\, \varepsilon_4 {\cdot} p_1 + \varepsilon_2 {\cdot} \varepsilon_3\, \varepsilon_1 {\cdot} p_3\, \varepsilon_4 {\cdot} p_2) \nn\\
& \qquad\, +s_{23} (\varepsilon_2 {\cdot} \varepsilon_4\, \varepsilon_1 {\cdot} p_2\, \varepsilon_3 {\cdot} p_4 + \varepsilon_3 {\cdot} \varepsilon_4\, \varepsilon_1 {\cdot} p_3\, \varepsilon_2 {\cdot} p_4 + \varepsilon_2 {\cdot} \varepsilon_1\, \varepsilon_3 {\cdot} p_1\, \varepsilon_4 {\cdot} p_2 + \varepsilon _3 {\cdot} \varepsilon_1\, \varepsilon_2 {\cdot} p_1\, \varepsilon_4 {\cdot} p_3) \nn\\
& \qquad\, +s_{13} ( \varepsilon_1 {\cdot} \varepsilon_4\, \varepsilon_2 {\cdot} p_1\, \varepsilon_3 {\cdot} p_4 + \varepsilon_3 {\cdot} \varepsilon_4\, \varepsilon_1 {\cdot} p_4\, \varepsilon_2 {\cdot} p_3 + \varepsilon_1 {\cdot} \varepsilon_2\, \varepsilon_3 {\cdot} p_2\, \varepsilon_4 {\cdot} p_1 + \varepsilon_3 {\cdot} \varepsilon_2\, \varepsilon_1 {\cdot} p_2\, \varepsilon_4 {\cdot} p_3 ) \Big).\nn 
\end{align}
Notice that as a consequence of double poles in $\varphi^{\text{gauge}}_{+,4}|_{\Lambda^{0}}$ the above computation resulted in a spurious tachyon pole in the $s_{12}$-channel, which ought to disappear in the final expression. The contributions proportional to $\Lambda^2$ give:
\be
{\bm\varphi}_3^+[\varphi^{\text{gauge}}_{+,4}|_{\Lambda^{2}}] = 4\Lambda^2 s_{23} \frac{\varepsilon_1 {\cdot} \varepsilon_2\, \varepsilon_3 {\cdot} \varepsilon_4}{s_{12} - \Lambda^2},
\ee
which are precisely the terms needed to cancel out the spurious pole.
The amplitude is computed by putting the above partial results together. This gives:
\begin{align}
{\cal A}^{\text{YM}}(1234) &= {\bm\varphi}_3^- [\PT(1234)]\, \Big( {\bm\varphi}_3^+[\varphi^{\text{gauge}}_{+,4}|_{\Lambda^{0}}] + {\bm\varphi}_3^+[\varphi^{\text{gauge}}_{+,4}|_{\Lambda^{2}}] \Big)\nn\\
&= -\frac{16}{s_{12} s_{23}} t_{8,\mu_1 \nu_1 \mu_2 \nu_2 \mu_3 \nu_3 \mu_4 \nu_4}\, p_1^{\mu_1} \varepsilon_1^{\nu_1}  p_2^{\mu_2} \varepsilon_2^{\nu_2}  p_3^{\mu_3} \varepsilon_3^{\nu_3} p_4^{\mu_4} \varepsilon_4^{\nu_4},
\end{align}
where for conciseness we introduced the $t_8$ tensor antisymmetric under the exchange of any pair of indices $(\mu_i, \nu_i) \leftrightarrow (\mu_j, \nu_j)$ and under $\mu_i \leftrightarrow \nu_i$, whose explicit form can be found in \cite{Schwarz:1982jn}. Note that the result is independent of $\Lambda$ in our normalization.

Let us move on to evaluating the $n=4$ gravity amplitude, which can be done in several ways. One of them is to evaluate $\bm\varphi_3^-[\widetilde{\varphi}^{\text{gauge}}_{-,4}]$, which is almost identical to the previous computation and yields:
\be
\bm\varphi_3^-[\widetilde{\varphi}^{\text{gauge}}_{-,4}] = -\frac{16}{s_{12} s_{23}} t_{8,\mu_1 \nu_1 \mu_2 \nu_2 \mu_3 \nu_3 \mu_4 \nu_4}\, p_1^{\mu_1} \tilde\varepsilon_1^{\nu_1}  p_2^{\mu_2} \tilde\varepsilon_2^{\nu_2}  p_3^{\mu_3} \tilde\varepsilon_3^{\nu_3} p_4^{\mu_4} \tilde\varepsilon_4^{\nu_4}.
\ee
Note that it is not of the same form as $\bm\varphi_3^+[\varphi^{\text{gauge}}_{+,4}]$ since the recursion relations are not $\pm$-symmetric. It gives us the gravity amplitude:
\begin{align}
{\cal A}^{\text{GR}}_4 &= \bm\varphi_3^-[\widetilde{\varphi}^{\text{gauge}}_{-,4}] \Big(  {\bm\varphi}_3^+[\varphi^{\text{gauge}}_{+,4}|_{\Lambda^{0}}] + {\bm\varphi}_3^+[\varphi^{\text{gauge}}_{+,4}|_{\Lambda^{2}}] \Big)\nn\\
&= -\frac{s_{12} s_{23}}{s_{13}} \widetilde{\cal A}^{\text{YM}}(1234)\, {\cal A}^{\text{YM}}(1234),\label{KLT-relation-4}
\end{align}
where the tilde indicates a replacement $\varepsilon_i^\mu \to \tilde{\varepsilon}_i^\mu$ (as remarked before, the result should be appropriately symmetrized to extract pure graviton contributions). The gravity amplitude is a quadratic combination of Yang--Mills amplitudes, which can be made manifest by inserting a resolution of identity into the intersection number $\la \widetilde{\varphi}^{\text{gauge}}_{-,4} | {\varphi}^{\text{gauge}}_{+,4} \ra_\omega$ as a special case of Proposition~\ref{proposition}:
\begin{align}
{\cal A}^{\text{GR}}_4 &= \Braket{ \widetilde{\varphi}^{\text{gauge}}_{-,4} | \PT(1234) }_\omega \Braket{ \PT(1234)^\vee | \PT(1234)^\vee }_\omega \Braket{ \PT(1234) | {\varphi}^{\text{gauge}}_{+,4} }_\omega.
\end{align}
Here the first and third intersection numbers compute gauge theory amplitudes, $\widetilde{{\cal A}}^{\text{YM}}(1234)$ and ${\cal A}^{\text{YM}}(1234)$ respectively, while the second one is easily evaluated to $-s_{12}s_{23}/s_{13}$, giving the required KLT relation from \eqref{KLT-relation-4}.

The above computations can be repeated at higher-$n$ using computer algebra software. Since this would not illustrate any new principle other than using recursion relations, we will not print the results here. In the following subsection we will consider intersection numbers of Kac--Moody correlators, which can be thought of as building blocks for more general twisted forms, to demonstrate the use of recursion relations.

The above computations already exemplify a general feature of $\varphi^{\text{gauge}}_{\pm,n}$: even though the twisted form itself depends on $\Lambda$, this mass scale decouples after evaluating intersection numbers. Since each order in $\Lambda$ is non-logarithmic, it produces tachyonic (massive) poles in the $+$ ($-$) case, which conspire to cancel out after summing over all terms. It remains an open question whether there exists a more efficient representation of $\varphi^{\text{gauge}}_{\pm,n}$ that makes this fact manifest.\footnote{For practical application it is of course sufficient to use ${\cal A}^{\text{YM}}(\alpha) = \lim_{\Lambda \to 0} \la \PT(\alpha) | \,d\mu_n \Pf {\bm\Psi}_{12} / (z_1{-}z_2)\ra_\omega$ and a similar expression for gravity, however this does not qualify as a satisfactory answer to the above question.} Perhaps the answer will come from formulating intersection theory of twisted forms on supermoduli spaces.

Let us briefly comment on the relation to the scattering equations formalism implied by \eqref{CHY-formula}. Since in the cases of our interest the left-hand side of \eqref{CHY-formula} is independent of $\Lambda$, also its right-hand side computes Yang--Mills and gravity amplitude exactly. Inside the residue we need to use $\lim_{\Lambda\to0} \varphi^{\text{gauge}}_{\pm,n} = d\mu_n \Pf {\bm\Psi}_{12} / (z_1 {-} z_2)$, which is indeed the factor proposed by Cachazo, He, and Yuan \cite{Cachazo:2013hca,Cachazo:2013iea}. An alternative proof of the fact that this formula computes tree-level amplitudes in Yang--Mills theory was given earlier by Dolan and Goddard \cite{Dolan:2013isa}.

In the case of bosonic strings one can also make use of a twisted form inherited from string theory \cite{green1988superstring}, which can be written concisely as:
\be
\varphi^{\text{bosonic}}_{+,n} = \Lambda^{n-2}\, d\mu_n\! \int \prod_{i=1}^{n} d\theta_i d\tilde\theta_i \exp \left( \sum_{j \neq i} \left(\frac{1}{\Lambda}\frac{ 2\theta_j\tilde\theta_j p_i {\cdot} \varepsilon_j}{z_i {-} z_j} + \frac{\theta_i \tilde\theta_i \theta_j \tilde\theta_j \varepsilon_i {\cdot} \varepsilon_j}{(z_i {-} z_j)^2} \right) \right),
\ee
using a pair of auxiliary Grassmann variables $\theta_i, \tilde\theta_i$ for each puncture (the form $\varphi^{\text{bosonic}}_{-,n}$ is obtained by replacing $\Lambda \to i\Lambda$). Unlike in the case of $\varphi^{\text{gauge}}_{\pm,n}$, intersection numbers computed with $\varphi^{\text{bosonic}}_{\pm,n}$ are no longer independent of $\Lambda$. As mentioned before, the low-energy limit, $\alpha' \to 0$ (or equivalently $\Lambda \to \infty$), of $\la \PT(\alpha) | \varphi^{\text{bosonic}}_{+,n}\ra_\omega$, $\la \widetilde{\varphi}^{\text{gauge}}_{-,n} | \varphi^{\text{bosonic}}_{+,n}\ra_\omega$, and $\la \widetilde{\varphi}^{\text{bosonic}}_{-,n} | \varphi^{\text{bosonic}}_{+,n}\ra_\omega$ coincide with gauge and gravity amplitudes. On the other hand, the massless limit, $\Lambda\to 0$, leads to scattering amplitudes in certain non-unitary theories, such as conformal gravity \cite{Berkovits:2004jj,Johansson:2017srf,Azevedo:2017lkz}. For instance, we have
\begin{align}
\Braket{ \PT(1234) | \varphi^{\text{bosonic}}_{+,n} }_\omega {=} {\cal A}^{\text{YM}}(1234) {-} 4 s_{13} \bigg(
\frac{(s_{12}\, \varepsilon_1 {\cdot} \varepsilon_2 {-} 2 p_1{\cdot}\varepsilon_2\, p_2{\cdot}\varepsilon_1)(s_{12}\, \varepsilon_3 {\cdot}\varepsilon_4 {-} 2 p_3{\cdot} \varepsilon_4\, p_4 {\cdot} \varepsilon_3)}{s_{12}^2 (s_{12} {-} \Lambda^2)}&\nn\\
+ \frac{(s_{13}\, \varepsilon_1 {\cdot} \varepsilon_3 {-} 2 p_1{\cdot}\varepsilon_3\, p_3{\cdot}\varepsilon_1)(s_{13}\, \varepsilon_4 {\cdot}\varepsilon_2 {-} 2 p_4{\cdot} \varepsilon_2\, p_2 {\cdot} \varepsilon_4)}{s_{13}^2 (s_{13}{-}\Lambda^2)}&\nn\\
+ \frac{(s_{23}\, \varepsilon_1 {\cdot} \varepsilon_4 {-} 2 p_1{\cdot}\varepsilon_4\, p_4{\cdot}\varepsilon_1)(s_{23}\, \varepsilon_2 {\cdot}\varepsilon_3 {-} 2 p_2{\cdot} \varepsilon_3\, p_3 {\cdot} \varepsilon_2)}{s_{23}^2 (s_{23}{-}\Lambda^2)}&\\
+ \frac{4(p_4{\cdot}\varepsilon_1\, s_{12} {-} p_{2}{\cdot}\varepsilon_1\, s_{23})(p_1{\cdot}\varepsilon_2\, s_{23} {-} p_{3}{\cdot}\varepsilon_2\, s_{12})(p_2{\cdot}\varepsilon_3\, s_{12} {-} p_{4}{\cdot}\varepsilon_3\, s_{23})(p_3{\cdot}\varepsilon_4\, s_{23} {-} p_{1}{\cdot}\varepsilon_4\, s_{12})}{\Lambda^2\, s_{12}^2\, s_{13}^2\, s_{23}^2}&
\bigg),\nn
\end{align}
in agreement with \cite{Huang:2016tag}. Surprisingly, even at finite $\Lambda$ some of the aforementioned intersection numbers compute amplitudes in quantum field theories recently identified by Azevedo et al. \cite{Azevedo:2018dgo}.

\subsection{\label{sec:Kac-Moody-currents}Intersection Numbers of Kac--Moody Correlators}

\textsc{Let us consider correlation functions} of Kac--Moody currents, see, e.g., \cite{doi:10.1142/S0217751X86000149,frenkel1992}.\footnote{To be more precise, we are referring to the correlators $\la {\cal J}^{c_1}(z_1) {\cal J}^{c_2}(z_2) \cdots {\cal J}^{c_n}(z_n) \ra$ of the Kac--Moody currents ${\cal J}^{c_i}(z_i)$ with the operator product expansion
	\be
	\la {\cal J}^{a}(z_i) {\cal J}^{b}(z_j) \ra = \frac{k\, \delta^{a b}}{(z_i {-} z_j)^2} + \frac{{f^{a b}}_{c}\, {\cal J}^{c}(z_j)}{z_i {-} z_j} + \ldots,
	\ee
	where ${f^{ab}}_{c}$ are structure constants of a compact simple Lie algebra and $k$ is the central extension.
}
We will be interested in coefficients of each trace structure $\text{Tr}(T^{c_{\alpha_1(1)}} T^{c_{\alpha_1(2)}} \cdots) \text{Tr}(T^{c_{\alpha_2(1)}}\allowbreak T^{c_{\alpha_2(2)}} \cdots) \cdots\allowbreak \text{Tr}(T^{c_{\alpha_{\text{T}}(1)}} T^{c_{\alpha_{\text{T}}(2)}} \cdots)$, which regarded as a twisted form in $H^{n-3}_{\pm\omega}$ read, up to an overall constant,
\be\label{multi-trace-PT}
\mathrm{PT}(\alpha_1 | \alpha_2 | \cdots | \alpha_{\mathrm{T}}) := \frac{d\mu_n}{\prod_{m=1}^{T} \prod_{i=1}^{|\alpha_m|} \left(z_{\alpha_m\! (i)} - z_{\alpha_m\! (i+1)}\right)}.
\ee
The set $\{\alpha_1, \alpha_2,\ldots, \alpha_{\mathrm{T}}\}$ is a partition of $n$ labels into $\mathrm{T}$ orderings $\{\alpha_m\}_{m=1}^{\mathrm{T}}$ with at least two labels each. The resulting multi-trace Parke--Taylor factors \eqref{multi-trace-PT} are $\SL(2,\C)$-invariant and can serve as building blocks for more complicated twisted forms. In the sequel we consider their intersection numbers.

Following the notation used in the previous subsection we use $\bm\varphi_{p}^{\pm}[\varphi_\pm]$ to denote partial results of the recursion relations from Theorem~\ref{theorem-31} associated to the boundary conditions $\bm\varphi_{n}^{\pm} = \varphi_\pm$. The choice of fibration bases was designed to simplify computation of planar amplitudes: using \eqref{FB-PT} and the definition \eqref{bold-varphi} we find:
\be\label{planar-recursion}
\bm\varphi_3^+[\PT(12\cdots n)] = \left( 0, 0, \ldots, 0, 1 \right)^\intercal,
\ee
which becomes a vector of $(n{-}3)!{-}1$ zeros and a single one in the final entry. This can be exploited when computing $\la \varphi_- | \PT(12\cdots n) \ra_\omega$ as only the final entry of $\bm\varphi_3^-[\varphi_-]$ needs to be computed in such cases.

We close this section by giving examples of intersection numbers for $n=4,5,6$. Computations of the results printed below typically take a fraction of a second using modern symbolic manipulation software.

\begin{example}
	For $n=4$ computing $\bm\varphi_3^-$ for all inequivalent Parke--Taylor forms gives:
	\begin{align}
	\bm\varphi_3^-[ \big( \PT(1234), \PT(1243), \PT(1423), \PT(12|34), \PT(13|24), \PT(14|23) \big) ]&\nn\\
	=\left(
	\dfrac{1}{s}{+}\dfrac{1}{t}\quad -\dfrac{1}{s}\quad  -\dfrac{1}{t}\quad  \dfrac{u}{s \left(s{+}\Lambda ^2\right)}\quad \dfrac{1}{u{+}\Lambda^2}\quad
	\dfrac{u}{t \left(t{+}\Lambda^2\right)}
	\right),&\label{KM-examples-4-minus}
	\end{align}
	where we used $s := s_{12}$, $t := s_{23}$, and $u := s_{13}$. Similarly, for $\bm\varphi_3^+$ we obtain:
	\begin{align}
	\bm\varphi_3^+[ \big( \PT(1234), \PT(1324), \PT(1342), \PT(12|34), \PT(13|24),\PT(14|23) \big) ]&\nn\\
	=\left(
	1\quad \dfrac{t}{u}\quad \dfrac{s}{u}\quad -\dfrac{t}{s{-}\Lambda ^2}\quad -\dfrac{s t}{u \left(u{-}\Lambda ^2\right)}\quad -\dfrac{s}{t{-}\Lambda ^2}
	\right).&\label{KM-examples-4-plus}
	\end{align}
	Computing all combinations of intersection numbers amounts to taking $\eqref{KM-examples-4-minus}^\intercal \eqref{KM-examples-4-plus}$, which gives:
	\be
	\left(
	\begin{array}{*6{>{}c}}
		\frac{1}{s}{+}\frac{1}{t} & -\frac{1}{s} & -\frac{1}{t} & \frac{u}{s \left(s{-}\Lambda ^2\right)} & \frac{1}{u{-}\Lambda ^2} &
		\frac{u}{t \left(t{-}\Lambda ^2\right)} \\
		-\frac{1}{s} & \frac{1}{s}{+}\frac{1}{u} & -\frac{1}{u} & \frac{t}{s \left(s{-}\Lambda ^2\right)} & \frac{t}{u \left(u{-}\Lambda
			^2\right)} & \frac{1}{t{-}\Lambda ^2} \\
		-\frac{1}{t} & -\frac{1}{u} & \frac{1}{t}{+}\frac{1}{u} & \frac{1}{s{-}\Lambda ^2} & \frac{s}{u \left(u{-}\Lambda ^2\right)} &
		\frac{s}{t \left(t{-}\Lambda ^2\right)} \\
		\frac{u}{s \left(s{+}\Lambda ^2\right)} & \frac{t}{s \left(s{+}\Lambda ^2\right)} & \frac{1}{s{+}\Lambda ^2} & -\frac{t u}{s
			\left(s{-}\Lambda ^2\right) \left(s{+}\Lambda ^2\right)} & -\frac{t}{\left(s{+}\Lambda ^2\right) \left(u{-}\Lambda ^2\right)} &
		-\frac{u}{\left(s{+}\Lambda ^2\right) \left(t{-}\Lambda ^2\right)} \\
		\frac{1}{u{+}\Lambda ^2} & \frac{t}{u \left(u{+}\Lambda ^2\right)} & \frac{s}{u \left(u{+}\Lambda ^2\right)} &
		-\frac{t}{\left(s{-}\Lambda ^2\right) \left(u{+}\Lambda ^2\right)} & -\frac{s t}{u \left(u{-}\Lambda ^2\right) \left(u{+}\Lambda
			^2\right)} & -\frac{s}{\left(t{-}\Lambda ^2\right) \left(u{+}\Lambda ^2\right)} \\
		\frac{u}{t \left(t{+}\Lambda ^2\right)} & \frac{1}{t{+}\Lambda ^2} & \frac{s}{t \left(t{+}\Lambda ^2\right)} &
		-\frac{u}{\left(s{-}\Lambda ^2\right) \left(t{+}\Lambda ^2\right)} & -\frac{s}{\left(t{+}\Lambda ^2\right) \left(u{-}\Lambda
			^2\right)} & -\frac{s u}{t \left(t{-}\Lambda ^2\right) \left(t{+}\Lambda ^2\right)} \\
	\end{array}
	\right).
	\ee
\end{example}

\begin{example}
	For $n=5$ we start by considering all $12$ inequivalent single-trace contributions. In the first step of the recursion for $\bm\varphi_p^-$ with $p{=}4$ we find:
	\begin{align}
	&\bm\varphi_4^-[ \big( \PT(12345), \PT(12354), \PT(12435), \PT(12453), \PT(12534), \PT(12543) \nn\\
	&\qquad\,\PT(14235), \PT(14253), \PT(14523), \PT(15234), \PT(15243), \PT(15423) \big) ] \nn\\
	&= \left(
	\begin{array}{*6{>{\displaystyle}c}}
	\frac{z_{13}}{s_{15} z_{14} z_{43}} & \frac{z_{13}}{s_{35} z_{14} z_{43}} & \frac{\left(s_{35}{+}s_{15}\right) z_{23}}{s_{35}
		s_{15} z_{24} z_{34}} & \frac{z_{23}}{s_{35} z_{24} z_{43}} & \frac{z_{31}}{s_{35} z_{14} z_{43}} & 0 \\
	\frac{\left(s_{45}{+}s_{15}\right) z_{13}}{s_{45} s_{15} z_{14} z_{43}} & \frac{z_{31}}{s_{45} z_{14} z_{43}} &
	\frac{z_{32}}{s_{15} z_{24} z_{43}} & \frac{z_{32}}{s_{45} z_{24} z_{43}} & 0 & \frac{z_{23}}{s_{45} z_{24} z_{43}} \\
	\end{array}\right. \nn\\
	&\qquad
	\left. \begin{array}{*6{>{\displaystyle}c}}
	\frac{\left(s_{35}{+}s_{15}\right) z_{12}}{s_{35} s_{15} z_{14} z_{24}} & \frac{z_{12}}{s_{35} z_{14} z_{42}} & 0 &
	\frac{z_{31}}{s_{15} z_{14} z_{43}} & \frac{z_{23}}{s_{15} z_{24} z_{43}} & \frac{z_{12}}{s_{15} z_{14} z_{42}} \\
	\frac{z_{21}}{s_{15} z_{14} z_{42}} & 0 & \frac{z_{21}}{s_{45} z_{14} z_{42}} & \frac{z_{31}}{s_{15} z_{14} z_{43}} &
	\frac{z_{23}}{s_{15} z_{24} z_{43}} & \frac{\left(s_{45}{+}s_{15}\right) z_{12}}{s_{45} s_{15} z_{14} z_{42}} \\
	\end{array}
	\right) dz_4,
	\end{align}
	which in the second step, $p{=}3$, gives:
	\begin{align}
	&\bm\varphi_3^-[ \big( \PT(12345), \PT(12354), \PT(12435), \PT(12453), \PT(12534), \PT(12543) &\nn\\
	&\qquad\,\PT(14235), \PT(14253), \PT(14523), \PT(15234), \PT(15243), \PT(15423) \big) ] \nn\\
	&= \medmath{\left(\!\!
	\begin{array}{*6{>{}c}}
	\frac{\frac{1}{s_{34}}{+}\frac{1}{s_{23}}}{s_{15}}{+}\frac{1}{s_{12} s_{34}} &
	\frac{\frac{1}{s_{35}}{+}\frac{1}{s_{23}}}{s_{14}}{+}\frac{1}{s_{12} s_{35}} &
	-\frac{\frac{1}{s_{35}}{+}\frac{1}{s_{34}}}{s_{12}}{-}\frac{1}{s_{34} s_{15}} & \frac{1}{s_{12} s_{35}} &
	-\frac{\frac{1}{s_{35}}{+}\frac{1}{s_{34}}}{s_{12}}{-}\frac{1}{s_{14} s_{35}} & \!\!\!\!\frac{1}{s_{12} s_{34}} \\
	\frac{\frac{1}{s_{45}}{+}\frac{1}{s_{34}}}{s_{12}}{+}\frac{\frac{1}{s_{15}}{+}\frac{1}{s_{45}}}{s_{23}}{+}\frac{1}{s_{34} s_{15}} &
	-\frac{\frac{1}{s_{23}}{+}\frac{1}{s_{12}}}{s_{45}} & -\frac{\frac{1}{s_{15}}{+}\frac{1}{s_{12}}}{s_{34}} & -\frac{1}{s_{12}
		s_{45}} & -\frac{1}{s_{12} s_{34}} & \!\!\!\!\frac{\frac{1}{s_{45}}{+}\frac{1}{s_{34}}}{s_{12}} \\
	\end{array}
	\right.}\nn\\
	&\qquad
	\left.
	\begin{array}{*6{>{}c}}
	-\frac{\frac{\frac{s_{23}}{s_{35}}+1}{s_{14}}{+}\frac{1}{s_{15}}}{s_{23}} & \frac{1}{s_{14} s_{35}} & \frac{1}{s_{14} s_{23}} &
	-\frac{\frac{1}{s_{34}}+\frac{1}{s_{23}}}{s_{15}}{-}\frac{1}{s_{14} s_{23}} & \frac{1}{s_{34} s_{15}} & \frac{1}{s_{23} s_{15}}
	\\
	-\frac{1}{s_{23} s_{15}} & 0 & -\frac{1}{s_{23} s_{45}} & -\frac{\frac{1}{s_{34}}+\frac{1}{s_{23}}}{s_{15}} & \frac{1}{s_{34}
		s_{15}} & \frac{\frac{1}{s_{15}}+\frac{1}{s_{45}}}{s_{23}} \\
	\end{array}
	\right).
	\end{align}
	Similarly, for $\bm\varphi_p^{+}$ and $p{=}4$ we obtain:
	\begin{align}
	&\bm\varphi_4^+[ \big( \PT(12345), \PT(12354), \PT(12435), \PT(12453), \PT(12534), \PT(12543) \nn\\
	&\qquad\,\PT(14235), \PT(14253), \PT(14523), \PT(15234), \PT(15243), \PT(15423) \big) ] \nn\\
	&= \left(
	\begin{array}{*6{>{\displaystyle}c}}
	0 & \frac{z_{13}}{z_{14} z_{43}} & \frac{z_{32}}{z_{24} z_{43}} & \frac{z_{23}}{z_{24} z_{43}} &
	\frac{\left(s_{25}{+}s_{35}\right) z_{13}}{s_{25} z_{14} z_{34}} & \frac{s_{35} z_{23}}{s_{25} z_{24} z_{43}} \\
	\frac{z_{13}}{z_{14} z_{43}} & \frac{z_{31}}{z_{14} z_{43}} & 0 & \frac{z_{32}}{z_{24} z_{43}} & \frac{s_{45}
		z_{31}}{s_{25} z_{14} z_{43}} & \frac{\left(s_{25}{+}s_{45}\right) z_{23}}{s_{25} z_{24} z_{43}} \\
	\end{array}
	\right.\nn\\
	&\qquad
	\left.
	\begin{array}{*6{>{\displaystyle}c}}
	\frac{z_{21}}{z_{14} z_{42}} & \frac{\left(s_{25}{+}s_{35}\right) z_{12}}{s_{25} z_{14} z_{42}} & \frac{s_{35} z_{21}}{s_{25}
		z_{14} z_{42}} & \frac{s_{35} z_{13}}{s_{25} z_{14} z_{43}} & \frac{s_{35} z_{32}}{s_{25} z_{24} z_{43}} & 0 \\
	0 & \frac{s_{45} z_{12}}{s_{25} z_{14} z_{42}} & \frac{\left(s_{25}{+}s_{45}\right) z_{21}}{s_{25} z_{14} z_{42}} & \frac{s_{45}
		z_{13}}{s_{25} z_{14} z_{43}} & \frac{s_{45} z_{32}}{s_{25} z_{24} z_{43}} & \frac{z_{12}}{z_{14} z_{42}} \\
	\end{array}
	\right) dz_4,
	\end{align}
	which in the final step, $p{=}3$, gives rise to
	\begin{align}
	&\bm\varphi_3^+[ \big( \PT(12345), \PT(12354), \PT(12435), \PT(12453), \PT(12534), \PT(12543) &\nn\\
	&\qquad\,\PT(14235), \PT(14253), \PT(14523), \PT(15234), \PT(15243), \PT(15423) \big) ] \nn\\
	&=\left(
	\begin{array}{*6{>{}c}}
	0 & 1 & \frac{s_{14}}{s_{24}} & -\frac{s_{14} \left(s_{13}{+}s_{35}\right)}{s_{24} s_{13}} & -\frac{s_{35}}{s_{25}}{-}1 &
	-\frac{s_{14} s_{35}}{s_{13} s_{25}} \\
	1 & -1 & \frac{s_{45}}{s_{24}} & \frac{s_{15} s_{45}{-}s_{12} s_{24}}{s_{13} s_{24}}{-}1 & -\frac{s_{45}}{s_{25}} & \frac{s_{12}
		s_{25}{-}s_{14} s_{45}}{s_{25} s_{13}}{+}1 \\
	\end{array}
	\right. \label{example-n-5-1}\\
	&\qquad
	\medmath{\left.
	\begin{array}{cccccc}
	-\frac{s_{14}}{s_{24}}{-}1 &\!\!\! \frac{\frac{s_{14} \left(s_{24}+s_{25}\right) \left(s_{13}+s_{35}\right)}{s_{13}
			s_{24}}-s_{23}}{s_{25}} & \frac{s_{35} \left(s_{12}+s_{15}\right)}{s_{13} s_{25}} & \frac{s_{35}}{s_{25}} & \frac{s_{14}
		\left(s_{24}+s_{25}\right) s_{35}}{s_{13} s_{24} s_{25}} & \!\!\!-\frac{s_{14} s_{35}}{s_{13} s_{24}} \\
	-\frac{s_{45}}{s_{24}} &\!\!\! \frac{\left(\frac{s_{13}+s_{14}}{s_{25}}-\frac{s_{15}}{s_{24}}\right) s_{45}}{s_{13}} & \frac{s_{12}
		s_{25}-\left(s_{13}+s_{14}\right) s_{45}}{s_{13} s_{25}} & \frac{s_{45}}{s_{25}} & \frac{\left(s_{14}
		s_{24}-\left(s_{13}+s_{15}\right) s_{25}\right) s_{45}}{s_{13} s_{24} s_{25}} & \!\!\!\frac{\left(s_{13}+s_{23}\right)
		s_{24}-s_{35} s_{45}}{s_{13} s_{24}} \\
	\end{array}\!\!
	\right)}.\nn
	\end{align}
	Considering all $10$ double-trace contributions, in the first step, $p{=}4$, for $\bm\varphi_p^-$ we obtain:
	\begin{align}
	&\bm\varphi_4^-[ \big( \PT(12|345), \PT(13|245), \PT(14|235), \PT(15|234), \PT(23|145) \nn\\
	&\qquad\,\PT(24|135), \PT(25|134), \PT(34|125), \PT(35|124), \PT(45|123) \big) ] \nn\\
	&= \left(
	\begin{array}{*5{>{\displaystyle}c}}
	\frac{z_{13} z_{32}}{s_{35} z_{12} z_{34}^2} & 0 & \frac{z_{12} z_{13}}{s_{35} z_{14}^2 z_{23}} &
	\frac{\frac{s_{45}}{z_{41}}{+}\frac{s_{35}+s_{15}}{z_{42}}{+}\frac{s_{25}}{z_{43}}}{s_{15} \left(s_{15}{+}\Lambda^2\right)} &
	\frac{z_{12} z_{13}}{s_{15} z_{14}^2 z_{23}} \\
	\frac{z_{13} z_{23}}{s_{45} z_{12} z_{34}^2} & \frac{z_{12} z_{23}}{s_{45} z_{13} z_{24}^2} & 0 &
	\frac{\frac{s_{45}+s_{15}}{z_{41}}{+}\frac{s_{35}}{z_{42}}{+}\frac{s_{25}}{z_{43}}}{s_{15} \left(s_{15}{+}\Lambda^2\right)} &
	\frac{\left(s_{45}{+}s_{15}\right) z_{12} z_{13}}{s_{45} s_{15} z_{14}^2 z_{23}} \\
	\end{array}\right.\\
	&\;\quad\left.
	\begin{array}{*5{>{\displaystyle}c}}
	\frac{\left(s_{35}{+}s_{15}\right) z_{12} z_{23}}{s_{35} s_{15} z_{13} z_{24}^2} & \frac{z_{13}}{z_{14} z_{43} \left(s_{25}{+}\Lambda^2\right)} & \frac{z_{13} z_{23}}{s_{15} z_{12} z_{34}^2} & \frac{s_{45} z_{14} z_{23} {-} s_{25} z_{12} z_{34}}{s_{35}
		z_{14} z_{42} z_{43} \left(s_{35}{+}\Lambda^2\right)} & \frac{z_{31}}{z_{41}z_{34}(s_{45}{+}\Lambda^2)} \\
	\frac{z_{12} z_{23}}{s_{15} z_{13} z_{24}^2} & \frac{z_{23}}{z_{24} z_{43} \left(s_{25}{+}\Lambda^2\right)} & \frac{z_{13}
		z_{23}}{s_{15} z_{12} z_{34}^2} & \frac{z_{23}}{z_{24} z_{43} \left(s_{35}{+}\Lambda^2\right)} &
	\frac{\frac{s_{45}+s_{15}}{z_{41}}{+}\frac{s_{25}}{z_{42}}{+}\frac{s_{35}+2 s_{45}}{z_{43}}}{s_{45} \left(s_{45}{+}\Lambda^2\right)} \\
	\end{array}
	\!\!\right)
	dz_4,\nn
	\end{align}
	which in the second step, $p{=}3$, gives:
	\begin{align}
	&\bm\varphi_3^-[ \big( \PT(12|345), \PT(13|245), \PT(14|235), \PT(15|234), \PT(23|145) \nn\\
	&\qquad\,\PT(24|135), \PT(25|134), \PT(34|125), \PT(35|124), \PT(45|123) \big) ] \nn\\
	&= \left(
	\begin{array}{ccccc}
	\frac{s_{34} \left(s_{45}{-}s_{13}\right)+\left(s_{12}{-}s_{45}\right) s_{25}}{s_{12} s_{34} s_{35} \left(s_{12}{+}\Lambda^2\right)}
	& 0 & \frac{s_{23} s_{24}{-}s_{35} s_{45}}{s_{14} s_{23} s_{35} \left(s_{14}{+}\Lambda^2\right)} & \frac{s_{23} s_{25}{-}s_{34}
		s_{45}}{s_{23} s_{34} s_{15} \left(s_{15}{+}\Lambda^2\right)} & \frac{s_{14} s_{24}{-}s_{25} s_{15}}{s_{14} s_{23} s_{15}
		\left(s_{23}{+}\Lambda^2\right)} \\
	\frac{s_{45} \left(s_{25}-s_{34}\right)+s_{13} s_{34}}{s_{12} s_{34} s_{45} \left(s_{12}{+}\Lambda^2\right)} &
	\frac{1}{s_{45} \left(s_{13}{+}\Lambda^2\right)} & \frac{1}{s_{23} \left(s_{14}{+}\Lambda^2\right)} & \frac{s_{23}
		\left(s_{25}-s_{34}\right)+s_{34} s_{14}}{s_{23} s_{34} s_{15} \left(s_{15}{+}\Lambda^2\right)} &
	\frac{\left(s_{23}-s_{14}\right) s_{24}+s_{25} s_{15}}{s_{23} s_{45} s_{15} \left(s_{23}{+}\Lambda^2\right)} \\
	\end{array}
	\right.\nn\\
	&\qquad\left.
	\begin{array}{ccccc}
	\frac{s_{35}+s_{15}}{s_{35} s_{15} \left(s_{24}{+}\Lambda^2\right)} & \frac{s_{14}+s_{34}}{s_{14} s_{34} \left(s_{25}{+}\Lambda^2\right)} & \frac{s_{12} s_{24}-s_{45} s_{15}}{s_{12} s_{34} s_{15} \left(s_{34}{+}\Lambda^2\right)} & \frac{s_{12}
		s_{25}-s_{14} s_{45}}{s_{12} s_{14} s_{35} \left(s_{35}{+}\Lambda^2\right)} & -\frac{s_{12}+s_{23}}{s_{12} s_{23}
		\left(s_{45}{+}\Lambda^2\right)} \\
	\frac{1}{s_{15} \left(s_{24}{+}\Lambda^2\right)} & \frac{1}{s_{34} \left(s_{25}{+}\Lambda^2\right)} & \frac{s_{12}
		s_{24}-\left(s_{34}+s_{45}\right) s_{15}}{s_{12} s_{34} s_{15} \left(s_{34}{+}\Lambda^2\right)} & \frac{1}{s_{12}
		\left(s_{35}{+}\Lambda^2\right)} & \frac{s_{12} s_{14}-s_{23} \left(s_{34}+s_{45}\right)}{s_{12} s_{23} s_{45} \left(s_{45}{+}\Lambda^2\right)} \\
	\end{array}
	\right).\label{example-n-5-2}
	\end{align}
	Finally, in the case of $\bm\varphi_p^+$ for $p{=}4$ we find:
	\begin{align}
	&\bm\varphi_4^+[ \big( \PT(12|345), \PT(13|245), \PT(14|235), \PT(15|234), \PT(23|145) \nn\\
	&\qquad\,\PT(24|135), \PT(25|134), \PT(34|125), \PT(35|124), \PT(45|123) \big) ] \nn\\
	&= \left(
	\begin{array}{*5{>{\displaystyle}c}}
	\frac{z_{13} z_{32}}{z_{12} z_{34}^2} & \frac{s_{35} z_{12} z_{23}}{s_{25} z_{13} z_{24}^2} & \frac{\left(s_{45}{+}s_{15}\right)
		z_{21} z_{13}}{s_{25} z_{14}^2 z_{23}} & \frac{s_{35} z_{23}}{z_{24} z_{34} \left(s_{15}{-}\Lambda^2\right)} & 0 \\
	\frac{z_{13} z_{23}}{z_{12} z_{34}^2} & \frac{\left(s_{35}{+}s_{15}\right) z_{12} z_{23}}{s_{25} z_{31} z_{24}^2} & \frac{s_{45}
		z_{12} z_{13}}{s_{25} z_{14}^2 z_{23}} & \frac{s_{45} z_{13}}{z_{14} z_{34} \left(s_{15}{-}\Lambda^2\right)} & \frac{z_{12}
		z_{13}}{z_{14}^2 z_{23}} \\
	\end{array}
	\right.\\
	&\qquad
	\left.
	\begin{array}{*5{>{\displaystyle}c}}
	\frac{z_{12} z_{23}}{z_{13} z_{24}^2} & \frac{s_{35} \left(s_{15} z_{13} z_{24}{+}s_{45} z_{12} z_{34}\right)}{s_{25} z_{14}
		z_{42} z_{43} \left(s_{25}{-}\Lambda ^2\right)} & \frac{s_{35} z_{13} z_{32}}{s_{25} z_{12} z_{34}^2} & \frac{s_{45} z_{13}
		z_{24}{+}s_{15} z_{12} z_{34}}{z_{14} z_{42} z_{43} \left(s_{35}{-}\Lambda ^2\right)} & \frac{s_{35} z_{23}}{z_{24} z_{34}
		\left(s_{45}{-}\Lambda ^2\right)} \\
	0 & \frac{s_{45} \left(s_{15} z_{14} z_{23} {-}s_{35} z_{12} z_{34}\right)}{s_{25} z_{14} z_{42} z_{43} \left(s_{25}{-}\Lambda
		^2\right)} & \frac{s_{45} z_{13} z_{32}}{s_{25} z_{12} z_{34}^2} & \frac{s_{45} z_{31}}{z_{14} z_{34} \left(s_{35}{-}\Lambda
		^2\right)} & \frac{s_{15} z_{12} z_{34}{+}s_{35} z_{23} z_{41}}{z_{14} z_{24} z_{34} \left(s_{45}{-}\Lambda ^2\right)} \\
	\end{array}
	\right)
	dz_4,\nn
	\end{align}
	and consequently
	\begin{align}
	&\bm\varphi_3^+[ \big( \PT(12|345), \PT(13|245), \PT(14|235), \PT(15|234), \PT(23|145) \nn\\
	&\qquad\,\PT(24|135), \PT(25|134), \PT(34|125), \PT(35|124), \PT(45|123) \big) ] \nn\\
	&=\! \medmath{\left(\!
	\begin{array}{ccccc}
	\frac{s_{14}}{s_{12}-\Lambda ^2} & \frac{s_{14} s_{35} \left(s_{12} \left(s_{23}+s_{34}\right)+s_{23} s_{15}\right)}{s_{13}
		s_{24} s_{25} \left(s_{13}-\Lambda ^2\right)} & \!\!\!\!\!\frac{s_{12} s_{45}+\left(s_{34}+s_{45}\right) s_{15}}{s_{25}
		\left(s_{14}-\Lambda ^2\right)} & \!\!\!\frac{s_{14} s_{35}}{s_{24} \left(s_{15}-\Lambda ^2\right)} & \!\!\!\frac{s_{35}}{s_{23}-\Lambda
		^2} \\
	\frac{s_{12}+s_{13}}{s_{12}-\Lambda ^2} & \frac{s_{23} \left(s_{13}+s_{23}\right) \left(s_{13}-s_{25}\right)
		s_{25}+\left(s_{14} \left(s_{12}{+}s_{23}\right)+\left(s_{23}{-}s_{12}\right) s_{25}\right) s_{34} s_{45}}{s_{13} s_{24} s_{25}
		\left(s_{13}-\Lambda ^2\right)} & -\frac{s_{45} \left(s_{12}+s_{15}\right)}{s_{25} \left(s_{14}-\Lambda ^2\right)} &\!\!\!
	\frac{\left(s_{12}+s_{14}\right) s_{45}}{s_{24} \left(s_{15}-\Lambda ^2\right)} & \!\!\!\frac{s_{45}-s_{12}}{s_{23}-\Lambda ^2} \\
	\end{array}
	\right.}\nn\\
	&\quad\left.
	\begin{array}{ccccc}
	\frac{s_{14} \left(s_{12} s_{34}+s_{23} \left(s_{34}+s_{45}\right)\right)}{s_{13} s_{24} \left(s_{24}-\Lambda ^2\right)} &
	\frac{s_{35} \left(s_{23} s_{15}+s_{12} \left(s_{45}+s_{15}\right)\right)}{s_{13} s_{25} \left(s_{25}-\Lambda ^2\right)} &
	\frac{s_{14} s_{35}}{s_{25} \left(s_{34}-\Lambda ^2\right)} & \frac{s_{23} s_{45}+s_{34} \left(s_{45}+s_{15}\right)}{s_{24}
		\left(s_{35}-\Lambda ^2\right)} & \frac{s_{14} s_{35}}{s_{13} \left(s_{45}-\Lambda ^2\right)} \\
	\frac{\left(s_{12} s_{34}-s_{14} s_{23}\right) s_{45}}{s_{13} s_{24} \left(s_{24}-\Lambda ^2\right)} & \frac{s_{45}
		\left(s_{23} s_{15}-s_{12} s_{35}\right)}{s_{13} s_{25} \left(s_{25}-\Lambda ^2\right)} & \frac{\left(s_{23}+s_{35}\right)
		s_{45}}{s_{25} \left(s_{34}-\Lambda ^2\right)} & -\frac{\left(s_{23}+s_{34}\right) s_{45}}{s_{24} \left(s_{35}-\Lambda
		^2\right)} & \frac{s_{34} s_{15}-s_{14} s_{35}}{s_{13} \left(s_{45}-\Lambda ^2\right)} \\
	\end{array}\!\!
	\right)\!.
	\end{align}
	Intersection numbers are obtained simply by contracting columns of the relevant $\bm\varphi_3^\pm$. For instance, we can compute $\la \PT(12|345)| \PT(12345) \ra_\omega$ using \eqref{example-n-5-1} and \eqref{example-n-5-2}:
	\be
	\bm\varphi_3^-[\PT(12|345)]^\intercal\, \bm\varphi_3^+[\PT(12345)] = \frac{\frac{s_{25}}{s_{34}}+\frac{s_{13}}{s_{45}}-1}{s_{12}\left(s_{12}{+}\Lambda^2\right)}.
	\ee
\end{example}

\begin{example}
	For $n=6$ the number of inequivalent trace structures proliferates rapidly and hence we only print three examples. For a single trace Parke--Taylor $\PT(135642)$ we find:
	\be
	\bm\varphi_5^-[\PT(135642)] = \left(\def\arraystretch{1.4}
	\begin{array}{*1{>{\displaystyle}c}}
		0 \\
		-\frac{z_{23}}{s_{46} z_{24} z_{35} z_{45}} \\
		\frac{z_{23}}{s_{56} z_{24} z_{35} z_{45}} \\
	\end{array}
	\right) dz_4\wedge dz_5
	\ee
	in the first step of the recursion, while in the following two we have:
	\be
	\bm\varphi_4^-[\PT(135642)] = \left(\def\arraystretch{1.4}
	\begin{array}{*1{>{\displaystyle}c}}
		\frac{z_{23}}{s_{35} s_{124} z_{24} z_{34}} \\
		-\frac{z_{23}}{s_{35} s_{46} z_{24} z_{34}} \\
		\frac{\left(s_{35}{+}s_{56}\right) z_{23}}{s_{35} s_{56} s_{124} z_{24} z_{34}} \\
		0 \\
		\frac{z_{23}}{s_{46} s_{123} z_{24} z_{34}} \\
		-\frac{z_{23}}{s_{56} s_{123} z_{24} z_{34}} \\
	\end{array}
	\right) dz_4,\quad
	\bm\varphi_3^-[\PT(135642)] =
	\left(\def\arraystretch{1.4}
	\begin{array}{*1{>{\displaystyle}c}}
		-\frac{1}{s_{12} s_{35} s_{124}} \\
		\frac{1}{s_{12} s_{35} s_{46}} \\
		-\frac{s_{35}{+}s_{56}}{s_{12} s_{35} s_{56} s_{124}} \\
		0 \\
		-\frac{1}{s_{12} s_{46} s_{123}} \\
		\frac{1}{s_{12} s_{56} s_{123}} \\
	\end{array}
	\right).
	\ee
	For a double-trace $\PT(13|5642)$ the first step, $p{=}5$, the recursion gives:
	\be
	\bm\varphi_5^-[\PT(13|5642)] =
	\left(\def\arraystretch{1.4}
	\begin{array}{*1{>{\displaystyle}c}}
		0 \\
		-\frac{z_{12} z_{23}}{s_{46} z_{13} z_{24} z_{25} z_{45}} \\
		\frac{z_{12} z_{23}}{s_{56} z_{13} z_{24} z_{25} z_{45}} \\
	\end{array}
	\right) dz_4 \wedge dz_5,
	\ee
	and in the two following steps we find:
	\be
	\bm\varphi_4^-[\PT(13|5642)] = \left(\!\!\def\arraystretch{1}
	\begin{array}{*1{>{\displaystyle}c}}
		0 \\
		0 \\
		0 \\
		0 \\
		\frac{z_{12} z_{23}}{s_{46} s_{123} z_{13} z_{24}^2} \\
		-\frac{z_{12} z_{23}}{s_{56} s_{123} z_{13} z_{24}^2} \\
	\end{array}\!\!
	\right)dz_4
	,\quad
	\bm\varphi_3^-[\PT(13|5642)] = \left(\!\!\def\arraystretch{1}
	\begin{array}{*1{>{\displaystyle}c}}
		0 \\
		0 \\
		0 \\
		0 \\
		\frac{1}{s_{46} s_{123} \left(s_{13}{+}\Lambda^2\right)} \\
		-\frac{1}{s_{56} s_{123} \left(s_{13}{+}\Lambda^2\right)} \\
	\end{array}\!\!
	\right).
	\ee
	Finally, for an example triple-trace $\PT(13|56|42)$ one has:
	\be
	\bm\varphi_5^-[\PT(13|56|42)] =
	\left(\def\arraystretch{1.4}
	\begin{array}{*1{>{\displaystyle}c}}
		\frac{z_{12} z_{23}}{z_{15} z_{24}^2 z_{35} \left(s_{56}{+}\Lambda^2\right)} \\
		\frac{z_{12} z_{14} z_{23}}{z_{13} z_{15} z_{24}^2 z_{45} \left(s_{56}{+}\Lambda^2\right)} \\
		\frac{z_{12} z_{23}
			\left(\frac{s_{26}}{z_{52}}+\frac{s_{16}+s_{56}}{z_{51}}+\frac{s_{36}}{z_{53}}+\frac{s_{46}}{z_{54}}\right)}{s_{56} z_{13}
			z_{24}^2 \left(s_{56}{+}\Lambda^2\right)} \\
	\end{array}
	\right) dz_4 \wedge dz_5
	\ee
	in the first step, while the second one gives:
	\be
	\bm\varphi_4^-[\PT(13|56|42)] = \left(\def\arraystretch{1.4}
	\begin{array}{*1{>{\displaystyle}c}}
		-\frac{\left(s_{124}{+}s_{234}\right) z_{12} z_{23}}{s_{124} s_{234} z_{13} z_{24}^2 \left(s_{56}{+}\Lambda^2\right)} \\
		-\frac{z_{12} z_{23}}{s_{234} z_{13} z_{24}^2 \left(s_{56}{+}\Lambda^2\right)} \\
		\frac{\left(s_{15} s_{124}{+}\left(s_{36}{-}s_{124}\right) s_{234}\right) z_{12} z_{23}}{s_{56} s_{124} s_{234} z_{13} z_{24}^2
			\left(s_{56}{+}\Lambda^2\right)} \\
		-\frac{z_{12} z_{23}}{s_{234} z_{13} z_{24}^2 \left(s_{56}{+}\Lambda^2\right)} \\
		-\frac{\left(s_{123}{+}s_{234}\right) z_{12} z_{23}}{s_{123} s_{234} z_{13} z_{24}^2 \left(s_{56}{+}\Lambda^2\right)} \\
		\frac{\left(s_{15} s_{123}{+}\left(s_{46}{-}s_{123}\right) s_{234}\right) z_{12} z_{23}}{s_{56} s_{123} s_{234} z_{13} z_{24}^2
			\left(s_{56}{+}\Lambda^2\right)} \\
	\end{array}
	\right) dz_4,
	\ee
	and in the third one we find:
	\be
	\bm\varphi_3^-[\PT(13|56|42)] = \left(\!\!\def\arraystretch{1.4}
	\begin{array}{*1{>{\displaystyle}c}}
		-\frac{s_{124}{+}s_{234}}{s_{124} s_{234} \left(s_{24}{+}\Lambda^2\right) \left(s_{56}{+}\Lambda^2\right)} \\
		-\frac{1}{s_{234} \left(s_{24}{+}\Lambda^2\right) \left(s_{56}{+}\Lambda^2\right)} \\
		\frac{s_{15} s_{124}{+}\left(s_{36}{-}s_{124}\right) s_{234}}{s_{56} s_{124} s_{234} \left(s_{24}{+}\Lambda^2\right) \left(s_{56}{+}\Lambda^2\right)} \\
		-\frac{1}{s_{234} \left(s_{24}{+}\Lambda^2\right) \left(s_{56}{+}\Lambda^2\right)} \\
		-\frac{s_{234} \left(s_{24}{+}\Lambda^2\right){+}s_{123} \left(s_{13}+s_{234}{+}\Lambda^2\right)}{s_{123} s_{234} \left(s_{13}{+}\Lambda^2\right) \left(s_{24}{+}\Lambda^2\right) \left(s_{56}{+}\Lambda^2\right)} \\
		\frac{s_{15} s_{123} \left(s_{13}{+}\Lambda^2\right){+}s_{234} \left(s_{46} \left(s_{24}{+}\Lambda^2\right){-}s_{123} \left(s_{13}{+}s_{24}{-}s_{246}{+}\Lambda^2\right)\right)}{s_{56} s_{123} s_{234} \left(s_{13}{+}\Lambda^2\right) \left(s_{24}{+}\Lambda^2\right)
			\left(s_{56}{+}\Lambda^2\right)} \\
	\end{array}\!\!
	\right).
	\ee
	Using the result of the recursion $\bm\varphi_3^+[\PT(123456)]$ from \eqref{planar-recursion} we find easily that the intersection number $\la \PT(13|56|42) | \PT(123456)\ra_\omega$ is given by
	\be
	\bm\varphi_3^-[\PT(13|56|42)]^\intercal\, \bm\varphi_3^+[\PT(123456)] = \frac{
		\frac{s_{26}+s_{46}}{\left(s_{13}{+}\Lambda^2\right) \left(s_{24}{+}\Lambda^2\right)}
		+\frac{s_{46}}{s_{123} \left(s_{13}{+}\Lambda^2\right)}
		+\frac{s_{15}}{s_{234} \left(s_{24}{+}\Lambda^2\right)}
		-\frac{1}{s_{24}{+}\Lambda^2}
	}{s_{56} \left(s_{56}{+}\Lambda^2\right)}.
	\ee
\end{example}

It remains an open question whether there exists a quantum field theory whose colour-ordered amplitudes are computed by intersection numbers of Kac--Moody correlators. In the cases $\la \PT(12\cdots n) |\allowbreak \PT(\alpha_1 | \alpha_2 | \cdots| \alpha_{\text{T}}) \ra_\omega$ a Lagrangian for such a theory was proposed in \cite{Azevedo:2018dgo}.

\pagebreak
\section{\label{sec:conclusion}Conclusion}

\textsc{In this work we introduced} a formulation of the quantum field theory S-matrix at tree-level in terms of intersection theory of twisted forms on the moduli space $\M_{0,n}$. It not only challenges our basic understanding of where scattering amplitudes come from, but also allows for their explicit computations using the newly-introduced recursion relations. Together with the earlier reformulation of the string theory S-matrix in terms of pairings of twisted homology and cohomology groups \cite{Mizera:2017cqs}, it completes a unified geometric framework for computation of scattering amplitudes from genus-zero Riemann surfaces. Let us summarize these findings on the following diagram:
\be
\begin{tikzcd}
\varphi_- {\in} H_{-\omega}^{n-3} \arrow[rrr, "\langle \varphi_- | \varphi_+ \rangle_\omega"{inner sep=4pt}, leftrightarrow] \arrow[ddd, "\int_{\Delta_-} \!\!\KN^{-1} \varphi_-"{rotate=90,anchor=south}, leftrightarrow] &  &  & \varphi_+ {\in} H_{\omega}^{n-3} \arrow[rrr, "\int_{\M_{0,n}}\!\! |\KN|^2\, \varphi_+ \wedge \overbar{\vartheta_+}", leftrightarrow] &  &  & \overbar{\vartheta_+} {\in} H_{\overbar{\omega}}^{n-3} \arrow[rrr, "\la \overbar{\vartheta_+} | \overbar{\vartheta_-} \ra_{\overbar{\omega}}"{inner sep=4pt}, leftrightarrow]  &  &  & \overbar{\vartheta_-} {\in} H_{-\overbar{\omega}}^{n-3} \arrow[ddd, "\int_{\Gamma_{\!-}} \!\!\overbar{\KN}^{-1} \, \overbar{\vartheta_-}"{rotate=-90,anchor=south}, leftrightarrow] \\
&  &  &  &  &  &  &  &  &  \\
&  &  &  &  &  &  &  &  &  \\
\Delta_- {\in} H^{-\omega}_{n-3} \arrow[rrr, "\la \Delta_+ \otimes \KN | \Delta_- \otimes \KN^{-1} \ra"'{inner sep=5pt}, leftrightarrow]  &  &  & \Delta_+ {\in} H^{\omega}_{n-3} \arrow[uuu, "\int_{\Delta_+} \!\!\KN\, \varphi_+"{rotate=90,anchor=south}, leftrightarrow] \arrow[rrr, "\la \Delta_+ \otimes \KN | \Gamma_+ \otimes \overbar{\KN} \ra"'{inner sep=5pt}, leftrightarrow] &  &  & \Gamma_+ {\in} H^{\overbar{\omega}}_{n-3} \arrow[uuu, "\int_{\Gamma_{\!+}} \!\!\overbar{\KN}\, \overbar{\vartheta_+} "{rotate=-90,anchor=south}, leftrightarrow] \arrow[rrr, "\la \Gamma_- \otimes \overbar{\KN}^{-1} | \Gamma_+ \otimes \overbar{\KN} \ra"'{inner sep=5pt}, leftrightarrow] &  &  & \Gamma_- {\in} H^{-\overbar{\omega}}_{n-3}
\end{tikzcd}\nn
\ee
The bottom and top rows feature homology and cohomology groups with different twistings. Each arrow represents a duality between two groups by a pairing: the top ones are intersection numbers of twisted forms, as well as closed-string integrals; the vertical arrows correspond to different types of open-string integrals; while the ones at the bottom are homological intersection numbers. As special cases of the Proposition~\ref{proposition}, each quadruple of such vector spaces (with possible duplicates) induces linear algebra relations between their pairings. For example, the square in the middle of the above diagram manifests the KLT relations \cite{Kawai:1985xq,Mizera:2017cqs}.

The key outstanding problem remains to find a worldsheet model reproducing intersection numbers. It is known that in the massless case the CHY formula \eqref{CHY-formula} has an interpretation as coming from a string theory in the ambitwistor space \cite{Mason:2013sva,Berkovits:2013xba,Ohmori:2015sha}, which suggests that its generalization to intersection numbers might exist.\footnote{Worldsheet models with a finite spectrum of fields were previously constructed by Ademollo et al. in the context of ${\cal N}=2$ string \cite{Ademollo:1975an,Ademollo:1976pp,Ooguri:1990ww}.} Such a model would have to involve two left-moving (holomorphic) modes, most likely in the formulation with local coefficients \eqref{cocycle-with-local-coefficient}, in which the correlation function of vertex operators would involve plane wave contributions $\prod_{i<j} |z_i {-} z_j|^{\pm 2\alpha' p_i {\cdot} p_j}$ cancelling between the $\pm$ forms. Despite some progress made in similar directions \cite{Siegel:2015axg,Huang:2016bdd,Leite:2016fno,Li:2017emw,Lee:2017utr,Casali:2017mss,Lee:2017crr}, consistent formulation of a worldsheet action giving rise to intersection numbers remains an open problem. We leave it as a challenge for the reader.

Since in this work we focused primarily on general properties of intersection numbers and ways of evaluating them, we barely scratched the surface of the space of quantum field theories whose amplitudes admit an intersection number interpretation. We hope that the computational techniques introduced in this work will enable further exploration of this space. In particular, it would be interesting to formulate a purely four-dimensional version of intersection numbers along the lines of \cite{Witten:2003nn,Roiban:2004yf}.

Finally, it is natural to wonder about generalizations of twisted geometries to supermoduli spaces and higher genera. We leave these open problems for future investigations.

\vfill

\pagebreak
\appendix

\titleformat{name=\section}[display]
{\normalfont}
{\footnotesize\textsc{Appendix \thesection}}
{0pt}
{\Large\bfseries}
[\vspace{-10pt}\color{Maroon}\rule{\textwidth}{.6pt}]

\section[Primer on Homology with Coefficients in a Local System]{\label{app:aspects}Primer on Homology with Coefficients in a Local System}

\textsc{In this appendix} we informally introduce relevant aspects of homologies with coefficients in a local system (for textbook treatment see \cite{aomoto2011theory}). In particular, we focus on their application to moduli spaces of punctured Riemann spheres, which allows us to study geometric and topological properties of integrals appearing in scattering amplitudes of string theory \cite{Mizera:2017cqs}, such as their singularity structure, analytic continuation, and behaviour in different kinematic limits.

Many of these results depend on constructing relations between string theory amplitudes, which can be thought of as a generalization of the famous Riemann bilinear relations \cite{riemann1857theorie} to the twisted case. For this reason we start by reviewing homology with constant coefficients on compact Riemann surfaces.

\subsection{\label{app:compact-Riemann-surfaces}Compact Riemann Surfaces}

\textsc{Originally written down in 1857} \cite{riemann1857theorie}, bilinear relations of Riemann give an identity between period integrals on compact Riemann surfaces. Here we present their more modern formulation in terms of homology theory. Consider a compact Riemann surface $\Sigma_g$ of genus $g$:

\be
\begin{aligned}
	\includegraphics[scale=1]{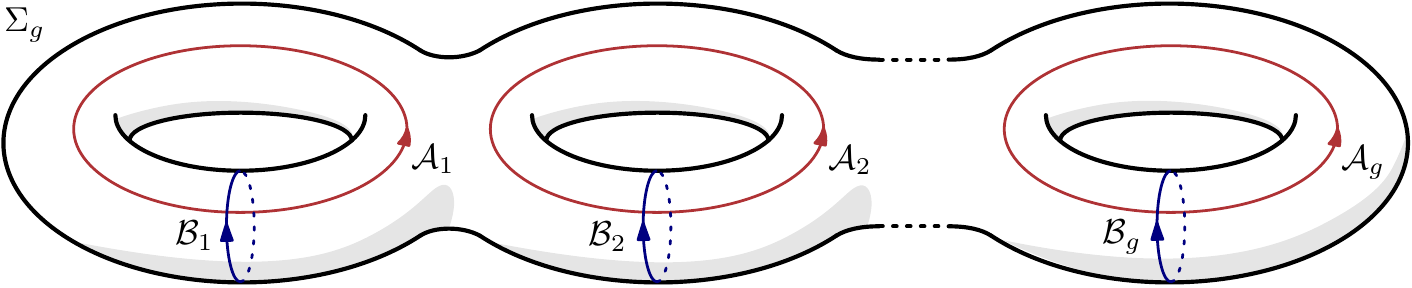}
\end{aligned}
\ee
On this surface, we indicated a canonical basis of \emph{cycles}, called ${\cal A}_i$ and ${\cal B}_i$ for $i=1,2,\ldots,g$. They are elements of the $1$-st singular homology group $H_1(\Sigma_g, \Z) := \ker \partial / \im \partial$ with integer coefficients:
\be
{\cal A}_1,\, {\cal A}_2, \,\ldots,\, {\cal A}_g,\, {\cal B}_1,\, {\cal B}_2,\, \ldots,\, {\cal B}_g \;\in\; H_1(\Sigma_g, \Z),
\ee
for the boundary operator $\partial$.
Similarly, we have the $1$-st de Rham cohomology group, $H^{1}(\Sigma_g, \Z) := \ker d / \im d$, with the differential $d$, whose canonical basis consists of 
\be
dz,\, z dz,\, \ldots,\, z^{g-1} dz,\, d\overbar{z},\, \overbar{z} d\overbar{z},\, \ldots,\, \overbar{z}^{g-1} d\overbar{z} \;\in\; H^{1}(\Sigma_g, \Z).
\ee
Element of this group are called \emph{cocycles}. It admits a Hodge decomposition $H^1(\Sigma_g, \Z) = H^{1,0}(\Sigma_g, \Z) \oplus H^{0,1}(\Sigma_g, \Z)$ into spaces of $(1,0)$- and $(0,1)$-forms. There are also other non-vanishing homology and cohomology groups, but they will not play a role here.

We can construct three types of pairings between elements of the above spaces. For instance, pairing a cycle $\cal C$ and a cocycle $\varphi$ yields the contour integral:
\be\label{cycle-cocycle-pairing}
\la{\cal C} | \varphi \ra := \oint_{\cal C} \varphi.
\ee
Similarly, a pairing of two cocycles $\varphi,\widetilde\varphi$ is given by the integral over the whole surface,\footnote{Note that this is not an inner product of $\widetilde\varphi$ and $\varphi$ on $L^2(\Sigma_g)$, which in this notation would be $\la\widetilde\varphi | \!\star\!\overbar{\varphi}\ra$ for a Hodge star operator $\star$.}
\be
\la \widetilde\varphi | \varphi \ra := \int_{\Sigma_g} \widetilde\varphi \wedge \varphi.
\ee
If both forms are (anti)holomorphic then it vanishes identically. Lastly, a pairing between two cycles $\cal C, \widetilde{\cal C}$ is the so-called \emph{intersection number},
\be\label{intersection-constant-coefficients}
\la {\cal C} | \widetilde {\cal C} \ra := \!\!\!\sum_{p \,\in\, {\cal C} \,\cap\, \widetilde {\cal C}} \pm 1.
\ee
It is an integer computed as a sum over all intersection points $p \in {\cal C} \cap \widetilde{\cal C}$, where each point is weighted with a sign depending on the relative orientation of the two cycles at the intersection point:
\be
\begin{aligned}\label{intersection-signs}
	\includegraphics[scale=1]{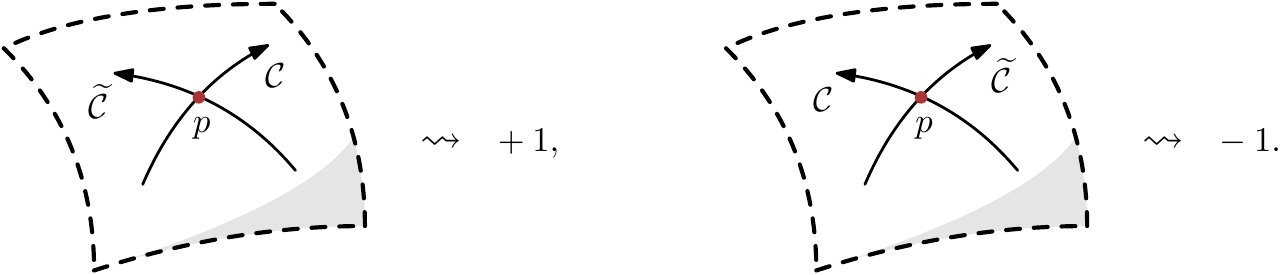}
\end{aligned}
\ee
Before using the above formula one ought to deform the cycles such that there is a finite number of transverse intersection points. It is straightforward to check that intersection numbers are topological invariants. They are antisymmetric in $\cal C$ and $\widetilde {\cal C}$, as are the other bilinears.

As a consequence of the Poincar\'e duality between the homology and cohomology groups,
\be
H_1(\Sigma_g, \Z) \cong H^1(\Sigma_g, \Z),
\ee
the above pairings are not mutually independent. Other than allowing for the integration \eqref{cycle-cocycle-pairing}, it means that for each cycle ${\cal C} \in H_1(\Sigma_g, \Z)$ there exists a cocycle $\eta({\cal C}) \in H^1(\Sigma_g, \Z)$ with support along $\cal C$ such that, in our conventions,
\be
\la {\cal C} | \varphi \ra = \la \eta({\cal C}) | \varphi \ra, \qquad \la {\cal C} | \widetilde{\cal C} \ra = \la \eta({\cal C}) | \widetilde{\cal C}\ra
\ee
for every $\varphi$ and $\widetilde{\cal C}$. Therefore, for a cocycle basis $\{\varphi_j\}_{j=1}^{2g}$ we can write
\be
\varphi_j = \sum_{k=1}^{2g} \alpha_{jk}\, \eta(\widetilde{\mathcal{C}}_{k})
\ee
in terms of a cycle basis $\{\widetilde{\mathcal{C}}_k\}_{k=1}^{2g}$ with some coefficients $\alpha_{jk}$. We can fix them by pairing up both sides with a basis of cycles, $\la {\cal C}_\ell | \,\bullet\, \ra$, which gives,
\be
\la {\cal C}_\ell | \varphi_j \ra = \sum_{k=1}^{2g} \alpha_{jk}\, \la {\cal C}_\ell | {\widetilde{\cal C}}_k \ra,
\ee
for $\ell=1,2,\ldots,2g$. It allows us to solve for $\alpha_{i\ell}$, which involves inverse of the intersection matrix $\mathbf{J}_{\ell k} := \la {\cal C}_\ell | {\widetilde{\cal C}}_k \ra$. We find
\be\label{cocycle-cycle-relation}
\varphi_j = \sum_{k,\ell=1}^{2g} \eta(\widetilde{\cal C}_k) \, \mathbf{J}^{-1}_{k\ell}\, \la {\cal C}_\ell | \widetilde\varphi_j \ra.
\ee
Finally, pairing up with a cocycle basis, $\la \widetilde\varphi_i| \,\bullet\, \ra$ yields
\be\label{quadratic-relation}
\la \widetilde\varphi_i | \varphi_j \ra = \sum_{k,\ell=1}^{2g} \la \widetilde\varphi_i | \widetilde{\cal C}_k \ra\, \mathbf{J}^{-1}_{k\ell}\, \la \mathcal{C}_\ell | \varphi_j \ra,
\ee
for every $i,j=1,2,\ldots,2g$. Hence we obtained relations that homologically split the complex integral into a quadratic combination of middle-dimensional integrals weighted by intersection numbers.

It is straightforward to evaluate the intersection matrix directly by specializing to the canonical homology basis,
\be
{\cal C}_i = \widetilde{\cal C}_i = \begin{dcases}
	{\cal A}_i \qquad & \text{if}\qquad i = 1,2,\ldots, g,\\
	{\cal B}_{i-g} \qquad & \text{if} \qquad i=g{+}1,g{+}2,\ldots, 2g.
\end{dcases}
\ee
Clearly, cycles of the same type never intersect and hence we have:
\be
\la \mathcal{A}_i | \mathcal{A}_j \ra = 0, \qquad \la \mathcal{B}_i | \mathcal{B}_j \ra = 0.
\ee
In order to evaluate the self-intersection numbers $\la \mathcal{A}_i | \mathcal{A}_i \ra = \la \mathcal{B}_i | \mathcal{B}_i \ra = 0$ it is sufficient to deform one copy of the cycle such that they do not overlap. For the remaining cycles we have:
\be
\la \mathcal{A}_i | \mathcal{B}_j \ra = \delta_{ij}, \qquad \la \mathcal{B}_i | \mathcal{A}_j \ra = -\delta_{ij},
\ee
since only the cycles associated to the same hole intersect. Note different signs coming from orientations of the cycles. Hence the intersection matrix reads:
\be
\J = \left(\begin{array}{rr}
	0 & \phantom{+}\I \\
	-\I & 0
\end{array}\right),
\ee
where $\I$ is the $g{\times} g$ identity matrix. Therefore the relation \eqref{quadratic-relation} in this basis becomes
\be
\int_{\Sigma_g} \widetilde\varphi_i \wedge \varphi_j = \sum_{k=1}^{g} \left( \oint_{{\cal A}_k} \!\widetilde \varphi_i \; \oint_{{\cal B}_k} \!\varphi_j - \oint_{{\cal B}_k} \!\widetilde\varphi_i\; \oint_{{\cal A}_k} \!\varphi_j \right).
\ee
Here we used antisymmetry $\la \widetilde\varphi_i | \widetilde{\cal C}_k \ra = - \la \widetilde{\cal C}_k | \widetilde\varphi_i \ra$. There are two special cases originally considered by Riemann: when both $\widetilde\varphi_i$ and $\varphi_j$ are holomorphic and hence the left-hand side vanishes, as well as when $\varphi_j = \star \overbar{\widetilde\varphi}_i$ and hence the left-hand side is always non-negative. These two identities are called Riemann bilinear relations \cite{riemann1857theorie}.

Another useful identity is obtained by contracting \eqref{cocycle-cycle-relation} with $\la {\cal C}| \,\bullet\, \ra$ for some fixed ${\cal C}$, giving:
\be
\la {\cal C} | \varphi_j \ra = \sum_{k,\ell=1}^{2g} \la {\cal C} | \widetilde{\cal C}_k\ra \, \mathbf{J}^{-1}_{k\ell}\, \la {\cal C}_\ell | \varphi_j \ra,
\ee
which expresses a single middle-dimensional integral in terms of a basis of $2g$ periods. This boils down to a counting problem, as one only needs to compute intersection numbers. Choosing the canonical homology basis yields
\be
\oint_{\cal C} \varphi_j = \sum_{k=1}^{g} \left( \la {\cal C}| {\cal B}_k \ra \oint_{{\cal A}_k} \!\varphi_j - \la {\cal C}| {\cal A}_k \ra \oint_{{\cal B}_k} \!\varphi_j \right).
\ee
Note that this is simply a relation between cycles, so the pairing with $\varphi_j$ is not strictly necessary. For completeness, let us mention that there exists an analogous formula for expressing an arbitrary cocycle in a cohomology basis. The above relations are special cases of Proposition~\ref{proposition}.

For more background on compact Riemann surfaces see, e.g., \cite{farkas2012riemann,eynard2018lectures}.

\subsection{\label{app:punctures}Adding Punctures}

\textsc{In order to compute} scattering amplitudes in string theory, one needs Riemann surfaces with punctures (removed marked points), which correspond to insertions of vertex operators. Here we consider the simplest case at genus zero, in which three points have been fixed, by convention to $0,1,\infty$, using the $\PSL(2,\C)$ redundancy and the fourth one lives on the surface $\Sigma_{0,3}$ with three punctures,
\be
\Sigma_{0,3} := \CP^1 \setminus \{ 0, 1, \infty \}.
\ee
It is in fact the simplest example of a moduli space, which are studied more comprehensively in the main body of this work.

We endow $\Sigma_{0,3}$ with an additional structure, which will allow us to describe its global properties using local data. It is called a \emph{local system} (or a locally-constant sheaf) ${\cal L}$. We refer the reader to Section~\ref{sec:intersection-numbers} for a more thorough discussion of local systems, while here we focus on the basic case of $\Sigma_{0,3}$. In the simplest case ${\cal L}$ is a one-dimensional Abelian representation of the fundamental group $\pi_1(\Sigma_{0,3})$, which associates a non-zero complex number to every path,
\be
{\cal L}:\; \pi_1(\Sigma_{0,3}) \;\to\; \C^\times.
\ee
If two paths $\gamma_1$ and $\gamma_2$ are homotopic then their \emph{local coefficients} are equal, ${\cal L}(\gamma_1) = {\cal L}(\gamma_2)$. Another property is that for every pair of paths $\gamma_1$ and $\gamma_2$, their local coefficients multiply to that of the composite path $\gamma_1 {\circ} \gamma_2$, i.e.,
\be\label{local-coefficients}
{\cal L}(\gamma_2)\, {\cal L}(\gamma_1) = {\cal L}(\gamma_1 {\circ} \gamma_2).
\ee
Hence paths contractible to a point have a unit local coefficient. 

One can describe a local system in terms of a one-form $\omega$, such that local coefficients are assigned via the integral
\be
{\cal L}_\omega : \; \gamma \;\mapsto\; \exp \int_\gamma \omega. 
\ee
Here exponentiation removes the zero complex number and gives the required multiplication property \eqref{local-coefficients}. To be more concrete, on $\Sigma_{0,3}$ we construct a form with logarithmic singularities at the positions of the punctures,
\be
\omega := s\, d\log z + t\, d\log(1{-}z),
\ee
where $s,t \notin \Z$ are constants. Notice that it also has a pole at infinity with residue $u := -s{-}t$. Therefore to every positively-oriented loop $\circlearrowleft_p$ around $p\in \{0,1,\infty\}$ the local system ${\cal L}_\omega$ associates the following numbers
\be\label{eq:appendix-local-system}
{\cal L}_\omega \left(\circlearrowleft_0, \circlearrowleft_1, \circlearrowleft_\infty\right) \;=\; \left(e^{2\pi i s}, e^{2\pi i t}, e^{2\pi i u}\right)
\ee
called \emph{monodromies} (or holonomies). For example, the fact that the three monodromies multiply to one, $e^{2\pi i (s+t+u)} =1$, is the same as saying that the sum of the three loops, $\circlearrowleft_0 {+} \circlearrowleft_1{+}\circlearrowleft_\infty =0$, can be deformed to a point.

In practical terms one can introduce the following picture, in which with dashed lines we denote branch cuts, crossing which corresponds to multiplying by a given monodromy,
\be
\begin{aligned}
	\includegraphics[scale=1]{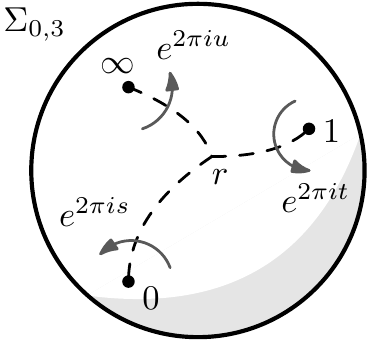}
\end{aligned}
\ee
Here the position of the point $r$ is arbitrary and it can in principle coincide with one of $0,1,\infty$. The above branch cuts can be thought of as those associated to the multi-valued function
\be\label{eq:appendix-u}
u(z) := z^s (1-z)^t,
\ee
and in fact an oriented path $(p,q)$ connecting two points $p$ and $q$ will have a local coefficient $\exp \int_p^q \omega = u(q)/u(p)$, which is a phase of $u(z)$.\footnote{Equivalently one can describe the local system ${\cal L}_\omega$ in the following way. Consider a (locally-finite) open cover $\Sigma_{0,3} = \bigcup_i U_i$. The space of horizontal sections of $\nabla_{-\omega} := d - \omega\wedge$, i.e., functions $u_i(z)$ satisfying $\nabla_{-\omega} u_i(z)=0$ on $U_i$, such that $u_i(z)$ and $u_j(z)$ agree on each intersection $U_i \cap U_j$ is called a rank-one locally-constant sheaf $\L_\omega$.} We can talk about paths (chains) with coefficients in the local system ${\cal L}_\omega$, denoted by
\be
(p,q) \otimes u(z),
\ee
which to every point $z$ along the path associate a coefficient $u(z)$. Notice that this coefficient changes as one crosses a branch cut. The advantage of assigning local coefficients to paths is that one can work directly in the space $\Sigma_{0,3}$, as opposed to its universal cover.

We can define a natural boundary operator, which here acts as
\be\label{eq:appendix-boundary}
\partial_{\omega} \Big( (p,q) \otimes u(z) \Big) = -p\otimes u(p) +  q \otimes u(q).
\ee
It satisfies $\partial_\omega^2 = 0$ and hence defines homology groups with coefficients in the local system ${\cal L}_\omega$, or simply \emph{twisted homologies}, $H_k(\Sigma_{0,3}, {\cal L}_\omega) := \ker \partial_{\omega} / \im \partial_{\omega}$. Its elements are called \emph{twisted cycles}. For example, $\circlearrowright_{p=0,1,\infty}$ are \emph{not} twisted cycles, since their boundary is non-zero (the phases of $u(z)$ at the beginning and the end of the loop are not the same, since a branch cut is crossed). Let us illustrate the topological part of three examples of twisted cycles below:
\be
\begin{aligned}
	\includegraphics[scale=1]{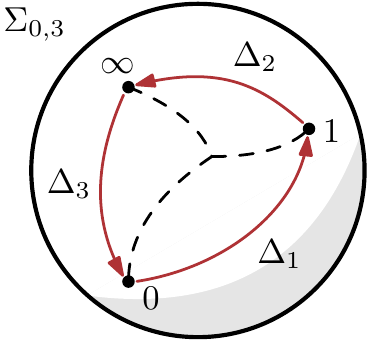}
\end{aligned}
\ee
These cycles are non-compact, since they have endpoints on the ramification points. We specify their local coefficients by selecting an arbitrary phase of $u(z)$, referred to as the \emph{loading}, as follows
\be
{\cal C}_1 := \Delta_1 \otimes u(z), \qquad {\cal C}_2 := \Delta_2 \otimes e^{-\pi i t} u(z), \qquad {\cal C}_3 := \Delta_3 \otimes e^{\pi i s}u(z).
\ee
Here the choice is made such that the loading is real-valued along each $\Delta_i$ if it lies on the real axis, i.e., we have $e^{-\pi i t} u(z) = z^s (z{-}1)^t \in \mathbb{R}$ on $(1,\infty)$ and $e^{\pi i s} u(z) = ({-}z)^s (1{-}z)^t \in \mathbb{R}$ on $(-\infty,0)$. Since the points $0,1,\infty$ are not in $\Sigma_{0,3}$, the cycles have zero boundary and hence belong to the $1$-st locally-finite twisted homology, $H^{\text{lf}}_1(\Sigma_{0,3},{\cal L}_\omega)$, isomorphic to $H_1(\Sigma_{0,3}, {\cal L}_\omega)$. Here the qualifier \emph{locally-finite} refers to the fact that its elements are non-compact.

It remains to find out what is the dimension of $H^{\text{lf}}_1(\Sigma_{0,3}, {\cal L}_\omega)$. Using the result that the twisted homology is concentrated in the middle dimension, i.e., $H^{\text{lf}}_k(\Sigma_{0,3},{\cal L}_\omega)$ vanishes for $k \neq 1$ \cite{aomoto1975vanishing}, we have:
\be
\chi(\Sigma_{0,3}) = \sum_{k=0}^{2} (-1)^{k} \dim  H^{\text{lf}}_k(\Sigma_{0,3}, {\cal L}_\omega) = -\dim H^{\text{lf}}_1(\Sigma_{0,3}, {\cal L}_\omega)
\ee
for the Euler characteristic $\chi(\Sigma_{0,3})$. It can be computed, for instance, by considering a surface of a tetrahedron with three points removed. As a CW complex it has one $0$-cell, six $1$-cells, and four $2$-cells and hence the Euler characteristic is $\chi(\Sigma_{0,3}) = 1-6+4 = -1$.

Another, more pedestrian, way of finding out that there is only one independent twisted cycle is by considering the following two cycles homologous to zero (deformable to a point):
\be
\begin{aligned}
	\includegraphics[scale=1]{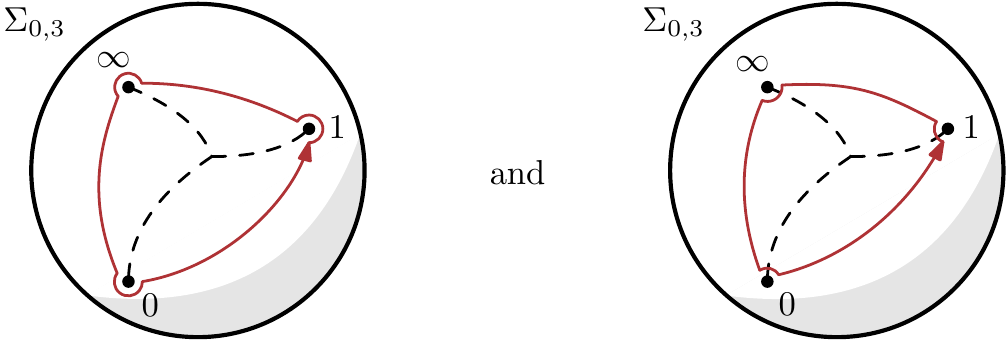}
\end{aligned}
\ee
They give two explicit linear relations between the twisted cycles of our interest,
\begin{align}
0 &= (\Delta_1 + \Delta_2 + \Delta_3) \otimes u(z) = {\cal C}_1 + e^{\pi i t}\, {\cal C}_2 + e^{-\pi i s}\, {\cal C}_3,\\
0 &= (\Delta_1 + e^{-2\pi i t} \Delta_2 + e^{2\pi i s} \Delta_3) \otimes u(z) = {\cal C}_1 + e^{-\pi i t}\, {\cal C}_2 + e^{\pi i s}\, {\cal C}_3,
\end{align}
which confirms that $\dim H^{\text{lf}}_1(\Sigma_{0,3}, {\cal L}_\omega) = 1$, as one can solve the above system for one twisted cycle. Here we used the fact that for a constant complex number $c$ we can freely move it as $c\Delta \otimes u_\Delta = \Delta \otimes c u_\Delta$.

It is natural to ask about the cohomological analogue of the above construction. We first introduce the connection $\nabla_{\omega} := d + \omega\wedge$ twisted by the one-form $\omega$. Since $\omega$ is closed, the connection is integrable, i.e., $\nabla_\omega^2 = 0$, and hence we can define \emph{twisted cohomology groups} $H^k(\Sigma_{0,3}, \nabla_{\omega}) := \ker \nabla_{\omega} / \im \nabla_{\omega}$. Its elements, called \emph{twisted cocycles}, are simply equivalence classes of $\nabla_\omega$-closed differential forms on $\Sigma_{0,3}$ modulo $\nabla_\omega$-exact terms. We will sometimes refer to the specific representative of this cohomology class as a \emph{twisted form}.

Twisted cohomology groups are isomorphic to their homology counterparts (and hence also concentrated in the middle dimension, $k{=}1$) by the following pairing. For $\Delta \otimes u_\Delta \in H_{1}^{\text{lf}}(\Sigma_{0,3},{\cal L}_\omega)$ and $\varphi \in H^1(\Sigma_{0,3}, \nabla_\omega)$ we have
\be\label{integral-pairing}
\Braket{ \Delta \otimes u_\Delta | \varphi } := \int_{\Delta \otimes u_\Delta} \varphi = \int_\Delta u_\Delta\, \varphi,
\ee
which is simply an integral of $\varphi$ along $\Delta$ weighted by the local coefficients. Upon identification of $s,t$ with Mandelstam invariants for massless kinematics, $s=(p_1 {+}p_2)^2$, $t=(p_2 {+} p_3)^2$, this pairing gives rise to four-point scattering amplitudes of open strings, which can be written formally as, e.g.,
\be\label{contour-integral}
\Braket{ (0,1) \otimes u(z) | d\log \frac{z}{z{-}1} } = \int_{0}^{1} z^{s-1} (1-z)^{t-1}\, dz = \frac{\Gamma(s)\Gamma(t)}{\Gamma(s+t)},
\ee
which is the famous Veneziano amplitude \cite{Veneziano:1968yb}. For the purposes of this appendix we specialize to massless external momenta, $p_i^2=0$, though it is not strictly necessary. The above integral is formal, because it does not converge in the physical kinematic region. We will address this issue shortly.

To make further progress we specialize to $s,t\in \mathbb{R}$ and introduce antiholomorphic versions of the above structures, namely $H^{\text{lf}}_k(\Sigma_{0,3}, \L_{\overbar\omega})$ and $H^k(\Sigma_{0,3}, \nabla_{\overbar{\omega}})$, where $\nabla_{\overbar{\omega}} := d + \overbar{\omega}\wedge$. The corresponding local system is defined by $\overbar{\omega}$ and hence the monodromies are inverses of those in ${\cal L}_\omega$. For example, going around a positively-oriented loop $\circlearrowleft_0$ has a monodromy $\overbar{e^{2\pi i s}} = e^{-2\pi i s}$ in ${\cal L}_{\overbar{\omega}}$. One can show that all of the above groups are isomorphic, see, e.g., \cite{aomoto2011theory},
\be\label{isomorphisms}
H^{\text{lf}}_k(\Sigma_{0,3}, \L_{\omega}) \cong H^k(\Sigma_{0,3}, \nabla_{{\omega}}) \cong
H^{\text{lf}}_k(\Sigma_{0,3}, \L_{\overbar\omega}) \cong H^k(\Sigma_{0,3}, \nabla_{\overbar{\omega}}).
\ee
We will often refer to the antiholomorphic versions as \emph{dual} twisted homologies and cohomologies. There is a natural counterpart of \eqref{integral-pairing}, written as $\braket{ \overbar{\widetilde{\varphi}} | \widetilde{\Delta}\otimes \overbar{u_{\widetilde{\Delta}}} }$ for $\widetilde{\Delta}\otimes \overbar{u_{\widetilde{\Delta}}} \in H_1^{\text{lf}}(\Sigma_{0,3}, \L_{\overbar{\omega}})$ and $\overbar{\widetilde{\varphi}} \in H^1(\Sigma_{0,3}, \nabla_{\overbar{\omega}})$.

The isomorphisms \eqref{isomorphisms} also give a way of defining bilinears of the holomorphic and antiholomorphic objects. The pairing between $\varphi$ and $\overbar{\widetilde{\varphi}}$, is the integral over the whole surface $\Sigma_{0,3}$ weighted by the monodromy-invariant function $|u(z)|^2$,
\be\label{complex-integral}
\braket{ \overbar{\widetilde\varphi} | \varphi } := \int_{\Sigma_{0,3}} |u(z)|^2\; \varphi \wedge \overbar{\widetilde\varphi}.
\ee
Formally it gives rise to four-point scattering amplitudes of closed strings, e.g.,
\be
\Braket{ d\log \frac{\overbar{z}}{\overbar{z}{-}1} | d\log \frac{z}{z{-}1} } = \int_{\mathbb{C}\setminus \{0,1\}} \!\!\!\!\!\! |z|^{2s-2} |1{-}z|^{2t-2}\, dz \wedge d\overbar{z} = -2\pi i  \frac{\Gamma(s)\Gamma(t)\Gamma(1{+}u)}{\Gamma(1{-}s)\Gamma(1{-}t)\Gamma({-}u)}
\ee
is the Virasoro--Shapiro amplitude \cite{Virasoro:1969me,Shapiro:1970gy}.

In a similar spirit, one can define intersection numbers of two types of twisted cycles, $\Delta \otimes u_\Delta$ and $\widetilde\Delta \otimes \overbar{u_{\widetilde\Delta}}$ via the following expression \cite{MANA:MANA19941660122,doi:10.1002/mana.19941680111}:
\be\label{twisted-intersection-pairing}
\Braket{ \Delta \otimes u_\Delta | \widetilde{\Delta} \otimes \overbar{u_{\widetilde\Delta}} } \;:= \sum_{p \in \Delta \cap \widetilde\Delta} \pm\, \frac{u_{\Delta}(p)\, \overbar{u_{\widetilde\Delta}(p)}}{|u(p)|^2}. 
\ee
Here the sum goes over all intersection points $p$, each contributing an orientation-dependent sign, as in \eqref{intersection-signs}, weighted by local coefficients. This definition agrees with \eqref{intersection-constant-coefficients} in the special case when all the intersection points lie on the same branch, but in general is a rational function of the monodromies.
For the sake of clarity, whenever the loading of a twisted cycle is exactly $u(z)$ or $\overbar{u(z)}$ in the holomorphic and antiholomorphic case respectively, we will omit it from the notation. For example, \eqref{contour-integral} will be written as $\braket{ (0,1) | d\log z/(z{-}1) }$, which should not be confused with \eqref{cycle-cocycle-pairing}.

There is an important caveat in that $\Sigma_{0,3}$ is non-compact. Therefore the above pairings are only well-defined if at least one of the objects involved is compact (in the case of cycles) or has compact support (in the case of cocycles). We discuss this topic next.

\subsubsection{\label{app:regularization}Regularization of Integrals}

\textsc{Let us illustrate} how to compactify a given (locally-finite) twisted cycle, $\Delta \otimes u_\Delta$, near one of its endpoints denoted by $p$. This can be achieved by attaching an infinitesimal, positively-oriented loop $\circlearrowleft_p$ around $p$, starting and ending at some point $p_\varepsilon$, as follows:
\be
\begin{aligned}
	\includegraphics[scale=1]{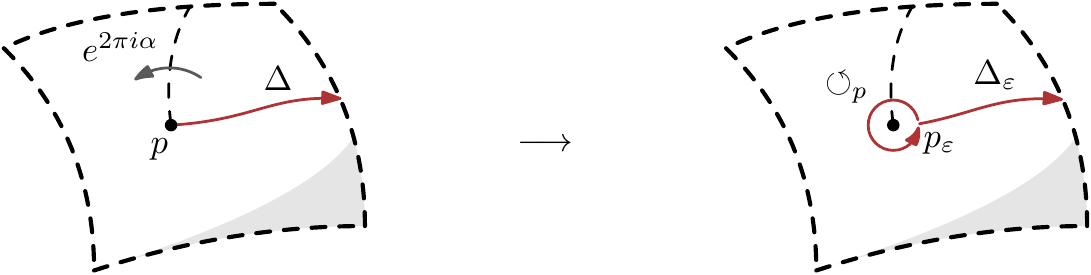}
\end{aligned}
\ee
Here we also illustrated the monodromy around $p$, equal to, say, $e^{2\pi i \alpha}$. The resulting twisted cycle is
\be\label{compactified-cycle}
\circlearrowleft_p \otimes\, u_{\circlearrowleft_p} + \Delta_\varepsilon \otimes u_\Delta,
\ee
where $\Delta_\varepsilon$ is a version of $\Delta$ with an endpoint moved from $p$ to $p_\varepsilon$. The loading of $\circlearrowleft_p \otimes\, u_{\circlearrowleft_p}$ can be determined simply by imposing that \eqref{compactified-cycle} is a twisted cycle, and in particular has no boundary, i.e.,
\be
0 = \partial_{\omega} \left( \circlearrowleft_p \otimes\, u_{\circlearrowleft_p} + \Delta_\varepsilon \otimes u_\Delta \right) = p_\varepsilon \otimes \left( - u_{\circlearrowleft_p} + e^{2\pi i \alpha} u_{\circlearrowleft_p} - u_\Delta \right).
\ee
The first two contributions are due to the beginning and end of the loop $\circlearrowleft_p$ respectively, while the third one comes from an endpoint of $\Delta_\varepsilon$ (there are no other boundaries since $\Delta \otimes u_\Delta$ was a twisted cycle to begin with). Solving for $u_{\circlearrowleft_p}$ we find
\be
u_{\circlearrowleft_p} = \frac{1}{e^{2\pi i \alpha}-1} u_\Delta.
\ee

Applying this procedure for all endpoints of a locally-finite twisted cycle gives a map called a \emph{regularization}:
\be
\text{reg}: \, H^{\text{lf}}_1(\Sigma_{0,3}, \L_\omega) \,\to\, H_1(\Sigma_{0,3}, \L_\omega).
\ee
For example, we have
\be\label{reg01}
\text{reg}\left( (0,1) \otimes u(z) \right) = \frac{\circlearrowleft_0}{e^{2\pi i s}-1} \otimes u(z) + (\varepsilon,1{-}\varepsilon) \otimes u(z) - \frac{\circlearrowleft_1}{e^{2\pi i t}-1} \otimes u(z) .
\ee
Here $\circlearrowleft_0$ and $\circlearrowleft_1$ are small loops around $0,1$ joining with the endpoints of the interval $(\varepsilon, 1{-}\varepsilon)$. Note that the coefficient of the loop $\circlearrowleft_1$ comes with a minus sign because the orientation of $(0,1)$ is opposite at $0$ and $1$. The only singularities that this contour can support are at $s,t\in \Z$ and come from the small neighbourhoods of the ramification points. The above twisted cycle is in fact equivalent to the Pochhammer contour \cite{Pochhammer1890}, see, e.g., Section~4.1 of \cite{Mizera:2017cqs}.\footnote{Let us clarify that in our conventions \eqref{reg01} is \emph{not} the same as
	\be\label{reg01-wrong}
	\left( \frac{\circlearrowleft_0}{e^{2\pi i s}-1} + (\varepsilon,1{-}\varepsilon) - \frac{\circlearrowleft_1}{e^{2\pi i t}-1} \right) \otimes u(z),
	\ee
	which is not a twisted cycle itself. For example, the left endpoint of $(\varepsilon,1{-}\varepsilon)$ is loaded with $e^{2\pi i s}u(z)$ in \eqref{reg01-wrong}, as opposed to just $u(z)$ in \eqref{reg01}.
}

Similarly, we can impose compact support on a given twisted cocycle $\varphi$ near each point $p$, i.e., impose that it vanishes in a small neighbourhood of this point. To be specific, we construct:
\be
\varphi - \nabla_{\omega} \Big( h_p\, \nabla^{-1}_\omega \varphi \Big),
\ee
where $h_p := \Theta(\varepsilon^2 - |z{-}p|^2)$ is the Heaviside step function having support on the small circle of radius $\varepsilon$ around $z=p$ and $\nabla^{-1}_\omega \varphi =: \psi_p$ denotes the unique holomorphic solution of the equation $\nabla_\omega \psi_p = \varphi$ near the point $p$. The resulting one-form is in the same cohomology class as $\varphi$, since they differ by a $\nabla_\omega$-exact term. It evaluates to:
\be
\left(1- h_p\right)\varphi - dh_p\, \nabla^{-1}_\omega \varphi.
\ee
Since $1{-}h_p = dh_p = 0$ in the small neighbourhood of the point $p$, the above form has compact support near $p$. In particular, the second term only has support on the loop $\circlearrowleft_p$.
Applying this procedure to all points removed from $\Sigma_{0,3}$ gives a map from the twisted cohomology to the twisted cohomology with compact support,
\be
\iota_\omega :\, H^1(\Sigma_{0,3}, \nabla_\omega) \to H^1_\text{c}(\Sigma_{0,3}, \nabla_\omega).
\ee
For instance, when $\varphi = d\log z/(z{-}1)$, we have:
\begin{align}\label{dlog-compactification}
\iota_\omega \!\left( d\log \frac{z}{z{-}1} \right) = &\,\Big( 1 - h_0 - h_1 - h_\infty\Big)\, d\log \frac{z}{z{-}1} - dh_0 \left(\frac{1}{s} + \frac{s{+}t}{s(s{+}1)}z + {\cal O}(z^2) \right)\nn\\
&- dh_1 \left(-\frac{1}{t}
+ \frac{s{+}t}{t(t{+}1)}(z{-}1) + {\cal O}((z{-}1)^2) \right) - dh_\infty \left( \frac{1}{u{+}1} \frac{1}{z} + {\cal O}(1/z^2) \right).
\end{align}
Here we expanded $\nabla_\omega^{-1} \varphi$ explicitly near each $p=0,1,\infty$. Note that all the poles in $s,t \in \Z_{\leq 0}$ and $u \in \Z_{<0}$ (whose exact positions depend on the choice of $\varphi$) have support on the loops $\circlearrowleft_p$ around their respective points.

Using the above compactification procedures one obtains well-defined integrals \eqref{contour-integral} and \eqref{complex-integral}, as well as the intersection number \eqref{twisted-intersection-pairing}. The former essentially corresponds to the ``finite part'' of an integral in the sense of Hadamard \cite{zbMATH02599971}. We can make use of these tools to study analytic properties of string integrals. For example, using \eqref{reg01} let us check the residues in the $s$-channel of the open-string amplitude $\braket{ (0,1) | \varphi }$, which only receive contributions from $\circlearrowleft_0$:
\begin{align}
\Res_{s=-n} \int_{0}^{1} z^s (1{-}z)^t\, \varphi &= \Res_{s=-n} \left( \frac{1}{e^{2\pi i s}-1} \oint_{\circlearrowleft_0} z^s (1{-}z)^t\, \varphi\right)\nn\\
&= \Res_{z=0} \left( z^{-n} (1{-}z)^t\, \varphi \right)\label{residue-01}
\end{align}
for $n \in \mathbb{Z}$ and hence the lowest resonance happens at $s = {-}1{-}\ord_{z=0}(\varphi)$, where $\ord_{z=p}(\varphi)$ denotes the order of the zero of $\varphi$ at $z{=}p$. Applying this formula to the case with $\varphi = d\log z/(z{-}1)$ reproduces the classic result \cite{green1988superstring}
\be
\Res_{s=-n} \int_{0}^{1} z^{s-1} (1{-}z)^{t-1} dz = \begin{dcases}
	0 & \text{if}\quad n < 0,\\
	\frac{1}{n!}\prod_{k=1}^{n} (k-t) & \text{if}\quad n \geq 0.
\end{dcases}
\ee
Analogous computation for closed strings will be described briefly, after we learn how to homologically split such integrals.

On top of manifesting their singularity structure in $s,t,u$, another useful application is in studying the low-energy limit of string theory integrals, given by the leading order as $\alpha' \to 0$ after rescaling $(s,t,u) \to \alpha'(s, t, u)$.\footnote{Strictly speaking here we extract the leading $1/\alpha'$ behaviour of string amplitudes. In cases when this order vanishes, low-energy limit is governed by subleading terms which can be extracted using an analogous method.} For example, the integral $\braket{ (0,1) | \varphi }$ behaves as
\begin{align}\label{open-string-low-energy}
\lim_{\alpha' \to 0} \int_{0}^1 z^{\alpha'\! s} (1{-}z)^{\alpha'\! t}\, \varphi &= \frac{1}{2\pi i \alpha' s} \oint_{\circlearrowleft_0} \lim_{\alpha' \to 0} u(z) \,\varphi - \frac{1}{2\pi i \alpha' t} \oint_{\circlearrowleft_1} \lim_{\alpha' \to 0} u(z) \,\varphi\nn\\
&= \frac{1}{\alpha' s} \Res_{z=0} \!\left( \lim_{\alpha' \to 0} \varphi\right) - \frac{1}{\alpha' t} \Res_{z=1} \!\left( \lim_{\alpha' \to 0}\varphi \right).
\end{align}
In the first line we regularized the contour as in \eqref{reg01}, which also allowed us to commute the limit with the integration. Since $u(z) \to 1$, the result localizes as a sum of residues around $0$ and $1$ of the form $\varphi$.

As another example, we can consider the complex integral $\braket{ \overbar{\varphi} | d\log z/(z{-}1) }$, which at the leading order equals
\begin{align}
\lim_{\alpha' \to 0} \int_{\C \setminus \{0,1\}} \!\!\!\!\!\!\!\!\!\!\!\! |z|^{2\alpha' s} |1{-}z|^{2\alpha' t} \, d\log \frac{z}{z{-}1} \wedge \overbar{\varphi} &= \int_{\C \setminus \{ 0,1\}}\lim_{\alpha' \to 0} |u(z)|^2 \; \left( - dh_0 \frac{1}{\alpha' s} + dh_1 \frac{1}{\alpha' t} \right) \overbar{\varphi} \nn\\
&= -2\pi i \left(\frac{1}{\alpha' s} \Res_{z=0} \!\left( \lim_{\alpha' \to 0} \varphi\right) - \frac{1}{\alpha' t} \Res_{z=1} \!\left( \lim_{\alpha' \to 0} \varphi \right) \right)\!.
\end{align}
In the first equality we imposed compact support on the twisted form $d\log z/(z{-}1)$ according to \eqref{dlog-compactification}, which allowed us to commute the limit with the integration. Since $|u(z)|^2 \to 1$, the right-hand side then becomes a total differential, which localizes on the two loops $\circlearrowleft_0$ and $\circlearrowleft_1$,\footnote{The sign on the right-hand side can be confirmed by writing a smooth version of $h_p$ such that it is equal to one inside the loop $\circlearrowleft_p$, to zero outside of a larger, but still infinitesimal, loop $\widetilde{\circlearrowleft}_p$ and interpolates between the two values on the annulus $D_p$ bounded by $\circlearrowleft_p$ and $\widetilde{\circlearrowleft}_p$. Then a two-form $ dh_p \wedge \overbar{\varphi}$ has support only on $D_p$ and hence
	\be
	\int_{D_p} \!\!\! dh_p \wedge \overbar{\varphi} = \int_{D_p} d\left( h_p\, \overbar{\varphi} \right) = \int_{\partial D_p} \!\!\! h_p\, \overbar{\varphi} = -\int_{\circlearrowleft_p} \overbar{\varphi},
	\ee
	where we used Stokes' theorem and the fact that $\partial D_p = \widetilde{\circlearrowleft}_p - \circlearrowleft_p$.
}
\be
\int_{\C \setminus \{ 0,1\}} d\left( \lim_{\alpha' \to 0} \left( -\frac{h_0}{\alpha' s} + \frac{h_1}{\alpha' t}\right)\, \overbar{\varphi} \right) = \frac{1}{\alpha' s}\oint_{\circlearrowleft_0} \lim_{\alpha' \to 0} \overbar{\varphi} - \frac{1}{\alpha' t}\oint_{\circlearrowleft_1} \lim_{\alpha' \to 0} \overbar{\varphi},
\ee
which can be written as a residue. To obtain the final result we also used the fact that
\be
\Res_{z=p} \!\left( \lim_{\alpha' \to 0} \overbar{\varphi}\right) = -\Res_{z=p} \!\left( \lim_{\alpha' \to 0} \varphi\right).
\ee
It agrees with \eqref{open-string-low-energy} up to the factor of $-2\pi i$.

\subsubsection{Analytic Continuation of Open-String Amplitudes}

\textsc{Compactifications discussed up to now} provide a way of making string theory integrals well-defined, though they might not be the most convenient ones to use in practice as they involve an infinitesimal procedure. Let us consider a more conventional way of performing analytic continuation of integrals, which uses Morse (or Picard--Lefschetz) theory, see, e.g., \cite{milnor2016morse,Witten:2010cx}.

We start by defining the so-called \emph{Morse function} on $\Sigma_{0,3}$, given by
\be\label{Morse-function-4pt}
\Re(\log u(z)) = \Re(s \log z + t \log (1{-}z)).
\ee
It is single-valued if $s,t\in \mathbb{R}$. This Morse function has a unique non-degenerate critical point, given by $\omega=0$, at the position
\be
z_\ast := \frac{s}{s{+}t},
\ee
which has index one, i.e., the Morse function has a shape of a saddle near $z_\ast$. Therefore we can define two flows from $z_\ast$ parametrized by a ``time'' variable $\tau$: a downward $(-)$ and an upward $(+)$ one, which satisfy 
\be
\frac{dz}{d\tau} = \mp \frac{1}{2}\frac{\partial \overbar{\log u(z)}}{\partial \overbar{z}}, \qquad \frac{d\overbar{z}}{d\tau} = \mp \frac{1}{2}\frac{\partial \log u(z)}{\partial z}.
\ee
It is straightforward to show that $d \Im (\log u(z))/d\tau = 0$, which is the stationary phase condition along the flow. In addition we have
\be
\frac{d \Re(\log u(z))}{d\tau} =  \frac{\partial \Re (\log u(z))}{\partial z} \frac{dz}{d\tau} + \frac{\partial \Re (\log u(z))}{\partial \overbar{z}} \frac{d\overbar{z}}{d\tau} = \mp 2\left|\frac{d \Re(\log u(z))}{dz}\right|^2\!\!\!,
\ee
which confirms that the Morse function is strictly decreasing (increasing) along the flow in the $-$ ($+$) sign case. Solutions to these equations can be found by evolving them directly such that $z=z_\ast$ at $\tau = \pm \infty$.

In the specific case \eqref{Morse-function-4pt}, the flows can be found uniquely (up to homotopy) without solving the gradient flow equations. Without loss of generality let us focus on the case of physical kinematics for massless external states, $s < 0$ and $t,u>0$, for which $1 < z_\ast < \infty$. In addition, in the neighbourhood of $z = 0$ the Morse function approaches $+\infty$, while in the neighbourhoods of $z = 1,\infty$, it goes to $-\infty$. Therefore, aside from passing through $z_\ast$, the path of steepest descent ${\cal J}$ has to have endpoints on $z=1,\infty$, while the path of steepest ascent ${\cal K}$ has both endpoints on $z=0$:
\be
\begin{aligned}
	\includegraphics[scale=1]{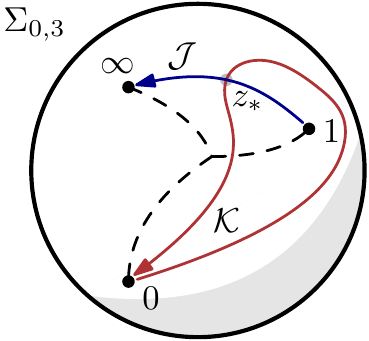}
\end{aligned}
\ee
cf. Figure~\ref{fig:Morse-function}. The orientations and loading of contours can be chosen arbitrarily. For example, we construct
\be\label{steepest-cycles}
{\cal J}\otimes u(z) \in H_1^{\text{lf}}(\Sigma_{0,3},\L_\omega), \qquad {\cal K}\otimes \overbar{u(z)} \in H_1^{\text{lf}}(\Sigma_{0,3},\L_{\overbar{\omega}}).
\ee
Since they intersect only once at the critical point $z_\ast$, their loading cancels and we have the intersection number
\be\label{J-K-intersection}
\braket{ {\cal J} | {\cal K} } = 1.
\ee
Analytic continuation of an integral over a given twisted cycle $\Delta \otimes u(z)$ boils down to expressing it in terms of ${\cal J} \otimes u(z)$. Given \eqref{J-K-intersection} we have simply
\be\label{Delta-J}
\Delta \otimes u(z) = \braket{ \Delta | {\cal K} } \, {\cal J} \otimes u(z).
\ee
For instance, let us consider $\Delta = (0,1)$ and compute the required intersection number. Regularizing $(0,1)\otimes u(z)$ near $z=0$ we have:
\be
\begin{aligned}
	\includegraphics[scale=1]{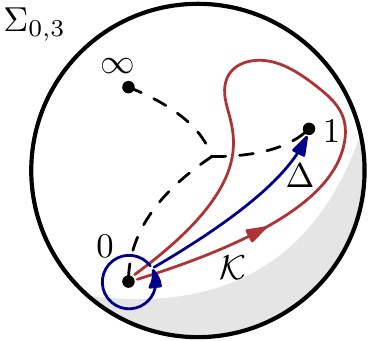}
\end{aligned}
\ee
There are two intersection points in the neighbourhood of $z=0$, which give:
\begin{align}
\braket{ (0,1) | {\cal K} } &= \Braket{ \frac{\circlearrowleft_0}{e^{2\pi i s}-1} | {\cal K} } \nn\\
&= -\frac{e^{2\pi i s}}{e^{2\pi i s} -1} + \frac{e^{-2\pi i t}}{e^{2\pi i s}-1},
\end{align}
where the denominators come from the normalization of the loop $\circlearrowleft_0$ and the numerators are obtained by keeping track of the branch cuts that each contour crosses before arriving at the intersection points. Simplifying this result and plugging back into \eqref{Delta-J} we find that an open-string integral can be analytically continued as
\be\label{analytic-continuation}
\int_{0}^{1} z^s (1{-}z)^t \, \varphi = e^{-\pi i t}\, \frac{\sin (\pi u)}{\sin (\pi s)} \int_{\cal J} z^s (1{-}z)^t \, \varphi
\ee
in the kinematic region of our interest.\footnote{Strictly speaking, the right-hand side of \eqref{analytic-continuation} converges only when $\Re(t) > -1{-}\ord_{z=1} (\varphi)$ and $\Re(u) > -1{-}\ord_{z=\infty} (\varphi)$. In order to further extend this region to $\Re(t),\Re(u) > 0$ one can express the right-hand side of \eqref{analytic-continuation} in terms of an integral of $d \log (z{-}1)$ using intersection numbers of twisted forms, as discussed in Section~\ref{sec:open-string-amplitudes}. Explicitly we have:
	\be
	\int_{0}^{1} z^s (1{-}z)^t \, \varphi = e^{-\pi i t}\, \frac{\sin (\pi u)}{\sin (\pi s)} \, \braket{ u\, d\log z | \varphi }_\omega \int_{\cal J} z^s (1{-}z)^t \, d\log (z{-}1).
	\ee
	Since we want to focus our discussion on the homological aspects, we will ignore this issue in the remainder of the appendix.}
Notice that the denominator $\sin (\pi s)$ reintroduces the $s$-channel poles that are absent from the integral over ${\cal J}$, while the numerator $\sin (\pi u)$ cancels the $u$-channel poles that are not present in the original integral over $(0,1)$.

The result of analytic continuation \eqref{analytic-continuation}, and in particular convergence of the right-hand side, allows us to study the high-energy limit $\alpha' \to \infty$ using a saddle-point expansion. We find
\be\label{analytic-continuation-result}
\lim_{\alpha' \to \infty} \int_{0}^{1} z^{\alpha'\!s} (1{-}z)^{\alpha' \!t} \, \varphi = \sqrt{\frac{-2\pi s t}{\alpha' u^3}} \frac{\sin (\pi \alpha' u)}{\sin (\pi \alpha' s)} e^{\alpha'\! \left(s\log (-s) + t\log t + u \log u \right)} \!\lim_{\alpha' \to \infty}\!\! \widehat{\varphi}(z_\ast),
\ee
where $\widehat{\varphi} dz := \varphi$ is evaluated at the critical point $z_\ast$. String theory amplitudes have an interesting high-energy limit \cite{Gross:1987kza,Gross:1987ar} with Stokes lines (places where ${\cal J}$ and ${\cal K}$ change discontinuously) when either of $\Re(s), \Re(t), \Re(u)$ changes sign, similar to those of the gamma function \cite{10.2307/51844}.

Similar computations are certainly possible on the covering space of $\Sigma_{0,3}$, though they involve an infinite number of saddles that resum to give the sine factors in front of \eqref{analytic-continuation-result}. The advantage of working with cycles with coefficients in a local system is that sines are automatically taken into account and the asymptotics is governed by a single saddle point.

Before discussing analytic continuation of complex integrals of the type \eqref{complex-integral} we need to introduce a way of homologically splitting such integrals.

\subsubsection{Homological Splitting}

\textsc{Originally introduced in 1985}, Kawai--Lewellen--Tye (KLT) relations \cite{Kawai:1985xq} give identities between open and closed string theory amplitudes. Similar relations were found independently by Fateev and Dotsenko \cite{Dotsenko:1984nm,Dotsenko:1984ad} in the context of conformal field theory, as well as Aomoto \cite{10.1093/qmath/38.4.385} for Selberg integrals. Here we review their modern generalization in terms of twisted homology theory \cite{cho1995,Mizera:2016jhj,Mizera:2017cqs}. Their derivation is entirely analogous to the one given in Appendix~\ref{app:compact-Riemann-surfaces}.

We start by using the twisted Poincar\'e duality and isomorphism between the two homology groups with local systems ${\cal L}_\omega$ and ${\cal L}_{\overbar{\omega}}$, 
\be
H^1(\Sigma_{0,3}, \nabla_{\omega}) \cong H_1^{\text{lf}}(\Sigma_{0,3}, {\cal L}_{\omega}) \cong H_1^{\text{lf}}(\Sigma_{0,3}, {\cal L}_{\overbar{\omega}}),
\ee
which allows us to construct $\eta(\widetilde{\cal C}) \in H^1(\Sigma_{0,3}, \nabla_\omega)$ for every $\widetilde{\cal C} \in H_1^{\text{lf}}(\Sigma_{0,3}, {\cal L}_{\overbar{\omega}})$ such that
\be
\braket{ {\overbar{\widetilde\varphi}} | \widetilde{\cal C} } = \braket{ {\overbar{\widetilde\varphi}} | \eta(\widetilde{\cal C}) }, \qquad
\braket{ {\cal C} | \widetilde{\cal C} } = \braket{ {\cal C} | \eta(\widetilde{\cal C}) }.
\ee
Similar expressions are valid for other pairings, but they will not be needed here. We write a given twisted cocycle $\varphi$ in terms of a basis of a dual twisted cycle $\widetilde{\cal C}$ as $\varphi = \alpha\, \eta(\widetilde{\cal C})$. Contracting both sides with $\braket{ {\cal C} |\, \bullet\, }$ and solving for $\alpha$ we have $\varphi = \eta(\widetilde{\cal C}) \braket{ {\cal C} | \varphi } / \braket{ {\cal C} | \widetilde{\cal C} }$. Finally evaluating $\braket{ \overbar{\widetilde{\varphi}} |\, \bullet\,}$ we find the identity:
\be\label{twisted-period-relation}
\braket{ \overbar{\widetilde{\varphi}} |\varphi } = \frac{\braket{ \overbar{\widetilde{\varphi}} | \widetilde{\cal C} } \braket{ {\cal C} | \varphi }}{\braket{ {\cal C} | \widetilde{\cal C}}}.
\ee
This is one of the simplest examples of the so-called \emph{twisted period relations} \cite{cho1995}, which generalize the Riemann bilinear relations \eqref{quadratic-relation} to the twisted case.

For example, let us apply the formula \eqref{twisted-period-relation} with $\mathcal{C} = (0,1) \otimes u(z)$ and $\widetilde{\cal C} = (0,1) \otimes \overbar{u(z)}$ by evaluating the required intersection number $\braket{ {\cal C} | \widetilde{\cal C}}$. Regularizing the first twisted cycle and deforming the second one into a sine-like curve we find three intersection points:
\be
\begin{aligned}
	\includegraphics[scale=1]{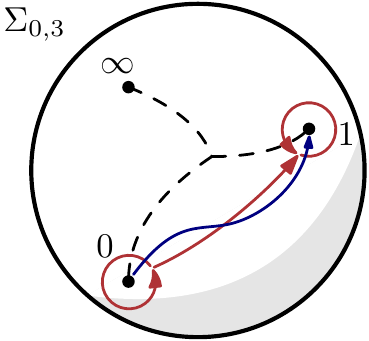}
\end{aligned}
\ee
Since all of them lie on the same branch of $u(z)$, after taking into account orientations we have simply:
\be
\braket{ (0,1) | (0,1) } = -\frac{1}{e^{2\pi i s}-1} - 1 - \frac{1}{e^{2\pi i t}-1}.
\ee
Simplifying and plugging back into the relation we find the identity
\be\label{KLT-relation-1}
\int_{\C \setminus \{0,1\} } |z|^{2s} |1{-}z|^{2t}\, \varphi \wedge \overbar{\widetilde{\varphi}} \,=\, 2i \frac{\sin (\pi s) \sin (\pi t)}{\sin (\pi u)} \left( \int_{0}^1 z^s (1{-}z)^t\, \varphi \right) \left( \int_{0}^1 \overbar{z}^s (1{-}\overbar{z})^t\, \overbar{\widetilde\varphi} \right).
\ee
This is a relation between a four-point scattering amplitudes of closed and open strings. We can use this representation to study factorization of the closed string amplitude. For example, in the $s$-channel both integrals receive only contributions from the loop $\circlearrowleft_0$ and we find:
\begin{align}
\Res_{s=-n}\!\! \int_{\C \setminus \{0,1\} } \!\!\!\!\!\!\!\!\!\!\!\! |z|^{2s} |1{-}z|^{2t}\, \varphi \wedge \overbar{\widetilde{\varphi}} &= \Res_{s=-n}\! \left(\! 2i \frac{\sin(\pi t)}{\sin(\pi u)} \frac{\sin(\pi s)}{(e^{2\pi i s}{-}1)(e^{-2\pi i s}{-}1)} \oint_{\circlearrowleft_0} \!\!\! z^s (1{-}z)^t \varphi \oint_{\circlearrowleft_0} \!\!\! \overbar{z}^s (1{-}\overbar{z})^t \overbar{\widetilde{\varphi}} \right)\nn\\
&= -2\pi i \Res_{z=0} \left( z^{-n} (1{-}z)^{t} \varphi \right)  \Res_{z=0} \left( z^{-n} (1{-}z)^{t} \widetilde\varphi \right)
\end{align}
for $n \in \mathbb{Z}$. In fact, this result is a consequence of a KLT relation between three-points amplitudes and \eqref{residue-01}.

Let us illustrate an alternative way of setting up the problem. We use ${\cal C} = (0,1)\otimes u(z)$ and $\widetilde{\cal C} = (1,\infty)\otimes e^{\pi i t} \overbar{u(z)}$, as well as construct a pair of (quasi-)orthonormal twisted cycles, such that
\be\label{orthonormality}
\braket{ {\cal C} | \widetilde{\cal C}^\vee } = 1, \qquad \braket{ {\cal C}^\vee | \widetilde{\cal C} } = -1.
\ee
In this way we have $\braket{ {\cal C} | \widetilde{\cal C} }^{-1} = -\braket{ {\cal C}^\vee | {\widetilde{\cal C}}^{\vee}}$ and the relation \eqref{twisted-period-relation} becomes a polynomial. Finding orthonormal cycles is straightforward:
\be
\begin{aligned}
	\includegraphics[scale=1]{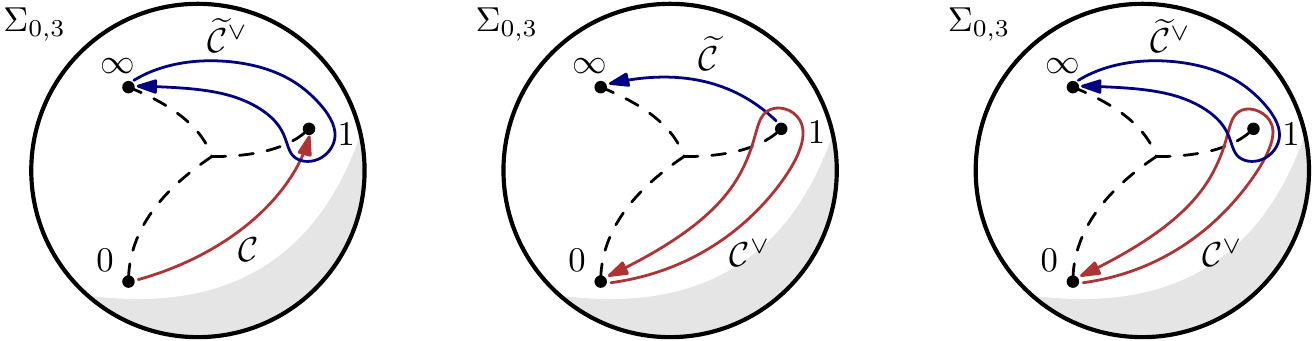}
\end{aligned}
\ee
Here ${\cal C}^\vee$ and $\widetilde{\cal C}^\vee$ are loaded with $e^{-\pi i t} u(z)$ and $\overbar{u(z)}$ respectively, so that \eqref{orthonormality} is satisfied. As indicated above, these two twisted cycles intersect at exactly two points, and hence we have
\be
\braket{ {\cal C}^\vee | {\widetilde{\cal C}}^{\vee} } = e^{-\pi it} - e^{\pi i t}.
\ee
The first contribution is simply the relative phase $e^{-\pi i t}$ between the two loadings, while the second is additionally multiplied by $e^{2\pi i t}$ since $\widetilde{\cal C}^\vee$ crosses a branch cut in the clockwise direction before reaching the second intersection point. This leads to the KLT relation
\be
\int_{\C \setminus \{0,1\} } |z|^{2s} |1{-}z|^{2t}\, \varphi \wedge \overbar{\widetilde{\varphi}} \,=\, 2i \sin (\pi t) \left( \int_{0}^1 z^s (1{-}z)^t\, \varphi \right) \left( \int_{1}^{\infty} \overbar{z}^s (\overbar{z}{-}1)^t\, \overbar{\widetilde\varphi} \right).
\ee
The difference to \eqref{KLT-relation-1} is in the poles that need to be cancelled on the right-hand side: here both open string integrals have poles in the $t$-channel, which are compensated by the zeros of $\sin(\pi t)$, while in \eqref{KLT-relation-1} both integrals have $s$- and $t$-channel poles, which requires the numerator of $\sin (\pi s) \sin (\pi t)$, as well as denominator $\sin(\pi u)$ to reintroduce the $u$-channel resonances.

Similar computations on the covering space of $\Sigma_{0,3}$ were considered before in \cite{Gaiotto:2013rk,Witten:2013pra}.

\subsubsection{Analytic Continuation of Closed-String Amplitudes}

\textsc{Having learned how to} homologically split complex integrals, analytic continuation of the four-point closed string amplitude becomes straightforward. Utilizing the steepest descent and ascent twisted cycles from \eqref{steepest-cycles} we can write
\be
\braket{ \overbar{\widetilde\varphi} | \varphi } = -\braket{ \overbar{\widetilde\varphi} | {\cal J} } \braket{ {\cal K} | {\cal K} } \braket{ {\cal J} | \varphi },
\ee
where the minus comes about because of the orthonormality conditions $\braket{ {\cal J}| {\cal K} } = - \braket{ {\cal K} | {\cal J} } = 1$. It remains to evaluate the intersection number $\braket{ {\cal K} | {\cal K} }$. To do so, we regularize ${\cal K} \otimes u(z)$ as
\be
\text{reg} \left( {\cal K} \otimes u(z) \right) = -\frac{e^{2\pi i t}-1}{e^{2\pi i s}-1} \circlearrowleft_0 \otimes\, u(z) + {\cal K}_\varepsilon \otimes u(z),
\ee
where $\circlearrowleft_0$ is an infinitesimal loop around $z=0$ starting and ending at the point $\varepsilon$, while ${\cal K}_\varepsilon$ is a small deformation of ${\cal K}$ which starts at $\varepsilon$, loops around $z=1$ and goes back to the point $\varepsilon$. The coefficient in front of $\circlearrowleft_0$ can be obtained in the same way as that in \eqref{compactified-cycle} by requiring that the result is a twisted cycle. Hence we have
\be
\begin{aligned}
	\includegraphics[scale=1]{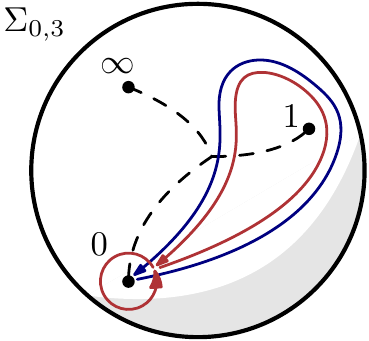}
\end{aligned}
\ee
The two cycles intersect at two points near $z=0$, giving
\be
\braket{ {\cal K} | {\cal K} } = \Braket{ -\frac{e^{2\pi i t}-1}{e^{2\pi i s}-1} \circlearrowleft_0 | {\cal K} } = -\frac{e^{2\pi i t}-1}{e^{2\pi i s}-1} \left( - e^{2\pi i s} + e^{-2\pi i t}\right).
\ee
Note that the phase factors come with different signs in the exponents because of the way ${\cal K}\otimes u(z)$ and ${\cal K}\otimes \overbar{u(z)}$ cross branch cuts before arriving at the intersection points. This computation gives us the analytic continuation of a closed-string amplitude:
\be
\int_{\C \setminus \{0,1\} } \!\!\!\! |z|^{2s} |1{-}z|^{2t}\, \varphi \wedge  \overbar{\widetilde{\varphi}} \,=\, 2i \frac{\sin (\pi t) \sin (\pi u)}{\sin (\pi s)} \left( \int_{\cal J} z^s (1{-}z)^t\, \varphi \right) \left( \int_{\cal J} \overbar{z}^s (1{-}\overbar{z})^t\, \overbar{\widetilde\varphi} \right).
\ee

Since the integrals on the right-hand side converge in the kinematic region of our interest, we can take the high-energy limit, which gives:
\begin{equation}
\lim_{\alpha' \to \infty} \int_{\C \setminus \{0,1\} } \!\!\!\!\!\!\!\!\!\! |z|^{2\alpha'\! s} |1{-}z|^{2\alpha'\! t}\, \varphi \wedge \overbar{\widetilde{\varphi}} = -\frac{4\pi i st}{\alpha' u^3}  \frac{\sin (\pi\alpha' t) \sin (\pi \alpha' u)}{\sin (\pi \alpha' s)} e^{2\alpha'\! \left(s\log (-s) + t\log t + u \log u \right)} \!\!\lim_{\alpha' \to \infty}\!\!  \widehat{\varphi}(z_\ast) \widehat{\overbar{\widetilde{\varphi}}}(z_\ast),
\end{equation}
where the forms on the right-hand side are evaluated at the critical point $z_\ast = s/(s{+}t)$.

\subsection{\label{app:generalization}Generalization to Moduli Spaces of Punctured Spheres}

\textsc{The space} $\Sigma_{0,3}$ we have been working with is an example of a moduli space of genus-zero Riemann surfaces with punctures,
\be
{\cal M}_{0,n} := \Conf_n(\CP^1) / \PSL(2,\C).
\ee
Using the $\PSL(2,\C)$ redundancy to fix positions of three punctures $(z_1, z_{n-1}, z_n)$ we can write it as
\be
{\cal M}_{0,n} = \{ (z_2, z_3, \ldots, z_{n-2}) \in (\mathbb{CP}^1)^{n-3} \;|\; z_i \neq z_j \text{ for } i \neq j \},
\ee
where we will use the same notation $z_i$ to denote inhomogeneous coordinates on each Riemann sphere $\CP^1$. In the simplest cases $\M_{0,3}$ is a point, while in the next-simplest one we have $\M_{0,4} \cong \Sigma_{0,3}$. Most of the results discussed in the previous subsections generalize in a natural way; we briefly outline some of them here. We start by endowing the moduli space with a rank-one local system ${\cal L}_\omega$,
\be
{\cal L}_\omega:\; \pi_1({\cal M}_{0,n}) \;\to\; \C^\times,
\ee
which associates a non-zero complex number $\exp \int_\gamma \omega$ to every path $\gamma \in \pi_1(\M_{0,n})$. The one-form $\omega$ is given by:
\be
\omega := \alpha'\!\!\! \sum_{1 \leq i<j \leq n} \!\!\! 2 p_i {\cdot} p_j \, d\log (z_i - z_j),
\ee
where we attached a momentum $p_i^\mu$ to each puncture $z_i$. We assume the momentum conservation condition $\sum_{i=1}^{n} p_i^\mu = 0$.

The fundamental group $\pi_1(\M_{0,n})$ is generated by the loops $\circlearrowleft_{ij}$ in which the puncture $z_j$ goes around $z_i$ in the positive direction and comes back to its original position. To each of them the above local system associates the number $\exp(4\pi i \alpha' p_i {\cdot} p_j)$. Since in generic $\PSL(2,\C)$-fixing infinity is not a special place, the loop $\circlearrowleft_{\infty i}$ where the puncture $z_i$ goes around infinity can be contracted to a point and hence should be represented by the coefficient $1$. This gives a constraint on the possible masses $m_i^2 := - p_i^2$ of external states:
\be
1 = \exp \int_{\circlearrowleft_{\infty i}} \!\!\omega = \exp (-4\pi i \alpha' m_i^2 ),
\ee
which implies that only $m_i^2 \in \Z/(2\alpha')$ are allowed. 

Twisted homology and cohomology groups, $H_{k}(\M_{0,n}, {\cal L}_\omega)$ and $H^{k}(\M_{0,n}, \nabla_\omega)$, as well as their locally-finite, compactly-supported, and complex-conjugated versions can be defined in the straightforward way. It can be shown that they are all isomorphic and concentrated in the middle dimension $k=n{-}3$ \cite{aomoto1975vanishing}.

Let us start the discussion with the case of homology. Organizing each puncture along the circle $\mathbb{RP}^1 \subset \CP^1$ on the Riemann surface in a given cyclic ordering $\alpha$ we obtain the topological cycles:
\be\label{Delta-disk}
\Delta(\alpha) := \{ (z_2, z_3, \ldots, z_{n-2}) \in \mathbb{R}^{n-3} \,|\, z_{\alpha(1)} < z_{\alpha(2)} < \cdots < z_{\alpha(n-1)} \},
\ee
where without loss of generality we fixed $z_n {=} \infty$ and $\alpha(n) {=} n$ for simplicity of notation. Twisted cycles are obtained by loading each $\Delta(\alpha)$ with a local coefficient $\exp \int_\gamma \omega$, where $\gamma$ starts at an arbitrary point $p$ and ends at the coordinate $(z_2, z_3, \ldots, z_{n-2})$. Since the choice of $p$ amounts to an overall phase, we can write twisted cycles as:
\be
\Delta(\alpha) \otimes \KN := \Delta(\alpha) \otimes e^{i\pi \phi(\alpha)} \!\!\!\!\! \prod_{1 \leq i < j \leq n} (z_i - z_j)^{2\alpha'\! p_i {\cdot} p_j},
\ee
where for physical applications we choose $\phi(\alpha)$ such that the coefficient is real-valued on $\Delta(\alpha)$, i.e., equal to the Koba--Nielsen factor $\prod_{i<j} |z_i {-} z_j|^{2\alpha' p_i \cdot p_j}$ with absolute values. By considering a single fibre of $\M_{0,n}$, say in the variable $z_2 \in \Sigma_{0,n-1}$, while the remaining punctures are held fixed, we find two twisted cycles homologous to zero:
\be
\begin{aligned}
	\includegraphics[scale=1]{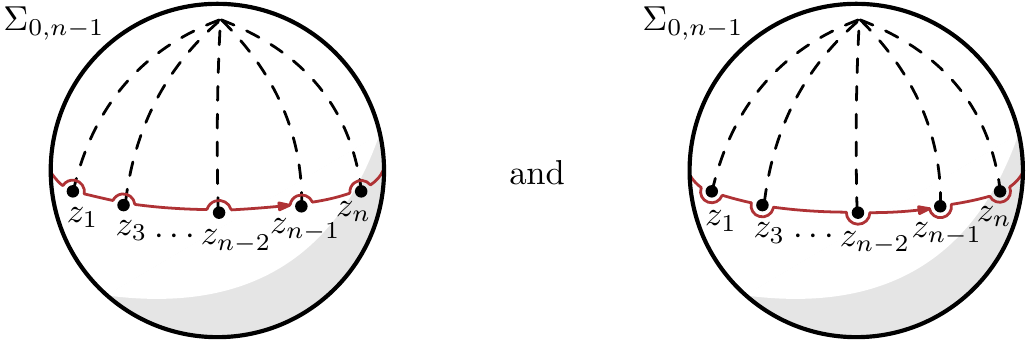}
\end{aligned}
\ee
since they can be deformed into a point near the north and south pole respectively. Dashed lines correspond to phases $e^{4\pi i \alpha' p_2 \cdot p_j}$ obtained when $z_2$ passes above $z_j$. Taking into account loading of each twisted cycle gives us two relations:
\begin{align}
0 = \bigg( \Delta(1234\cdots n) + e^{\pm 2\pi i \alpha' p_2 \cdot p_3 } \Delta(1324\cdots n ) &+ e^{\pm 2\pi i \alpha' p_2 \cdot (p_3 + p_4) } \Delta(1342\cdots n ) \nn\\
& + \ldots + e^{\pm 2\pi i \alpha' p_2 \cdot \sum_{j=3}^{n} p_j} \Delta(134\cdots n2)\bigg)  \otimes \text{KN}\label{cycle-relations}
\end{align}
involving $n{-}1$ terms each, where the signs $\pm$ are always the same in each exponent. Aomoto showed that similar identities for other fibres generate all relations between twisted cycles and the dimension of twisted homology group $H_{n-3}^{\text{lf}}(\M_{0,n},{\cal L}_\omega)$ becomes $(n{-}3)!$ \cite{aomoto1977structure,10.1093/qmath/38.4.385,aomoto1987gauss}, see also \cite{Plahte:1970wy}. The relations \eqref{cycle-relations} can be given some physical interpretation: for example their leading $\alpha'$-order is the photon decoupling identity in the field-theory limit \cite{Stieberger:2009hq,BjerrumBohr:2009rd}.

Alternatively, one can compute the dimension of the twisted homology and cohomology groups topologically, since we have
\be
\chi(\M_{0,n}) = (-1)^{n-3} \dim H^{\text{lf}}_{n-3}(\M_{0,n}, {\cal L}_\omega).
\ee
The Euler characteristic $\chi(\M_{0,n})$ can be determined by considering a fibration $\M_{0,n} \to \M_{0,n-1}$ with a fibre being a Riemann sphere with $n{-}1$ points removed, $\Sigma_{0,n-1}$. Given that the Euler characteristic of $\Sigma_{0,n-1}$ is equal to that of a sphere, $\chi(\CP^1) = 2$, minus one for each removed point, i.e., $\chi(\Sigma_{0,n-1}) = 3{-}n$, we find the recursion
\be
\chi(\M_{0,n}) = (3{-}n)\, \chi(\M_{0,n-1}),
\ee
since the Euler characteristics on a fibre bundle multiply. Together with the boundary condition $\chi(\M_{0,3})=1$, it gives us the required answer, $\chi(\M_{0,n}) = (-1)^{n-3} (n{-}3)!$.

Let us consider pairings between twisted homology and cohomology groups on $\M_{0,n}$. As a natural extension of the ones used in the previous subsections, we have the middle-dimensional integrals:
\be
\braket{ \Delta \otimes \KN | \varphi } := \int_{\Delta} \KN\, \varphi, \qquad \braket{ \overbar{\widetilde{\varphi}} | \widetilde{\Delta} \otimes \overbar{\KN} } := \int_{\widetilde{\Delta}} \overbar{\KN}\, \overbar{\widetilde{\varphi}},
\ee
where $\Delta\otimes \KN \in H_{n-3}^{\text{lf}}(\M_{0,n}, {\cal L}_\omega)$ and $\varphi \in H^{n-3}(\M_{0,n}, \nabla_\omega)$, while the tilded versions belong to the antiholomorphic counterparts. These are open-string integrals (after embedding $\M_{0,n} {\hookrightarrow} \CP^{n(n-3)/2}$ it is sometimes convenient to think of them as Mellin transforms from the extended moduli space to the kinematic space). Similarly, a bilinear between the two twisted cohomology groups is given by
\be
\braket{ \overbar{\widetilde{\varphi}} | \varphi } := \int_{\M_{0,n}} \!\!\!\!\! |\KN|^2\, \varphi \wedge \overbar{\widetilde{\varphi}},
\ee
which is a closed-string integral, involving a single-valued combination $|\KN|^2 = \allowbreak \prod_{i<j}\allowbreak (|z_i {-} z_j|^2)^{2\alpha' p_i {\cdot} p_j}$. One can make sense of these formal integrals by regularizing them in a manner similar to that discussed in Appendix~\ref{app:regularization}, namely by constructing compact cycles or forms with compact support (see also \cite{Brown:2018omk,Sen:2019jpm} for recent related approaches). Note that a new problem arising for $n{>}4$ is the existence of exceptional divisors, i.e., boundaries of $\M_{0,n}$ corresponding to a collision of three or more punctures. One can deal with this issue either using local blow-ups \cite{PMIHES_1969__36__75_0,Brown:2009qja}, or a fibration of the moduli space.

Intersection numbers of twisted cycles have a definition analogous to \eqref{twisted-intersection-pairing}, as a sum over all intersection points weighted by signs and phase factors,
\be
\Braket{ \Delta \otimes \KN_\Delta | \widetilde{\Delta} \otimes \overbar{\KN_{\widetilde{\Delta}}} } \;:= \sum_{p \in \Delta \cap \widetilde\Delta} \pm\, \frac{\KN_{\Delta}(p)\, \overbar{\KN_{\widetilde\Delta}(p)}}{|\KN(p)|^2}, 
\ee
where we indicated explicitly loading by branches of the Koba--Nielsen to avoid ambiguities. In evaluation of such intersection numbers one can exploit the fact that the calculation localizes on points $p {\in} \Delta {\cap} \widetilde{\Delta}$. In an infinitesimal neighbourhood of each $p$ one can treat $\M_{0,n}$ as a product space of $n{-}3$ one-dimensional spaces. The contribution from around $p$ becomes a product of intersection numbers on each of these one-dimensional spaces computed in the same way as in the previous subsections. This problem was addressed systematically in \cite{Mizera:2017cqs}. In particular, intersection numbers associated to $\Delta(\alpha)$ from \eqref{Delta-disk} have a simple combinatorial structure in terms of Feynman-like diagrams \cite{Mizera:2016jhj}. For instance, the diagonal entries are known to take the form:
\be
\Braket{\Delta(\alpha) \otimes \KN | \Delta(\alpha) \otimes \overbar{\KN} } = \left(\frac{i}{2}\right)^{n-3} \sum_{{\cal T}} \frac{\prod_{v \in {\cal T}} C_{(|v|-3)/2}}{\prod_{e \in {\cal T}}{ \tan(\pi \alpha' p_e^2)}},
\ee
where the sum goes over all trees $\cal T$ with odd-valent vertices, which are planar with respect to the ordering $\alpha$. For each tree the denominator involves ``propagators'' of the form $1/\tan(\pi \alpha p_e^2)$, where $p_e^\mu$ is the momentum flowing through the edge $e \in {\cal T}$, while the numerator is a product of Catalan numbers $C_{(|v|-3)/2}$ for each internal vertex $v \in {\cal T}$ of valency $|v|$ \cite{Mizera:2016jhj,Mizera:2017cqs}. Other intersection numbers have similar combinatorial expansions that can be computed to high multiplicity using the computer code attached to \cite{Mizera:2016jhj}.

Other than analytic continuation, intersection numbers find applications in basis expansion of integrals and KLT relations. For example, using the Proposition~\ref{proposition} the latter can be written as\footnote{%
	A more symmetric version of KLT relations was given in \cite{Kawai:1985xq,10.1093/qmath/38.4.385}, which in our language can be stated as the following resolution of identity:
	\be
	\I = \sum_{\alpha,\beta} e^{i\pi F(\alpha|\beta)}\, |\Delta(\beta) \otimes \overline{\text{KN}} \ra\, \la \Delta(\alpha) \otimes \text{KN} |.
	\ee
	The sum goes over all $(n{-}1)!/2$ inequivalent permutations $\alpha,\beta$ (such that, say, $\{1,n{-}1,n\}$ appear in the same cyclic order in all $\alpha,\beta$). Fixing $\alpha(n){=}\beta(n){=}n$, the phase $F(\alpha|\beta)$ is the sum of $2\alpha' p_i {\cdot} p_j$ for all pairs $\{i,j\}$ in which $i,j \in \{1,2,\ldots,n{-}1\}$ appear in different order in $\left(\alpha(1),\alpha(2),\ldots,\alpha(n{-}1)\right)$ than in $\left(\beta(1),\beta(2),\ldots,\beta(n{-}1)\right)$.}
\be
\int_{\M_{0,n}} \!\!\!\!\! |\KN|^2\; \varphi \wedge \overbar{\widetilde{\varphi}} \;=\; \sum_{\alpha, \beta} \left(\int_{\Delta(\alpha)} \!\!\! \overbar{\KN}\; \overbar{\widetilde{\varphi}} \right) \mathbf{H}_{\alpha\beta} \left(\int_{\Delta(\beta)} \!\!\! \KN\; \varphi \right),
\ee
where the sum goes over two $(n{-}3)!$ bases of twisted homologies, which we chose to consist of twisted cycles of the form \eqref{Delta-disk}, and the intersection matrix is given by
\be
\mathbf{H}_{\alpha\beta} = \la \Delta(\alpha)^\vee \otimes \KN \,|\, \Delta(\beta)^\vee \otimes \overbar{\KN} \ra, \qquad \mathbf{H}_{\alpha\beta}^{-1} = \la \Delta(\alpha) \otimes \KN \,|\, \Delta(\beta) \otimes \overbar{\KN} \ra.
\ee
Here $\{\Delta(\alpha)^\vee \otimes \KN\}$ denote twisted cycles orthonormal to $\{\Delta(\beta) \otimes \overbar{\KN}\}$ with respect to the intersection pairing and similarly for the other basis.\footnote{For instance, a natural choice of such orthonormal bases is given by paths of steepest descent ${\cal J}_\alpha$ and ascent ${\cal K}_\beta$. Using the \emph{positive kinematics} $\{\,p_i {\cdot} p_j >0 \;|\; i,j=1,2,\ldots,n{-}1,\; \{i,j\}\neq \{1,n{-}1\}\}$ and the $\SL(2,\C)$ fixing $(z_1, z_{n-1}, z_n)=(0,1,\infty)$, the paths of steepest descent ${\cal J}_\alpha$ coincide with $\Delta(1,\alpha,n{-}1,n)$ for $(n{-}3)!$ permutations $\alpha$ \cite{Cachazo:2016ror}. Loading each ${\cal J}_\alpha$ and ${\cal K}_\beta$ with the same branch of the Koba--Nielsen factor we have $\la {\cal J}_{\alpha} \otimes \KN | {\cal K}_\beta \otimes \overbar{\KN} \ra = \delta_{\alpha\beta}$.} The above KLT formula is nothing but a generalization of the Riemann bilinear relations to the twisted case.

Finally, let us comment on a correspondence between twisted forms and twisted cycles. The real section of the moduli space, $\M_{0,n}(\mathbb{R})$, is decomposed into $(n{-}1)!/2$ chambers, equal to $\Delta(\alpha)$ for all cyclically-inequivalent permutations $\alpha$. Their compactification can be described combinatorially as a polytope called the \emph{associahedron} \cite{10.2307/1993608,10.2307/1993609}. One can construct a differential form with logarithmic singularities along the boundaries of $\Delta(\alpha)$, called the \emph{Parke--Taylor form}:
\be
\text{PT}(\alpha) := \frac{d\mu_n}{\prod_{i=1}^{n} \left( z_{\alpha(i)} - z_{\alpha(i+1)} \right)},
\ee
where $d\mu_n := (z_1 {-} z_{n-1})(z_{n-1} {-} z_n)(z_1 {-} z_n) \bigwedge_{i=2}^{n-2} dz_i$ is a measure making the resulting form $\SL(2,\C)$-invariant. Such forms can be treated as elements of $H^{n-3}(\M_{0,n}, \nabla_\omega)$. They satisfy cohomology relations similar to \eqref{cycle-relations}:
\be
0 = p_2 {\cdot} p_3\, \PT(1324\cdots n ) + p_2 {\cdot} (p_3 {+} p_4)\, \PT(1342\cdots n ) 
+ \ldots + p_2 {\cdot} \left(\textstyle\sum_{j=3}^{n} p_j\right)\, \PT(134\cdots n2),
\ee
since for massless kinematics the right-hand side is proportional to $\nabla_\omega \PT(134\cdots n)$ and thus cohomologous to zero.

The chambers $\Delta(\alpha)$ and twisted forms $\PT(\alpha)$ are natural counterparts, in the sense that the leading $\alpha'$-order of all their possible pairings coincides (up to an overall normalization) with the intersection number of twisted forms $\braket{ \PT(\alpha) | \PT(\beta) }_\omega$ studied in the main text of this paper, namely:
\begin{align}
\lim_{\alpha' \to 0} (\alpha')^{n-3} \Braket{ \Delta(\alpha) \otimes \KN | \PT(\beta)}
&= \lim_{\alpha' \to 0} (-2\pi i \alpha')^{n-3} \Braket{ \Delta(\alpha) \otimes \KN | \Delta(\beta) \otimes \overbar{\KN}} \nn\\
&= \lim_{\alpha' \to 0} \left(\frac{-\alpha'}{2\pi i}\right)^{n-3}\!\! \Braket{ \PT(\alpha) | \PT(\beta)}\\
&= \Braket{\PT(\alpha) | \PT(\beta)}_\omega.\nn
\end{align}
For more discussion on homological aspects of string-theory integrals at genus zero, see, e.g., \cite{Mizera:2017cqs,Brown:2018omk}.

\pagebreak
\bibliographystyle{JHEP}
\bibliography{references}

\end{document}